\documentclass{article}
\usepackage{dcolumn}
\usepackage{bm}

\usepackage{verbatim}       

\usepackage[colorlinks=true, linkcolor=blue, citecolor=blue, urlcolor=blue, linktoc=page]{hyperref}
\usepackage[hyphenbreaks]{breakurl}

\usepackage[dvips]{graphicx}
\usepackage{amsthm}
\usepackage{amssymb}
\usepackage{amstext}
\usepackage{psfrag}
\usepackage{amsmath}
\usepackage{amsfonts}
\usepackage{mathrsfs}
\usepackage{cite}

\newcommand{\p}{\partial}

\newcommand{\dd}{{\rm d}}
\newcommand{\bd}{\begin{definition}}                
\newcommand{\ed}{\end{definition}}                  
\newcommand{\bc}{\begin{corollary}}                 
\newcommand{\ec}{\end{corollary}}                   
\newcommand{\bl}{\begin{lemma}}                     
\newcommand{\el}{\end{lemma}}                       
\newcommand{\bp}{\begin{proposition}}            
\newcommand{\ep}{\end{proposition}}                
\newcommand{\bere}{\begin{remark}}                  
\newcommand{\ere}{\end{remark}}                     

\newcommand{\bt}{\begin{theorem}}
\newcommand{\et}{\end{theorem}}

\newcommand{\be}{\begin{equation}}
\newcommand{\ee}{\end{equation}}

\newcommand{\bit}{\begin{itemize}}
\newcommand{\eit}{\end{itemize}}
\newtheorem{theorem}{Theorem}[section]
\newtheorem{corollary}[theorem]{Corollary}
\newtheorem{lemma}[theorem]{Lemma}
\newtheorem{proposition}[theorem]{Proposition}
\theoremstyle{definition}
\newtheorem{definition}[theorem]{Definition}
\theoremstyle{remark}
\newtheorem{remark}[theorem]{Remark}
\newtheorem{example}[theorem]{Example}
\newtheorem*{remark*}{Remark}

\begin{document}
%


\title{Causality of spacetimes admitting a parallel null vector and weak KAM theory\footnote{Revised and extended version of arXiv:1211.2685.}}

\author{E. Minguzzi\footnote{Dipartimento di Matematica, Universit\`a degli Studi di Pisa, Largo B. Pontecorvo 5,
I-56127 Pisa, Italy. E-mail:
ettore.minguzzi@unipi.it}}


\date{}
\maketitle

\noindent The spacetimes admitting a parallel null vector connect relativistic and non-relativistic physics. We prove that their Lorentzian distance is determined by the least action of a mechanical system on a quotient classical spacetime, and that the timelike eikonal equation corresponds to the Hamilton-Jacobi equation. The causal ladder is characterized through the least action: strong causality coincides with stable causality and with the lower semi-continuity of the least action, and global hyperbolicity with the coercivity of the action functional, so that Tonelli's theorem becomes the implication from global hyperbolicity to causal simplicity. In the weak KAM correspondence, the Lax-Oleinik evolution is the causal propagation of an achronal graph, and the weak KAM solutions are horizons whose generators project on the calibrated curves. For compact configuration space and arbitrary time dependence we obtain global viscosity solutions of the Hamilton-Jacobi and timelike eikonal equations, and via the area theorem we prove that under the null convergence condition they are as smooth as the Lagrangian, the dynamics being a geodesic flow seen from a moving frame. For time-periodic systems Ma\~n\'e's critical value is the threshold at which a quotient spacetime passes from totally vicious to globally hyperbolic, and the periodic weak KAM solutions are compact horizons, which under the null convergence condition are degenerate Killing horizons. The Newtonian $N$-body problem in three or more dimensions satisfies the null convergence condition, which excludes entire fixed-time minimizers. Regarding stationary spacetimes as Maupertuis reductions, we show that they are globally hyperbolic if and only if their Sagnac ratio is smaller than one, and that the projections of the horizons give synchronizations of the stationary clocks which minimize the maximal anisotropy of the one-way speed of light.



\newpage

\tableofcontents


\section{Introduction}

Let ${Q}$ be a $d$-dimensional manifold (the space) endowed with a
positive definite metric $a_t$, a 1-form field $b_t$ and a potential
function $V(t,q)$, all possibly time dependent and of class $C^r$, $r \ge 2$. On
the classical spacetime $E=T \times {Q} $, $T$ connected interval of
the real line, let $t$ be the time coordinate and let
$e_0=(t_0,q_0)$ and $e_1=(t_1,q_1)$ be events, the latter in the
future of the former, i.e.\ $t_1>t_0$. Consider the classical
mechanics action functional
\begin{equation} \label{cxs}
\mathcal{S}_{e_0,e_1}[q]=\int_{t_0}^{t_1} [\frac{1}{2}\,
a_t(\dot{q},\dot{q}) +b_t(\dot{q})-V(t,q)] \dd t,
\end{equation}
on the space $C^1_{e_0,e_1}$ of $C^1$ curves $q:[t_0,t_1] \to {Q}$
with fixed endpoints $q(t_0)=q_0$, $q(t_1)=q_1$. The $C^1$
stationary points are smoother than the Lagrangian (namely
$C^{r+1}$, see \cite[Theor. 1.2.4]{jost98}), and, by Hamilton's
principle, they solve the Euler-Lagrange equations (see Eq.
(\ref{ele})) of the Lagrangian
\begin{equation} \label{clas}
L(t,q, v)=\frac{1}{2}\, a_t(v,v) +b_t(v)-V(t,q).
\end{equation}
Historically, this has proved to be one of the most important
variational problems, because the systems of particles
subject to (possibly time dependent) holonomic constraints move
according to Hamilton's principle with a Lagrangian of the form
(\ref{clas}) (see \cite{goldstein65}).

Brinkmann \cite{brinkmann25} (see also
\cite{eisenhart38,walker50,zakharov73}) considered the most general
$(d+2)$-dimensional spacetime $(M,g)$ admitting a covariantly constant
lightlike vector field $n$. He proved that it is locally isometric
to the spacetime $M:=T\times {Q} \times \mathbb{R}$ with coordinates
$(t,q,y)$, metric
\begin{equation} \label{eis}
g=a_t \!-\!\dd t \otimes (\dd y-\!b_t) -\!(\dd y-\!b_t) \otimes \dd
t-2V \dd t ^2 ,
\end{equation}
and time orientation given by the global timelike vector
$W=-[V-\frac{1}{2}]\p_y+\p_t$, $g(W,W)=-1$. The covariantly constant
future directed lightlike vector is $n=\p_y$. Indeed, as the metric
does not depend on $y$, $n$ is a Killing vector field, $L_n g=0$, and
since $(L_n g)(X,Y)=g(\nabla_X n,Y)+g(X,\nabla_Y n)$ the bilinear form
$(X,Y)\mapsto g(\nabla_X n,Y)$ is antisymmetric. On the other hand,
the 1-form $n^\flat=g(\cdot,n)=-\dd t$ is closed and, the
Levi-Civita connection being torsion free,
$\dd n^\flat(X,Y)=g(\nabla_X n,Y)-g(\nabla_Y n,X)$, thus the same
bilinear form is symmetric. Therefore it vanishes, namely $\nabla n=0$.

Spacetimes of the form $T\times {Q} \times \mathbb{R}$, endowed with
the metric (\ref{eis}),  and generically denoted in the following as
$(M,g)$,  will be referred to as {\em generalized gravitational
waves}, Eisenhart spacetimes or Brinkmann spacetimes. It is
understood that these spacetimes need not solve the Einstein
equations, nor need they be four-dimensional. The terminology
just recalls that the gravitational wave solutions
of the Einstein equations are special cases of the spacetimes
considered here.



The fields $a_t$, $b_t$ and $V$ entering the spacetime metric are
the same that define the Lagrangian (\ref{clas}). Eisenhart
\cite{eisenhart29}
proved that the spacelike geodesics project to the stationary points
of the associated classical Lagrangian problem. Similar connections
were rediscovered from a different perspective by authors working on
Newton-Cartan theory, namely on the geometrized formulation of
Newtonian gravity \cite{cartan23,cartan24,friedrichs27,havas64,dombrowski64,trautman65,kunzle72,kunzle76,ehlers81,malament12},
and on the Bargmann structures \cite{duval85,duval91}. In  \cite{minguzzi06d} we proved
that analogous results hold in the timelike and lightlike cases. The
lightlike case is the most convenient, as it allows us to use methods
from causality theory to attack problems of classical Lagrangian
systems and, conversely, to use results on classical
Lagrangian problems to infer the causal properties of the spacetimes
of Brinkmann type. In \cite{minguzzi06d} we suggested that a
dictionary could be built between the mathematics of Lagrangian
mechanical systems, and the causality of generalized gravitational
wave spacetimes. This work is meant as a significant step in
this direction.

As we have just mentioned, Brinkmann and other authors proved that every
spacetime admitting a parallel null vector $n$ is locally isometric
to a generalized gravitational wave. This result extends globally
under two assumptions. Assumption (a) requires that $(M,g)$ be a principal
$(\mathbb{R},+)$-bundle, $\pi: M\to E$, the group action being given
by the flow $\varphi_s$ generated by $n$, see \cite{minguzzi06d}. If
$M$ is strongly causal, the existence of the quotient manifold and of the
smooth projection follows from standard results of
manifold theory \cite[Theor. 9.16]{lee03}. Let $t:M\to \mathbb{R}$ be such that $\dd t=-g(\cdot, n)$; as
$\p_n t=- g(n,n)=0$, $t$ is the pullback of a function on $E$, which
we denote in the same way. Assumption (b) asks for the existence of
a complete vector field $v:E\to TE$ such that $\dd t[v]=1$
(a Newtonian flow or frame), and $Q$ is then defined as the
quotient of $E$ under this flow.

The proof of the global isometry \cite{minguzzi06d}, obtained
through the  detailed construction of the coordinate system
$(t,q,y)$, gives a lot of insight into the invariant properties of
the mechanical system whose Lagrangian is given by Eq. (\ref{clas}).
To change the flow means to change the
``body frame'' with respect to which the natural motion described by
the Lagrangian is observed. Such changes imply corresponding changes
in the Lagrangian itself but not in the dynamics (Sect. \ref{gau}). For simplicity, we
shall assume that a choice of Newtonian flow has been made, and
hence that the space $Q$ has been defined from the start.

With respect to other works in gravitational waves, e.g.
\cite{ehrlich92,ehrlich92b,ehrlich93,flores03,flores06} \cite[Chap.\
13]{beem96}, here we assume neither $\p_t a_t=0$ nor $b_t=0$.
Nevertheless, in the time independent case
it is often natural to add the condition $b_t=0$; indeed, for a
mechanical system subject to time independent constraints
$b_t=0$ (see \cite{goldstein65}). These mechanical systems are
sometimes called {\em natural} \cite{greenwood77}.

There seems to be some confusion in the literature concerning the
possible simplifications of the metric. Indeed, some simplifications
might hold locally while failing globally. This is the case with the
condition $b_t=0$ as well as with the issue of rewriting the metric
in Rosen coordinates, an interesting problem to which we shall later
return (Sect.\ \ref{ros}).

The Eisenhart metric takes its simplest and most symmetric form in
the case of a free particle in Euclidean space: $Q=\mathbb{R}^d$,
$a_{t, bc} = \delta_{bc}$, $b_c = 0$, for $b,c=1,\ldots, d$, and $V
= const$. Remarkably, in this case the Eisenhart metric becomes the
Minkowski metric.

%

The spacetimes admitting a covariantly constant null vector are
important because, on the one hand, they include the gravitational
plane waves as the most physically interesting subfamily and,
on the other hand, they provide exact classical
backgrounds for string theory (vanishing of $\alpha'$ corrections).
After some pioneering works
\cite{penrose65,ehlers62,ehrlich92}, the
causal aspects of these spacetimes have received considerable
attention \cite{hubeny02b,hubeny03,flores03,flores06}.
The causal behavior of the spacetime determines its causal
boundary, whose knowledge is fundamental in some applications to
theoretical physics, such as the AdS/CFT correspondence.

Among the questions on the causal behavior of a
spacetime, that of whether it is distinguishing
is particularly important since, if distinction fails,
the boundary construction of Geroch, Kronheimer and Penrose cannot be
applied. This problem will be reduced to the verification of a lower
semi-continuity property of the mechanical least action
\cite{minguzzi06d} (Hamilton's principal function) $S: E \times E
\to [-\infty,+\infty]$ given by
\begin{align*}
S(e_0,e_1)&=\inf_{q \in C^1_{e_0,e_1}}
\mathcal{S}_{e_0,e_1}[q],\qquad  \textrm{ for } \ t_0<t_1,& \\
S(e_0,e_1)&= 0, \qquad \textrm{ for } \ t_0=t_1 \textrm{ and } q_0=
q_1,\\
S(e_0,e_1)&=+\infty, \qquad \textrm{ elsewhere}.
\end{align*}

Actually, we shall establish a formula (Eq.\ (\ref{cft})) which
connects the function $S$ with the Lorentzian distance of the
spacetime. This result provides the clearest evidence of the
useful interplay between classical Lagrangian problems and
spacetimes of Brinkmann type. More importantly, since most of the
causality properties of a spacetime can be expressed in terms of the
Lorentzian distance (see \cite{beem96} and \cite{minguzzi06f}), this
result suggests reformulating them as conditions on the least
action $S$ alone, and then inferring them from the
behavior of the metric coefficients $a_t$, $b_t$ and $V$.
%

\begin{quote}
{\em History of this work.} The main ideas of this work were presented with the talk ``A connection between Lorentzian distance and mechanical least action'' at the meeting ``Non-commutative structures and non-relativistic (super)symmetries'', Tours, June 21--25, 2010 \cite{minguzzi10e}, and with the talk ``Causality of gravitational waves and weak KAM theory'' at the conference ``Weak KAM Theory in Italy'', Cortona, September 11--17, 2011 \cite{minguzzi11cortona}. There we introduced the connection between the Lorentzian distance and the mechanical least action (Section \ref{s:dist}), that between the eikonal equation for the Lorentzian distance and the Hamilton-Jacobi equation (Subsection \ref{pqz}), and we linked the Lax-Oleinik semigroup with the boundary of the causal future of the initial condition on the Eisenhart spacetime (Theorem \ref{lo1}). We also suggested that regularity results could be imported from those on horizons in general relativity.

One key idea was that the smoothness of the weak KAM solutions in the compact periodic case could be obtained by regarding their graphs as compact horizons of the Eisenhart spacetime. To this end we first had to prove that compact horizons are smooth under the null convergence (or null energy) condition, a condition which can be formulated in Lagrangian terms (Prop.\ \ref{ncc}). This was an open problem in mathematical relativity at the time, which we solved in 2014 \cite{minguzzi14b,minguzzi14d} (see also the independent work by Larsson \cite{larsson14}).

After these meetings, in November 2012, we posted the preprint arXiv:\allowbreak 1211.2685. That version contained the material which is now in Sections 1--4, in Subsections \ref{ss:imp}, \ref{ss:coe} and \ref{ss:jac}, and in Section \ref{s:exa}, namely the results which are mostly on causality theory. We posted that paper somewhat prematurely to answer some questions by colleagues on the causality of gravitational waves that we happened to have solved in the manuscript. The title already mentioned weak KAM theory but there was little of the connection in it; in fact, in a front-page footnote we explained that the material on the weak KAM connection would be added in a forthcoming version.

The skeleton of the other sections, on homogeneous pp-waves, on the Maupertuis spacetime and on weak KAM theory, was ultimately left in the drawer, as we planned to complete it before submitting the work for publication. Unfortunately, we put aside this work to develop other projects. We have finally been able to return to it, so this paper represents the completion of the original project (though not of its consequences, which have largely yet to be explored). In the present version the 2012 material has been revised, the skeleton has been completed (the remaining parts of Sections 5 and 6, and Sections 8 and 9), and Sections 10, 11 and 12, which contain the main new results, have been added. Section 12 answers a question raised in an unpublished 2011 manuscript on the synchronization of clocks, now revised and posted as the companion paper \cite{minguzzisync}.
\end{quote}

The interest in the geometry of the spacetimes admitting a parallel null vector has not faded since the first version of this work was posted. Without any claim of completeness, we mention the studies on the geodesic completeness of pp-waves and Brinkmann spacetimes and on the Ehlers-Kundt conjecture \cite{candela13,samann16,leistner16,costaesilva18,flores20,mehidi22,frauenberger26,cairo26b}; the works on the global structure, splitting and classification of these spacetimes \cite{costaesilva17,globke16,roche23}, among which \cite{costaesilva17} makes use of results of the first version of this paper; the developments on the Eisenhart lift, on the Bargmann structures and on the non-relativistic symmetries \cite{bekaert13,duval14,cariglia14,cariglia15,cariglia16,zhang17,aazami24}; and the works in which ideas of Aubry-Mather and weak KAM theory are imported into Lorentzian geometry \cite{cui14,jin18,bernard18,suhr19,zhu25} (see \cite{minguzzi19} for a recent account of causality theory). Very recently, Cairo \cite{cairo26} proved that the gravitational pp-waves which are not plane waves are not distinguishing, and deduced a null version of Anderson's rigidity theorem conjectured by Costa e Silva and Flores. Her argument rests on an elementary criterion, proved directly in her paper, which expresses the existence of a timelike curve over a given spatial path through the action of that path (cf.\ Lemma \ref{pcf}), and she recalls the translation of the causal ladder into properties of the least action given in the first version of this work.

A comment is in order concerning the Lorentzian eikonal equation. The existence of global viscosity solutions of the timelike eikonal equation on general globally hyperbolic spacetimes was claimed in \cite{zhu25}. However, Galloway, McCann and Ohanyan \cite{galloway26} have recently constructed examples which contradict a key lemma of that work and hence, as they observe, the results which depend on it. Therefore, the existence results for the eikonal equation obtained in this work for the spacetimes admitting a parallel null vector (Corollary \ref{eikc}) are not special cases of a general theorem. They follow from the correspondence with the Hamilton-Jacobi equation and, ultimately, from the classical Lax-Oleinik formula, which in these spacetimes describes the boundary of a causal future.

Let us outline the content of the first part of the work. In Subsection \ref{ndx} we show that the Legendre transform on the classical base corresponds to the musical isomorphism on the light cone, and that the Young-Fenchel inequality is nothing but a well known inequality of Lorentzian geometry. In Section \ref{s:dist} we prove the formula (\ref{cft}) which relates the Lorentzian distance of $(M,g)$ with the least action $S$ (Theorem \ref{vfg}), and we show that the timelike eikonal equation satisfied by the Lorentzian distance is equivalent to the Hamilton-Jacobi equation for the least action (Subsection \ref{pqz}).

In Section \ref{s:c1} we study the $C^1$ solutions of the Hamilton-Jacobi equation through the null hypersurfaces of $(M,g)$, we present a strategy to simplify the dynamics by passing to a reference frame co-moving with the E.-L.\ flow (Subsection \ref{gau}), and we show that the problem of recasting the metric in a global Rosen form is equivalent to that of finding global solutions to the Hamilton-Jacobi equation having complete characteristics (Subsection \ref{ros}).

In Section \ref{s:cau} we express the causality properties through the least action: strong causality coincides with stable causality and is equivalent to the lower semi-continuity of $S$ (Theorems \ref{bup} and \ref{psf2}), distinction is equivalent to a one-sided lower semi-continuity property (Theorem \ref{psf}), and global hyperbolicity is equivalent to the coercivity of the action functional (Theorem \ref{glo}). Tonelli's classical theorem in the calculus of variations then corresponds to the well known result that global hyperbolicity implies causal simplicity. Sections \ref{s:com}, \ref{mau} and \ref{s:exa} are devoted to completeness, to the spacelike dimensional reduction to the Maupertuis spacetime in the time independent case, and to examples which show that the levels of the causal ladder are distinct.

Let us outline the content of the second part of the work, which develops the dictionary in the direction of weak KAM theory \cite{fathi08}. In Section \ref{s:future_sets} we show that the closed future sets of $(M,g)$ are the epigraphs of the functions on $E$ which are dominated by the Lagrangian, and that the Lax-Oleinik evolution of an initial condition is the boundary of the causal future of its graph (Theorem \ref{lo1}). The classical properties of the Lax-Oleinik semigroup are then recovered, with no compactness or completeness assumption, from the properties of the achronal boundaries: the minimizers are the projections of the generators, the solution is not differentiable precisely at the points reached by more than one generator, and its semi-concavity is the semi-convexity of horizons (Theorem \ref{lo2}).

In Section \ref{s:subsolutions} we show that the dominated functions are viscosity subsolutions, that the weak KAM solutions correspond to the horizons generated by past inextendible lightlike geodesics, and that the Cauchy problem is solved by the causal future of the initial condition (Theorem \ref{cau}). If $Q$ is compact, with no assumption on the time dependence of the Lagrangian, we obtain global weak KAM solutions (Theorem \ref{glob}), and from them time functions on $(M,g)$ which are global viscosity solutions of the timelike eikonal equation (Corollary \ref{eikc}).

In Section \ref{s:per} we consider the Lagrangians which are periodic in time, with $Q$ compact. The spacetime $(M,g)$ admits the isometries $(t,q,y)\mapsto (t+P,q,y-cP)$, where $P$ is the period, which give the quotient spacetimes $M_c$. Through the coordinate $\tilde y=y+ct$, all the spacetimes $M_c$ can be realized on the same manifold $(\mathbb{R}/P\mathbb{Z})\times Q\times\mathbb{R}$, where the metric of $M_c$ is that of the Eisenhart spacetime of $L+c$, so that on this manifold the light cones get narrower as the constant $c$ increases (Eq.\ (\ref{gc})).

We prove that there is a unique value of $c$ for which there are weak KAM solutions compatible with the quotient, that it coincides with Ma\~n\'e's critical value $c(L)$ (Theorem \ref{crit}; the equivalent characterizations of $c(L)$ for autonomous systems are collected in Prop.\ \ref{critchar}), and that it marks a sharp transition in the causal ladder: $M_c$ is totally vicious for $c<c(L)$, globally hyperbolic for $c>c(L)$, and chronological but not strongly causal for $c=c(L)$ (Theorem \ref{cau2}). The chronological relation of $M_c$ is described by Ma\~n\'e's potential, and at the critical value strong causality fails exactly on the fibers of the projected Aubry set (Prop.\ \ref{aub}). The periodic weak KAM solutions are compact horizons of the critical spacetime, and they become compact Cauchy horizons if the light cones are suitably narrowed in their past (Theorem \ref{hor}).

In Section \ref{s:ncc} we calculate the Ricci tensor of $(M,g)$, which turns out to be the pullback of a tensor on the classical spacetime $E$, and we show that the null convergence condition is a condition on the mechanical system: it is equivalent to the non-negativity of the Ricci curvature of the configuration space together with a discriminant inequality involving the gyroscopic and potential terms, Eq.\ (\ref{nec}), and for a particle in a Newtonian gravitational field it is the non-negativity of the mass density (Prop.\ \ref{ncc}). For these spacetimes the null convergence condition coincides with the timelike convergence condition and with the dominant energy condition.

The Raychaudhuri equation becomes an equation for the Laplacian of the solutions of the Hamilton-Jacobi equation. Using the area theorem and the regularity results for horizons of \cite{minguzzi14d}, we prove that, on a compact configuration space and for arbitrary time dependence, under the null convergence condition the global weak KAM solutions are as differentiable as the Lagrangian and unique up to additive constants, and the system is the geodesic flow of a Riemannian manifold with non-negative Ricci curvature observed from a moving frame (Theorem \ref{smoa}). The proof applies the area theorem to the horizons of the system and of its time reversal: the volume of the configuration space can only decrease along the former and only increase along the latter, hence it is constant, and the equality case of the area theorem gives the regularity.

For time-periodic systems we obtain in addition that the Aubry set is the whole classical spacetime and that after one period the frame of calibrated curves is displaced by an isometry (Theorem \ref{smo}). As far as we know this is a new type of regularity result in weak KAM theory; Bochner's theorem on harmonic 1-forms is recovered as a special case (for a related Laplacian estimate on closed manifolds with non-negative Ricci curvature see \cite{chengwei26}).

Mechanically, the calibrated curves form a rigid reference frame which is inertial: relative to it every motion is a constant-speed geodesic of a fixed Riemannian manifold, and the Lax-Oleinik semigroup becomes the Hopf-Lax semigroup of that manifold (Remark \ref{mec}). The known properties of these geodesic flows then translate into properties of the original system: the speed relative to the frame is a first integral, the periodic solutions of a time-periodic system correspond to the invariant geodesics of an isometry in the sense of Grove, and the topological entropy of the period map is that of the geodesic flow (Corollary \ref{dyn}). The compact horizons of the periodic case are then degenerate Killing horizons generated by a parallel null vector field, and by narrowing the light cones in their past we obtain compact Cauchy horizons of this type in spacetimes which satisfy the null convergence condition (Remark \ref{sgr} and Example \ref{kch}).

The Laplacian comparison holds also in non-compact settings. For instance, the Eisenhart spacetime of the Newtonian $N$-body problem in three-dimensional space, with the collisions removed, is a vacuum solution of the Einstein equations (Subsection \ref{yak}). Hence the null convergence condition holds, and the Laplacian comparison shows that the global viscosity solutions of the Hamilton-Jacobi equation of Maderna and Venturelli \cite{maderna12,maderna20} are superharmonic with respect to the mass metric (Prop.\ \ref{nbody}). The same two facts, the null convergence condition and the homogeneity of the Newtonian potential, show by a second variation argument that in dimension $d\ge 3$ no motion defined on $\mathbb{R}$ minimizes the fixed-time action between any two of its points, so that the spacetime has no complete lightlike lines transverse to $n$; the same holds for the pair potentials $r^{-\alpha}$, $0<\alpha\le d-2$, which include the Newtonian potential of $\mathbb{R}^d$ (Prop.\ \ref{nline}). This contains, for $d\ge 3$, the theorem of da Luz and Maderna on the absence of entire free time minimizers \cite{daluz14}, recently reproved through a splitting theorem by Assouline and Mazzucchelli \cite{assouline26}.

In Section \ref{s:sync} we apply these results to the stationary spacetimes $-(\dd\tau-b)^2+a$ over a compact Riemannian manifold $(Q,a)$, which are the Maupertuis reductions of the spacetimes $(M,g)$ of the Lagrangians $\frac{1}{2}a(v,v)+b(v)$. We prove that Ma\~n\'e's critical value is the threshold of distinction for the Maupertuis spacetimes (Theorem \ref{mcau}), and that a stationary spacetime is totally vicious, chronological but non-distinguishing, or globally hyperbolic according to whether its Sagnac ratio, i.e.\ the maximal ratio between the circulation of $b$ along a closed curve and its length, is larger than, equal to, or smaller than one (Corollary \ref{scau}).

This classification can also be derived from the work of Garfinkle and Harris on the fundamental cocycle \cite{garfinkle97,harris15}, whose weight coincides with the Sagnac ratio; what the present approach adds is the identification of the Sagnac ratio with Ma\~n\'e's critical value and its characterization as the minimal norm of the drift form over the synchronizations (Prop.\ \ref{sagn}).

The projections of the horizons of $(M,g)$ provide synchronizations of the clocks of the stationary observers, the eikonal synchronizations, which minimize the maximal anisotropy of the one-way speed of light, with the same anisotropy at almost every point, and which are time functions with Cauchy slices whenever the spacetime is distinguishing; we call this result the synchronization theorem (Theorem \ref{sync}). Under the null convergence condition the eikonal synchronization is smooth, unique up to constants, and coincides with the harmonic synchronization, which minimizes the mean square anisotropy and solves the wave equation (Theorem \ref{syncs}), while in general the harmonic synchronization need not be a time function (Example \ref{counter}).

If $\beta<1$, the slices of an eikonal synchronization are orthogonal, in a weak sense which allows for caustics, to a congruence of observers freely falling in the metric normalized by the Killing field, and under the null convergence condition they are totally geodesic and the stationary spacetime is ultrastatic; we compare these facts with the theory of maximal hypersurfaces in Subsection \ref{ss:geo}.

Finally, the stationary spacetime can also be represented as a Maupertuis spacetime with a potential, the representations corresponding to the conformal factors, which the clocks can measure through the redshift (Lemma \ref{rep}). Each representation gives a synchronization theorem in which the slices have constant lapse with respect to the corresponding metric, the minimal lapse being determined by a redshift-weighted Sagnac condition (Theorem \ref{synchi}). For the physical metric and $\beta<1$ we obtain stationary synchronous reference systems, whose slices are orthogonal to observers freely falling in the physical spacetime and whose time is their proper time, so that the clocks are syntonized and synchronized coherently. The null convergence condition of the corresponding Eisenhart spacetime, namely the Galilean convergence condition of the corresponding classical spacetime, is a curvature condition of a new type, distinct from the energy conditions of the stationary spacetime, under which these synchronizations are smooth. We characterize it (Theorem \ref{curvchi}): on a compact space it holds precisely when the drift form and the potential are, up to a gauge transformation, the field of inertial forces of a rigid motion along a Killing vector field of a metric with non-negative Ricci curvature, a purely Riemannian rigidity result, for which we also give an independent elliptic proof (Theorem \ref{riemrig}). In this proof one starts from the harmonic synchronization, obtained from a Poisson equation, and the curvature condition needs to be checked only along its anisotropy field to conclude that it is eikonal and that this field is a Killing vector field.

\medskip

{\em Overview of the results on the ``relativistic causality''--``weak KAM'' connection.} For the reader's convenience we collect the main results on this connection and the way they depend on each other. The three spacetime constructions used in the work, and their time coordinates, are summarized in Table \ref{tab:three}. Readers coming from weak KAM theory, who may be less familiar with causality theory, can use the dictionary of Tables \ref{dic} and \ref{dic2} at the end of the work to translate the causal notions into mechanical ones; Remark \ref{mec}, which reads Theorem \ref{smo} in purely mechanical terms, can serve as an entry point.
\begin{itemize}
\item Theorem \ref{lo1}: the Lax-Oleinik evolution of an initial condition is the boundary of the causal future of its graph in $(M,g)$; it rests on the identity between Lorentzian distance and least action of Section 2.
\item Theorem \ref{glob}: on a compact $Q$ there are global weak KAM solutions, whose graphs are horizons of $(M,g)$, for arbitrary time dependence.
\item Theorems \ref{crit}, \ref{cau2} and \ref{hor}: for time-periodic Lagrangians Ma\~n\'e's critical value is the unique level of periodic weak KAM solutions and the threshold of the causal ladder of the quotients $M_c$, whose periodic weak KAM solutions give compact horizons. These results use Theorems \ref{lo1} and \ref{glob}.
\item Theorems \ref{smoa} and \ref{smo}: on a compact configuration space, under the null convergence condition the global weak KAM solutions are smooth and unique up to constants, and the spacetime splits; in the periodic case the Aubry set is the whole classical spacetime. They use Theorems \ref{glob} and \ref{hor} and the area and regularity theorems for horizons of \cite{minguzzi14d}.
\item Corollary \ref{dyn}, Remark \ref{sgr} and Example \ref{kch}: under the null convergence condition the properties of the geodesic flows of manifolds with non-negative Ricci curvature become properties of the mechanical system, and the compact horizons of the periodic case are degenerate Killing horizons. They use Theorems \ref{rig}, \ref{smoa} and \ref{smo}.
\item Theorem \ref{mcau}, Corollary \ref{scau} and Theorems \ref{sync}, \ref{syncs}: the autonomous case and the Maupertuis reduction of Section \ref{mau} give the causal classification of stationary spacetimes by the Sagnac ratio, the synchronization theorem, and, through Theorem \ref{smo}, its smooth version under the curvature condition.
\item Theorems \ref{synchi} and \ref{curvchi}: the Maupertuis reduction with a potential gives the synchronization theorem for every conformal representative of a stationary spacetime, in particular for the physical metric, and its smooth version under the Galilean convergence condition, namely the null convergence condition of the corresponding Eisenhart spacetime. They use Theorems \ref{mcau}, \ref{smo} and \ref{rig}. Theorem \ref{riemrig} gives an independent elliptic proof of the rigidity statements of Theorems \ref{syncs} and \ref{curvchi}, which uses neither the spacetime nor the area theorem.
\end{itemize}

\begin{table}[ht]
\centering
\footnotesize
\begin{tabular}{>{\raggedright\arraybackslash}p{2.1cm}>{\raggedright\arraybackslash}p{4.3cm}>{\raggedright\arraybackslash}p{2.8cm}>{\raggedright\arraybackslash}p{1.0cm}}
\hline
Spacetime & Manifold and metric & Time coordinates & Sect.\\
\hline
Eisenhart spacetime $(M,g)$ of $L$ & $T\times Q\times \mathbb{R}\ni (t,q,y)$, metric (\ref{eis}) & $t$ (semi-time function), $y$ (affine parameter of $n=\p_y$) & 1--9, 11\\[4pt]
Periodic quotient $M_c$ & $M/\langle\sigma_c\rangle\simeq (\mathbb{R}/P\mathbb{Z})\times Q\times\mathbb{R}$, $\sigma_c(t,q,y)=(t+P,q,y-cP)$, metric (\ref{gc}) & $\tilde y=y+ct$, invariant under $\sigma_c$ & 10, 11\\[4pt]
Stationary (Maupertuis) spacetime $\mathcal{M}_e$, $(\mathcal{M},\hat g)$ & $\mathbb{R}\times Q\ni(\tilde y,q)$, reduction of $(M,g)$ along $\p_t$ for autonomous $L$ at energy $e$, metric (\ref{gmo}); for $L_b$ and $e=\frac{1}{2}$, $\hat g=-(\dd\tau-b)^2+a$; for $L_\chi$, Eq.\ (\ref{lchi}), and $e=0$, again $\hat g$, conformal to $g_\chi=\chi\hat g$ & $\tilde y=\tau$ (Killing time), synchronizations $\tau-\alpha$ & 6, 12\\
\hline
\end{tabular}
\caption{The three spacetime constructions used in this work.}
\label{tab:three}
\end{table}

We refer the reader to \cite{minguzzi19} for most of the
conventions used in this work. In particular, by ($C^k$)  spacetime
we mean a connected, paracompact, Hausdorff, time-oriented
Lorentzian  ($C^k$)  manifold without boundary of arbitrary
dimension $n \geq 2$ and signature $(-,+, \dots, +)$. Indices
$a,b,c$ take the values $1,2,\ldots, d$ where $d$ is the dimension
of the configuration space $Q$, while (spacetime) greek indices
$\alpha,\beta,\cdots,$ take the values $0,1,\ldots, d+1$.

 A tensor field over a manifold is {\em smooth} if its degree of
differentiability is the maximal one compatible with the degree of
differentiability of the manifold. Thus, if the manifold is $C^k$,
by smooth vector field we mean a $C^{k-1}$ vector field, namely one
for which its components with respect to a coordinate basis are
$C^{k-1}$. In this respect, the generalized gravitational wave
spacetime $(M,g)$, which has been introduced in this section, is
$C^{r+1}$  where the fields $a_t,b_t, V$, entering the spacetime
metric and the Lagrangian are $C^r$. Thus the spacetime metric and
the Lagrangian have the same degree of differentiability. Since we
assume that $r\ge 2$ we can safely speak of Levi-Civita connection,
 Riemann tensor and exponential map.

The subset symbol $\subset$ is reflexive, thus $X \subset X$. The
boundary of a subset $A$ of a topological space is denoted $\dot{A}$
or $\p A$. Given two events $x,y\in M$, by $x<y$ we mean that
there is a future directed causal curve joining $x$ and $y$, and we
write $x\le y$ (also denoted $(x,y) \in J$) if $x<y$ or $x= y$.
If there is a timelike curve joining the events $x$ and $y$ we write
$x\ll y$ or $(x,y)\in I$. Thus the chronological and causal relations $I$ and $J$ are subsets of $M\times M$, while $I^{\pm}(x)$ and $J^{\pm}(x)$ denote the chronological and causal future and past of a point $x$, e.g.\ $J^{+}(x)=\{y: (x,y)\in J\}$, and similarly for subsets of $M$. The horismos relation is the
difference $J\backslash I$ (we do not denote it $E$, as this letter is reserved for the classical spacetime), and as it is well known
\cite{hawking73}, $(x,y)\in J\backslash I$ if and only if there is an
achronal lightlike geodesic connecting $x$ to $y$; we write $E^{\pm}(x)=J^{\pm}(x)\backslash I^{\pm}(x)$. By convention, timelike, causal and lightlike vectors are non-zero,
and the curves of the corresponding causal types
are, unless otherwise stated, future directed and regular.
Lines are inextendible curves which maximize the Lorentzian distance
between any pair of their points. Rays are defined analogously but
are only required to be past or future inextendible.

%

\subsection{Some relevant semi-time functions}

On a spacetime a semi-time function, according to the terminology
introduced by Seifert \cite{seifert77}, is a function $f: M \to
\mathbb{R}$ that increases over every timelike curve $x\ll y
\Rightarrow f(x)<f(y)$. By continuity (i.e. by $J\subset
\overline{I}$), every semi-time function is non-decreasing over
every causal curve, that is $x\le y \Rightarrow t(x)\le t(y)$.

An important property of the spacetime $(M,g)$ is the presence of
the semi-time function $t:M\to \mathbb{R}$. If $\gamma$ is a causal
curve then $\dd t[\gamma']=-g(n,\gamma')\ge 0$ where the equality
holds if and only if $\gamma'\propto n$. In particular, since the integral
lines of $n$ are diffeomorphic to $\mathbb{R}$ the spacetime is
causal. It is often useful to regard the spacetime $M$ as a
principal bundle $\pi: M \to E$ over the group $(\mathbb{R},+)$
giving the translations generated by the parallel vector $n$.



The hypersurfaces $t=const.$, denoted $\mathcal{N}_t$, are lightlike
as $\dd t[n]=-g(n,n)=0$, and totally geodesic. Indeed, if $\eta$ is
a geodesic with starting point in $\mathcal{N}_t$ and there tangent
to that hypersurface we have $\dd t[\eta']=-g(n,\eta')=cnst$, as $n$
is covariantly constant. However, at the starting point
$g(n,\eta')=0$, thus $t$ is constant over $\eta$.

%
%

Under some additional conditions we can find other interesting
semi-time functions. We recall that a time function is a continuous
function which increases over every casual curve, that is
$x<y\Rightarrow t(x)<t(y)$. It is well known that a spacetime is
stably causal if and only if it admits a time function
\cite{bernal04,minguzzi09c}.

\begin{proposition} \label{tff}
Let $B<+\infty$ be a constant such that for every $(t,q)\in E$,
$V+\frac{1}{2}a_t^{-1}(b_t,b_t)\le B$. Then $y+B t$ is a semi-time
function and if $B'>B$ then $y+B't$ is a time function. Thus if a
constant $B$ with this property exists, $(M,g)$ is stably causal.
\end{proposition}

\begin{proof}
Let $\gamma:I\to M$, $\lambda \to \gamma(\lambda)$,  be a causal
curve. If ${\gamma}'\propto n$ at the considered event then clearly
both functions have positive derivative at the event. Otherwise
$\gamma$ can be parametrized with respect to $t$ in a neighborhood
of the event and the causality condition reads
\begin{equation}
0\le -g(\dot{\gamma},\dot{\gamma})=2[\dot{y}-L(t,q(t),\dot{q}(t))],
\end{equation}
which becomes
\[\dot{y}\ge L(t,q(t),\dot{q}(t))\ge
-[V+\frac{1}{2}a_t^{-1}(b_t,b_t)]\ge -B,\] from which we get easily
the desired conclusion.
\end{proof}

\subsection{Legendre transform as the musical isomorphism on the light
cone} \label{ndx}

Let $\pi_T: E\to T$ be the projection on the first factor of
$E=T\times Q$. Every (local) section $\sigma:T\to E$ represents a
motion on the classical spacetime $E$. Its tangent vector has the
form $w=\frac{\p}{\p t}+\dot{q}$ where $\dot{q}\in TQ$ is expressed
in local coordinates as $\dot{q}=\dot{q}^k\p/\p{q^k}$. This example
suggests to look in detail to tangent vectors $w\in TE$ satisfying
$\dd t[w]=1$ as they can be written as
\[
w=\frac{\p}{\p t}+v
\]
with $v \in TQ$. In \cite{minguzzi06d} we suggested to represent
these vectors through their light lift on $M$. In other words, if
$e=(t,q)$, $w\in TE_{e}$, $\dd t[w]=1$, $y\in \mathbb{R}$, then
there is one and only one tangent vector $w^L\in TM_{(e,y)}$ which
is lightlike and such that $\pi_*(w^L)=w$. This vector can be easily
found by writing it as $w^L=w+\alpha n$, $n=\p_y$, and by fixing
$\alpha$ through the condition $g(w^L,w^L)=0$. The result is
\begin{equation} \label{flf}
w^L=\frac{\p}{\p t}+v+L(t,q,v) \, n.
\end{equation}

Let us now consider a slice $N_t\subset M$ where the semi-time
function $t$ is constant. This slice can be regarded as a fiber
bundle $\pi: N_t\to Q_t$, $Q_t=\{t\}\times Q$, with structure group
$(\mathbb{R}, +)$ generated by the action of the Killing field $n$.
A (abelian) connection $\omega_t$ on $N_t$ is a 1-form field
satisfying the properties $L_n\omega_t=0$, $\omega_t(n)=1$ (see
\cite{kobayashi63}), and it can  be written $\omega_t=\dd y-p_t$
where $p_t$ is a 1-form field over $Q_t$ (the minus potential).

This example suggests to consider at $(t,q,y)\in N_t$, 1-forms of
the form
\[
\omega=\dd y-p
\]
with $p\in T^*Q_q$, that is, those 1-forms that satisfy
$\omega(n)=1$. As done above we represent them through the only
1-form on $M$ which  restrict to $\omega$ on ${N}_t$ and which is
lightlike accordingly to the contravariant metric
\begin{equation}
g^{-1}=a_t^{-1}- [\frac{\p}{\p t}-a_t^{-1}(\cdot, b_t)]\otimes
\frac{\p}{\p y}-\frac{\p}{\p y}\otimes [\frac{\p}{\p
t}-a_t^{-1}(\cdot, b_t)]+[2V+a_t^{-1}(b_t, b_t)](\frac{\p}{\p y})^2.
\end{equation}
This unique lightlike 1-form is
\begin{equation} \label{glg}
\omega^L=\dd y-p+H(t,q,p) \dd t,
\end{equation}
where
\begin{equation} \label{nqe}
H(t,q,p)=\frac{1}{2} a_t^{-1}(p-b_t, p-b_t)+V(t,q),
\end{equation}
as it can be easily proved writing $\omega^L=\omega+\alpha \dd t$
and by fixing $\alpha$ through the condition
$g^{-1}(\omega^L,\omega^L)=0$. Clearly, $H(t,q,p)$ is the Legendre
transform of $L(t,q,v)$, namely the Hamiltonian of the mechanical
system on the base.
\begin{remark} \label{nyg}
If we weaken the condition of being lightlike to that of being
causal  then we find
\begin{equation} \label{gcg}
\omega^C=\dd y-p+F\dd t,
\end{equation}
and the causality condition reads $H(t,q,p)\le F$.
\end{remark}

Given a lightlike vector $w^L$ one can obtain a lightlike 1-form by
using the musical isomorphism between $TM$ and $T^*M$ provided by
the spacetime metric: $(w^L)^\flat=g(\cdot,w^L)$. Conversely, given a
lightlike 1-form $\omega^L$ we can obtain a corresponding lightlike
vector again through the musical isomorphism
$(\omega^L)^\sharp=g^{-1}(\cdot,\omega^L)$. It is now easy to check
that
\begin{align*}
-(\frac{\p}{\p t}+v+L(t,q,v) \, n)^\flat&=\dd y-p+H(t,q,p) \dd t \ \
\textrm{where} \ p=\frac{\p L}{\p
v}=a_t(\cdot,v)+b_t, \\
-(\dd y-p+H(t,q,p) \dd t)^\sharp&= \frac{\p}{\p t}+v+L(t,q,v) \, n\
\ \textrm{where} \ v=\frac{\p H}{\p p}=a_t^{-1}(p-b_t).
\end{align*}
The previous results can be summarized as follows

\begin{theorem} \label{nce}
To every vector $w\in TE$, $\dd t[w]=1$, corresponds one and only
one lightlike vector $w^L\in TM_x$, $x\in \pi^{-1}(e)$, such that
$\pi_*(w^L)=v$. This vector is future directed and given by Eq.
(\ref{flf}).

To every 1-form $\omega \in T^*N_t$, $\omega(n)=1$,  where $N_t$ is
a slice of constant time $t$, there corresponds one and only one
lightlike 1-form $\omega^L \in T^*M$ such that the pullback of
$\omega^L$ to $N_t$ under the inclusion $N_t \hookrightarrow M$,
coincides with $\omega$ (i.e. $\omega^L(V)=\omega(V)$ on every
vector $V$ tangent to $N_t$). This 1-form is given by Eq.
(\ref{glg}).

The (minus) musical isomorphism restricted to the light cone, namely
to lightlike vectors and lightlike 1-forms, acts a  Legendre
transformation for the components.
\end{theorem}

This result clarifies that the Lagrangian or the Hamiltonian point
of views are essentially the same. In the former on works preferably
on the spacetime tangent bundle, and in particular with the vectors
tangent to the light cone, while in the latter one works preferably
in the spacetime cotangent bundle, and in particular with the planes
tangent to the light cone.

\subsubsection{The Young-Fenchel inequality}
The Young-Fenchel inequality has a simple spacetime interpretation.
Given two future directed lightlike vectors $v,w$, the scalar
product satisfies $g(v,w)\le 0$ with equality if and only if $v$ and
$w$ are proportional \cite{hawking73} (this statement can be easily
proved in local Minkowski coordinates). Taking any lightlike vector
of the form $w^L$ and any 1-form of type $\omega^L$ we have
$\omega^L(w^L)=g((\omega^L)^\sharp,w^L)=-g(-(\omega^L)^\sharp,w^L)\ge
0$, because $-(\omega^L)^\sharp$ is future directed.  Taking into
account the expressions for $\omega^L$ and $w^L$, the inequality
reads
\[
L(t,q,v)+H(t,q,p)-p(v)\ge 0,
\]
where the equality holds if and only if $v^L\propto-
(\omega^L)^\sharp$, which implies $v^L=-(\omega^L)^\sharp$ that is
$v=\p H/\p p$.

\subsubsection{The velocity potential}

Through the musical isomorphism on the light cone (Legendre
transform) we can pass from a section $p: Q\to T^*Q$ to a section
$v: Q\to TQ$ and conversely. The 1-form fields which are exact
$p(q)=\dd f$, where $d$ is the exterior differentiation on $Q$, will
have special relevance in connection with the Hamilton-Jacobi
equation. In  this case the velocity field will take the form
\[
v=a_t^{-1}(\dd f-b_t)=a_t^{-1}( D^t f- b_t)
\]
where $D^t$ is the Levi-Civita covariant derivation compatible with
$a_t$. We shall say that $f$ is the velocity potential for the field
of velocities and we shall say that in this case the velocity field
is {\em vortex free at large}. If $p(q)$ is only closed then we
shall say that $v(q)$ is {\em vortex free}. Of course, if $Q$ is
simply connected then any vortex free field of velocities is vortex
free at large.

 We shall prove that the solution to the
Hamilton-Jacobi equation has indeed the physical meaning  of a
velocity potential, and that passing to the frame determined by the
velocity field $v(q)$ we can remove completely the $b_t$ term from
the Lagrangian (Theor. \ref{ndp}).

Locally, the condition of being vortex free is preserved in time.
Indeed, more generally, the circulation $\int_\alpha p$ where
$\alpha$ is a closed curve on $Q$ is preserved following the E.-L.
solutions with initial condition given by the velocity field $v(q)$
defined on the image of $\alpha$. Indeed, this circulation invariant
is known under the name  of {\em  invariant integral of
Poincar\'e-Cartan } \cite{arnold78}. It must be observed that the
solutions to the E.-L. equation with initial condition $v(q)$ might
develop caustics. Thus, if $v$ is defined on the whole $Q_{t_0}$,
even if the map $Q_{t_0} \to Q_t$ induced by the E.-L. flow were
surjective, it could be non-injective. As a result a closed curve
$\alpha(t)$ on $Q_t$ might be the image of an open curve on
$Q_{t_0}$. As a consequence, even though the initial condition had
vanishing circulation over every closed path, after some time this
property might not hold anymore. This failure, and the subsequent
fact that the velocity potential does not exist for all times, will
be reflected by the generic impossibility of finding a $C^1$
solution to the Hamilton-Jacobi equation defined on the whole time
axis.

\subsection{The light lift and the action functional}

A basic idea that we shall use is that of  {\em light lift}
\cite{minguzzi03b,minguzzi06,minguzzi06d}. It has been introduced in
\cite{minguzzi03b} for the case of spacelike dimensional reduction
(which leads to the relativistic Lorentz force equation) and since
then it has been used to solve problems of existence and
multiplicity for stationary points of the charged particle action
\cite{minguzzi06,giambo07}. It has been introduced  in the lightlike
dimensional reduction context of this work in \cite{minguzzi06d} and
then taken up again  \cite{flores07}.

In short, given a $C^1$ curve on the quotient manifold generated by
the Killing vector, the classical spacetime $E$ in our case, one
seeks the (unique in the present lightlike dimensional reduction
case) $C^1$ lightlike curve (the {\em light lift}) that projects on
it (Fig. \ref{fre}).

\begin{figure}[ht]
\begin{center}
\psfrag{M}{{\small $M$}} \psfrag{T}{{\small $t$}} \psfrag{M0}{{\small
$x_0$}} \psfrag{N}{{\small $n=\p_y$}} \psfrag{Q0}{{\small $e_0$}}
\psfrag{Q1}{{\small $e_1$}} \psfrag{Q}{{\small $E$}}
\psfrag{P}{{\small $\pi$}} \psfrag{Ga}{{\small $e(t)$}}
\includegraphics[width=6cm]{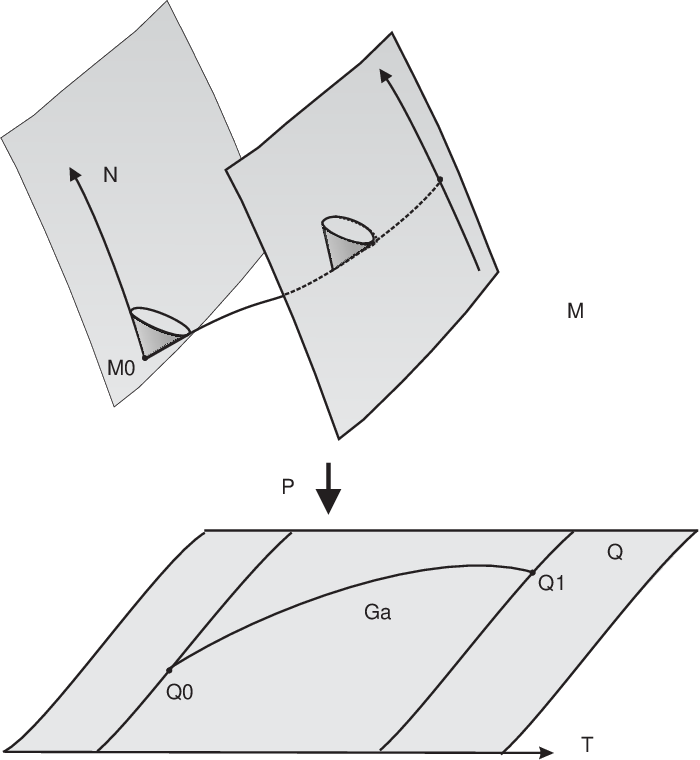}
\end{center} \caption{The light lift.} \label{fre}
\end{figure}

Now, it happens that the extra-coordinate along the light lift is
proportional to the action as calculated on the base curve, a fact
which allows us to relate the stationary points of the geodesic
functional on the full spacetime with the stationary points of the
action on the base by means of a lightlike version of the more
common timelike Fermat's principle \cite{kovner90,perlick90}. The
reader is referred to \cite{minguzzi06d} for the proof of these
results in the lightlike dimensional reduction case.

\begin{lemma}
Let $q: [t_0,t_1] \to Q$ be a $C^1$  curve of endpoints $q_0$ and
$q_1$ and let $e(t)=(t,q(t))$, be the corresponding  $C^1$ curve on
$E$ then the curve
\begin{equation}
x(t)=(t,q(t),y_0+ \mathcal{S}_{e_0,e(t)}[q\vert_{[t_0,t]}])
\end{equation}
gives the unique lightlike curve (the light lift) which projects on
$e(t)$   and starts from $x_0=(t_0,q_0,y_0)$. Conversely, every
$C^1$ lightlike curve with tangent vector nowhere proportional to
$n$ is the light lift of its ($C^1$) projection on the base $E$.
\end{lemma}

Note that, as claimed, the extra-coordinate of the light lift is
related to the classical action functional. The origin of this
result can be easily grasped by noting that any causal curve
$\gamma$ which can be parametrized by $t$ satisfies (see Eq.
(\ref{eis}))
\begin{equation} \label{mod}
-g(\dot{\gamma},\dot{\gamma})=2[\dot{y}-L(t,q(t),\dot{q}(t))].
\end{equation}
We recall a result obtained in \cite{minguzzi06d} (see also Prop.
\ref{nqa}). The correspondence holds also for minimizing curves (see
corollary \ref{vqz}).

\begin{theorem} \label{kdf}
Every lightlike geodesic of $(M,g)$ not coincident with a flow line
of $n=\p_y$, admits as affine parameter the function $t$, and for
any such curve $x(t)=(t,q(t),y(t))$, the function $q(t)$ is a
$C^{r+1}$ stationary point of functional (\ref{cxs}) on
$C^1_{e_0,e_1}$ for any pair $e_0,e_1$, $t_0<t_1$, on the projection
$e(t)=(t,q(t))$, and $x(t)$ is the light lift of $e(t)$, that is
\[y(t)=y_0+ \mathcal{S}_{e_0,e(t)}[q\vert_{[t_0,t]}].\]
Conversely, given a  stationary point\footnote{In \cite{minguzzi06d}
we required the stationary point to be $C^2$ but this condition can
be removed since any $C^1$ stationary point is actually $C^{r+1}$
since both $L$ and $\p_{v} L$ are $C^r$ \cite[theorem
1.2.4]{jost98}.} $q(t)$ of functional (\ref{cxs}) on
$C^1_{e_0,e_1}$, the light lift $x(t)=(t,q(t),y_0+
\mathcal{S}_{e_0,e(t)}[q\vert_{[t_0,t]}])$ is an affinely
parametrized lightlike geodesic of $(M,g)$ necessarily not
coincident with a flow line of $n$.
%
\end{theorem}

In mathematical relativity a causal curve $\gamma: (a,b)\to M$ is
{\em future extendible} if it admits a future endpoint $p$, namely
$\lim_{s\to b} \gamma(s)=p$. We have therefore a notion of
inextendibility for causal curves \cite{hawking73}. Fortunately, for
causal geodesics the concept of geodesic inextendibility (that is
maximality) coincides with that of inextendibility for causal curves
\cite[Lemma 8, Chap. 5]{oneill83}.

An immediate consequence of the previous theorem is that the
projection of a {\em maximal} (i.e. inextendible, in relativists'
terminology) lightlike geodesic not coincident with an integral line
of $n$ is a {\em maximal} solution to the E.-L. equations, and
conversely, the light lift of a maximal solution to the E.-L.
equations is an inextendible lightlike geodesic not coincident with
an integral line of $n$ .

%
%

\section{Relationship between Lorentzian distance and least action} \label{s:dist}

As a first step we are going to study the causal relations on
$(M,g)$. As we shall see the function $S$ will play a key role. It
is convenient to introduce suitable causal relations on $E$,
although $E$ is not a Lorentzian manifold.
%
%
Let us define
\begin{align}
I^{+}(e)&= \{(t(e),+\infty)\times {Q} \} , \\
 J^{+}(e)&=\{e\}\cup I^{+}(e) ,\\
 E^{+}(e)&=\{ e\}.
\end{align}

Thus $e_1 \in I^{+}(e_0)$ if and only if $t_1>t_0$,  $e_1 \in J^{+}(e_0)$ if and only if
$t_1>t_0$ or $e_1=e_0$, and $e_1 \in E^{+}(e_0)$ if and only if $e_1=e_0$.
Observe that $I^{+}(e)$ is open while $J^{+}(e)$ is not closed.
Note, moreover, that $I^{+}(e)=\textrm{Int} J^{+}(e)$ and
$\dot{I}^{+}(e)=\dot{J}^{+}(e)$ is the set $t=t(e)$.

Let us denote with $r_x$ the image of the future directed lightlike
ray starting from $x\in M$ generated by the lightlike vector field
$n$.


We have (see also \cite[Prop. 4.3]{flores07})

\begin{lemma} \label{pcf}
For every $x_0=(e_0, y_0) \in M$,
\begin{align}
I^{+}(x_0) &= \{x_1: y_1-y_0>S(e_0,e_1) \textrm{ and
} e_1 \in I^{+}(e_0) \}, \label{cgf1}\\
J^{+}(x_0) &\subset \{x_1: y_1-y_0\ge S(e_0,e_1)
 \},\label{cgf2} \\
E^{+}(x_0) &\subset r_{x_0} \cup \{x_1: y_1-y_0=S(e_0,e_1) \}.
\label{cgf3}
\end{align}
Analogous past versions hold.
\end{lemma}

\begin{remark}
Note that in Eq. (\ref{cgf2}) the set on the right-hand side is the
same of $\{x_1: y_1-y_0\ge S(e_0,e_1) \textrm{ and } e_1 \in
J^{+}(e_0) \}$ because $y_1-y_0$ is finite while for $e_1 \notin
J^{+}(e_0)$, $S(e_0,e_1)=+\infty$.
\end{remark}

\begin{proof}
If $x_1=(e_1,y_1)\in I^{+}(x_0)$ let $x(t)$ be a  $C^1$ timelike
curve connecting $x_0$ to $x_1$. The function $t$ can be taken as
parameter because $t$ is increasing over timelike curves, in
particular $t_1>t_0$ i.e. $e_1 \in I^{+}(e_0)$. Thus
$x(t)=(t,q(t),y(t))$, where $e(t)=(t,q(t))$ is the projected curve
on $E$. The condition of being timelike reads, see Eq. (\ref{mod}),
$\dot{y}>L$ from which it follows
$y_1-y_0>\mathcal{S}_{e_0,e_1}[q]\ge S(e_0,e_1)$. Conversely, assume
$x_1$ is such that $y_1-y_0>S(e_0,e_1)$ and $t_1>t_0$. Let
$e(t)=(t,q(t))$ be any $C^1$ curve connecting $e_0$ and $e_1$ such
that $y_1-y_0>\mathcal{S}_{e_0,e_1}[q]\ge S(e_0,e_1)$. Its light
lift $x(t)=(t,q(t),y_0+\mathcal{S}_{e_0,e(t)}[q\vert_{[t_0,t]}])$ is
a $C^1$ lightlike curve which connects $x_0$ to $(e_1,
y_0+\mathcal{S}_{e_0,e_1}[q])$. Note that this point is in the same
fiber of $x_1$ but in the past of it because
$y_1>y_0+\mathcal{S}_{e_0,e_1}[q]$, thus composing the curve $x(t)$
with a segment of the fiber one gets a causal curve joining $x_0$ to
$x_1$ which is not a lightlike geodesic (otherwise it would be
$t_0=t_1$) thus $x_1\gg x_0$.

If $x_1=(e_1,y_1)\in J^{+}(x_0)$ let $x(t)$ be a  $C^1$ causal curve
$\gamma$ connecting $x_0$ to $x_1$. If it is not a lightlike
geodesic then $x_1 \in I^{+}(x_0)$ and hence $x_1$ belong to the
right-hand side of Eq. (\ref{cgf1}) which is included in the
right-hand side of (\ref{cgf2}). If it is a lightlike geodesic then
$g(\gamma',n)=const. \le 0$. If this constant is  zero then $\gamma$
is necessarily a segment of the fiber starting from $x_0$. The whole
ray starting from $x_0$ is included in the right-hand side of
(\ref{cgf2}). If it is different from zero then $\dd
t[\gamma']=-g(\gamma',n)>0$ thus $t$ is an affine parameter. Let
$x(t)=(t,q(t),y(t))$ be the curve $\gamma$ parametrized  with
respect to $t$. The condition of being causal reads, see Eq.
(\ref{mod}), $\dot{y}\ge L$ from which it follows $y_1-y_0\ge
\mathcal{S}_{e_0,e_1}[q]\ge S(e_0,e_1)$. The last inclusion for
$E^{+}(x_0)=J^{+}(x_0)\backslash I^{+}(x_0)$ follows from the former
equations.
\end{proof}

\subsection{Relationship between the l.s.c.\  of the action functional and the u.s.c.\ of the Lorentzian length}

The relationship between the Lorentzian length functional $l$ and
the action functional $\mathcal{S}$ is given by the following
maximization result

\begin{theorem} \label{jsc}
Let $x_1 \in J^{+}(x_0)$ and $t_0<t_1$, thus in
particular,
$y_1-y_0\ge S(e_0,e_1)$. Let $e(t)=(t,q(t))$ be a $C^1$ curve which
is the projection of some $C^1$ causal curve connecting $x_0$ to
$x_1$, then $y_1-y_0\ge \mathcal{S}_{e_0,e_1}[{q}]$. Among all the
$C^1$ causal curves $x(t)=(t,q(t),y(t))$, connecting $x_0$ to $x_1$,
which project on $e(t)$, the causal curve $\gamma(t)=(t,q(t),y(t))$
with
\begin{equation} \label{vyg}
y(t)=y_0+
\mathcal{S}_{e_0,{e}(t)}[{q}\vert_{[t_0,t]}]+\frac{t-t_0}{t_1-t_0}(y_1-y_0-\mathcal{S}_{e_0,e_1}[{q}])
\end{equation}
is the one and the only one that maximizes the Lorentzian length.
The maximum is
\begin{equation} \label{psg}
l(\gamma)=\{2(y_1-y_0-\mathcal{S}_{e_0,e_1}[q])(t_1-t_0)\}^{1/2}.
\end{equation}
\end{theorem}

\begin{remark}
Thus in particular Eq.(\ref{vyg}) can be rewritten in the equivalent
form
\begin{equation} \label{vyg2}
y(t)=y_0+
\mathcal{S}_{e_0,{e}(t)}[{q}\vert_{[t_0,t]}]+\frac{l(\gamma)^2}{2(t_1-t_0)^2}
(t-t_0).
\end{equation}
\end{remark}

\begin{proof}
Let $\eta(t)=(t,q(t),w(t))$ be a $C^1$ causal curve connecting $x_0$
to $x_1$, then since it is causal by Eq. (\ref{mod}), $\dot{w}\ge
L$, and integrating, $y_1-y_0\ge \mathcal{S}_{e_0,e_1}[{q}]$.

The curve  $\gamma$ is causal because (use Eq. (\ref{mod}))
\[
-g(\dot{\gamma},\dot{\gamma})=\!\frac{2}{t_1\!-\!t_0}
(y_1-y_0-\mathcal{S}_{e_0,e_1}[{q}])\ge 0 ,
\]
taking the square root and integrating one gets Eq. (\ref{psg}). If
$\tilde\gamma=(t,q(t),\tilde{y}(t))$ is another $C^1$ timelike curve
connecting $x_0$ to $x_1$ and projecting on $e(t)$
\[
-g(\dot{\gamma},\dot{\gamma})=\!\frac{2}{t_1\!-\!t_0}
(y_1-y_0-\mathcal{S}_{e_0,e_1}[{q}])=\!\frac{1}{t_1\!-\!t_0}\!\int_{t_0}^{t_1}\!\!\!\!\![-g(\dot{\tilde{\gamma}},\dot{\tilde{\gamma}})]\dd
t.
\]
Using the Cauchy-Schwartz inequality $\int_{t_0}^{t_1}\! [
-g(\dot{\tilde{\gamma}},\dot{\tilde{\gamma}})]\dd t \ge
(t_1-t_0)^{-1}l(\tilde{\gamma})^2$, replacing in the above
equation, taking the square root and integrating $l(\gamma)\ge
l(\tilde{\gamma})$, thus $\gamma$ is longer than $\tilde\gamma$. In
order to prove the uniqueness note that the equality sign in
$l(\gamma)\ge l(\tilde{\gamma})$ holds if and only if it holds in the
Cauchy-Schwarz inequality which is the case if and only if
$g(\dot{\tilde{\gamma}},\dot{\tilde{\gamma}})=const.$, that is if and only if
$\dot{\tilde{y}}-L=const.$ which once integrated, and using suitable
boundary conditions, gives Eq. (\ref{vyg}).
\end{proof}

As it is well known the functional $l$ is not lower semi-continuous
on the set of connecting continuous causal curves  with the $C^0$
topology. The reason is that in any neighborhood of the curve it is
possible to find a curve which is lightlike and connects the same
endpoints. Analogously, the functional $\mathcal{S}$ is not upper
semi-continuous on the set of $C^1$  curves endowed with the $C^0$
topology. The reason is that in any neighborhood of a curve $q$ one
can always find a rapidly oscillating $C^1$ curve which makes the
functional $\mathcal{S}$ arbitrarily large thanks to the
contribution of the kinetic energy.

It is also known that the functional $l$ is upper semi-continuous on
the set of connecting continuous causal curves  with the $C^0$
topology \cite{beem96,minguzzi07c}.  Although we shall not use this
result, we are  going to prove that the functional $\mathcal{S}$ is
lower semi-continuous. The proof will be based on  the limit curve
theorem in Lorentzian geometry. The notion of continuous causal
curve was introduced in \cite{hawking73} through convex
neighborhoods. Following \cite{minguzzi19}, we define a continuous causal curve as an absolutely continuous curve whose tangent vector, which exists almost everywhere, is causal and future directed. This definition is equivalent to that of \cite{hawking73}, and such curves admit locally Lipschitz parametrizations, e.g.\ by the arc length of a Riemannian metric \cite{minguzzi19}.

\begin{theorem}
Let $e_0=(t_0,q_0)$, $e_1=(t_1,q_1)$, $t_0<t_1$. The functional
$\mathcal{S}_{e_0,e_1}[q]$ is lower semi-continuous in the $C^0$
topology on the $C^1$ connecting curves $q$. More precisely, let $q:
[t_0,t_1]\to Q$, be a $C^1$ connecting curve: $q(t_0)=q_0$,
$q(t_1)=q_1$, and let $d$ be the distance of a Riemannian metric on $Q$. For every $\epsilon>0$ there is $\delta>0$ such that any $C^1$ connecting curve $q':
[t_0,t_1]\to Q$ which satisfies $d(q'(t),q(t))<\delta$ for every $t\in[t_0,t_1]$ (equivalently, whose graph is contained in a sufficiently small neighborhood of the graph $\{(t,q(t))\}$ in $E$) satisfies
$\mathcal{S}_{e_0,e_1}[q']>\mathcal{S}_{e_0,e_1}[q]-\epsilon$.
\end{theorem}

The closeness condition cannot be imposed only on the spatial images. For $Q=\mathbb{R}$, $L=\frac{1}{2}v^2$, $[t_0,t_1]=[0,1]$ and the endpoints $0$ and $1$, the curves $q(t)=t^2$ and $\tilde q(t)=t$ have the same image $[0,1]$, but $\mathcal{S}[q]=\frac{2}{3}$ while $\mathcal{S}[\tilde q]=\frac{1}{2}$.

\begin{proof}
Let $K$ be a compact
neighborhood of the image of $q(t)$.
Let $h$ be a Riemannian metric on $Q$ such that for every $v\in TK$,
$t\in [t_0,t_1]$, $h(v,v)\le a_t(v,v)$. Since all the Riemannian distances induce the same uniform topology on compact sets, we can take $d=d_h$ near the image of $q$. Suppose by contradiction that the statement is false. There is an
$\epsilon>0$ and a sequence of connecting curves $q_n$ such that $d_h(q_n(t),q(t))<1/n$ for every $t\in[t_0,t_1]$ (so that their images are contained in $K$ for large $n$) and
$\mathcal{S}_{e_0,e_1}[q_n]<\mathcal{S}_{e_0,e_1}[q]-\epsilon$. Let
$y_0\in \mathbb{R}$; the  light lifts $x_n(t)$ of $(t, q_n(t))$
starting from $x_0=(t_0,q_0,y_0)$ are lightlike curves thus, by the
limit curve theorem \cite{beem96,minguzzi07c}, there is a continuous
causal curve $\eta(\lambda)$ starting from $x_0$ which is either future
inextendible or it reaches some point in $\pi^{-1}(e_1)$. Every point of $\eta$ is a limit of points $x_n(t_n)$, and since $q_n\to q$ uniformly and $q$ is continuous, its projection on $E$ is the limit of $(t_n,q_n(t_n))$, namely a point of the graph of $q$. Thus $\eta$ projects on the graph $\{(t,q(t)): t\in I\}$, where $I$ is some interval containing $t_0$ (possibly $I=\{t_0\}$). As a consequence, if it connects $x_0$ to $\pi^{-1}(e_1)$ then $I=[t_0,t_1]$, and the $y$ coordinate on the last point on $\pi^{-1}(e_1)$ is bounded from below by the $y$-coordinate of the last point of the light lift of $(t,q(t))$ (e.g.\ Theor. \ref{jsc} ), that is $y_0+\mathcal{S}_{e_0,e_1}$. This is a contradiction with  $\mathcal{S}_{e_0,e_1}[q_n]<\mathcal{S}_{e_0,e_1}[q]-\epsilon$ according to which its value should be no larger than $y_0+\mathcal{S}_{e_0,e_1}[q]-\epsilon$ (recall that by the limit curve theorem the convergence is uniform on compact subsets).

Let $\overline{V}$ be an upper bound for $V(t,q)$ on K and let $B^2$
be an upper bound of $h^{-1}(b_t,b_t)$ on $K$. The remaining
possibility is that $t_1\notin I$.  In this case, as $\eta$ is
inextendible the $y$ coordinate must be unbounded from above. In
particular, there is some $\bar{t}\in [t_0,t_1)$ for which
$(\bar{t}, q(\bar{t}), \bar{y})$ belongs to $\eta(\lambda)$, for
$\bar{y}>y_0+\mathcal{S}_{e_0,e_1}[q]+(\frac{B^2}{2}+\overline{V})(t_1-t_0)$.
We can find a sequence $\epsilon(n)$ which goes to $0$ as $n\to
+\infty$ such that
$y(x_n(\bar{t}+\epsilon(n)))>y_0+\mathcal{S}_{e_0,e_1}[q]+(\frac{B^2}{2}+\overline{V})(t_1-t_0)$.


For any given path (image of $q_n(t)$) the kinetic energy is
minimized by that reparametrization  which makes the speed constant
(Cauchy-Schwarz inequality). As a consequence,  for every
$n$
\begin{align*}
\mathcal{S}_{e(\bar{t}+\epsilon(n)),e_1}[q_n]&\ge \int_{\bar{t}+\epsilon(n)}^{t_1}
[\frac{1}{2}
a_t(\dot{q}_n,\dot{q}_n)+b_t(\dot{q}_n)-V(t,q_n(t))] \dd t \\
&\ge \int_{\bar{t}+\epsilon(n)}^{t_1} [\frac{1}{2} h(\dot{q}_n,\dot{q}_n)-B
\sqrt{h(\dot{q}_n,\dot{q}_n)}-\overline{V}] \dd t \\
&\ge \frac{(l_h[q_n\vert_{[\bar{t}+\epsilon(n), t_1]}])^2}{2
(t_1-\bar{t}-\epsilon(n))} - B
l_h[q_n\vert_{[\bar{t}+\epsilon(n), t_1]}]-\overline{V} (t_1-\bar{t}-\epsilon(n))\\
&\ge -(\frac{B^2}{2}+\overline{V})(t_1-\bar{t}-\epsilon(n)) \ge  -(\frac{B^2}{2}+\overline{V})(t_1-t_0)
\end{align*}
where $l_h$ is the Riemannian $h$-length. Finally, the last point of
$x_n(t)$ has $y$-coordinate
\[
y_0+\mathcal{S}_{e_0,e_1}[q_n]> \bar{y}-\epsilon/2 +
\mathcal{S}_{e(\bar{t}+\epsilon(n)),e_1}[q_n]\ge
y_0+\mathcal{S}_{e_0,e_1}[q]-\epsilon/2,
\]
a contradiction.
\end{proof}

\begin{remark}
The family of $C^1$ connecting curves is somewhat small, as it is
not preserved under limits. Nevertheless, the previous proof works
also if this family is replaced by those connecting curves on $E$
which are projections of continuous causal curves. This is the
family of continuous almost everywhere differentiable curves whose
derivative is $L^2$ (the continuous causal curve is $H^{1,2}$ and so
is its projection by \cite[Theor.\ 2.24]{buttazzo98}, where the role
of $\theta$ is played by the projection $\pi$). This last family
coincides with the most natural family of curves for the study of
variational problems for which the Lagrangian is of classical type,
i.e., quadratic in the velocities (see the discussion in \cite[Chap.
2]{buttazzo98}).  Nevertheless, the whole point of working on
$(M,g)$ rather that $E$ is that this approach makes it very easy to
deal with limit curves, as only continuous or $C^1$ causal curves
need to be considered and the limit of a sequence of continuous
causal curves is still a continuous causal curve. The smoothness
(i.e.\ $C^{r+1}$) of the limit is often inferred from some
achronality property, which implies that it is actually a lightlike
geodesic. As the mathematics simplifies, the geometrical content
becomes much more transparent.

\end{remark}

We observe that the found inversion of properties, namely the fact that  $\mathcal{S}$
is lower semi-continuous while $l$ is upper semi-continuous,  is
reflected by the minus sign in front of $\mathcal{S}$ in Eq.\
(\ref{psg}).

\subsection{Relation between least action and Lorentzian distance, and between Hamilton-Jacobi and eikonal equations} \label{pqz}

We are ready to establish the relation between the least action $S$
and the Lorentzian distance $d: M\times M\to [0,+\infty]$.

\begin{theorem} \label{vfg}
Let $x_0,x_1 \in M$,  $x_0=(e_0,y_0)$, $x_1=(e_1,y_1)$ then if $x_1
\in J^{+}(x_0)$,
\begin{equation} \label{cft}
d(x_0,x_1)=\sqrt{2[y_1-y_0-S(e_0,e_1)](t_1-t_0)}.
\end{equation}
 In particular, $S(e_0,e_1)=-\infty$ if and only if
$d(x_0,x_1)=+\infty$.
\end{theorem}

\begin{proof}
If $x_1 \in J^{+}(x_0)$ then $y_1-y_0-S(e_0,e_1)\ge 0$ and $e_1 \in
J^{+}(e_0)$ by Lemma \ref{pcf}. Let us consider separately the cases
$y_1-y_0-S(e_0,e_1)= 0$ and $y_1-y_0-S(e_0,e_1)> 0$. In the former
case the right-hand side vanishes and by Lemma \ref{pcf} we have
$x_1 \notin I^{+}(x_0)$, thus $d(x_0,x_1)=0$, i.e., in this case the
formula is verified. In the latter case by Lemma \ref{pcf} $x_1 \in
I^{+}(x_0)$ and $t_1>t_0$. Now, the Lorentzian distance is usually
defined as the least-upper bound of the Lorentzian lengths of the
causal connecting  curves. Nevertheless, since $x_1 \in I^{+}(x_0)$,
and since every connecting causal curve which is not a lightlike
geodesic (hence of length zero) can be replaced by a connecting
timelike curve with no less Lorentzian length \cite{lerner72}, the
least-upper bound of the Lorentzian lengths can be taken over the
connecting timelike curves. These curves can be parametrized by $t$
and by Theorem \ref{jsc} the Lorentzian distance is the least-upper
bound of the right-hand side of   Eq. (\ref{psg}) over the set
$C^1_{e_0,e_1}$, from which thesis follows. Note that the proof
works even if $S(e_0,e_1)=-\infty$.

\end{proof}

\begin{corollary} \label{vqz}
Let $\gamma(t)=(t,q(t),y(t))$, $t \in [t_0,t_1]$, be a $C^1$  causal
curve projecting on $e(t)=(t,q(t))$. We have
\begin{equation} \label{gcv}
d(x_0,x_1)^2-l(\gamma)^2\ge
2\{\mathcal{S}_{e_0,e_1}[q]-S(e_0,e_1)\}(t_1-t_0).
\end{equation}
Moreover, the equality sign holds if and only if (i) or (ii), where (i):
$d(x_0,x_1)=-S(e_0,e_1)=+\infty$  and (ii): $\gamma$ has
extra-coordinate dependence
\begin{equation}
y(t)=y_0+\mathcal{S}_{e_0,e(t)}[q\vert_{[t_0,t]}]+ c(t-t_0),
\end{equation}
for a suitable constant $c \ge0$ (necessarily  related to the length
of $\gamma$ by $l(\gamma)^2=c\,2 (t_1-t_0)^2$).  Finally, $\gamma$
is Lorentzian distance maximizing if and only if (ii) and $q$ is action
minimizing.
%
\end{corollary}

\begin{proof}
The equation (\ref{gcv})  as well as the study of the equality sign
follows from theorems \ref{jsc} and \ref{vfg}. The last statement is
a trivial consequence of Eq. (\ref{gcv}) if it is taken into account
that $\mathcal{S}_{e_0,e_1}[q]-S(e_0,e_1)\ge 0$.
\end{proof}

%
%


\begin{corollary} \label{hgf}
The function  $S$ is upper semi-continuous everywhere but on the
diagonal of $E \times E$  and satisfies the triangle inequality: for
every $e_0,e_1,e_2 \in E$
\[
 S(e_0,e_2)\le S(e_0,e_1)+S(e_1,e_2),
\]
with the convention that $(+\infty)+(-\infty)=+\infty$.
\end{corollary}

\begin{proof}
The case in which some of the term  equals $+\infty$ it is readily
verified considering the various sign cases for the time
differences. Let us assume that all the terms are bounded from above
and hence that $e_1 \in J^{+}(e_0)$ and $e_2 \in J^{+}(e_1)$.

 Since the set of curves $e(t)=(t,q(t))$ which
connect $e_0$ to $e_2$ contains the subset of curves passing through
$e_1$ the triangle inequality is obvious (but note that it was
important to define $S(e_1,e_2)=0$ for $e_1=e_2$). The upper
semi-continuity of ${S}$ at $(e_0,e_1)$ with $t_0<t_1$ is immediate
from Eq. (\ref{cft}) and the lower semi-continuity of $d$, it
suffices to choose $y_0$ and $y_1$ so that $x_0=(e_0,y_0)$ and
$x_1=(e_1,y_1)$ are chronologically related. The upper
semi-continuity of ${S}$ at $(e_0,e_1)$ with (i) $t_1<t_0$ or with
(ii) $t_0=t_1$ and $q_0\ne q_1$, follows from the fact that
$S(e_0,e_1)=+\infty$.
\end{proof}

Eq. (\ref{cft}) establishes a relation between the Lorentzian
distance and the  least action that explains the many analogies
between the two functions, from the continuity properties to the
triangle inequalities.

Both functions satisfy suitable differential equations. In a
distinguishing spacetime  $d(x_0,x)$ regarded as a function of $x$
coincides with the local distance function \cite{beem96} because
 there is a neighborhood $V \ni x_0$ such that no
causal curve starting from $x_0$ can return in $V$ after escaping it
\cite{hawking73,minguzzi19}. Moreover, this same local Lorentzian
distance function satisfies in a neighborhood of $x_0$ the timelike
eikonal equation \cite{erkekoglu03},
\[g(\nabla d,\nabla d)+1=0,\] while it is well known that locally
$S(e_0,e)$ regarded as a function of $e=(t,q)$ satisfies the
Hamilton-Jacobi equation\footnote{For global existence results of
solutions to the Hamilton-Jacobi equation see \cite{crandall83}.}.
Thanks to Eq. (\ref{cft}) they can be regarded as the same equation,
indeed a calculation gives
\begin{equation}
g(\nabla d,\!\nabla d)+1=\!\frac{2(t-t_0)^2}{d^2}\,[\,\frac{\p S}{\p
t}+ H(t,q,\frac{\p S}{\p q})].
\end{equation}

\section{\texorpdfstring{$C^1$}{C1} solutions to the Hamilton-Jacobi equation and null hypersurfaces} \label{s:c1}

Let $N$ be a $C^{r+1}$ spacetime. Let $S$ be a $C^k$ hypersurface in
$N$, $1\le k \le r$,  namely the image of   a
 co-dimension one embedded $C^k$ submanifold, $\phi: \phi^{-1}(S)\to
N$, where $\phi$ is $C^k$.

\begin{remark}
For every $p\in S$ there is a neighborhood $V\ni p$ and a $C^k$
function $f: V\to \mathbb{R}$ such that $S\cap V=f^{-1}(0)$.
\end{remark}

\begin{proof}
For every $p\in S$ we can find a neighborhood $U\subset N$ of $p$,
and a smooth (i.e. $C^r$) nowhere vanishing vector field whose
integral lines intersect $S\cap U$ only once. The flow $\varphi_t$
generated by this field is $C^r$ \cite[Theor. 17.19]{lee03}, thus
the map $(t,s)\to \varphi_t\circ \phi(s)$ is $C^k$. By the implicit
function theorem \cite[Theor. 2.2, Chap. 7]{gilardi92}, the
parameter $t$ of the integral lines of this field, with the zero
value fixed on $S\cap U$, provides a $C^k$ function $f:V\to
\mathbb{R}$, for some open set $V$, $p\in V\subset U$, such that
$f^{-1}(0)=S\cap V$.
\end{proof}

Since $k\ge 1$, the pushforward of the tangent space at
$\phi^{-1}(p)$, namely the tangent space to $S$ at $p\in S$, is the
kernel of the $C^{k-1}$ 1-form $\dd f$. In particular any other
1-form on $T^*N_p$ with the same kernel is proportional to $\dd f
(p)$. The pullback $\phi^{*}g$ is a $C^{k-1}$ metric which is
degenerate at $p$ if and only if $\dd f$ is a lightlike 1-form at
$p$ \cite{hawking73}. If this is the case for every $p\in S$, then
$S$ is called $C^k$ {\em null (or lightlike) hypersurface} in $M$.
The $C^{k-1}$ lightlike vector field on $V\cap S$ orthogonal to $\dd
f$ is $W=-g^{-1}(\cdot, \dd f)$, and it is tangent to $S$ because
$(\dd f)(W)=-g^{-1}(\dd f, \dd f)=0$. If needed we redefine the sign
of $f$  so as to make $W$ future directed. Thus $g^{-1}(\cdot, \dd
f)$, dual to $\dd f$, is past directed.

If $k\ge 2$ the integral curves of $W$ on $S$ are lightlike pregeodesics. Indeed, the 1-form $W^\flat=g(\cdot,W)=-\dd f$ is closed, and since the Levi-Civita connection is torsion free, $g(\nabla_X W,Y)=g(\nabla_Y W,X)$ for every $X,Y$. Taking $X=W$ we obtain $g(\nabla_W W,Y)=g(\nabla_Y W,W)=\frac{1}{2}\, Y\big(g(W,W)\big)$. The $C^1$ function $g(W,W)=g^{-1}(\dd f,\dd f)$ vanishes on $S$, thus at the points of $S$ its differential is proportional to $\dd f=-W^\flat$, and hence $\nabla_W W$ is proportional to $W$. If $\dd f$ is lightlike on an open set, as in the next lemma, then $g(W,W)=0$ there and $\nabla_W W=0$, namely $W$ is geodesic.

\begin{lemma} \label{njx}
Let $\dd f$ be an exact $C^0$  1-form defined on some open set
$\pi^{-1}(V)\subset M$, $V\subset E$, where $(M,g)$, $M=E\times
\mathbb{R}$, is a generalized gravitational wave spacetime, and
assume that $\dd f$   is a connection for the $(\mathbb{R},+)$
bundle $\pi: \pi^{-1}(V)\to V$, namely $\dd f(n)=1$ and $\forall
\Delta y\in \mathbb{R}$, $\varphi^*_{\Delta y}\dd f=\dd f$ where
$\varphi_{\Delta y}$ is the flow of $n$, then
\begin{itemize}
\item[(a)] $df$ is lightlike if and only if $f=y-u(t,q)$ for some $C^1$
function $u:V\to \mathbb{R}$, where $u$ satisfies the
Hamilton-Jacobi equation: $\frac{\p u}{\p t}+H(t,q,\frac{\p u}{\p
q})=0$,
\item[(b)] $df$ is causal if and only if $f=y-u(t,q)$ for some $C^1$
function $u:V\to \mathbb{R}$, where $u$  is a subsolution to the
Hamilton-Jacobi equation: $\frac{\p u}{\p t}+H(t,q,\frac{\p u}{\p
q})\le 0$.
\end{itemize}
\end{lemma}


\begin{proof}
The hypersurface $N_t\cap \pi^{-1}(V)$ is naturally included in $M$,
and under the assumptions the pullback of $\dd f$ under this
inclusion is a connection for the $(\mathbb{R},+)$ bundle $\pi:
N_t\cap \pi^{-1}(V) \to V$.

Suppose that $\dd f$ is lightlike. By Theorem \ref{nce} $\dd f$
reads $\dd f=\dd y-p_t+H(t,q,p_t)\dd t$ for some time dependent
1-form field $p_t\in T^*Q$. Let $f=y-u(t,q,y)$, plugging this
expression into the previous equation we obtain that $u$ is
independent of $y$ and satisfies the Hamilton-Jacobi equation. The
converse is trivial by Theorem \ref{nce}.

Suppose that $\dd f$ is causal. By Remark \ref{nyg} it takes the
form $\dd f=\dd y-p_t+F(t,q)\dd t$, where $F\ge H(t,q,p_t)$ for some
time dependent 1-form field $p_t\in T^*Q$. Let $f=y-u(t,q,y)$,
plugging this expression into the previous equation we obtain that
$u$ is independent of $y$ and is a subsolution to the
Hamilton-Jacobi equation. The converse is trivial by Remark
\ref{nyg}.
\end{proof}

The next result clarifies the connection between lightlike
hypersurfaces transverse to the flow of $n$ and the Hamilton-Jacobi
equation.

\begin{theorem} \label{lkp}
Let $V=I\times Q$, $I$ connected open subset of $T=\mathbb{R}$,  and
let $S$ be a $C^1$ hypersurface on $\pi^{-1}(V)$ which intersects
each integral line of $n$ in $\pi^{-1}(V)$ once and only once and
transversally. Then $S$ is the image of a map $(t,q)\to
(t,q,u(t,q))$ for some $C^1$ function $u: V\to \mathbb{R}$.
Moreover, $S$ is lightlike if and only if  $u$ satisfies the
Hamilton-Jacobi equation and has causal normals (i.e. with tangent
spaces which are spacelike or lightlike) if and only if $u$ is a
H.-J.\ subsolution.

 Conversely, given a $C^1$ function $u: V\to \mathbb{R}$  its graph regarded as a
subset of $M$ is a $C^1$  hypersurface which is lightlike if $u$
satisfies the Hamilton-Jacobi equation, while it has causal normals
if $u$ is just a subsolution.

In the lightlike case $S$ is generated by inextendible achronal
lightlike geodesics whose projections give maximal solution to the
E.-L.\ equations on $V$ which are action-minimizing between any pair
of  points. These projections are the {\em characteristics} in the
sense that they satisfy $p=\p L/\p \dot{q} =\p u/\p q$. Finally, $S$
is achronal and if $[t_0,t_1]\in I$, then
\[
u(t_1,q_1)\le \inf_{q, \, q(t_1)=q_1}
[u(t_0,q(t_0))+\int_{t_0}^{t_1} L(t,q(t),\dot{q}(t))\, \dd t]
\]
where the infimum is taken over the $C^1$ maps $q:[t_0,t_1]\to Q$,
such that $q(t_1)=q_1$. The equality is attained by some connecting curve, namely the infimum is a minimum equal to $u(t_1,q_1)$, if and only if the
characteristic passing through $(t_1,q_1)$ extends to the past up to
time $t_0$, in which case the minimum is attained on  that
characteristic (this is the case if the E.-L. flow is complete on
$[t_0,t_1]$).

\end{theorem}

Without completeness, the equality of $u(t_1,q_1)$ with the infimum does not imply that the infimum is attained. For $Q=(0,\infty)$, $L=\frac{1}{2}v^2$ and the global solution $u(t,x)=x-t/2$ of $u_t+\frac{1}{2}u_x^2=0$, with $t_0=0$ and $(t_1,x_1)=(1,1)$, the infimum over the curves ending at $x_1$ is $\inf_{x_0>0}[x_0+\frac{1}{2}(1-x_0)^2]=\frac{1}{2}=u(1,1)$, but it is not attained: the characteristic $x(t)=t$ leaves $Q$ at $t=0$.

In the previous statement it is understood that the properties of
inextendibility and maximality are referred to the portion of
spacetime  comprised in the interval $I$. We remark that we do not
assume neither that $Q$ is compact, nor that the E.-L. flow is
complete.

\begin{proof}
For simplicity we give the proof with $V=E$. The map $\phi:
\phi^{-1}(S)\to M$ defining the hypersurface is $C^1$ thus $\pi\circ
\phi:S\to E$ is $C^1$ and locally invertible (with $C^1$ inverse) by
transversality and the implicit function theorem. The inverse exists
globally, thus $\pi\circ \phi$ provides a $C^1$ diffeomorphism. The
$C^1$ map $\phi \circ (\pi\circ \phi)^{-1}: E\to M$ reads $(t,q)\to
(t,q,u(t,q))$ for some $C^1$ map $u: E\to \mathbb{R}$.

The hypersurface $S\subset M$ has equation $y-u(t,q)=0$, thus at any
point its tangent plane is the kernel of the $C^0$ 1-form $\dd
(y-u(t,q))=\dd y-\frac{\p u}{\p t}\dd t-\frac{\p u}{\p q}$. From
Eqs. (\ref{glg}) and (\ref{gcg}) we find that this 1-form is
lightlike if and only if $u$ satisfies the Hamilton-Jacobi equation,
and causal if and only if $u$ is a H.-J.\ subsolution.

The claim ``given a $C^1$ function $u: V\to \mathbb{R}$  its graph
regarded as a subset of $M$ is a $C^1$  hypersurface which is
lightlike if $u$ satisfies the Hamilton-Jacobi equation, while it
has causal normals if $u$ is just a subsolution'' is trivial given
the fact that this graph has equation $f=0$ with $f=y-u(t,q)$ and
given Theorem \ref{nce} and Remark \ref{nyg}.

Let us consider the lightlike case. Let $p\in S$, and let
$\gamma:(a',c')\to S$, $t\to \gamma(t)$, $\gamma(b)={p}$, $b\in
(a',c')$ be a maximal integral curve passing through $p$ of the
$C^0$ lightlike vector field $W$ tangent to $S$. Since $W$ is
transverse to the fibers and lightlike we have $\dd t(W)=-g(W,n)>
0$, thus the curve $\gamma$ can be assumed to be parametrized with
the semi-time function $t$ so that $\dot{\gamma}\propto W$ (since
$W$ is only $C^0$ and not necessarily Lipschitz we cannot use the
uniqueness of the solution to the Cauchy problem, but we shall find
in a moment that the curve, once parametrized with respect to $t$,
must be a geodesic and hence that it is uniquely determined).

Let us prove that $\gamma$ is achronal, and hence that it is an
inextendible lightlike geodesic (note that we do not assume that $u$
is $C^2$). We shall prove it by proving the achronality of $S$.

Let us suppose by contradiction that there are two events
$\check{p}=(a,\check{q}, u(a,\check{q}))$, and $\hat{p}=(c,\hat{q},
u(c,\hat{q}))$ in  $S$ which are chronologically related. There is a
$C^1$ timelike curve $\sigma: [a,c]\to M$ joining $\check{p}$ with
$\hat{p}$. We can assume that $\sigma$ is parametrized with $t$
because, being a timelike curve, its tangent vector has negative
scalar product with $n=-g^{-1}(\cdot, \dd t)$.

Let $\tilde{\sigma}=\pi\circ \sigma$, and let us write
$\tilde{\sigma}(t)=(t,q(t))$. Since in Eq. (\ref{mod}),
$-g(\dot{\tilde\sigma},\dot{\tilde\sigma})>0$, we have
$u(c,\hat{q})-u(a,\check{q})> \int_a^c L(t,q,\dot{q}) \dd t$.
However, since $u$ satisfies the Hamilton-Jacobi equation
\[
u(c,\hat{q})-u(a,\check{q})=\int_{\tilde{\sigma}} \dd u=\int_a^c
[\frac{\p u}{\p q}(\dot{q})-H(t,q,\frac{\p u}{\p q})]\dd t,
\]
thus by the Young-Fenchel inequality
\[
u(c,\hat{q})-u(a,\check{q})-\int_a^c L(t,q,\dot{q}) \dd t= \int_a^c
[\frac{\p u}{\p q}(\dot{q})-H(t,q,\frac{\p u}{\p
q})-L(t,q,\dot{q})]\dd t\le 0.
\]
The contradiction proves that $S$ is achronal and hence that
$\gamma$ is an achronal lightlike geodesic. Since $\dot{\gamma}$ is
lightlike it can be written $\p_t+v+L(t,q,v)n$, and the hyperplane
orthogonal (and tangent) to it is the kernel of the 1-form (theorem
\ref{nce}) $-(\frac{\p}{\p t}+v+L(t,q,v) \, n)^\flat=\dd
y-p+H(t,q,p)$ where $p=\p L/\p v$. But since $S$ has equation
$y-u(t,q)=0$ this 1-form must coincide with $\dd (y-u(t,q))$ (as
they have the same kernel and the same coefficient in $\dd y$) thus
$p=\p u/\p q$ on the projection of $\gamma$, i.e. the projections
are characteristics.

The fact that there can  be only one inextendible lightlike geodesic
passing through a point of an achronal hypersurface $S$, can be
easily proved showing that the presence of another geodesic, not
coincident with the first one, would imply that $S$ is not achronal
(through a typical corner argument: recall that any two events
causally related but not chronologically related, are joined only by
achronal lightlike geodesics).

For every $C^1$ map $q:[t_0,t_1]\to Q$, such that $q(t_1)=q_1$, and
every $\epsilon>0$, the event $(t_1,q_1,
u(t_0,q(t_0))+\int_{t_0}^{t_1} L(t,q(t),\dot{q}(t))\, \dd
t+\epsilon)$ stays in the chronological future of $(t_0,q_0,
u(t_0,q(t_0))$ (see Eq. (\ref{cgf1}) or (\ref{mod})) and for
$\epsilon=0$ it stays in the causal future of $(t_0,q_0,
u(t_0,q(t_0))$ (consider the light lift of $(t,q(t))$. Since $S$ is
achronal $u(t_1,q_1)\le u(t_0,q(t_0))+\int_{t_0}^{t_1}
L(t,q(t),\dot{q}(t))\, \dd t$. If the equality holds for some curve
$q(t)$ ending at $q_1$, namely if the infimum is attained, then its light lift gives a lightlike curve
connecting $x_0=(t_0,q_0, u(t_0,q(t_0))$ to $x_1=(t_1,q_1,
u(t_1,q(t_1))$, and since they belong to $S$ which is achronal, this
light lift must be a lightlike generator of $S$ (otherwise take
$x_0'<x_0$ along the generator passing through $x_0$,  and
$x_1'>x_1$ along the generator passing through $x_1$; then $x_0'\ll
x_1'$ which gives a contradiction) and its projection $(t,q(t))$ is
therefore a characteristic by the argument given above. Conversely,
if there is a characteristic $(t,q(t))$ connecting $(t_0,q_0)$ with
$(t_1,q_1)$, then
\begin{align*}
u(t_1,q_1)-u(t_0,q_0)&=\int_{(t,q(t))} \dd u=\int_{t_0}^{t_1}
[\frac{\p u}{\p q}(\dot{q})-H(t,q,\frac{\p u}{\p q})]\dd
t\\
&=\int_{t_0}^{t_1} [p(t)(\dot{q})-H(t,q,p(t))]\dd t=\int_{t_0}^{t_1}
L(t,q,\dot{q})\dd t,
\end{align*}
that is, the equality is attained on the characteristic.
\end{proof}

\begin{proposition}
Let  $V=I\times Q$, $I$ connected open subset of $T=\mathbb{R}$, and
suppose that $u:V\to \mathbb{R}$ is a $C^1$ subsolution to the
Hamilton-Jacobi equation, then the Hamilton's principal function
$S(e_0,e_1)$ is finite on $V^2$ whenever $t_0<t_1$. In particular,
$u(t_1,q_1)-u(t_0,q_0)\le S(e_0,e_1)$.
\end{proposition}

\begin{proof}
Let $e(t)=(t,q(t))$ be a ($C^1$) curve on the classical spacetime
$V$ connecting $e_0=(t_0,q_0)$ to $e_1=(t_1,q_1)$. We have
\[
u(t_1,q_1)-u(t_0,q_0)-\mathcal{S}_{e_0,e_1}[q]= \int_{t_0}^{t_1}
[\frac{\p u}{\p q}(\dot{q})+\frac{\p u}{\p t}-L(t,q,\dot{q})]\dd t.
\]
Since $u$ is a subsolution $\frac{\p u}{\p t}\le -H(t,q,\frac{\p
u}{\p q})$ thus
\[
u(t_1,q_1)-u(t_0,q_0)-\mathcal{S}_{e_0,e_1}[q]\le \int_{t_0}^{t_1}
[\frac{\p u}{\p q}(\dot{q})-H(t,q,\frac{\p u}{\p
q})-L(t,q,\dot{q})]\dd t\le 0.
\]
Where for the last step we used the Young-Fenchel inequality. Taking
the infimum of $u(t_1,q_1)-u(t_0,q_0)\le \mathcal{S}_{e_0,e_1}[q]$
over all possible connecting curves, we obtain
$u(t_1,q_1)-u(t_0,q_0)\le S(e_0,e_1)$, thus $S(e_0,e_1)\ne -\infty$.

\end{proof}

The next result clarifies that stable causality is a necessary
condition on $(M,g)$ for the existence of a $C^1$ time-local
solution to the Hamilton-Jacobi equation.


\begin{theorem} \label{bbb}
Let $V=I\times Q$, $I$ connected open subset of $T=\mathbb{R}$, and
suppose that $u:V\to \mathbb{R}$ is a $C^1$ subsolution to the
Hamilton-Jacobi equation, then the function $f:M\to \mathbb{R}$
defined by $f(t,q,y)=y-u(t,q)$ is a $C^1$ semi-time function, whose gradient is everywhere timelike
if and only if $u$ is a strict subsolution, in which case $f$ is a time function.
Furthermore,  suppose that $u$ is a $C^1$ subsolution, then for
every constant $\alpha>0$, the function $f+\alpha t$ is a $C^1$ time
function with timelike gradient.
 As a consequence, the spacetime
$\pi^{-1}(V)$ endowed with the induced metric, is a stably causal
spacetime.

Suppose $V=E$. If additionally the E.-L. flow on $E$ is complete,
(or, which is the same, $M$ is null geodesically complete (Theorem
\ref{rco})), then $M$ is globally hyperbolic  and on $V$ one can
actually find a smooth (i.e.\ $C^{r+1}$) strict subsolution.
\end{theorem}

\begin{proof}
Let $u$ be a $C^1$ subsolution. The $C^0$ 1-form $\dd f$ and hence
the vector $W=-g^{-1}(\cdot, \dd f)$ $=-\nabla f$ is causal  and
timelike if and only if $u$ is a strict subsolution (Eqs.
(\ref{glg}) and (\ref{gcg})). Moreover, $W$ is future directed
because the  scalar product with $n$ is -1 and hence negative. If
$\gamma: I\to \pi^{-1}(V)$, $s \to \gamma(s)$, is a future directed
causal curve then $\p_s f=\dd f(\gamma_s)= -g(W,\gamma_s)\ge 0$
which proves that $f$ is a semi-time function. If $u$ is a strict subsolution, $W$ is timelike and the inequality is strict, thus $f$ is a time
function with timelike gradient.

Let $\alpha>0$. We have $\p_s (f+\alpha t)=(\dd f +\alpha \dd
t)(\gamma_s)=-g(W,\gamma_s)+\alpha \dd t(\gamma_s)$. The term $\dd
t(\gamma_s)$ is non-negative because $t$ is a semi-time function,
and the term  $-g(W,\gamma_s)$ is strictly positive unless
$\gamma_s\propto W$. However, since $\dd f(n)=1$, $W$ is not
proportional to $n$, thus if $\gamma_s\propto W$ then $\dd t
(\gamma_s)>0$. We conclude that since  $\p_s (f+\alpha t)>0$, the
function $f+\alpha t$ is a $C^1$ time function with timelike
gradient. The existence of  such function implies that $\pi^{-1}(V)$
is stably causal \cite{hawking73}.

Let us prove that if we have a $C^1$ subsolution on $E$ and if  the
E.-L. flow on $E$ is complete then $M$ is globally hyperbolic. The
idea is to show that if $F=y-u(t,q)+\alpha t$ the set $F^{-1}(0)$,
necessarily acausal as $F$ is a time function, is in fact a Cauchy
hypersurface. Due to \cite[Property 6]{geroch70} we have just to
show that every inextendible lightlike geodesic on $M$ intersects
it. If the geodesic is an integral line of $n$ this is obvious. If
it is not then $t$ provides an affine parameter which takes all the
values in $\mathbb{R}$ and since $y-u(t,q)$ is a semi-time function
the conclusion follows from continuity.

Every globally hyperbolic  spacetime admits a smooth Cauchy temporal function
$T$, namely a smooth time function with timelike gradient whose level sets are Cauchy hypersurfaces \cite{bernal04,fathi12}. Every
integral line of $n$ is an inextendible causal curve, thus it intersects the Cauchy hypersurface $T^{-1}(0)$ exactly once, and transversally since $\nabla T$ is timelike. Thus $T^{-1}(0)$ is the graph $y=u(t,q)$ of a function $u$ which, by the implicit function theorem, is smooth. By Theorem \ref{lkp} (or Eqs.
(\ref{glg})-(\ref{gcg})) this is actually a strict subsolution as
its normals are timelike.
\end{proof}

\begin{remark}
A $C^1$ subsolution which is not strict can still give a time function $f=y-u$, whose gradient is then lightlike somewhere. Let $Q$ be a compact Riemannian manifold with metric $a$, $L=\frac{1}{2}a(v,v)$, so that $g=a-2\dd t\dd y$, and let $u(t,q)=-t^3/3$. Then $u_t+H(q,\dd_q u)=-t^2\le 0$, with equality at $t=0$, and $f=y+t^3/3$ satisfies $g^{-1}(\dd f,\dd f)=-2t^2$, which vanishes at $t=0$. Nevertheless $f$ is a time function: along a future directed causal curve $t$ and $y$ do not decrease, if the endpoints have different values of $t$ then $t^3$ increases strictly, and if they have the same value of $t$ the curve is a non-trivial future directed segment of an integral line of $n=\p_y$, along which $y$ increases strictly.
\end{remark}

\subsection{The light cone as the Monge cone for the H.-J.\ equation}

Let us consider a first order partial differential equation (PDE) on
$\mathbb{R}\times Q$
\begin{equation} \label{mhb}
F(t,q, a,b)=0, \qquad a=u_t, \ b=u_q,
\end{equation}
where $F$ is a $C^2$ function with the property $\p_a F\ne 0$. In
the case of the H.-J. equation we have
\begin{equation}
F=a+ H(t,q,b).
\end{equation}

Let $y=u(t,q)$ be a solution to the PDE (\ref{mhb}), and let
$(t_0,q_0, y_0)$ be a point on its graph, that is $y_0=u(t_0,q_0)$.
The tangent plane to the graph is the kernel of the 1-form $\dd
y-b-a \dd t$ where we regard $b=\p_q u$ as an element of $T^*Q$.
Furthermore, $a$ and $b$ are constrained at each point $(t_0,q_0,
y_0)$ as in Eq. (\ref{mhb}). As the pair $(a, b)$ solving Eq.
(\ref{mhb}) varies, the tangent planes at ${(t_0,q_0, y_0)}$
envelope a cone which is called the {\em Monge cone} of the
first-order PDE \cite{courant62b}. In our H.-J. case, with the
Hamiltonian given by Eq. (\ref{nqe}), the condition $F=0$ implies
that these planes are determined by the kernel of $\dd y-b+H(t,q,b)
\dd t$ and we already know, from the study of section \ref{ndx},
that they are tangent to the light cone of the Eisenhart's metric at
$(t_0,q_0, y_0)$. We conclude that the Monge cone coincides with the
light cone for the spacetime $(M,g)$.

According to theory of characteristics for the PDE (\ref{mhb}),
the Monge cone is tangent to the graph of any solution of the PDE.
The tangent vector at a point of the graph which belongs to the
intersection between the tangent plane to the graph and the Monge
cone determines a special direction, whose integral lines are called
{\em the characteristics} of the PDE solution. The method of
characteristics inverts this development and builds the solution
from the characteristics issued from the graph of the initial
condition \cite{courant62b,evans98}.

It is clear that the developments of the previous section fits this
general construction once the light cone and the Monge cone are
identified. Indeed, in the previous section we have found that the
graph of a solution  to the H.-J. equation is a lightlike
hypersurface. The characteristics are the lightlike geodesics
running on the lightlike hypersurface.

As we just mentioned, the method of characteristics allows us to
convert the PDE into a system of ordinary differential equations
(ODE). While the usual approach fixes a coordinate chart and works
locally in some space $\mathbb{R}^k$, we reduce here the PDE to an
ODE which determines curves $\Gamma: U \to T\times Q\times
\mathbb{R}\times T^*Q$, $U\subset T$. Each of them might be  called
{\em characteristic strip}. Its projection $\gamma: U\to T\times
Q\times \mathbb{R}$ is the {\em characteristic curve} (which is
tangent to the Monge cone) and the projection $c: U\to T\times Q$ is
the {\em base characteristic} \cite{benton77}. We take advantage of
the special form of the Lagrangian (Hamiltonian) to assign to $Q$ a
(time dependent) affine connection $D^t$ induced from $a_t$. It
makes sense to take derivatives of tensor fields with this
connection at any time. Using it the ODE for the curve $t\to
(t,q(t),y(t),p(t))$ obtained with the method of characteristics
\cite[Eq. 8.3, Chap. I]{benton77} reads
\begin{align}
\dot{t} &=1, \label{jdb}\\
\dot{q}&= a_t^{-1}(\cdot, p-b_t), \label{jdc}\\
\dot{y}& = L(t,q,\dot{q}), \label{mkk}\\
\frac{D^t}{\dd t} p&= (D^t  b_t)(a_t^{-1}(\cdot, p-b_t))-\p_q V.
\label{jdd}
\end{align}
It is understood that in last expression, if expressed in coordinate
form, the contravariant index of
 $a_t^{-1}(\cdot, p-b_t)(=\dot{q})$ is contracted
with the covariant index of $b_t$ and not with that of $D^t$.
Equations (\ref{jdc}) and (\ref{jdd}) are Hamilton's equations,
which joined together give the Euler-Lagrange equation
\begin{equation} \label{ele}
a_t(\cdot, \frac{D^t}{\dd t} \dot{q})= F_t(\cdot, \dot{q})-(\p_t
a_t)(\cdot,\dot{q})-( \p_t b_t+\p_q V).
\end{equation}
In this expression $F_t=\dd b_t$, where $d$ is the exterior
differentiation on $Q$ (thus $d$ does not differentiate with respect
to the time dependence of $b_t$).

\begin{remark}
The whole section \ref{ndx} and the above considerations could be
easily generalized to the case in which the field coefficients
entering the spacetime metric, $a_t$, $b_t$, $V$, depend also on the
extra coordinate $y$. The dependence of the Lagrangian and
Hamiltonian on these fields would not change. One would still
recover that the Monge cone of the differential equation
$u_t+H(t,q,u,u_t,u_q)=0$ is $g$ and hence a Lorentzian cone. In this
case the characteristics are still null geodesics but on the
quotient $E$ the base characteristics are interpreted as solutions
to a problem of control. Indeed, in this case Eq. (\ref{mkk}) is no
more decoupled with the other equations. We shall leave this
interesting generalization for future work.
\end{remark}

Let $Q_{t_0}=\{t_0\}\times Q$ and let $u_{t_0}: Q_{t_0}\to
\mathbb{R}$ be a $C^2$ function. The subset of $T\times Q\times
\mathbb{R}\times T^*Q$ given by $\mathscr{S}_0=\{(t_0,q, u_{t_0}(q),
\p_q u_{t_0}(q)): q\in Q\}$ provides the initial condition for the
ODE above. The method of characteristics consists integrating the
ODE and in proving that $y$, as a function of the base point, is a
solution to the PDE, at least in some neighborhood of the  initial
base manifold $Q_{t_0}$. In this respect it is useful to note that
the proof of existence and uniqueness works also non-locally
provided: (i) the flow on $E$ obtained with the method of
characteristics has non singular Jacobian, that is, provided one
excludes focusing points; (ii) one localizes the solution in a
region of $E$ that can be reached by the characteristics. More
precisely, with the method of characteristics it is possible to
prove the following  theorem whose proof does not differ
significatively from the standard ones
\cite{hartman02,courant62b,caratheodory65,evans98,benton77}.
Unfortunately, the given references do not formulate it with this
degree of generality.

\begin{theorem} \label{bkz}
Let $Q_{t_0}=\{t_0\}\times Q$ and let $u_{t_0}: Q_{t_0}\to
\mathbb{R}$ be a $C^2$ function. Let $e_0=(t_0,q_0)\in Q_{t_0}$ and
and let $\psi(t,q_0)$ be the base characteristic curve passing
through $q_0$. There is an open neighborhood $W\subset E$, $W\supset
Q_{t_0}$, with the property that for each  $e_1=(t_1,q_1)\in W$
there is one and only one $e_0\in Q_{t_0}$ such that
$e_1=\psi(t_1,q_0)$ and the base characteristic connecting $e_0$ to
$e_1$ is entirely contained in $W$. The map $\psi:(t,q_0)\to
(t,q(t,q_0))$ is such that $q(t,q_0)$ is differentiable with respect
to $q_0$ and of maximum rank (i.e. it is a local diffeomorphism).
For every open set $V$ with these same properties there is a a
unique $C^2$ function $u: V\to \mathbb{R}$ which solves the H.-J.
equation with initial condition $u(t_0,q)=u_{t_0}(q)$. This function
is obtainable with the method of characteristics.
\end{theorem}

We remark that $V$ does not need to be projectable on $T$, that is,
it is not necessarily of the form $\pi_T^{-1}(\pi_{T}(V))$. For this
reason, this theorem proves the existence and uniqueness of
solutions to the H.-J. equation only in a spacetime-local sense.
Indeed, a time-local version would certainly require more
assumptions for otherwise, according to Theorem \ref{bbb}, any
generalized gravitational wave spacetime $(M,g)$ would be stably
causal, which is not true.

The next corollary clarifies the good local causal behavior of the
spacetimes under study.
(actually we could prove causal continuity using some later results).

\begin{corollary}
On $(M,g)$ every slice $N_t$ admits a projectable neighborhood
$\pi^{-1}(V)$, $Q_t\subset V\subset E$, which is stably causal.
\end{corollary}

\begin{proof}
Follows at once from Theorem \ref{bbb} and the fact that we can
construct a $C^2$ solution of the H.-J.\ equation over some open
neighborhood $V$ of $Q_t$ using the method of characteristics.
\end{proof}

\begin{corollary}
Under the assumption of Theorem \ref{bkz}, if $Q$ is compact then
the $C^2$ functions there cited are defined on projectable
neighborhoods, that is $u(t,q)$ solves the H.-J.\ equation
time-locally.
\end{corollary}

\begin{proof}
Every point $e\in Q_{t_0}$ is contained in a rectangular open set
$O(e)=(\check{t}(e), \hat{t}(e))\times U(e)$, $O(e)\subset V$, where $U(e)\subset Q$
is an open set, $t_0\in (\check{t}(e), \hat{t}(e))$ and $O(e)\subset
V$. The compact set $Q$ can be covered with a finite number of sets
of the form $U(e_i)$, then defined $\check{t}=\max \check{t}(e_i)$
and $\hat{t}=\min \hat{t}(e_i)$, we have that
$V'=\pi_T^{-1}((\check{t}, \hat{t}))\subset V$  is projectable.
\end{proof}

\begin{remark}
Let us denote with $C^{1,Lip}$ the space of differentiable functions
with locally (uniformly) Lipschitz partial derivatives. The optimal
version of Theorem \ref{bkz} is due to Severini \cite{severini16},
Wazewski \cite{wazewski35,wazewski38} and Digel \cite{digel38}. It
is obtained by replacing $C^2$ for $u_{t_0}$ and $u$ with
$C^{1,Lip}$, and by assuming that on $W$ the map $\psi: (t,q_0)\to
(t,q(t,q_0))$ is such that $q(t,q_0)$ is a local (uniform)
Lipomorphism for any given $t$ (see also
\cite{hartman02,lions82,pelczar91} and the references of
\cite{crandall83}). Often this theorem is formulated with stronger
assumptions in order to obtain time-local solutions
\cite{szarski59}. It seems to us that a relatively simple
proof of this theorem could pass through the Lipschitz version of
Frobenius theorem \cite{simic96,rampazzo07}.
\end{remark}

\subsection{Gauges, reference frames, and the geometrization of dynamics} \label{gau}

In this work we have first introduced the Lagrangian problem, and
then we have built a classical spacetime $E=T\times Q$,
$T=\mathbb{R}$, and an extended relativistic spacetime $M=E \times
\mathbb{R}$, as tools to study it. Nevertheless, we have mentioned
that we can, in fact, follow a  different path.

Indeed, we can start from a spacetimes $M$ which admits a
covariantly constant null vector $n$ with open $\mathbb{R}$ orbits,
and in fact such that $M$ is turned into an abelian $(\mathbb{R},+)$
bundle over a quotient space $E$. The space $E$ is then interpreted
as the classical spacetime, and on it one can naturally define a
function $t: E\to T$, defined through $\dd t=-g(\cdot, n)$, which is
interpreted as a classical time with its absolute simultaneity
slices. A complete splitting of $E$, $\pi_Q: E\to Q$, $E\sim T\times
Q$, is provided by a complete vector field $v: E\to TE$ such that
$\dd t(v)=1$ (the Newtonian flow). This field represents a flow,
which defines a frame of reference, namely it specifies the motion
of the points which we are going to regard as `at rest' with respect
to the frame. The diffeomorphism  $E\sim T\times Q$ is thus not a
natural one, in fact it depends, or better it defines, the frame
chosen where $Q$ has to be interpreted as the ``body space'' or the
``reference frame space'' (Fig.\ \ref{tw}).

\begin{figure}[ht]
\begin{center}
\includegraphics[width=8cm]{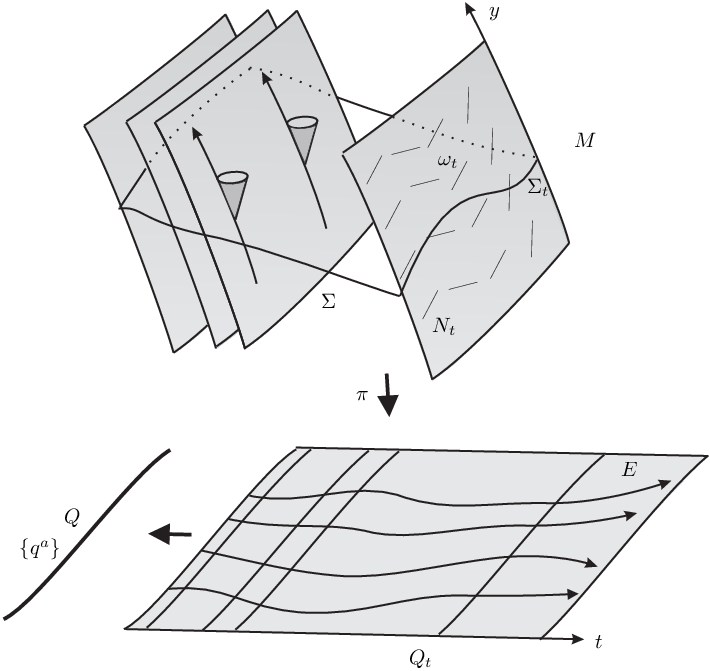}
\end{center}
\caption{The flow on the base $E$ generates the splitting
$E=\mathbb{R}\times Q$. The figure depicts the section $\Sigma:E\to
M$, of the fiber bundle $\pi:M\to E$, and its slices
$\Sigma_t:Q_t\to N_t$. The 1-forms $\omega_t=\dd y-b_t$ over $N_t$
are represented through their kernel.} \label{tw}
\end{figure}

This discussion clarifies that while the Lagrangian function depends
on the choice of coordinates, namely in the way we split the trivial
bundles $\pi:M\to E$ and $t:E\to T$, the dynamics, captured in the
spacetime lightlike geodesics, is independent of such choice.
Therefore, it is natural to investigate whether there are
particularly simple choices for those splittings which simplify the
dynamics. We shall devote this section to  answer to this problem.

\subsubsection{Change of gauge}

The chosen splitting of the fiber bundle $\pi:M\to E$ will be
referred to as a ``gauge'' and the change of splitting as a ``change
of gauge''. A change of gauge amounts to a redefinition of the
coordinate $y$, namely there is a function $\alpha: \mathbb{E}\to
\mathbb{R}$ such that
\begin{equation} \label{jiy}
y'=y+\alpha(t,q),
\end{equation}
whereas the other coordinates are left unchanged: $t'=t$, $q'=q$.
Since our old coordinates of the form $(t,q,y)$ provided a $C^{r+1}$
atlas, the new coordinates of the form $(t,q,y')$ provide a
$C^{r+1}$ atlas only if $\alpha$ is $C^{r+1}$, otherwise the new
atlas is a $C^k$ atlas, $k<r+1$, where $k$ is the degree of
differentiability of $\alpha$.

Under the above change of section the components of the spacetime
metric change as follows
\begin{align}
a'_t&=a_t, \\
b_t'&= b_t+\p_q \alpha, \\ \label{nhg}
 V'&=V-\p_t\alpha.
\end{align}
The Lagrangian changes by a `total differential' (as expected, this
change does not affect the action and hence the dynamics)
\begin{equation}
L'(t,q,v)=L(t,q,v)+[\p_q\alpha (v) + \p_t\alpha] =L+\frac{\dd }{\dd
t} \alpha.
\end{equation}
The Hamiltonian description changes as follows
\begin{align}
p'&=p+\p_q \alpha,\\
 H'(t,q,p')&=H(t,q,p'-\p_q \alpha)-\p_t \alpha. \label{bkc}
\end{align}

\subsubsection{Change of reference frame}

Coming to the   bundle $\pi:E\to T$, the chosen splitting will be
referred to as ``reference frame'' or ``observer'' and the change of
splitting as a ``change of reference frame'' or a ``change of
observer''. We remark that the splitting is induced by a projection
$\pi_Q:E\to Q$ on a quotient manifold, rather than  by a section
$\sigma: T\to E$, because $\pi:E\to T$ is not a principal bundle. To
change splitting means to change projection $\pi_{Q'}:E\to Q'$. Of
course, $Q$ and $Q'$ are diffeomorphic, as for any $t\in T$, $Q$ is
diffeomorphic to $Q_t$, $Q'$ is diffeomorphic to $Q'_t$ and
$Q_t=Q'_t$. On $Q'$ we are given coordinate charts $\{ q^{k'} \}$,
and the time dependent coordinate transformation (diffeomorphism)
$\nu_t: Q' \to Q$ can be written
\begin{equation} \label{coo}
q'=q'(t,q)=\nu_t^{-1}(q),
\end{equation}
whereas the other coordinates are left unchanged: $y'=y$, $t'=t$.
Let $\hat{v}(t,q)$ be the velocity of the Newtonian flow associated
to $Q'$ as seen from $Q$. That is, if we write the inverse map as
$q=q(t,q')=\nu_t(q')$, then $\hat{v}=\p q/\p t$ (this is an element
of $TQ$, and by means of the diffeomorphism, it can be regarded as
an element $\nu_{t*}^{-1} \hat{v}$ of $TQ'$). This change modifies
the fields entering the Lagrangian and the spacetime metric as
follows
\begin{align*}
a'_t&=\nu_t^{*} a_t, \\
b'_t&=\nu_t^{*} [b_t+a_t(\cdot,\hat{v})],\\
V'&=\nu_t^{*}[V-b_t(\hat{v})-\frac{1}{2} a_t(\hat{v},\hat{v})].
\end{align*}
If the diffeomorphism (\ref{coo}) is $C^k$, $1\le k\le r+1$, then
the new fields are $C^{k-1}$. The Lagrangian changes as follows
\[
L'(t,q', v')=L(t,\nu_t(q'),\nu_{t*} v'+\hat{v}),
\]
and the Hamiltonian description changes as follows
\begin{align*}
p'&=\nu_t^{*} p, \\
H'(t,q',p')&=\nu_t^{*}
[H(t,q,\nu_t^{-1 *} p')-(\nu_t^{-1 *} p')(\hat{v})].
\end{align*}

\subsubsection{Geometrization of dynamics  through \texorpdfstring{$C^2$}{C2} solutions of the H.-J. equation}


The existence of a solution to the Hamilton-Jacobi equation
determines a distinguished section for the bundle $\pi: M\to E$, and
the base characteristics provide a flow on $E$ which determines a
projection $\pi_{Q}:E\to Q$, and hence a splitting of the bundle
$\pi_T:E\to T$. In this section we wish to show that we can take
advantage of these splittings to simplify the dynamics.

Let us first  use the arbitrariness in the choice of gauge.



\begin{theorem}
Let $u:V\to \mathbb{R}$ be a $C^k$, $1\le k\le r+1$, solution of the
H.-J. equation where $V\supset Q_{t_0}$ is an open set with the
properties enumerated in Theorem \ref{bkz} (there is at least one
such neighborhood). Let us redefine $y\to y'=y-u(t,q)$, then
$a_t'=a_t$ and the new ($C^{k-1}$) Lagrangian takes the Ma\~ne form
\[
L'(t,q,v)=\frac{1}{2}a_t(v+{b'}_t^\sharp,v+{b'}_t^\sharp),
\]
where ${b'}_t^\sharp=a_t^{-1}(\cdot, b'_t)$. The Hamiltonian takes
the form
\[
H'(t,q,p')=\frac{1}{2}a_t^{-1}(p',p')-a_t^{-1}(p',b_t').
\]
\end{theorem}

\begin{remark}
In the $k=1$ case the Lagrangian $L'$ is only $C^0$ (but smooth in
$v$). Nevertheless, the Lagrangian problem makes still sense since
$L'$ being continuous is integrable, and the dynamics depends on the
minimization of the action. One can therefore write the action as
usual, express it in terms of the old Lagrangian, and show from
there a number of results such as the $C^{r+1}$ nature of the
stationary points.
\end{remark}

\begin{proof}
We already know that under a  change of gauge, $a'_t=a_t$. The
change of gauge is $y'=y+\alpha(t,q)$ with $\alpha(t,q)=-u(t,q)$. We
can rewrite equation (\ref{bkc}) as
\begin{align*}
\frac{1}{2} {a}_t^{-1}(p'-b'_t, p'-b'_t)+V'(t,q)&=
H'(t,q,p')=H(t,q,p'+\p_q u)+\p_t u\\&=H(t,q,p'+\p_q u)-H(t,q,\p_q
u).
\end{align*}
Plugging $p'=0$ into the equation we obtain the identity
$\frac{1}{2} {a}_t^{-1}(b'_t, b'_t)+V'=0$. Thus the new Lagrangian
is
\[
L'(t,q,v)=\frac{1}{2}\, a_t(v,v) +b'_t(v)+\frac{1}{2}
a_t({b'}_t^\sharp,{b'}_t^\sharp)=\frac{1}{2}a_t(v+{b'}_t^\sharp,v+{b'}_t^\sharp).
\]
Equation (\ref{nhg}) shows that $b'_t$ is $C^{k-1}$ thus $L'$ is
$C^{k-1}$.
 The derivation of the Hamiltonian is straightforward.
\end{proof}

It must be noted that an Hamiltonian of the form $H'$ admits the
constant functions as solutions to the Hamilton-Jacobi equation.
Indeed, the lightlike hypersurface of equation $y=u(t,q)$, under the
change of coordinate, is determined by the new equation $y'=0$.

The mentioned lightlike hypersurface is generated by lightlike lines
which project into maximal solution to the E.-L.\ equations. The
idea is to use this flow of characteristics as the reference frame.
We expect that this choice could simplify the motion of particles
`moving outside the flow'. This flow could be defined just in the
neighborhood of a point $e\in E$ of interest, thus the same is true
for the required solution to the H.-J.\ equation. For simplicity we
give the proof of the next theorem in the case in which we have a
solution of the H.-J.\ equation in a time-local sense, that is,
defined over a whole connected open interval $I$ of the real line.


\begin{theorem} \label{ndp}
Let $V=[t_0,t_1]\times Q$, and suppose that $u:V\to \mathbb{R}$ is a
$C^k$, $2\le k\le r+1$, solution of the H.-J. equation. Suppose that
the base characteristics induced on $V$ by $u$ are defined on the
whole interval $[t_0,t_1]$ (this is assured if the E.-L.\ flow is
complete on the interval, which in turn is the case if $Q$ is
compact). Identify $Q_{t_0}$ with the body frame manifold $Q'$ of
this (base) characteristic flow $\nu_t: Q'\to Q$. Then $\nu_t: Q'\to
Q$  is a $C^{k-1}$ diffeomorphism (including the time dependence)
and redefined
\begin{align}
y'&=y-u(t,q),\\
q'&=\nu_t^{-1}(q),
\end{align}
the Lagrangian in the new variables reads
\[
L'(t,q',v')=\frac{1}{2}a'_t(v',v').
\]
where $a_t'$ is $C^{k-2}$. That is, the solutions of the E.-L.\
equations, as seen from the new frame, appear as geodesics on the
new time dependent Riemannian geometry of metric
$a'_t=\nu_t^{*}a_t$.
\end{theorem}

\begin{proof}
Let us identify $Q'$ with $Q_{t_0}$. The characteristics are
solutions to the E.-L.\ equations, pass through every point of $V$,
and establish a one to one correspondence between $Q_{t_0}$ and
$Q_{t}$. Since $Q$ can be identified with any $Q_t$, this
correspondence is the map $\nu_t: Q'\to Q$. This is a $C^{k-1}$ flow
of equation $\dot{q}=a_t^{-1}(\cdot, \p_q u-b_t)$, thus the map
$\nu_t$ is $C^{k-1}$ (see \cite[Theor.\ 17.19]{lee03}) and mixed
derivatives involving $t$ exist and are continuous up to the $k$th
order \cite[Chap.\ 5, Cor.\ 3.2]{hartman64}, that is, both $\nu_t,
\p_t \nu_t: Q'\to Q$ are $C^{k-1}$. The metric $a'_t$ is therefore
$C^{k-2}$. After the change of gauge $y'=y-u$ the previous theorem gives $V=-\frac{1}{2}a_t^{-1}(b_t,b_t)$ and the velocity of the characteristic flow is $\hat v=-b_t^\sharp$ (we denote the new fields omitting the prime). Thus by the formulas for the change of reference frame, $b_t+a_t(\cdot,\hat{v})=0$ and $V-b_t(\hat{v})-\frac{1}{2} a_t(\hat{v},\hat{v})=0$, that is, in the new frame the 1-form field and the potential vanish.
\end{proof}

According to Theorem \ref{bkz}, given any event $e_0\in E$,
$e_0=(t_0,q_0)$, by choosing a sufficiently smooth initial condition
$u_{t_0}$ it is indeed possible to find a $C^2$ solution $u(t,q)$ of
the H.-J.\ equation which is defined over a neighborhood of $e_0$.
As a consequence, at any event $e_0$ we can observe the motion from
the frame given by the characteristics of $u$ to conclude that it
locally looks like a geodetic motion in an time dependent geometry.

It must be stressed that, generically, the frame will be local in
space unless we can find a time-local (rather than spacetime-local)
$C^2$ solution to the H.-J. equation, and it will last only a finite
time interval  because the flow might develop caustics or some
characteristic might go to infinity in a finite time (blow up).
%
%
%
%
%
%

\subsubsection{Motion in a time dependent geometry}
The previous section suggest to study the  E.-L.\ equation for the
Lagrangian $L(t,q,v)=\frac{1}{2}a_t(v,v)$, that  is (see Eq.
\ref{ele}),
\begin{equation} \label{ela}
\frac{D^t}{\dd t} \,\dot{q}= -(a_t^{-1}\p_t a_t)(\dot{q}),
\end{equation}
where $D^t$ is the Levi-Civita connection for $a_t$. A special case
is that of a uniformly expanding or contracting geometry:
$a_t=s^2(t) a$, where $s(t)>0$ is a scale factor and $a$ is a
Riemannian metric. In this case $D^t$ is independent of time and
coincides with the affine connection for $a$. The equation of motion
becomes
\begin{equation}
\frac{D}{\dd t} \dot{q}= -2(\p_t \ln s) \dot{q},
\end{equation}
and redefined $\eta=\int_0^t \frac{\dd r}{s(r)^2}+cnst.$, we obtain
\begin{equation}
\frac{D}{\dd \eta} \frac{\dd q}{\dd \eta} = 0.
\end{equation}
In other words, with respect to the  comoving observer, the
expansion of geometry can be removed with a redefinition of time. By
using a convenient time parameter $\eta$, the motion of neighboring
particles which move outside the flow appears as geodesic (this
result can also be understood rewriting the spacetime metric as
$g=s^2(a-\frac{\dd t}{s^2}\otimes \dd y-\dd y\otimes \frac{\dd
t}{s^2})$ and recalling that null geodesics get just reparametrized
under conformal changes).

More generally, the right-hand side of Eq.\ (\ref{ela}) provides the
acceleration induced from the dynamics of geometry, and in order to
appreciate its effect one can choose coordinates $q$ so at to
diagonalize the metric $a_t$ at the event $e$ of interest. Once
$a_t$ is made Euclidean at $e$ one can decompose the matrix
$(a_t^{-1}\p_t a_t)$ in its antisymmetric, symmetric and traceless,
proportional to the identity, components where the latter represents
the expansive term. One can work out the effect of each term. For
instance, the antisymmetric part induces a kind of Coriolis force.

\subsection{Local and global existence of Rosen coordinates}
\label{ros}

The subject of this work is the study of the manifold $M=T\times {Q}
\times \mathbb{R}$ endowed with the Brinkmann's (Eisenhart's) metric
\begin{equation}
g=a_t \!-\!\dd t \otimes (\dd y-\!b_t) -\!(\dd y-\!b_t) \otimes \dd
t-2V \dd t ^2 .
\end{equation}
A well known question is whether this spacetime metric can be
rewritten, under a change of variables, in the simplified Rosen form
\begin{equation} \label{nnh}
g=a'_{t} \!-\!\dd t \otimes \dd y' -\!\dd y' \otimes \dd t.
\end{equation}
where $\p_{y'})_{t,q'}$ is still the null Killing field
$n=\p_{y})_{t,q}$.  It must be noted that under the assumption
$\p_y=\p_{y'}$, the coordinate $t$ does not change because $\dd
t=-g(\cdot, \p_y)=-g(\cdot, \p_{y'})=\dd t'$, namely $t$ and $t'$
differ by an irrelevant additive constant. We shall speak of {\em
Rosen coordinates} $(t,q',y')$, but we do not mean with this
terminology that a single coordinate patch suffices to cover $Q$.

A proof that, at least locally, this simplification can indeed be
accomplished can be found in \cite[Chap.\ 10]{zakharov73},
\cite{blau09}, \cite[Chap.\ 4]{griffiths91}, \cite[Sect.\
20.5]{dinverno92}, \cite[Sect.\ 24.5]{stephani82}. These proofs are
given under some additional assumptions on the Brinkmann form (four
dimensionality, time independence of the space metric, flatness of
the space metric, $b_t=0$, quadratic dependence of $V(t,q)$ on $q$,
etc.). These restriction arise naturally if one imposes the vacuum
Einstein equations (which we do not impose). The mentioned proofs
are quite technical. As we shall realize in a moment, this
transformation problem is in fact quite geometrical although, to
appreciate it, the  general framework of this paper will be
required.

Unfortunately, the fact that such local transformation can hardly be
globalized is not so often mentioned. A relevant exception is
Penrose \cite{penrose76}. He reminds us  that Rosen showed that any
non-flat vacuum metric of the form (\ref{nnh}) necessarily
encounters singularities when one attempts to extend the range of
coordinates in order to obtain a geodesically complete spacetime.
Unfortunately, Rosen interpreted it as evidence that gravitational
plane waves do not exist in general relativity. It was later shown
by Robinson  \cite{bondi59c} that these singularities are due to a
mere coordinate effect  and that in the four-dimensional case,
vacuum metrics admit a global coordinate chart in which they take
the form
\[
g=\sum_{i=1}^{2} (\dd q_i)^2  \!-\!\dd t \otimes \dd y-\!\dd y
\otimes \dd t-2(\sum_{i=1}^{2} h_{ij}(t) q_i q_j) \,\dd t ^2,
\]
that is, they  fall into the class of Brinkmann's metrics studied in
this work.

Let us return to our spacetimes $(M,g)$. We are going to show that
the possibility of finding local (global) Rosen coordinates is
equivalent to the possibility of finding local (resp. global)
solutions to the Hamilton-Jacobi equation. The existence of local
solutions to the H.-J.\ equation is then assured by the method of
characteristics.

\begin{theorem} \label{rosthm}
At any point $p\in M$, the spacetime $(M,g)$ admits ($C^{r-1}$)
coordinates in a neighborhood of $p$ which allow us to rewrite $g$
in the Rosen form (\ref{nnh}) with a $C^{r-2}$ space metric.

The spacetime $(M,g)$  admits global (time-local at $t=t_0$) Rosen
coordinates if and only if there is a global (resp. time-local at
$t=t_0$) solution of the H.-J.\ equation with Hamiltonian
(\ref{nqe}) whose base characteristics are complete (resp. defined
on a common time interval neighborhood of $t_0$). (for the details
on the differentiability properties see the proof)
\end{theorem}

\begin{proof}
Let $p=(t_0,q',y)$. Since the Lagrangian is $C^r$, $r\ge 2$, the
flow $(t,q')\to (t,q(t,q'))$ induced by the base characteristics is
$C^{r-1}$.  We can always choose an initial condition $u_{t_0}$,
$u_{t_0}(q')=y$, which is $C^{r}$ and prove that the obtained
solution of the H.-J.\ equation is $C^r$ in a neighborhood of
$(t,q')$. The new coordinate $y'=y-u(t,q)$ is therefore $C^r$ while
the map $q(t,q')$ is $C^{r-1}$ together with its (local) inverse
$q'(t,q)$. The fact that  with this change  the metric simplifies to
the Rosen form follows from the spacetime-local version of Theorem
\ref{ndp}.

As for the last statement, we give the proof in the global case, the
time-local case being analogous. Suppose that there are $C^1$
coordinates through which the metric can be written in Rosen form
for some continuous metric coefficients. The hypersurface
$\mathcal{N}$ of equation $y'=0$ is $C^1$, lightlike and transverse
to $\p_{y'})_{t,q'}=n$. Let us look at this hypersurface using   the
original Brinkmann coordinates. According to Theorem \ref{lkp} there
is a global $C^1$ solution of the H.-J.\ equation whose graph
coincides with $\mathcal{N}$. Using again the Rosen form, the base
characteristics of this solution are the curves $q'=const.$ thus, by
assumption, they are defined on the whole time axis.

Conversely, suppose that $u(t,q)$ is a $C^k$, $2\le k\le r+1$,
solution of the H.-J.\ equation and that the base characteristics
are complete. According to Theorem \ref{ndp} there are $C^{k-1}$
coordinates $y',q'$ which bring the metric in Rosen form.

\end{proof}

\section{The causal hierarchy for spacetimes admitting a parallel
null vector} \label{s:cau}

We have already pointed out that  the spacetime $(M,g)$ is causal.
In this section we wish to establish the position of the generalized
gravitational wave spacetime $(M,g)$ in the causal hierarchy  of
spacetimes \cite{hawking73,minguzzi19}. We shall see that, at least
for the lower levels, the spacetime is  the more causally well
behaved the better the continuity properties of the associated least
action $S$. The identities
$\overline{I^{\pm}(x)}=\overline{J^{\pm}(x)}$,
$\overline{I}=\overline{J}$, will be used without further
mention \cite{hawking73,minguzzi19}. Since the causality results
depend only on the conformal class of the metric, most of the
results of this section will immediately extend to spacetimes which
are conformal to those considered here.


A spacetime is non-total imprisoning if no future inextendible
causal curve can be contained in a compact set. Replacing {\em
future} with {\em past} gives an equivalent property
\cite{beem76,minguzzi07f}. Every distinguishing spacetime is
non-total imprisoning and every non-total imprisoning spacetime is
causal \cite{minguzzi07f}.

\begin{theorem} \label{blf}
The spacetime $(M,g)$ is non-total imprisoning.
\end{theorem}

\begin{proof}
Suppose, by contradiction, that there is a future (or past)
inextendible causal curve contained in a compact set $C$,  then,
according to \cite[Theor.\ 3.9]{minguzzi07f}, there is an
inextendible achronal lightlike geodesic $\gamma$ entirely contained
in $C$ with the property that, chosen $p\in \gamma$ and $q\in
\gamma$ with $q< p$, the portion of $\gamma$ after $p$ accumulates
on $q$, in particular $q \in \overline{J^{+}(x)}\backslash\{p\}$.
The geodesic $\gamma$ cannot coincide with an integral line of $n$
because $y$ is continuous and would increase along the curve. In the
other cases $t$ provides an affine parameter for $\gamma$, it is
$t(q)<t(p)$, and since $t$ is continuous and cannot decrease along a
causal curve we find again that this case does not apply. The
contradiction proves that $(M,g)$ is non-total imprisoning.
\end{proof}

\begin{lemma} \label{prf}
For every $e_0,e_1\in E$ we have
\begin{align}
\liminf_{e \to e_1}\, S(e_0,e)&= S(e_0,e_1) \quad \textrm{or}\quad
-\infty , \label{bpo1}\\
\liminf_{e \to e_0}\, S(e,e_1)&= S(e_0,e_1) \quad \textrm{or}\quad
-\infty , \label{bpo3}\\
\liminf_{(e,e') \to (e_0,e_1)} S(e,e')&= S(e_0,e_1) \quad
\textrm{or} \quad -\infty. \label{bpo2}
\end{align}
Moreover,
\begin{equation}
\overline{J}=\{(x_0,x_1): y_1-y_0\ge \liminf_{(e,e') \to
(e_0,e_1)} S(e,e') \}. \label{prf2}
\end{equation}
For every $x_0=(e_0, y_0) \in M$,
\begin{align}
\overline{J^{+}(x_0)} &= \{x_1: y_1-y_0\ge \liminf_{e \to e_1}
S(e_0,e)
 \},\label{prf1} \\
\overline{J^{-}(x_1)} &= \{x_0: y_1-y_0\ge \liminf_{e \to e_0}
S(e,e_1)
 \}.\label{prf3}
\end{align}
\end{lemma}

\begin{proof}
Proof of Eq.\ (\ref{prf1}), the proof of Eq.\ (\ref{prf3}) being
analogous. Let us first prove the inclusion $\{x_1: y_1-y_0\ge
\liminf_{e \to e_1} S(e_0,e) \}\subset \overline{J^{+}(x_0)}$. Let
$x_1$ be such that $y_1\ge y_0+ \liminf_{e \to e_1} S(e_0,e)$ so
that $\liminf_{e \to e_1} S(e_0,e)\ne +\infty$.

There are two cases, either $\liminf_{e \to e_1} S(e_0,e)$ is finite
or it is $-\infty$. In the former case let  $\epsilon>0$ and let
$U\ni e_1$ be an open set, then there is $\hat{e}\in U$ such that $
S(e_0,\hat{e}) < \liminf_{e \to e_1} S(e_0,e)+\epsilon$. Note that
we can assume $S(e_0,\hat{e})\ne -\infty$ otherwise it would be,
from the arbitrariness of $U$, $\liminf_{e \to e_1}
S(e_0,e)=-\infty$. Since $\liminf_{e \to e_1} S(e_0,e)\ne +\infty$,
we have that $S(e_0,\hat{e})<+\infty$ and hence $\hat{t}> t_0$ or
$\hat{e}=e_0$. We can assume that we can always choose $\hat{e}\ne
e_0$ otherwise from the arbitrariness of $U$ and $\epsilon$,
$e_0=\hat{e}=e_1$ and taking the limit of the inequality, $
S(e_0,\hat{e}) < \liminf_{e \to e_1} S(e_0,e)+\epsilon$, we get
$0=S(e_0,e_1)\le \liminf_{e \to e_1} S(e_0,e)$, thus $y_1\ge y_0$
and the points $(e_0,y_1)$ with $y_1\ge y_0$ are included in
$J^{+}(x_0)$. Thus let us assume the other case, $\hat{t}> t_0$. By
Eq.\ (\ref{cgf1}) the point $\hat{x}=(\hat{e}, y_1+\epsilon)$
belongs to $I^{+}(x_0)$. Since $\epsilon$ and $U$ are arbitrary,
$x_1\in \overline{I^{+}(x_0)}=\overline{J^{+}(x_0)}$.

If  $\liminf_{e \to e_1} S(e_0,e)=-\infty$ then there is a sequence
$\hat{e}_n \to e_1$ such that $S(e_0,\hat{e}_{n}) \to -\infty$.
Since $S(e_0,\hat{e}_{n})<+\infty$, $t_n>t_0$ or $\hat{e}_n=e_0$. We
can assume that the latter possibility does not apply for no value
of $n$ because if there were a subsequence $\hat{e}_k$ with that
property $S(e_0,\hat{e}_{k})=S(e_0,{e}_{0})=0$ and could not
converge to $-\infty$. By Eq.\ (\ref{cgf1})  the points
$\hat{x}_n=(\hat{e}_n,y_1)$ are such that for sufficiently large
$n$, $\hat{x}_{n}\in I^{+}(x_0)$ but $\hat{x}_n \to x_1$, thus
$x_1\in \overline{I^{+}(x_0)}=\overline{J^{+}(x_0)}$.

For the converse, let $x_1\in
\overline{J^{+}(x_0)}=\overline{I^{+}(x_0)}$. This means that there
is a sequence of points $\hat{x}_n \in I^{+}(x_0)$ such that
$\hat{x}_n \to x_1$. By Eq.\ (\ref{cgf1}) $\hat{y}_n-y_0>
S(e_0,\hat{e}_n)$, and since $\hat{e}_n\to e_1$, $\liminf_{n\to
+\infty} S(e_0,\hat{e}_n)\le y_1-y_0$ from which thesis follows.

Proof of Eq.\ (\ref{bpo1}), the proof of Eq (\ref{bpo3}) being
analogous. Assume that $\liminf_{e \to e_1} S(e_0,e)< S(e_0,e_1)$
(note that it can be $S(e_0,e_1)=+\infty$) then, given $y_0$, we can
choose $y'$ such that
\[
y_0+\liminf_{e \to e_1} S(e_0,e)< y'<y_0+S(e_0,e_1)
\]
thus defined $x'=(e_1,y')$ and $x_0=(e_0,y_0)$ we have $x'\in
\overline{J^{+}(x_0)}\,\backslash J^{+}(x_0)$. By the limit curve
theorem \cite{beem96,minguzzi07c} there is a past inextendible
lightlike ray $\eta$ ending at $x'$ such that $\eta \subset
\overline{J^{+}(x_0)}$. This null geodesic must belong to the null
congruence generated by $n$, otherwise taking $\tilde{y}$ such that
\[
y_0+\liminf_{e \to e_1} S(e_0,e)< y'<\tilde{y}<y_0+S(e_0,e_1)
\]
any point of $\eta\backslash\{x'\}$ would be connected to
$\tilde{x}=(e_1,\tilde{y})$ by a timelike curve and thus
$\tilde{x}\in I^{+}(x_0)$, a contradiction with Eq.\ (\ref{cgf1}).
Since $\eta$ belongs to the congruence and it is past inextendible,
every point of the form $(e_1,y)$ with $y \le y'$ belongs to
$\overline{J^{+}(x_0)}$, and hence from Eq.\ (\ref{prf1}) we get
$\liminf_{e \to e_1} S(e_0,e)=-\infty$.

Proof of Eq.\ (\ref{prf2}). Let us first prove the inclusion
\[\{(x_0,x_1): y_1-y_0\ge \liminf_{(e,e')
 \to (e_0,e_1)} S(e,e') \} \subset \overline{J}.\] Let
 $x_0,x_1$,
be such that $y_1\ge y_0+ \liminf_{(e,e') \to (e_0,e_1)} S(e,e')$ so
that \[\liminf_{(e,e') \to (e_0,e_1)}S(e,e')\ne +\infty.\]

There are two cases, either $\liminf_{(e,e') \to (e_0,e_1)} S(e,e')$
is finite or it is $-\infty$. In the former case let  $\epsilon>0$
and  let $U,V,$ be open sets such that $U\times V \ni (e_0,e_1)$,
then there is $(\hat{e},\hat{e}')\in U\times V$ such that \[
S(\hat{e},\hat{e}') < \liminf_{(e,e') \to
(e_0,e_1)}S(e,e')+\epsilon.\] Note that we can assume
$S(\hat{e},\hat{e}') \ne -\infty$ otherwise it would be, from the
arbitrariness of $U\times V$, $\liminf_{(e,e') \to
(e_0,e_1)}S(e,e')=-\infty$. Since \[\liminf_{(e,e') \to
(e_0,e_1)}S(e,e')\ne +\infty,\] we have that
$S(\hat{e},\hat{e}')<+\infty$ and hence $\hat{t}'> \hat{t}$ or
$\hat{e}=\hat{e}'$.

We can assume that we can always choose $\hat{e}\ne \hat{e}'$
otherwise from the arbitrariness of $U\times V$ and $\epsilon$,
$e_0=e_1$ and $\hat{e}=\hat{e}'$ and taking the limit of the
inequality, $ S(\hat{e},\hat{e}') < \liminf_{(e,e') \to
(e_0,e_1)}S(e,e')+\epsilon$, we get $0=S(e_0,e_1)\le \liminf_{(e,e')
\to (e_0,e_1)}S(e,e')$, thus $y_1\ge y_0$ and the points $(e_0,y_1)$
with $y_1\ge y_0$ are included in $J^{+}(x_0)$. Thus let us assume
the other case, $\hat{t}> \hat{t}'$. By Eq.\ (\ref{cgf1}) the points
$\hat{x}=(\hat{e}, y_0)$ and $\hat{x}'=(\hat{e}', y_1+\epsilon)$ are
chronologically related. Since $\epsilon$ and $U\times V$ are
arbitrary, $(x_0,x_1)\in \overline{I}=\overline{J}$.

If  $\liminf_{(e,e') \to (e_0,e_1)}S(e,e')=-\infty$ then there is a
sequence $(\hat{e}_n,\hat{e}'_n) \to (e_0,e_1)$ such that
$S(\hat{e}_{n},\hat{e}'_{n}) \to -\infty$. Since
$S(\hat{e}_{n},\hat{e}'_{n})<+\infty$, we have $t_n'>t_n$ or
$\hat{e}_n=\hat{e}'_{n}$. We can assume that the latter possibility
does not apply for no value of $n$ because if there were a
subsequence $(\hat{e}_k,\hat{e}'_k)$ with that property
$S(\hat{e}_k,\hat{e}'_k)=S(\hat{e}_k,\hat{e}_k)=0$ and could not
converge to $-\infty$. By Eq.\ (\ref{cgf1})  the points
$\hat{x}_n=(\hat{e}_n,y_0)$, $\hat{x}_n'=(\hat{e}_n,y_1)$, are such
that for sufficiently large $n$, $\hat{x}_{n}'\in I^{+}(\hat{x}_n)$
but $(\hat{x}_n,\hat{x}_n') \to (x_0,x_1)$, thus $(x_0,x_1)\in
\overline{I}=\overline{J}$.

For the converse, let $(x_0,x_1)\in
\overline{J}=\overline{I}$. This means that there is a
sequence of points $(\hat{x}_n,\hat{x}_n' ) \in I$ such that
$(\hat{x}_n,\hat{x}_n' )\to (x_0,x_1)$. By Eq.\ (\ref{cgf1})
$\hat{y}'_n-\hat{y}_n> S(\hat{e}_n,\hat{e}'_n)$, and since
$(\hat{e}_n,\hat{e}_n' )\to (e_0,e_1)$, $\liminf_{n\to +\infty}
S(\hat{e}_n,\hat{e}'_n)\le y_1-y_0$ from which thesis follows.

Proof of Eq.\ (\ref{bpo2}). Assume that $\liminf_{(e,e') \to
(e_0,e_1)}S(e,e')< S(e_0,e_1)$ (note that it can be
$S(e_0,e_1)=+\infty$) then we can choose $\Delta y>0$, $y_0$, $y_1$,
such that
\[
\liminf_{(e,e') \to (e_0,e_1)}S(e,e')< y_1-y_0-\Delta y<
y_1-y_0<S(e_0,e_1)
\]
thus defined $x=(e_0,y_{0}+\Delta y/2)$ and  $x'=(e_1,y_1-\Delta
y/2)$ we have $(x,x')\in \overline{J}\,\backslash J$. By the
limit curve theorem \cite{beem96,minguzzi07c} there is a future
inextendible causal curve $\eta$ starting at $x$ and a past
inextendible causal curve $\eta'$ ending at $x'$ such that for every
$w \in \eta$ and $w'\in \eta'$, $(w,w') \in \overline{J}$. Let
us define $x_0=(e_0,y_0)$ and $x_1=(e_1,y_1)$. Let us show that
either $\eta$ or $\eta'$ is a lightlike geodesic generated by $n$.
If not there are $w \in \eta\backslash\{x\}$ and $w' \in
\eta'\backslash\{x'\}$ such that $(x_0,w) \in I$ and
$(w',x_1)\in I$ so that, since $I$ is open,   we have
$(x_0,x_1) \in I$ and hence by Eq.\ (\ref{cgf1})
$y_1-y_0>S(x_0,x_1)$, a contradiction.

Let us assume that $\eta'$ belongs to the congruence generated by
$n$, the case of $\eta$ being analogous. Since $\eta'$ is past
inextendible every point of the form $\tilde{x}=(e_1,y)$ with $y \le
y_1-\Delta y/2$, is such that $(x,\tilde{x})\in \overline{J}$,
thus from Eq.\ (\ref{prf1}) we get $\liminf_{(e,e') \to
(e_0,e_1)}S(e,e')=-\infty$.

\end{proof}
%
%


\begin{lemma} \label{vtr}
If $\liminf_{e \to e_1}\, S(e_0,e)=-\infty$ and $e_1'$ is such that
$t_1'\ge t_1$ then $\liminf_{e \to e_1'}\, S(e_0,e)=-\infty$. Under
the strict inequality $t_1<t_1'$ we have the stronger conclusion
$S(e_0',e_1')=-\infty$.
Analogously, if $\liminf_{e \to e_0}\, S(e,e_1)=-\infty$ and $e_0'$
is such that $t_0'\le t_0$ then $\liminf_{e \to e_0'}\,
S(e,e_1)=-\infty$. Under the strict inequality $t_0'<t_0$ we have
the stronger conclusion $S(e_0',e_1')=-\infty$.
\end{lemma}

\begin{proof}
Let $U\ni e_1'$ be an open set and let $e\in U$, $t(e)>t_1'\ge t_1$,
so that $S(e_1,e)<+\infty$ and we  can find  $\Delta y>S(e_1,e)$
such that for every $y_1$, $(e,y_1+\Delta y) \in I^{+}((e_1,y_1))$.
Since $\liminf_{e \to e_1}\, S(e_0,e)=-\infty$, for every $y_1$,
$(e_1,y_1)\in \overline{J^{+}((e_0,y_0))}$ and hence $(e,y_1+\Delta
y) \in {I^{+}((e_0,y_0))}$. Since $y_1$ is arbitrary, chosen any
$r\in \mathbb{R}$, $(e,r) \in {I^{+}((e_0,y_0))}$, and since $U$ is
arbitrary, $(e_1',r) \in \overline{I^{+}((e_0,y_0))}$. From Eq.
(\ref{prf1}) we get $\liminf_{e \to e_1'}\, S(e_0,e)=-\infty$.

If $t_1<t_1'$ then $S(e_1,e_1')<+\infty$ and we can find $\Delta
y>S(e_1,e)$. For every $y_1$, $(e_1',y_1+\Delta y)\in
I^{+}((e_1,y_1))$ but $(e_1,y_1)\in \overline{J^{+}((e_0,y_0))}$
thus $(e_1',y_1+\Delta y)\in I^{+}((e_0,y_0))$ and from Eq.
(\ref{cgf1}) and the arbitrariness of $y_1$ we get
$S(e_0,e_1')=-\infty$.
\end{proof}

%

\begin{lemma} \label{vtr2}
If $\liminf_{(\check{e},\hat{e})\to (e_0,e_1)}\,
S(\check{e},\hat{e})=-\infty$ and $(e_0',e_1')$ is such that
$t_0'\le t_0$ and  $t_1\le t_1'$ then
$\liminf_{(\check{e},\hat{e})\to (e_0',e_1')}\,
S(\check{e},\hat{e})=-\infty$. Under the strict inequalities
$t_0'<t_0$ and $t_1<t_1'$ we have the stronger conclusion
$S(e_0',e_1')=-\infty$.
\end{lemma}

\begin{proof}
The assumption $\liminf_{(\check{e},\hat{e})\to (e_0,e_1)}\,
S(\check{e},\hat{e})=-\infty$ implies that $t_0\le t_1$. Let $U,V,$
be open sets, $U\times V\ni (e_0',e_1')$ and let
$(\hat{e},\check{e})\in U\times V$, such that $t(\check{e})<t_0'\le
t_0$ and $t(\hat{e})>t_1'\ge t_1$, so that
$\max(S(\check{e},e_0),S(e_1,\hat{e}))<+\infty$ and we can find
$\Delta y> \max(S(\check{e},e_0),S(e_1,\hat{e}))$ so that for every
$\check{y},\hat{y}$, $(\check{e},\check{y}-\Delta y) \in
I^{-}((e_0,\check{y}))$ and $(\hat{e},\hat{y}+\Delta y) \in
I^{+}((e_1,\hat{y}))$.  Since $\liminf_{({e},{e}')\to (e_0,e_1)}\,
S({e},{e}')=-\infty$, for every $\check{y},\hat{y}$,
$((e_0,\check{y}),(e_1,\hat{y}))\in \overline{J}$ and hence
$((\check{e},\check{y}-\Delta y),(\hat{e},\hat{y}+\Delta y)) \in
{I}$.  From the arbitrariness of $\check{y}$ and $\hat{y}$ we
have that for every $\check{r},\hat{r}\in\mathbb{R}$,
$((\check{e},\check{r}),(\hat{e},\hat{r})) \in {I}$. Since
$U\times V$ is arbitrary, $((e_0',\check{r}),(e_1',\hat{r})) \in
\overline{I}$ and from Eq. (\ref{prf2}) we get thesis.

For the last statement, since $t_0'<t_0$ and $t_1<t_1'$, we have
$e_0'\ne e_1'$ and there is a constant $C<+\infty$ such that
$S(e_0',e_0)<C$ and $S(e_1,e_1')<C$. Let $(\check{e}_n,\hat{e}_n)
\to (e_0,e_1)$ be a sequence such that $\lim
S(\check{e}_n,\hat{e}_n)\to -\infty$. Since $S$ is upper
semi-continuous outside the diagonal (corollary \ref{hgf}) we can
assume $S(e_0', \check{e}_n)<C$ and $S(\hat{e}_n,e_1')<C$. From the
triangle inequality we get
\[
S(e_0',e_1')\le S(e_0', \check{e}_n)+S(\check{e}_n,\hat{e}_n)+
S(\hat{e}_n,e_1')\le 2C+ S(\check{e}_n,\hat{e}_n),
\]
and thus $S(e_0',e_1')=-\infty$.

\end{proof}

An immediate consequence of lemmas \ref{prf} and \ref{vtr2} is

\begin{proposition} \label{bfj}
If $S(e_0,e_1)$ is finite for every $e_0,e_1 \in E$ with $t_0<t_1$
then $S: E\times E\to (-\infty,+\infty]$  is (everywhere) lower
semi-continuous.
\end{proposition}

\begin{proof}
If $S$ were not lower semi-continuous at $e_0,e_1$ then
$\liminf_{(\check{e},\hat{e})\to (e_0,e_1)}\,
S(\check{e},\hat{e})<S(e_0,e_1)$ which would imply $t_0\le t_1$ and,
by Lemma \ref{prf}, $\liminf_{(\check{e},\hat{e})\to (e_0,e_1)}\,
S(\check{e},\hat{e})=-\infty$ and from Lemma \ref{vtr2} would give
for chosen $e_0', e_1'$ with $t_0'<t_0$ and $t_1<t_1'$,
$S(e_0',e_1')=-\infty$, a contradiction.
\end{proof}

%

%
%
%
%


\subsection{Equivalence between stable and strong causality}

It is well known that in a generic spacetime the relation
$\overline{J}$ is not transitive. Indeed, the smallest closed
and transitive relation which contains $J$, denoted $K$ here and $K^{+}$ by
Sorkin and Woolgar \cite{sorkin96}, is particularly important.
Seifert \cite{seifert71}  had also introduced a closed and
transitive relation denoted $J_S$, whose antisymmetry is
equivalent to stable causality and hence to the existence of a time
function \cite{hawking74,minguzzi07}. The equivalence between the
antisymmetry of $K$, called $K$-causality, and stable causality
has been recently established in \cite{minguzzi08b}, and will be
central to establish the equivalence between strong causality and
stable causality for generalized gravitational wave spacetimes.

Usually, $K$ does not identify with any simple relation
constructed in terms of causal curves, but for causally simple
spacetimes and a few other exceptions. Fortunately, for the
generalized gravitational wave spacetimes considered in this work
the next result holds

\begin{theorem}
On the spacetime $(M,g)$ the relation $\overline{J}$ is
transitive and thus coincident with $K$.
\end{theorem}

\begin{proof}
Let $(x_0,x_1)\in \overline{J}$ and $(x_1,x_2)\in \overline{J}$,
so that $t_0\le t_1\le t_2$. \\ If $\liminf_{({e},{e}')\to
(e_0,e_1)}\, S({e},{e}')$ $=-\infty$ or $\liminf_{({e},{e}')\to
(e_1,e_2)}\, S({e},{e}')$ $=-\infty$ then by Lemma \ref{vtr2},
$\liminf_{({e},{e}')\to (e_0,e_2)}\, S({e},{e}')=-\infty$ and hence
by Eq.  (\ref{prf2}), $(x_0,x_2)\in \overline{J}$. If, on the
contrary, both liminf are finite then by Lemma \ref{prf} they
coincide respectively with $S(e_0,e_1)$ and $S(e_1,e_2)$ and hence
\begin{align*}
y_1-y_0&\ge S(e_0,e_1),\\
y_2-y_1&\ge S(e_1,e_2),
\end{align*}
which, using the triangle inequality for $S$ give
\[
y_2-y_0\ge S(e_0,e_1)+S(e_1,e_2)\ge S(e_0,e_2),
\]
and hence by Eq. (\ref{prf2}), $(x_0,x_2)\in \overline{J}$.
\end{proof}

The previous result simplifies considerably the causal ladder for
generalized gravitational wave spacetimes. A strongly causal
spacetime for which $\overline{J}$ is transitive is called {\em
causally easy}. It has been proved \cite{minguzzi08b} that causal
continuity implies causal easiness which implies stable causality.
Thus the previous theorem implies

\begin{theorem} \label{bup}
For the generalized gravitational wave spacetime $(M,g)$, strong
causality and stable causality are equivalent (they are actually
equivalent to causal easiness).
\end{theorem}

This result implies that the infinite causality levels that one may
construct between strong causality and stable causality are actually
all coincident for this type of spacetime.

It is therefore interesting to establish under which conditions
strong  causality holds.

\begin{theorem} \label{psf2}
The spacetime $(M,g)$ is strongly  causal at an event $x$ if and only if it is
strongly  causal at every other event $y$ on the same time slice
(i.e.  $t(y)=t(x)$) if and only if  $S:E\times E \to [-\infty,+\infty]$ is
lower semi-continuous at $(e_x,e_x)$. In particular strong causality
holds if and only if $S$ is lower semi-continuous on the diagonal $\{(e,e): e
\in E\}$.
\end{theorem}

\begin{proof}

Assume that for some event $x$, $\liminf_{(e,e') \to (e_x,e_x)}
S(e,e')=-\infty$ and let $y$ be an event such that $t(y)=t(x)$. By
lemma \ref{vtr2} \[\liminf_{(e,e') \to (e_y,e_x)}
S(e,e')=\liminf_{(e,e') \to (e_x,e_y)} S(e,e')=\liminf_{(e,e') \to
(e_y,e_y)} S(e,e')=-\infty.\]

Strong causality is violated at $p$ if and only if there is a point $q\in
J^{-}(p)$, $q\ne p$, such that $(p,q) \in \overline{J}$ (see
\cite[Lemma 4.16]{penrose72} or the proof of \cite[theorem
3.4]{minguzzi07b}).

By lemmas \ref{prf} and \ref{vtr2}, strong causality is violated at
$x$ if and only if $\liminf_{(e,e') \to (e_x,e_x)} S(e,e')=-\infty$
(which,  by Eq. (\ref{bpo2}), holds if and only if $S$ is not lower
semi-continuous at $e_x$). Indeed, if the latter equality holds then
$w=(e_x,y_x-1)$ is such that $(w,x) \in J$, $w\ne x$ and $(x,w)
\in \overline{J}$. Conversely, if there is some $w$ with these
properties then as $(w,x) \in J$, $t_w \le t_x$, and since
$(x,w) \in \overline{J}$ and $t$ is continuous, $t_x\le t_w$, so
that they stay in the same time slice. Using again $(w,x) \in J$
it follows that  $w$ stays in the past lightlike ray generated by
$n$ that ends at $x$, so that $e_w=e_x$. Thus $(x,w) \in
\overline{J}$, $w\ne x$, and Eq. (\ref{prf2}) together with
$S(e_x,e_x)=0$ implies that $\liminf_{(e,e') \to (e_x,e_x)}
S(e,e')=-\infty$.
\end{proof}

\subsection{Distinguishing spacetimes}

Let us recall \cite{kronheimer67,minguzzi07e} that a spacetime is
future distinguishing if $I^{+}(x)=I^{+}(y) \Rightarrow x=y$; past
distinguishing if $I^{-}(x)=I^{-}(y) \Rightarrow x=y$; and weakly
distinguishing if `$I^{+}(x)=I^{+}(y)$ and $I^{-}(x)=I^{-}(y)$'$
\Rightarrow x=y$. The spacetime is distinguishing if it is both
future and past distinguishing, namely if `$I^{+}(x)=I^{+}(y)$ or
$I^{-}(x)=I^{-}(y)$'$ \Rightarrow x=y$. Another useful
characterization of future distinction is the fact that at the point
$x$ there are arbitrarily small neighborhoods such that no causal
curve issued from $x$ can escape and later reenter the neighborhood
\cite{hawking73,minguzzi19}. Similar  characterizations hold for
past distinction and distinction (but not for weak distinction).

There are other useful characterizations. Let us recall that the
relations $D_f=\{(x,y): y \in \overline{I^{+}(x)}\}$,
$D_p=\{(x,y): x \in \overline{I^{-}(y)}\}$, and $D=D_f\cap D_p$ are transitive \cite{minguzzi07b}. Moreover, $D_f$ ($D_p$) is antisymmetric if and only if the spacetime is future (resp.
past) distinguishing, and $D$ is antisymmetric if and only if the spacetime is
weakly distinguishing \cite{minguzzi07b}.

It is clear from Lemma \ref{pcf} that the chronological relation is
determined by the least action $S$. The conformal structure follows
also from the chronological relation and hence from $S$ provided the
spacetime is distinguishing (Malament's theorem \cite{malament77b},
see \cite{minguzzi19}). It is possible to completely
characterize the distinction of the spacetime $(M,g)$ using the
function $S$.

\begin{theorem} \label{psf}
The spacetime $(M,g)$ is future  (resp. past) distinguishing at
$x=(e_1,y_1)$ if and only if $S(e_1,\cdot)$ (resp. $S(\cdot,e_1)$) is lower
semi-continuous at $e_1$.
\end{theorem}
%
%
%

\begin{proof}
Assume $(M,g)$ is not future distinguishing at $x=(t_x,q_x,y_x)$.
Then there is $z=(t_z,q_z,y_z)$ such that $I^{+}(x)=I^{+}(z)$ which
implies $t_x=t_z$ because $t$ is a semi-time function. The violation
of future distinction implies the existence of a sequence of
timelike curves $\tilde\gamma_n$ of endpoints $x$ and $x_n$ such
that $x_n \to x$, and  $z$ is an accumulation point for
$\tilde\gamma_n$. By the limit curve theorem
\cite{beem96,minguzzi07c}, since the spacetime is chronological,
there is a lightlike past ray ending at $x$ entirely contained in
$\overline{I^{+}(x)}$. Necessarily this lightlike ray $r$ coincides
with the portion of the fiber passing through $x$ which stays in the
causal past of $x$. Indeed, if it were different, as $-g(n,r')=\dd
t[r']$ is a positive constant, $t$ would be an affine parameter for
$r$ and thus $r$ would contain points with $t<t_x$ which is
impossible since no points of this kind can be contained in
$\overline{I^{+}(x)}$. From Eq. (\ref{prf1}) we get
\[
\liminf_{e \to e_x} S(e_x,e)=-\infty,
\]
which violates lower semi-continuity as $S(e_x,e_x)=0$.
%
%
%

Conversely, assume $S(e_x,\cdot)$ is not lower semi-continuous at
$e_x$, then from Lemma $\ref{prf}$ it follows that the whole past
lightlike ray $\eta$ generated by $n$ and ending at $x$ is contained
in $\overline{J^{+}(x)}$. Take $w\in \eta\backslash\{x\}$ then as
$(w,x) \in J$ and $w \in \overline{J^{+}(x)}$ we have
$I^{+}(x)=I^{+}(w)$, so that future distinction is violated at $x$.
\end{proof}

\subsection{Reflectivity, causal continuity and independence of time}

A spacetime is future reflective if $x \in \overline{I^{-}(y)}
\Rightarrow y \in \overline{I^{+}(x)}$, past reflective if $x \in
\overline{I^{-}(y)} \Leftarrow y \in \overline{I^{+}(x)}$ and
reflective if it is both past and future reflective.

Equivalently \cite{minguzzi07b}, future reflectivity reads
$D_f=\overline{J}$ and past reflectivity reads
$D_p=\overline{J}$.


From Eqs. (\ref{prf1}) and (\ref{prf3}) we get

\begin{theorem} \label{bqa}
The spacetime $(M,g)$ is future reflective if and only if for every $e_0,e_1\in
E$
\begin{equation}
\liminf_{e \to e_1} S(e_0,e)\le\liminf_{e \to e_0} S(e,e_1),
\end{equation}
and past reflective if and only if for every $e_0,e_1\in E$
\begin{equation} \label{ion}
\liminf_{e \to e_1} S(e_0,e)\ge\liminf_{e \to e_0} S(e,e_1).
\end{equation}
Moreover, in the former case
\begin{equation} \label{nhf}
\liminf_{e \to e_1} S(e_0,e)=\liminf_{(e,e') \to (e_0,e_1)} S(e,e'),
\end{equation}
while in the latter case
\begin{equation}
\liminf_{e \to e_0} S(e,e_1)=\liminf_{(e,e') \to (e_0,e_1)} S(e,e').
\end{equation}

\end{theorem}

An important  case is that of autonomous Lagrangians

\begin{theorem} \label{byt}
If the Lagrangian $L(t,q,v)$ does not depend on time then the
spacetime $(M,g)$ is reflective.
\end{theorem}

\begin{proof}

Let $x_1 \in \overline{I^{+}(x_0)}$, $x_1=(t_1,q_1,y_1)$, and make a
coordinate change change $y'= y+C t$, $t'=t$, so that
$\p_t)_{y'}=\p_t)_{y} -C\p_y$. From Eq. (\ref{eis}) it follows that
this operation changes the Lagrangian by a constant and thus keeps
it independent of time. With $C=V(q_1)-1/2$, $\p_t)_{y'}$ is
timelike at $x_1$. For this reason with no loss of generality we can
assume that  $\p_t$ is timelike at $x_1$. The sequence
$w_n=(\tau_n,q_1,y_1)$ with $\tau_n>t_1$ and $\tau_n\to t_1$ stays
in the integral line of $\p_t$ passing through $x_1$, thus as $\p_t$
is timelike in a neighborhood of $x_1$, $w_n \in I^{+}(x_1)$ and
hence $w_m\in I^{+}(x_0)$. Let $\sigma_n=(t, q_n(t),y_n(t))$,
$w_n=\sigma_n(\tau_n)$, $y_n(\tau_n)=y_1$, $q_n(\tau_n)=q_1$,
$\sigma_n: [t_0,\tau_n]\to M$ be a sequence of timelike curves
starting from $x_0$ of final endpoints $w_n$. Consider the curves
$\tilde{\sigma}_n(t)=(t, q_n(t-t_1+\tau_n),y_n(t-t_1+\tau_n))$,
$\sigma: [t_0+t_1-\tau_n,t_1]\to M$ which are obtained translating
$\sigma_n$ backward under the flow of $\p_t$ by a parameter
$\tau_n-t_1$. By construction the curves $\tilde{\sigma}_n$ end at
$x_1$ and start at $(t_0+t_1-\tau_n, q_0,y_0)$ which converges to
$x_0$ for $\tau_n\to t_1$. Finally, $\tilde{\sigma}_n$ is timelike
because the causal character of a vector is preserved under the
action of the flow of a Killing  field, and in our case $\p_t$ is
Killing. We conclude that $x_0\in \overline{I^{-}(x_1)}$ and hence
past reflectivity holds. The proof of the other direction is
analogous.
\end{proof}

A theorem by Clarke and Joshi \cite[Prop. 3.1]{clarke88} states that
every spacetime admitting a complete timelike Killing vector field
is reflective. Their theorem is in a way connected to the above
result. Indeed, one could hope to prove Theorem \ref{byt} by showing
that under independence of time for $L$ the spacetime $(M,g)$ admits
a complete timelike Killing vector field. This alternative strategy
seems to work only in particular cases. In fact note that while,
under time independence, the vector $\p_t+k\p_y$ is  complete and
Killing for any constant $k$, it is not necessarily globally
timelike (although it can be made timelike at some event for some
$k$) unless the potential is bounded from below \cite[Theor.
2.8]{minguzzi06d}.


A spacetime which is distinguishing and reflective is called
causally continuous. From theorems \ref{psf} and \ref{bqa} we obtain

\begin{theorem} \label{xxd}
The spacetime $(M,g)$ is causally continuous if and only if
\begin{equation} \label{uiy}
\liminf_{e \to e_1} S(e_0,e)=\liminf_{e \to e_0} S(e,e_1),
\end{equation}
and this quantity vanishes for $e_0=e_1$.
\end{theorem}

\begin{remark} Since under Eq. (\ref{uiy}) the Eq. (\ref{nhf}) holds,
then $S$, under the assumptions of Theorem \ref{xxd}, is lower
semi-continuous on the diagonal, which is equivalent to stable
causality. This fact is consistent with the well known result that
causal continuity implies stable causality.
\end{remark}

\begin{corollary} \label{pxa}
If $S:E\times E\to [-\infty,+\infty]$ is lower semi-continuous then
$(M,g)$ is causally continuous.
\end{corollary}

\begin{proof}
If $S$ is lower semi-continuous then
\[
\liminf_{e \to e_1} S(e_0,e)=S(e_0,e_1)=\liminf_{e \to e_0}
S(e,e_1),
\]
and this quantity vanishes for $e_0=e_1$, because by definition
$S(e,e)=0$.
\end{proof}

Under reflectivity  it is especially important to establish if a
spacetime is distinguishing as this property would imply causal
continuity. Fortunately, the next proposition shows that it is
necessarily to check for distinction at just one point for every
time slice (as it happens for strong causality, see Theorem
\ref{psf2}).

\begin{proposition}
Assume the spacetime $(M,g)$ is reflective. If $(M,g)$ is past or
future distinguishing at an event $x$ then it is past and future
distinguishing at every other event $y$ on the same time slice (i.e.
$t(y)=t(x)$). Moreover, if this is not the case then the events in
the time slice have the same chronological past and the same
chronological future.
\end{proposition}

\begin{proof}
Assume that future distinction is violated at $x$ and let us prove
that past distinction is violated at $y$ with $t(y)=t(x)$. The
violation of future distinction at $x$ implies (see Lemma \ref{prf}
and Theorem \ref{psf})
\[
\liminf_{e \to e_x} S(e_x,e)=-\infty,
\]
thus  by Lemma \ref{vtr}
\begin{equation} \label{khg1}
\liminf_{e \to e_y} S(e_x,e)=-\infty,
\end{equation}
and using past reflectivity, that is Eq. (\ref{ion}),
\begin{equation}\label{khg2}
\liminf_{e \to e_x} S(e,e_y)=-\infty,
\end{equation}
and using again Lemma \ref{vtr}
\[
\liminf_{e \to e_y} S(e,e_y)=-\infty,
\]
which from Lemma \ref{prf} and Theorem \ref{psf} implies that past
distinction is violated at $y$. Analogously, if past distinction is
violated at $x$ then future distinction is violated at $y$. As $y$
is arbitrary and can be taken equal to $x$, future distinction is
violated at a point if and only if past distinction is violated at the point,
from which the first statement follows. As for the last statement,
the previous analysis shows that if future or past distinction is
violated at a point then  future distinction  is violated at $x$,
and Eqs. (\ref{khg1}) and (\ref{khg2}) imply (i) $y\in
\overline{J^{+}(x)}$ and (ii) $x \in \overline{J^{-}(y)}$. The
analogous omitted part of proof gives (a) $y\in \overline{J^{-}(x)}$
and (b) $x\in \overline{J^{+}(y)}$ from (i) and (b) we get
$I^{+}(y)=I^{+}(x)$ and from (ii) and (a) we get
$I^{-}(x)=I^{-}(y)$.
\end{proof}

\begin{corollary}
In the time independent case if past or future distinction holds at
a point then distinction holds at every point.
\end{corollary}

\begin{proof}
If past or future distinction holds at $x$ then it holds on the
whole time slice passing through $x$. As $\p_t$ is Killing (not
necessarily timelike) every event is obtained from  a point on
the time slice of $x$ by applying the flow $\phi_s$ of the
Killing field $\p_t$. As the maps $\phi_s$ are isometries one gets
thesis.
\end{proof}


\subsection{A partially time-independent case: subquadratic potentials}

In this subsection we investigate the  special case $\p_t
a_t=b_t=0$, and obtain some useful result which improve those
obtained in \cite{flores03}. Further result will be given in the
next sections, e.g.\ Theor. \ref{dis} and Cor. \ref{hyt}. The
assumption $\p_t a_t=b_t=0$ will be explicitly stated wherever it is
used.

\begin{definition}
We shall say that the functional $\mathcal{S}_{e_0,e_1}$ is {\em
coercive} at $(e_0,e_1)$ if given any sequence of $C^1$ curves
$q_n:[t_0,t_1] \to Q$ such that $\mathcal{S}_{e_0,e_1}[q_n]$ is
bounded from above, the images $q_n([t_0,t_1])$ are all contained in
a  compact subset of $Q$.  We shall say that the functional
$\mathcal{S}_{e_0,e_1}$ is {\em coercive} at if it is coercive
everywhere on $E\times E$.
\end{definition}

\begin{definition}
Let $I$ be an interval of the real line, and let $(Q,a)$ be a
Riemannian space. A potential $V(t,q)$ defined over $I\times Q$ is
said to be
\begin{align*}
&\qquad \textrm{{\em almost quadratic}} &\textrm{ if }& &V(t,q)&\le c_1 D(q_B,q)^2+c_2,\\
&\qquad \textrm{{\em subquadratic}}  &\textrm{ if }& &V(t,q)&\le c_1
\,
o(D(q_B,q)^2)+c_2,\\
&\qquad \textrm{{\em superquadratic}} &\textrm{ if }&  &V(t,q)&\ge
c_1 \omega(D(q_B,q)^2)+c_2,
\end{align*}
where $q_B\in Q$, $c_1, c_2$, are positive constants, and $D(q_0,q)$
is the Riemannian distance calculated through $a$. If mentioning
these properties no mention is made of $I$ then it is understood
that they hold for any compact interval $I$.
\end{definition}

Clearly if $Q$ and $I$ are compact then $V$ is almost quadratic, and
it cannot be neither subquadratic nor superquadratic.

It is understood that the little-$o$ and little-$\omega$ Landau
notation used above refers to the limit $(t,q) \to+\infty$ on
$\mathbb{R}\times Q$, with the Alexandrov one-point compactification
topology for the point $+\infty$ (its open sets are the complements
to compact sets). In particular, since $D$ is continuous, for fixed
$q_B$, $D(q_B,q)\to +\infty$ implies $q\to +\infty$, but the
converse holds if and only if $(Q,a)$ is complete (Hopf-Rinow
theorem).

\begin{remark}
The reference point $q_B$ can be chosen arbitrarily because, if
$q'_B, q_B''\in Q$ are any two points
\[
c_1' D(q_B',q)^2+c_2'\le c_1' [D(q_B',q_B'')+D(q_B'',q)]^2+c'_2\le
c_1''D(q_B'',q)^2+c_2'',
\]
where $c_1''=c_1'+{c_1'}^2$ and $c_2''=c_2'+2D(q_B',q_B'')^2$.
\end{remark}

The next result allows us to establish the lower semi-continuity of
the least action, and hence to infer strong causality. According to
our Theorem \ref{bup} this fact actually implies stable causality.
The theorem rephrases and improves some results contained in
\cite{flores03} where the key role of the quadraticity of the
potential was recognized.

\begin{theorem} \label{ber}
Let us consider the special case $\p_t a_t=b_t=0$. Suppose that for
some compact time interval $I$ the potential $V(t,q)$ on $I\times Q$
is
\begin{itemize}
\item[(a)] {\bf Almost quadratic}.  There is a constant $\epsilon>0$ such that the least
action $S({e_0,e_1})$, is finite on $(I\times Q)^2$ provided
$t_0<t_1<t_0+\epsilon$  (as a consequence $S$ is lower
semi-continuous on the diagonal on $(I\times Q)^2$, see Theor.\
\ref{bfj}). Moreover, under the same conditions on $t_1$, if $(Q,a)$
is complete then $\mathcal{S}_{e_0,e_1}$  is coercive.
\item[(b)] {\bf Subquadratic}. Same as (a) but we can take $\epsilon=+\infty$.
\item[(c)] {\bf Superquadratic}. Let $\tilde{e}\in I\times Q$ be such that $\tilde{t}\in \textrm{Int}\, I$. Then $S(\tilde{e},\cdot)$ and $S(\cdot,\tilde{e})$
are not lower semi-continuous at $\tilde{e}$.
\end{itemize}
\end{theorem}

\begin{proof}
Let $(e_0,e_1)\in (I\times Q)^2$ with $t_0<t_1$, and let us set
$q_B=q_0$. Let $q: [t_0,t_1]\to Q$ be a $C^1$ curve connecting $q_0$
to $q_1$, and let $\tilde{q}$  be a point on the image of $q(t)$
with maximum distance from $q_B$ (we recall that the Riemannian
distance is continuous). For any given path (image of $q(t)$) the
kinetic energy is minimized by that reparametrization  which makes
the speed constant (Cauchy-Schwarz inequality). As a consequence,
the kinetic energy satisfies the lower bound
\begin{align*}
T[q]&=\int_{t_0}^{t_1} \frac{1}{2} a(\dot{q},\dot{q}) \dd t\ge
\frac{1}{2} \frac{l[q]^2}{t_1-t_0} \ge \frac{1}{2}
\frac{[D(q_0,\tilde{q})+D(\tilde{q},q_1)]^2}{t_1-t_0}\\
& \ge \frac{[2D(q_B,\tilde{q})-D(q_0,q_1)]^2}{t_1-t_0},
\end{align*}
where $l[q]$ is the Riemannian length of the path.

Proof of (a). For an almost quadratic potential we get
\begin{align*}
\mathcal{S}_{e_0,e_1}[q]&\ge T[q]-\int_{t_0}^{t_1} V(t,q(t))\dd t\,
\ge T[q]-\int_{t_0}^{t_1}[c_1 D(q_B,q(t))^2+c_2] \dd t \\
& \ge T[q]-\int_{t_0}^{t_1}[c_1 D(q_B,\tilde{q})^2+c_2] \dd t
\\
&\ge\frac{[2D(q_B,\tilde{q})-D(q_0,q_1)]^2}{t_1-t_0}-[c_1
D(q_B,\tilde{q})^2+c_2](t_1-t_0).
\end{align*}
If $t_1<t_0+\epsilon$ with $\epsilon <2/\sqrt{c_1}$ then the
right-hand side is bounded from below by a constant independent of
$\tilde{q}$ and hence of $q(t)$. Taking the infimum over all paths
connecting $q_0$ to $q_1$ we obtain that $S(e_0,e_1)\ne -\infty$.
Suppose $(Q,a)$ is complete and let us make the same choice for
$\epsilon$. The last inequality proves that if
$\mathcal{S}_{e_0,e_1}[q]<C$ for some constant $R$,  then
$D(q_B,\tilde{q})<R(C,t_0,t_1, c_1,c_2)$, where the constant in the
last equation does not depend on the curve $q(t)$. Thus, as for
every $t$, $D(q_B,q(t))\le D(q_B,\tilde{q})$, all connecting curves
are  contained in a ball of radius $R$, which is compact by the
Hopf-Rinow theorem.

Proof of (b). Suppose that the potential is subquadratic. Let us
observe that for each $\delta>0$ there is some compact set $K\subset
Q$ such that $V(t,q)\le \delta c_1
 D(q_B,{q})^2+c_2$ for $(t,q)\in I\times (Q\backslash K)$. Since over
 the compact $I\times K$ the potential $V(t,q)$ attains a maximum,
 there is a constant $c_2'(\delta)$ such that $V(t,q)\le \delta c_1
 D(q_B,{q})^2+c_2'(\delta)$. Thus
\begin{align*}
\mathcal{S}_{e_0,e_1}[q]&\ge T[q]-\int_{t_0}^{t_1} V(t,q(t))\dd t\,
\ge T[q]-\int_{t_0}^{t_1}[\delta c_1 D(q_B,q(t))^2+c_2'(\delta)] \dd t \\
& \ge T[q]-\int_{t_0}^{t_1}[\delta c_1
D(q_B,\tilde{q})^2+c_2'(\delta)] \dd t
\\
&\ge\frac{[2D(q_B,\tilde{q})-D(q_0,q_1)]^2}{t_1-t_0}-[\delta c_1
D(q_B,\tilde{q})^2+c_2'(\delta)](t_1-t_0).
\end{align*}
Choosing $\delta<\frac{4}{(t_1-t_0)^2c_1}$ we obtain that the
right-hand side is bounded from below by a constant independent of
$\tilde{q}$ and hence of $q(t)$. Taking the infimum over all paths
connecting $q_0$ to $q_1$ we obtain that $S(e_0,e_1)\ne -\infty$.
Arguing as in (a) we also obtain that  $\mathcal{S}_{e_0,e_1}[q]$ is
coercive.

Proof of (c). Since  the potential is superquadratic $Q$ is non
compact. Let $q_B\in Q$ be any point,  even if $(Q,a)$ is not
complete there is a small compact ball $B$ of radius $r$ centered at
$q_B$. Let us define $t_0\in \textrm{Int} I$,  $e_0=(t_0,q_B)$,
$e_{1}=(t_0+3\epsilon, q_B)$ and let $\epsilon>0$. We wish to prove
that $S(e_0,e_1(\epsilon))$ goes to $-\infty$ for $\epsilon \to 0$
(and hence $e_1\to e_0$), thus proving that $S(e_0,\cdot)$ is not
lower semi-continuous, the proof to the dual claim being analogous.


Let us observe that for each $\delta$, such that $0<\delta<c_1
\epsilon^2/2$ there is some compact set $K\subset Q$, $q_B\in K$,
such that $V(t,q)\ge \frac{1}{\delta}  c_1 D(q_B,{q})^2+c_2$ for
$(t,q)\in I\times (Q\backslash K)$. Let $C$ be a lower bound for
$V(t,q)$  on the compact set  $I\times K$, then $F:=\textrm{min}(C,
c_2)$ is a lower bound for $V(t,q)$ all over $I\times Q$. Let us
consider a curve $q:[t_0,t_0+3\epsilon]\to Q$ such that for $t\in
[t_0,  t_0+\epsilon]$ the point $q(t)$ starts from $q_B$, moves at
constant speed till it reaches a point $\hat{q}\notin K$, there it
stays at rest during the interval $[t_0+\epsilon, t_0+2\epsilon]$,
and then it returns to $q_B$ at constant speed along the initial
path. Thus
 \begin{align*}
S(e_0,e_1)&\le \mathcal{S}_{e_0,e_1}[q]\le
T[q]-\int_{t_0}^{t_0+3\epsilon} V(t,q(t))\dd t\,
\\& \le  \frac{D(q_B,\hat{q})^2}{\epsilon}-\int_{t_0+\epsilon}^{t_0+2\epsilon}
[\frac{1}{\delta} c_1 D(q_B,q(t))^2+c_2] \dd t -2\epsilon F\\
& \le \frac{D(q_B,\hat{q})^2}{\epsilon}-\epsilon
[\frac{1}{\delta} c_1 D(q_B,\hat{q})^2+c_2] -2\epsilon F \\
&\le -\frac{D(q_B,\hat{q})^2}{\epsilon}-\epsilon (c_2 +2
F)\le-\frac{r^2}{\epsilon}-\epsilon (c_2 +2 F) .
\end{align*}
This inequality proves that $S(e_0,e_1(\epsilon))\to -\infty$  for
$\epsilon \to 0$.
\end{proof}


Using the recently proved equivalence between stable and
$K$-causality \cite{minguzzi08b, minguzzi09c} we can now prove the
stable causality of a large subfamily of Brinkmann-Eisenhart
spacetimes.

\begin{corollary} \label{njx2}
In the special case $\p_t a_t=b_t=0$, if the potential $V(t,q)$ is
almost quadratic then the spacetime is stably causal.
\end{corollary}

\begin{proof}
Follows from Theorems \ref{ber} case (a), \ref{bup} and \ref{psf2}.
\end{proof}

\begin{remark}
We stress that the previous result for the small family of
gravitational plane waves was already known. In fact Ehrlich and
Emch proved that for $Q=\mathbb{R}^2$, $a_{ij}=\delta_{ij}$,
$V(t,q)=f(t) (w^2-z^2)+g(t)w z$ the spacetime is even causally
continuous \cite[Theor.\ 6.9]{ehrlich92}. Similarly, Hubeny,
Rangamani and Ross \cite{hubeny04} showed that if $Q=\mathbb{R}^d$,
$a_{ij}=\delta_{ij}$, $V(t,q)=A_{ij}(t) x^i x^j$, then the spacetime
is stably causal. Contrary to these references we do not assume any
symmetry, nor an exact quadratic dependence. These last authors
claim \cite{hubeny04} that a proof of the above corollary can be
found in \cite{flores03}, but in this paper the authors do not
mention stable causality, and prove just strong causality. In fact,
we have seen that the above result depends on the recently proved
equivalence between $K$-causality and stable causality
\cite{minguzzi08b, minguzzi09c} which we used to obtain Theorem
\ref{bup}.
\end{remark}

%

The following result improves \cite[Prop. 2.1]{flores03}, because
our superquadraticity condition is less restrictive, the Riemannian
manifold $(Q,a)$ is not assumed complete, and  $V(t,q)$ is not
necessarily non-negative.

\begin{theorem} \label{dis}
Let us consider the special case $\p_t a_t=b_t=0$. If $V(t,q)$ is
superquadratic for every $t\in I$, with $I\subset \mathbb{R}$ open
interval,  then $(M,g)$ is not distinguishing on any point of
$\pi^{-1}(I\times Q)$.
\end{theorem}

\begin{proof}
It follows from Theorem \ref{ber} and Theorem \ref{psf}.
\end{proof}

\subsection{Global hyperbolicity and causal simplicity: from coercivity to Tonelli's theorem}

A spacetime is causally simple if it is causal and
$\overline{J}=J$ (see \cite{minguzzi19}). The
level of causal simplicity has been characterized in
\cite{minguzzi06d}. Here we provide a shorter proof of the following
theorem taking advantage of the previous results.

\begin{theorem} \label{cga2}
The spacetime $(M,g)$ is causally simple if and only if the following
properties hold
\begin{itemize}
\item[(a)] if $S(e_0,e_1)$, $t_0<t_1$,
is finite then the functional $\mathcal{S}_{e_0,e_1}$ attains its
infimum  at a certain (non necessarily unique) ${q}(t) \in
C^1_{e_0,e_1}$, i.e. $\mathcal{S}_{e_0,e_1}[{q}]=S(e_0,e_1)$,
\item[(b)]
$S$ is lower semi-continuous.

\end{itemize}
\end{theorem}

\begin{proof}
From Eq.\ (\ref{cgf2}) and (\ref{prf2}) we have
\[
\overline{J}\backslash J \supset \{(x_0,x_1):
\liminf_{(e,e') \to (e_0,e_1)} S(e,e') \le y_1-y_0 < S(e_0,e_1)\},
\]
thus causal simplicity implies (b). Assume  $S(e_0,e_1)$, $t_0<t_1$,
is finite and for any $y_0$ define $x_0=(e_0,y_0)$ and
$x_1=(e_1,y_0+S(e_0,e_1))$. By Eq.\ (\ref{prf2}), $(x_0,x_1) \in
\overline{J}=J$. However, from Eq.\ (\ref{cgf1}) we infer
$x_1 \notin I^{+}(x_0)$ thus $x_1\in E^{+}(x_0)$ and there is an
achronal lightlike geodesic $\gamma$ connecting $x_0$ to $x_1$
necessarily (because $t_1>t_0$) with tangent vector nowhere parallel
to $n$. Let $(t,q(t))$ be the $C^1$ curve projection of the geodesic
$\gamma$ on $E$. By corollary \ref{vqz},
$\mathcal{S}_{e_0,e_1}[q]=S(e_0,e_1)$ thus (a) holds.

For the converse, since by (b) $S$ is lower semi-continuous, Eq.
(\ref{prf2}) gives
\[
\overline{J}= \{(x_0,x_1):  y_1-y_0 \ge  S(e_0,e_1)\}.
\]
Given $(x_0,x_1)\in \overline{J}$ we have by the previous
equation, $S(e_0,e_1)<+\infty$, thus either $e_0=e_1$ and hence
$(x_0,x_1)\in J$ and we have finished, or $t_0<t_1$. In this
last case, if $S(e_0,e_1)=-\infty$ then $y_1-y_0>S(e_0,e_1)$ and $x_1\in I^+(x_0)$ by Eq.\ (\ref{cgf1}). If $S(e_0,e_1)$ is finite, by (a) there is some $C^1$ curve $(t,q(t))$ which connects
$e_0$ to $e_1$ and such that
$\mathcal{S}_{e_0,e_1}[{q}]=S(e_0,e_1)$. Then the light lift
$\gamma(t)=(t,q(t),y_0+\mathcal{S}_{e_0,e(t)}[q\vert_{[t_0,t]}])$ is
a causal curve connecting $x_0$ to $(e_1,y_0+S(e_0,e_1))$, and since $y_1\ge y_0+S(e_0,e_1)$ it can be continued with the future-directed lightlike segment generated by $n$ which ends at $x_1$. Thus $(x_0,x_1)\in J$, that is,
$\overline{J}=J$.
\end{proof}

In order to assure the existence of a minimizer for a variational
problem it is common to add some coercivity assumption which
guarantees that the minimizing sequence does not escape to infinity.


\begin{proposition} \label{nwx}
Assume that $\mathcal{S}_{e_0,e_1}$ is coercive then the same is
true for $\mathcal{S}_{e_0',e_1'}$ where $e_0',e_1'\in E$ are any
points such that $e_0'\in J^+(e_0)$, $e_1'\in J^-(e_1)$,
$t_0'<t_1'$.
\end{proposition}

\begin{proof}
Let $e_n'=(t,q_n'(t))$,  $q_n':[t_0',t_1'] \to Q$, be a sequence of
$C^1$ curves connecting $e_0'$ to $e_1'$, such that
$\mathcal{S}_{e_0',e_1'}[q_n']$ is bounded from above by $C\in
\mathbb{R}$. Let $(t,\alpha(t))$ be a $C^1$ curve which connects
$e_0$ to $e_0'$ and $(t,\beta(t))$ be a $C^1$ curve which connects
$e_1'$ to $e_1$, then the curves $q_n=\beta\circ q_n'\circ \alpha$
have an action functional bounded  above by $C+
\mathcal{S}_{e_0,e_0'}[\alpha]+\mathcal{S}_{e_1',e_1}[\beta]$ thus,
by coercivity of $\mathcal{S}_{e_0,e_1}$, the curves $q_n$ and hence
the curves $q_n'$ are all contained in a compact set.
\end{proof}

\begin{proposition} \label{cor2}
If $\mathcal{S}_{e_0,e_1}$, $t_0<t_1$, is coercive, then
$S(e_0,e_1)$ is finite.
\end{proposition}

\begin{proof}
Let $q_n:[t_0,t_1] \to Q$ be a sequence such that $\lim_{n \to
+\infty} \mathcal{S}_{e_0,e_1}[q_n]=S(e_0,e_1)$. Since $t_0<t_1$, we
have that $S(e_0,e_1)<+\infty$ and hence
$\mathcal{S}_{e_0,e_1}[q_n]$ is bounded from above. By coercivity
there is a compact $K\subset Q$ such that the images of the curves
$q_n$ are all contained in $K$, and hence the corresponding curves
on $E$, $e_n(t)=(t,q_n(t))$, are all contained in the compact set
$\hat{K}=[t_0,t_1]\times K$. By compactness of $\hat{K}$ and
continuity of $a_t$, $b_t$ and $V(t,q)$ on $\hat{K}$ we can find a
($C^0$) space metric $h: K\to T^*K\otimes T^*K$, such that
$h(v,v)<a_t(v,v)$ for every $v\in TK$. Let $B$ be an upper bound for
$\sqrt{h^{-1}(b_t,b_t)}$ on $\hat{K}$, and let $\overline{V}$ be an
upper bound for $V(t,q)$ in $\hat{K}$. Using  the Cauchy-Schwarz
inequality $\vert b_t(\dot{q})\vert \le \sqrt{h^{-1}(b_t,b_t)}
\sqrt{h(\dot{q},\dot{q})}$, thus
\begin{align*}
\mathcal{S}_{e_0,e_1}[q_n]&=\int_{t_0}^{t_1} [\frac{1}{2}
a_t(\dot{q}_n,\dot{q}_n) +b_t(\dot{q}_n)-V(t,q_n(t))] \dd t\\
&\ge\int_{t_0}^{t_1} [\frac{1}{2} h(\dot{q}_n,\dot{q}_n) -B
\sqrt{h(\dot{q},\dot{q})}-\overline{V}] \dd t\\
&\ge \frac{1}{2} \frac{l_h[q_n]^2}{t_1-t_0}-B
l_h[q_n]-\overline{V}(t_1-t_0)\ge-(t_1-t_0)(\frac{B^2}{2}+\overline{V}),
\end{align*}
where $l_h[q_n]$ is the $h$-length of the path $q_n$.
\end{proof}

A spacetime is globally hyperbolic if it is causal and for every
$x_0,x_1\in M$, $J^{+}(x_0)\cap J^{-}(x_1)$ is compact
\cite{bernal06b}.
 Equivalently, global hyperbolicity can be defined as follows
\cite[Cor.\ 3.3]{minguzzi08e}:  a spacetime is globally hyperbolic
if it is non-total imprisoning and for every $x_0,x_1\in M$,
$\overline{I^{+}(x_0)\cap I^{-}(x_1)}$ is compact.

\begin{theorem} \label{glo}
The spacetime $(M,g)$ is globally hyperbolic if and only if for every
$e_0,e_1\in E$, $t_0<t_1$, 
the functional
$\mathcal{S}_{e_0,e_1}$ is coercive.
\end{theorem}

\begin{proof}
Assume $(M,g)$ is globally hyperbolic and let $e_0,e_1\in E$,
$t_0<t_1$, so that $S(e_0,e_1)<+\infty$. Choose any $y_0 \in
\mathbb{R}$ and define $x_0=(e_0,y_0)$.

Let $e_n(t)=(t,q_n(t))$ be $C^1$ curves which connect $e_0$ to $e_1$
such that $\mathcal{S}_{e_0,e_1}[q_n]\le C$ for some constant $C$.
The light lifts $\gamma_n(t)=(t,q_n(t),y_0+
\mathcal{S}_{e_0,e(t)}[q_n(t)])$ are causal curves which connect
$x_0$ to $w_n=(e_1,y_0+ \mathcal{S}_{e_0,e_1}[q_n])$.

Choose $y_1$ such that $y_1-y_0>C> S(e_0,e_1)$, and define
$x_1=(e_1,y_1)$ so that by Eq.\ (\ref{cgf1}), $x_1\in I^{+}(x_0)$.
Since $y_1>y_0+C\ge y_0+ \mathcal{S}_{e_0,e_1}[q_n]$, we have
$w_n\in J^{-}(x_1)$ (the point $x_1$ can be reached from $w_n$ by
moving forward along the fiber generated by $n$) thus the images of
the curves $\gamma_n$ are contained in the compact $J^{+}(x_0)\cap
J^{-}(x_1)$. The images of the curves $e_n$ are all contained in the
compact $\pi(J^{+}(x_0)\cap J^{-}(x_1))$, and finally if $\pi_Q:
E\to Q$ is the natural projection of the splitting $E=T\times Q$,
the images of the curves $q_n$ are all contained in the compact
$\pi_Q(\pi(J^{+}(x_0)\cap J^{-}(x_1)))$, that is,
$\mathcal{S}_{e_0,e_1}$ is coercive.


Let us prove the converse. From Theorem \ref{blf} $(M,g)$ is
non-total imprisoning (alternatively, from Proposition \ref{bfj} and
the finiteness  of $S(e_0,e_1)$ for $t_0<t_1$, it follows that $S$
is lower semi-continuous, in particular $(M,g)$ is strongly causal
(theorem \ref{psf2})  and hence non-total imprisoning
\cite{hawking73,minguzzi07f}). We have to show that for every
$x_0,x_1\in M$, $\overline{I^{+}(x_0)\cap I^{-}(x_1)}$ is compact.
Clearly, we can assume $x_0 \ll x_1$, as the empty set is compact,
thus $t_0<t_1$.

%

Let $\gamma_n$ be a sequence of  timelike  curves connecting $x_0$
to $x_1$. Since the curves $\gamma_n$ are timelike they can be
parametrized by $t$ as $\gamma_n(t)=(t,q_n(t),y_n(t))$ for $C^1$
functions $q_n$ and $y_n$. By Theorem \ref{jsc} we have
$\mathcal{S}_{e_0,e_1}[q_n]\le y_1-y_0$, and by coercivity the
images of the curves $q_n$ are all contained in a compact set
$K\subset Q$.

By Prop. \ref{cor2} $S$ is finite, and by Prop. \ref{bfj}, $S$ is
lower semi-continuous. As  $([t_0,t_1]\times K)^2$ is compact the
function $S$ reaches a minimum there which must be finite since
$S\ne -\infty$ on this set. Thus there is a constant $N<0$ such that
$S(e, e')>N$ for $(e,e') \in ([t_0,t_1]\times K)^2$.

Since $y_n(t)-y_0\ge \mathcal{S}_{e_0,e_n(t)}[q_n]\ge S(e_0,
e_n(t))>N$ and $y_1-y_n(t)\ge \mathcal{S}_{e_n(t),e_1}[q_n]\ge
S(e_n(t), e_1)>N$
 (for the first steps use Theorem \ref{jsc}) we have that
the curves $\gamma_n$ are all contained in the compact set
$[t_0,t_1]\times K \times [y_0+N,y_1-N]$.
\end{proof}

With this result we can now interpret the well known result ``global
hyperbolicity $\Rightarrow$ causal simplicity'' \cite[Prop.\
3.16]{beem96} as a typical Tonelli's type result for the existence
of minimizers \cite{buttazzo98} where the typical ingredients for
obtaining  the existence of minimizers are: (a) the coercivity and
(b) the lower semi-continuity of the action functional. Here, we do
not need to mention the latter condition because of the special form
of the Lagrangian.

\begin{corollary}
If for every $e_0,e_1\in E$, functional $\mathcal{S}_{e_0,e_1}$ is
coercive then its attains its infimum provided $t_0<t_1$.
\end{corollary}

It must be remarked that here the notion of coercivity is rather
weak and we do not have to bother on the usual variational space of
absolutely continuous curves as all our curves are $C^1$ (minimizers
being projection of geodesics are actually $C^{r+1}$).

Given our improved definition of subquadraticity the next result improves slightly that of \cite{flores03}.
\begin{corollary} \label{hyt}
Let us consider the special case $\p_t a_t=b_t=0$. If $(Q,a)$ is
complete and $V(t,q)$ is subquadratic then $(M,g)$ is globally
hyperbolic.
\end{corollary}

\begin{proof}
Follows at once from Theorem \ref{ber} and Theorem \ref{glo}.
\end{proof}

The next result implies  that sufficiently small phase sandwiches of
gravitational pp-waves are globally hyperbolic.

\begin{corollary} \label{hyy}
Let us consider the special case $\p_t a_t=b_t=0$, where $(Q,a)$ is
complete and $V(t,q)$ is almost quadratic. Then for each $t$ there
is some maximal open interval $(a,b)\ni t$ such that $(M,g)$
restricted to this  classical time (i.e.\ phase) interval is
globally hyperbolic.
\end{corollary}

\begin{proof}
By Theorem \ref{ber} case (a), choosing $I$ compact interval
containing $t$ in its interior, we can find $\epsilon$ such
$(t-\epsilon,t+\epsilon)\subset I$ and $\mathcal{S}_{e_0,e_1}$ is
coercive whenever $t_0,t_1\in I$, $t_0<t_1<t_0+\epsilon$. As a
consequence the sandwich spacetime projecting on
$(t-\epsilon/2,t+\epsilon/2)$ is globally hyperbolic by Theorem
\ref{glo}. The existence of a maximal time projectable globally
hyperbolic subset follows from a standard application of the
Hausdorff maximal chain principle on the family of open globally
hyperbolic time-projectable subsets ordered through inclusion.
\end{proof}

\begin{corollary}
If $Q$ is compact then $(M,g)$ is globally hyperbolic.
\end{corollary}

\begin{proof}
Immediate, because the functional $\mathcal{S}_{e_0,e_1}$ is clearly
coercive.
%
%
%
\end{proof}

In this section we   proved that the  causality properties of the
spacetime $(M,g)$ can be ultimately recast as properties of the
least action $S$ and the action functional $\mathcal{S}$. Therefore,
we have
 reduced the problem of determining the causal type of a
spacetime to a  problem in mechanics. We can  use here several well
known result to establish if $\mathcal{S}$ is coercive. For
instance, Tonelli's superlinearity condition
\[
L(t,q,\dot{q})\ge c_0 \,\Vert\dot{q}\Vert^{m}-c_1 \ \textrm{ for constants }
m>1, \ c_0>0, \ c_1\ge 0,\
\]
where $\Vert\cdot\Vert$ is the norm of a complete Riemannian metric on $Q$ and the inequality holds for $t$ in the relevant time interval, assures coercivity: an upper bound on the action bounds the length of the curve by H\"older's inequality, and the Hopf-Rinow theorem provides the compact set. The completeness of the metric cannot be dropped: for $Q=(0,1)$ and $L=\frac{1}{2}\dot q^2$ the curves which start and end at $1/2$ and reach $1/n$ have bounded action but are not contained in any compact subset of $Q$. Unfortunately,  this result is somewhat weak for
our purposes because it can be used only if the energy potential is
bounded from above. In this sense one needs the more accurate
analysis that we provided for the $\p_t a_t=b_t=0$ case, and which
takes advantage of the fact that $L$ is not a general Lagrangian but
one having the classical quadratic dependence on the velocities
given by Eq. (\ref{clas}).

\section{Completeness} \label{s:com}

The proof of the next result goes as in \cite[Sect.\
4.1]{minguzzi06d} where the treatment is somehow more complicated by
the presence of a conformal factor. We recall that the hypersurface
$\mathcal{N}_{t'}$ is made of the events $x$ such that $t(x)=t'$,
and it is a totally geodesic hypersurface. Notice that the geodesics
of the next proposition can be spacelike.

\begin{proposition} \label{nqa}
Every  geodesic on $(M,g)$ not tangent to (in which case it would be
entirely contained in)  some totally geodesic submanifold
$\mathcal{N}_{t}$, admits the function $t$ as affine parameter and
once so parametrized projects on a solution to the E-L equations.
\end{proposition}

\begin{proof}
Let $\gamma$ be a geodesic not tangent to $\mathcal{N}_{t'}$ for
some (and hence every) $t'$, and let $\gamma'$ be its tangent
vector, then $\dd t[\gamma']=-g(\gamma',n)=cnst.\ne 0$ because $n$
is covariantly constant, thus $t$ is an affine parameter for
$\gamma$. In what follows we shall assume that $\gamma$ is
parametrized by $t$. Given the interval  $[t_0,t_1]$ the curve
$\gamma$ being a geodesic is a stationary point for the action (we
denote with a prime differentiation with respect to the generic
parameter $\lambda$ and with a dot differentiation with respect to
$t$)
\begin{align*}
\mathcal{I}[\eta]&=\frac{1}{2}\int_{\lambda_0}^{\lambda_1}
g({\eta}',{\eta}')
\,\dd \lambda=\frac{1}{2}\int_{t_0}^{t_1} g(\dot{\eta},\dot{\eta}) (t')\, \dd t\\
&= \int_{t_0}^{t_1}[L(t,q(t),\dot{q}(t))- \dot{y}] \,(t') \,\dd t.
\end{align*}
Here we have used the fact that, since $\vert \dd t[\gamma']\vert
=cnst.> 0$ and $[t_0,t_1]$ is a compact set,  the same is true for
all the curves in a given $C^1$ variation provided the variation is
sufficiently small. That is, we can assume that all the curves in
the variation can be parametrized by $t$. The first variation of
$\mathcal{I}$ around $\gamma$ must vanish and this is true in
particular for the variations such that for all the longitudinal
curves the parametrization is such that $t'=1$. Now, for every $C^1$
variation $(t,q(t,r))$ on the base around the projection of
$\gamma=(t,q(t,0),y(t))$ one has the  $C^1$ variation of $\gamma$,
$(t,q(t,r),y(t))$, that projects on it, thus the variation of
\[
\int_{t_0}^{t_1}L(t,q(t),\dot{q}(t)) \dd t- y_1+y_0
\]
must vanish on the projection of $\gamma$, hence thesis.

\end{proof}

Since $a_t$ is non-singular  the Euler-Lagrange equations can be
regarded as describing a flow on $TE$, the so called {\em
Euler-Lagrange flow}. We shall say that this flow is complete on
$(t_0,t_1)$ if it is complete as a vector field in
$T((t_0,t_1)\times Q)$, namely, if every every local solution to the
E-L equations in the time interval $(t_0,t_1)$ can be extended to
the whole interval. We shall say that the E-L flow is complete if it
is complete on $\mathbb{R}$. Analogous definitions with a future or
past attribute can be given but will not be considered here.

\begin{proposition} \label{nqb}
Every  geodesic $\lambda \to \gamma(\lambda)$ on $(M,g)$  tangent to
(and hence contained in) some totally geodesic submanifold
$\mathcal{N}_{t}$, and not coincident with an integral line of $n$
is spacelike and projects under $\pi: M\to E$, on a geodesic
$\lambda \to  q(\lambda)$ of $(Q_t,a_t)$, $\lambda$ being an affine parameter for both. Furthermore, denoting with a prime the derivative with respect to $\lambda$, the
coordinate $y$ reads
\[
y(\lambda)=y_0+\int_{q\vert_{[0,\lambda]}} b_t+k\lambda -\frac{1}{2}\int_0^\lambda\!\! \int_0^\sigma (\p_t a)_t(q'(\rho),q'(\rho))\, \dd \rho\, \dd \sigma ,
\]
where $k\in \mathbb{R}$ is an arbitrary
constant and $(\p_t a)_t$ is the derivative of the metric with respect to $t$ evaluated at the given time. In particular, if the metric does not depend on time the last term vanishes. Conversely, every geodesic on $(Q_t,a_t)$ lifted to a
curve on $\mathcal{N}_{t}$, using  the previous expression for
$y(\lambda)$, gives a spacelike geodesic.
\end{proposition}

\begin{proof}
The geodesic $\gamma$ is a stationary point of the action functional
\begin{align*}
\mathcal{I}[\eta]&=\frac{1}{2}\int_{\lambda_0}^{\lambda_1}\!\!\!
g({\eta}',{\eta}') \,\dd \lambda=
\frac{1}{2}\int_{\lambda_0}^{\lambda_1}\!\!\!\! \dd \lambda
\,[a_t(\frac{\dd q}{\dd \lambda}, \frac{\dd q}{\dd \lambda})
\!-\!2\frac{\dd t}{\dd \lambda}  [\frac{\dd y}{\dd
\lambda}-\!b_t(\frac{\dd q}{\dd \lambda})] -2V (\frac{\dd t}{\dd
\lambda})^2],
\end{align*}
and since $\gamma$ is contained in $\mathcal{N}_t$ we already know
that $\frac{\dd t}{\dd \lambda}=0$ on it. The variation with respect
to $q$ proves that the projection $q(\lambda)$ is a geodesic on
$(Q_t,a_t)$ (the terms containing $\p_q b$ and $\p_q V$ are multiplied by $\frac{\dd t}{\dd\lambda}=0$). It cannot degenerate into a fixed point since $\gamma$
is not an integral line of $n$, thus $g(\gamma',\gamma')=a_t(q',q')>0$. The variation with respect to $y$ gives $\frac{\dd^2 t}{\dd\lambda^2}=0$, consistently with $\frac{\dd t}{\dd\lambda}=0$. The variation with
respect to $t$, evaluated at $\frac{\dd t}{\dd\lambda}=0$, gives
\[
\frac{\dd}{\dd\lambda}\Big[\frac{\dd y}{\dd
\lambda}-b_t(\frac{\dd q}{\dd \lambda})\Big]=-\frac{1}{2}(\p_t a)_t(q',q'),
\]
since the derivatives with respect to $t$ of the terms containing $b$ and $V$ are multiplied by $\frac{\dd t}{\dd\lambda}$. Integrating twice the claims follow. For instance, for $Q=\mathbb{R}$ and $g=e^{2t}\dd x^2-2\dd t\dd y$, the geodesic on $t=0$ with $x(\lambda)=\lambda$ has $y(\lambda)=y_0+k\lambda-\lambda^2/2$.

\end{proof}

%
%

\begin{theorem} \label{rco}
The E.-L. flow is complete if and only if the spacetime $(M,g)$  is null
geodesically complete.  In this case $(M,g)$ is causally
geodesically complete.

The E-L flow is complete and the Riemann spaces $(Q_t, a_t)$ are
complete if and only if $(M,g)$ is geodesically complete.
\end{theorem}

\begin{proof}
The lightlike geodesic generated by $n=\p_y$ are complete because
the coordinate $y$ is an affine parameter and it takes values in
$\mathbb{R}$. The causal geodesics  not generated by $n$ admit $t$
as affine parameter, project to solutions of the E.-L. equations and
are complete if so are their projections. Indeed, any causal
geodesic $\gamma$ is a lift of its projection where the coordinate
$y$ has dependence (see Eq.\ (\ref{mod}))
$y(t)=y_0+\mathcal{S}_{e_0,e(t)}[q\vert_{[t_0,t]}]-\frac{g(\dot{\gamma},\dot{\gamma}}{2}(t-t_0)$.
Moreover, any solution of the E-L equations has a light lift which
is a lightlike geodesic not coincident with a flow line of $n$.
These facts imply the  E.-L. flow is complete if and only if the spacetime
$(M,g)$  is null geodesically complete, and that in this case
$(M,g)$ is actually causally geodesically complete.

The last statement follows from propositions \ref{nqa} and
\ref{nqb}: in the latter case the projection is a geodesic of the complete metric $a_t$, and the coordinate $y$, given by the formula of Proposition \ref{nqb}, is finite on every finite interval of the affine parameter.
\end{proof}






\subsection{Space imprisonment and completeness} \label{ss:imp}

In the following $I\subset \mathbb{R}$ is an interval of the time
axis (not necessarily proper). For short, by solution $q: I \to Q$
to Euler-Lagrange equations (ELE, Eq.\ (\ref{ele})) we mean a curve
which solves ELE and which is at least $C^2$, and hence $C^{r+1}$,
over $I$.

%

\begin{definition}
A solution to ELE $q: I \to  Q$ is future (past) {\em inextendible}
or {\em maximal} if no other solution $q' : I' \to Q$ coinciding
with $q(t)$ on $I\cap I'$ exists such that there is a $t' \in I'$,
$t'>t$ (resp. $t'<t$) for any $t \in I$. A solution which is both
past and future inextendible is said to be inextendible.
\end{definition}

The ELE satisfies the hypothesis of the   Picard-Lindel\"of theorem
\cite{cronin94}, thus joining the uniqueness properties of the
solution with Zorn lemma, we obtain that through  each point on $TQ$
passes one and only one inextendible solution to the ELE. Note that
if $q(t)$ is future (past) inextendible then $\sup_I t \notin I$
($\inf_I t \notin I$).

\begin{definition}
If $I$ is not bounded from above (below) then we say that the
solution to the ELE $q:I \to  Q$ is future (past) {\em complete}.
\end{definition}

\begin{definition} \label{tcom}
The dynamical system is $[t_0,t_1]$-{\em complete} if every solution
to the ELE, $q: I \to  Q$, such that $I\cap [t_0,t_1] \ne \emptyset
$, can be extended to $I'$, $[t_0,t_1] \subseteq I'$. The dynamical
system is {\em complete} or {\em singularity free} if it is
$[t_0,t_1]$-{\em complete} for arbitrary $t_0, t_1$.
\end{definition}

\begin{definition}
If the inextendible solution to the ELE $q:I \to  Q$ has an image
contained in a compact $K$ for $t>t_K$ ($t<t_K$) then we say that
$q(t)$ is future (past) {\em imprisoned}.
\end{definition}

\begin{definition}
The inextendible solution to the ELE $q:I \to  Q$ is said to be
future (past) {\em recurrent} or {\em partially imprisoned} if there
is a compact $K$ such that no matter how large $t< \bar{t}=\sup_I t$
there is a $t'
>t$ such that $q(t) \in K$.
\end{definition}

\begin{remark}
In the Alexandrov one-point compactification,  a solution to the ELE
is not future (past) imprisoned if and only if it has $\{\infty \}$ as
accumulation point in the future (past) time direction, and it is
not future (past) recurrent if it has $\{\infty \}$ as limit point
in the future (past) time direction.
\end{remark}

\begin{remark}
Although the inextedible solution to the ELE  $q:I \to Q$ is not
future imprisoned there may well exist a compact $K$ such that no
matter how large  $t < \bar{t}=\sup_I t$, the curve returns in $K$
for a suitable $t'>t$. However, if $q(t)$ is incomplete,
 this fact implies that the velocity must grow towards infinity as
 $t \to \bar{t}<+\infty$.
\end{remark}

\begin{lemma} \label{tem}
Let $\tilde{K}=[t_0,\bar{t}] \times K$, $K\subset Q$ a compact, and
let $B: \tilde{K} \to \overbrace{T^{*}Q\otimes \cdots \otimes
T^{*}Q}^{m \ {\rm factors}} $, , be a continuous time dependent
covariant tensor. Moreover,   let $h: \tilde{Q} \to T^{*}Q\otimes
T^{*}Q$ be a continuous time dependent metric tensor on $Q$, then
there is a constant $C>0$ such that for any $q=(t,s) \in \tilde{K}$,
$v \in TQ_q$, it is $B(q)(v, \cdots,v) \le C \{h(q)(v,v)\}^{m/2}$.
In particular any pair of metrics $h$, $h'$ of the above form is
Lipschitz equivalent, that is, there are constant $C,C'>0$, such
that for any $e=(t,q) \in \tilde{K}$, $v \in TQ_q$, it is $
\frac{1}{C}h'(v,v)\le h(v,v)\le C' h'(v,v)$.
\end{lemma}

\begin{proof}
For any point $e=(t,q) \in\tilde{K}$ consider  the continuous
function $B(u,\cdots,u)$ on the  the compact set made of unitary
vectors $u \in TQ_q$, $h(e)(u,u)=1$. $B(u,\cdot,u)$ reaches a
maximum $M$ and hence choosing $C_e>M$, we have $B(u,\cdot,u)<C_e$.
By continuity the same inequality must hold, with the same $C_e$, in
an open neighborhood of $e$. Thus since $[t_0,\bar{t}]\times K$ is
compact there is a $C$ such that for any $e \in \tilde{K}$,  $v \in
TQ_q$, it is $B(v/\sqrt{h(v,v)},\cdot,v/\sqrt{h(v,v)})<C$ from which
thesis follows.
\end{proof}

\begin{theorem} \label{jhy}
If the future (past) inextendible solution $q: I \to Q$ to ELE is
future (past) incomplete, then $q:I \to Q$ is not future imprisoned
(i.e. it escapes every compact). In particular if $Q$ is compact
every inextendible solution to ELE is complete. Finally, if the
system is also autonomous then $q(t)$ is not recurrent, that is it
coverges to the boundary $\{\infty \}$ in the Alexandrov topology.
\end{theorem}

\begin{proof}
We give the proof in the future case, the past case being analogous.
Let $t_0 \in I$, $A^{+}=[t_0,+\infty)\cap I$ and define
$q^{+}=q\vert_{A^{+}}$, so that $q(t)$ is future inextendible if and only if
$q^{+}(t)$ is future inextendible.

We mentioned that the integration of the ELE is equivalent to the
integration of a time dependent Hamiltonian vector field $X$ on
$T^{*}Q$, or equivalently to the integration of the field $(1,X)$ on
$\mathbb{R}\times T^{*}Q$, see Eqs.\ (\ref{jdb}), (\ref{jdc}) and
(\ref{jdd}), that is, we can introduce a new time parameter $\tau$
and a new differential equation $\dd t/\dd \tau=1$ to obtain a
vector field independent of $\tau$ on a $(d+2)$-dimensional
manifold. Then \cite[Prop.\ 1.10, Chapt.\ 5]{godbillon69} proves
that if the closure of the image of $(q^{+},\dot{q}^{+}): A^{+} \to
TQ$ is compact then $I$ in not bounded from above.

Let us assume by contradiction that $\sup_{I} t=\bar{t}<+\infty$ and
there is a compact $K \subset Q$ and a $t_K \in [t_0, \bar{t})$,
such that for $t>t_K$, $q^{+}(t) \in K$. The ELE (\ref{ele}) implies
\begin{equation}
\frac{\dd }{\dd t} a(t)(\dot q, \dot q) = - \{ (\p_t a)(\dot q,\dot
q)+2(\p_t b)(\dot q)+2\dd V(\dot q) \} .
\end{equation}
Here  $\p_t a$, $\p_t b$ and $\dd V$ are continuous covariant time
dependent tensor fields on $\tilde{K}$ thus we can apply Lemma
\ref{tem} to obtain  that  suitable constants $C_0,C_2>0$ exist such
that
\begin{eqnarray*}
\vert \frac{\dd }{\dd t} a(t)(\dot q, \dot q) \vert\le \vert (\p_t
a)(\dot q,\dot q)\vert+2 \vert (\p_t b)(\dot q)\vert+2 \vert\dd
V(\dot q) \vert \le C_0+C_2 {a(t)(\dot q, \dot q)} .
\end{eqnarray*}
Then \[ a(t)(\dot q, \dot q)\le \frac{C_0}{C_2}(e^{C_0C_2
(t-t_k)}-1)+a(t_K)(\dot q, \dot q) e^{C_0C_2 (t-t_k)}
\]
which implies, since $t$ is bounded from above, that  $a(\dot q,\dot
q)$ is bounded and hence that $(q^{+},\dot{q}^{+}): A^{+} \to TQ$
can not escape every bounded set on $TQ$ in its domain of
definition. The contradiction proves that $q^{+}(t)$ cannot be
future imprisoned in a compact set (if $Q$ is compact the
contradiction proves that no inextendible solution to ELE is future
or past incomplete).

Let us show that for an autonomous system the future incomplete
inextendible curve $q(t)$ cannot be recurrent. Of course this
hypothesis makes sense only if $Q$ is not compact. Let $K$ be a
compact set and assume that $q(t)$ returns to $K$ indefinitely.
Since $Q$ is not compact we can find an open set $A$ of compact
closure $\overline{A} \ne Q$ such that $K \subset A$. Let $h$ be a
Riemannian metric on $Q$ and denote with $d_{h}$ the distance. Then
$d_h(K, \overline{A^{C}})>\epsilon>0$ being the distance between a
compact and a disjoined closed set is positive. Note that $q(t)$
must leave indefinitely the compact $\overline{A}$ and then return
indefinitely to $K \subset \overline{A}$ (which implies that $q(t)$
is recurrent also with respect to $\overline{A}$), thus the image of
the curve inside $\overline{A}$ has infinite length as it goes from
$K$ to $A^{C}$ and back indefinitely. Nevertheless, it does this
oscillation in a finite time $\le \vert \bar{t}-t_0 \vert$, a fact
which implies that a sequence $t_n \to \bar{t}$ must exist such that
$h(\dot q(t_n),\dot q(t_n)) \to +\infty$, $q(t_n) \in \overline{A}$.
But this result is in contradiction with the conservation of energy
(as it follows from ELE)
\begin{equation}
\frac{1}{2} a(\dot q,\dot q) +V= E
\end{equation}
which implies that on the compact $\overline{A}$, $a(\dot q,\dot
q)<C>0$ for a suitable constant $C$ and because of Lemma \ref{tem},
$h(\dot q(t_n),\dot q(t_n))<C'>0$.
\end{proof}

%
%

\begin{corollary}
If $Q$ is compact then $(M,g)$ is geodesically complete.
\end{corollary}

\begin{proof}
If $Q$ is compact then $(Q,a_t)$ is complete for every $t$ by the
Hopt-Rinow theorem. The E.-L.\ flow is complete by Theorem \ref{jhy}
thus, by Theorem \ref{rco}, $(M,g)$ is geodesically complete.
\end{proof}

\begin{definition}
A point $p \in Q$ is a future (past) {\em endpoint} for the curve
$q:I \to Q$,  if for every neighborhood $\mathcal{U}$ of $q$, there
is a $t(\mathcal{U})\in I$ such that for $t'\ge t$ (resp. $t'\le
t$), $q(t) \in \mathcal{U}$.
\end{definition}

\begin{proposition} \label{end}
Let  $q: I \to  S$ be a future (past) inextendible and future (past)
incomplete solution to the ELE, then it has no future (past)
endpoint.
\end{proposition}

\begin{proof}
Indeed, if $p$ is a future endpoint then $q(t)$ is future imprisoned
in the closure of a neighborhood of $p$, a contradiction with
theorem \ref{jhy}.
\end{proof}

\subsection{Completeness and least action finiteness imply coercivity} \label{ss:coe}

Suppose that coercivity fails at $(e_0, e_1)$, then there is some
constant $K>0$ and  a sequence of $C^1$ curves
$\alpha_n:[t_0,t_1]\to Q$, connecting $q_0$ to $q_1$, not all
contained in a compact set such that
$\mathcal{S}_{e_0,e_1}[\alpha_n]\le K$. From this fact it follows
that $S(e_0,e_1)\le K$. We define $C(e_0,e_1)$ to be the infimum of
all the constants $K$ with such non coercivity property. If
coercivity holds at $(e_0, e_1)$ then we set $C(e_0,e_1)=+\infty$.
By construction $S(e_0,e_1)\le C(e_0,e_1)$.

\begin{lemma} \label{lem}
Let $e_0,e_1\in E$ with $t_0<t_1$ be such that $S(e_0,e_1)$ is
finite and assume that the least action $S$ is lower semi-continuous
at $(e_0,e_1)$. Suppose that coercivity fails at $(e_0, e_1)$, then
for every $y_0,y_1\in \mathbb{R}$, such that $y_1-y_0=C(e_0,e_1)$
defined $x_0=(e_0,y_0)$, $x_1=(e_1,y_1)$, there is a past
inextendible lightlike  ray $\gamma^{x_1}$, ending at $x_1$ and
contained in $\overline{J^{+}(x_0)}$.
\end{lemma}

\begin{proof}
Since  $S(e_0,e_1)\le C(e_0,e_1)$ we have that  $C$ is finite. Let
$H_i$ be a sequence of compact sets $H_i\subset H_{i+1}$ such that
$\cup_i H_i=Q$ and every compact set is contained in some $H_i$ (the
existence of such a sequence characterizes the property of
hemicompactness which is implied by the second countability and the
local compactness of $Q$). Let $K_i>C$, $K_i\to C$, and take for
each $i$, a $C^1$ curve $q_i(t)$ which is not entirely contained in
the compact $H_i$ and is such that $\mathcal{S}_{e_0,e_1}[q_i]\le
K_i$.

The $C^1$ curves  $e_n(t)=(t,q_n(t))$ connect $e_0$ to $e_1$ and are
such that \[\liminf_n \mathcal{S}_{e_0,e_1}[q_n]\le  C,\]  and the
curves $q_n$ are not all contained in a compact set. Actually,
\[\liminf_n \mathcal{S}_{e_0,e_1}[q_n]=  C,\] otherwise we could find a
subsequence which escapes every compact and whose action is bounded
by a constant smaller than $C$, in contradiction with the definition
of $C$. The inequality $\limsup_n \mathcal{S}_{e_0,e_1}[q_n]\le  C$
follows from $\mathcal{S}_{e_0,e_1}[q_i]\le K_i$, thus $\lim
\mathcal{S}_{e_0,e_1}[q_n]=  C$.

The light lifts $\gamma_n(t)=(t,q_n(t),y_0+
\mathcal{S}_{e_0,e(t)}[q_n(t)])$ are causal curves which connect
$x_0$ to $w_n=(e_1,y_0+ \mathcal{S}_{e_0,e_1}[q_n])$ and  are not
all contained in a compact set. Choose $y_1=y_0+C$ and define
$x_1=(e_1,y_1)$ so that $w_n \in
 J^{+}(x_0)$ and $w_n\to x_1$.
Since the curves $\gamma_n$ are not all contained in a compact, by
the limit curve theorem \cite{minguzzi07c} there is
some past inextendible continuous causal curve $\gamma^{x_1}$ to
which some subsequence $\gamma_k$ converges uniformly on compact
subsets (e.g. with respect to the parametrizations  obtained from a
complete Riemannian metric measuring the length of the curve from
its future endpoint, see \cite{minguzzi07c}). Moreover, for every
$y\in \gamma^{x_1}$, we have $y\in \overline{J^{+}(x_0)}$. Now note
that if $\gamma^{x_1}$ were not achronal then there would be some
$z\in \gamma^{x_1}$ such that $z \in I^{-}(x_1)$. It would be
possible to find $w$, such that $z\ll w\ll x_1$ and to join $x_0$ to
$x_1$ with a sequence of causal curves $\eta_n$ which follow
$\gamma_n$ till they reach a neighborhood of $z$ in the past of $w$,
pass through $w$ and take the same final timelike path from $w$ to
$x_1$. Now observe that the curves $\gamma_k$ converge uniformly to
$\gamma^{x_1}$ on compact subsets thus they escape every compact set
before coming close to $z$, and so the same can be said of the
sequence $\eta_n$. As their last timelike segment from $w$ to $x_1$
is independent of $n$ and timelike, they can be deformed to give a
sequence $\beta_n=(t,r_n(t),s_n(t))$ of timelike curves with
projections $r_n(t)$ escaping every compact and connecting $x_0$
with some point $x_1'$ before $x_1$ in the lightlike fiber of $e_1$.
Integrating Eq. (\ref{mod}) we have $C=y_1-y_0>y_1'-y_0\ge
\mathcal{S}_{e_0,e_1}[r_n]$. We have therefore shown that $C$ is not
the infimum of all the constants with the above mentioned non
coercivity property. This contradiction proves that $\gamma^{x_1}$
is achronal and hence a past lightlike ray.
\end{proof}

\begin{theorem} \label{nqz}
Let $e_0,e_1\in E$ with $t_0<t_1$ be such that $S(e_0,e_1)$ is
finite. The $[t_0,t_1]$-completeness of the E-L flow (Definition \ref{tcom}) and the lower
semi-continuity of the least action $S$ at $(e_0,e_1)$ implies the
coercivity of $\mathcal{S}_{e_0,e_1}$ and the fact that the infimum
$S(e_0,e_1)$ is attained by some minimizer.
\end{theorem}

\begin{proof}
Assume by contradiction that $\mathcal{S}_{e_0,e_1}$ is not
coercive. From Lemma \ref{lem} we infer the existence of a past
lightlike ray $\gamma^{x_1}$ ending at $x_1$ and contained in
$\overline{J^{+}(x_0)}$.

However, as $\gamma^{x_1}\subset \overline{J^{+}(x_0)}$ by Eq.\
(\ref{prf1}) and the lower semi-continuity of $S$ it follows that
$\gamma^{x_1}$ is nowhere tangent to $n$, thus, by Proposition
\ref{nqa}, it admits as affine parameter the function $t$. Its projection is a solution of the E.-L.\ equations defined on an interval which contains $t_1$, thus by $[t_0,t_1]$-completeness it extends to an open interval containing $[t_0,t_1]$, and the past inextendibility implies that the affine
parameter $t$ takes values smaller than $t_0$ which is impossible
again by $\gamma^{x_1}\subset \overline{J^{+}(x_0)}$. The
contradiction proves that $\mathcal{S}_{e_0,e_1}$ is coercive.

We know that $S(e_0,e_1)$ is finite, $S$ is lower semi-continuous at
$(e_0,e_1)$ and $\mathcal{S}_{e_0,e_1}$ is coercive. Let (they are
not the same curves as above) $e(t)=(t,q_n(t))$ be a sequence of
$C^1$ curves which connect $e_0$ to $e_1$ and such that
$\mathcal{S}_{e_0,e_1}[q_n]\to S(e_0,e_1)$. By coercivity the curves
$q_n$ are all contained in a compact set $K\subset Q$. Let $y_0\in
\mathbb{R}$; the light lifts
$\gamma_n(t)=(t,q_n(t),y_0+\mathcal{S}_{e_0,e(t)}[q_n\vert_{[t_0,t]}])$
connect $x_0=(e_0,y_0)$ to
$z_n=(e_1,y_0+\mathcal{S}_{e_0,e_1}[q_n])$. Thus if
$x_1=(e_1,y_0+S(e_0,e_1))$ we have as
$y_0+\mathcal{S}_{e_0,e_1}[q_n]\ge y_0+S(e_0,e_1)$ that $z_n\in
J^{+}(x_1)\cap J^{+}(x_0)$ and $z_n \to x_1$. By the limit curve
theorem \cite{minguzzi07c}  either there is a continuous causal
curve joining $x_0$ to $x_1$ or there is a past inextendible
continuous causal curve $\gamma^{x_1}$ ending at $x_1$ and such that
$\gamma^{x_1}\subset \overline{J^{+}(x_0)}$. Let us consider the
latter case. The curve $\gamma^{x_1}$ must be achronal otherwise it
would be possible to connect $x_0$ to $x_1$ with a timelike curve in
contradiction with Eq.\ (\ref{cgf1}). Thus $\gamma^{x_1}$ is a past
lightlike ray which by Eq.\ (\ref{prf1}) and the lower
semi-continuity of $S$  is nowhere tangent to $n$, thus, by
proposition \ref{nqa}, it admits as affine parameter the function
$t$. As above, the $[t_0,t_1]$-completeness and the past inextendibility imply
that the affine parameter $t$ takes values smaller than $t_0$
which is impossible again by $\gamma^{x_1}\subset
\overline{J^{+}(x_0)}$. The  contradiction proves that the latter
case does not apply and that there is a causal curve joining $x_0$
to $x_1$. Again, this curve must be an achronal lightlike geodesic
$\gamma(t)=(t,q(t),y(t))$, otherwise $x_0$ and $x_1$ would be
connected by a timelike curve in contradiction with Eq.\
(\ref{cgf1}). Therefore, the curve $\gamma(t)$ is the light lift of
$(t,q(t))$, thus $y_0+\mathcal{S}_{e_0,e_1}[q]=y_1=y_0+S(e_0,e_1)$
and hence thesis.

\end{proof}

\begin{corollary} \label{dfc}
If the E-L flow is complete and for every $e_0,e_1\in E$, such that
$t_0<t_1$, we have that $S$ is finite, then $\mathcal{S}_{e_0,e_1}$
is coercive and $(M,g)$ is globally hyperbolic.
\end{corollary}

\begin{proof}
The finiteness of $S$ implies its lower semi-continuity (Prop.
\ref{bfj}), and lower semi-continuity, the finiteness and the
completeness properties imply coercivity (Theor. \ref{nqz}) hence
thesis.
\end{proof}

%
%

\subsection{Completeness of maximal solutions to ELE} \label{ss:cele}

Let $(Q,a_t)$ be a manifold endowed with a time dependent complete
Riemannian metric. In a recent work \cite{minguzzi13} we established
a result on the completeness of maximal solutions to ordinary second
order equations over $(Q,a_t)$, namely to equations of the form
\[
\frac{D^{(t)} }{\dd t}\, \dot{q}=f(t,q(t), \dot{q}(t)).
\]
where,  for any given $t$, $D^{(t)}$ is the connection of $(Q,a_t)$.

As an application we obtained the next results on the completeness
of maximal solutions to the ELE (\ref{ele}).

\begin{theorem} \label{bix}
Let ${Q}$ be a $d$-dimensional manifold (the space) endowed with the
(possibly time dependent) positive definite metric $a_t$, 1-form
field $b_t$ and potential function $V(t,q)$ (all $C^r$, $r \ge 2$).
Let us suppose that $(Q,a_t)$ are complete for each $t$ and let
$\rho_t(p,q)$ be the corresponding Riemannian distance. Let us fix
$p\in Q$ and let us suppose that for each time interval $[-r,r]$
there is a continuous non-decreasing  function $g_r:[1,+\infty) \to
[1,+\infty)$ with the property that
\[
G_r(y):=\int_1^y \frac{1}{x g_r(x)}\, \dd x ,
\]
diverges for $y\to +\infty$ (the typical choice will be
$g_r=cnst.\ge 1$) and such that for every $(t,q,v)\in [-r,r]\times
TQ$
\begin{align}
\pm[ (\partial_t a_t)(v,v) ]&\le 2 g_r(1+\rho_t(p,q)^2) \,
a_t(v,v),\\
\Vert (\p_t b_t)(t,q) \Vert_t,\, \Vert (\p_q V)(t,q)\Vert_t &\le
g_r(1+\rho_t(p,q)^2) [1+\rho_t(p,q)],
\end{align}
where for $w \in TM$, $\Vert w\Vert_t:= \sqrt{a_t(w,w)}$. Then the
E.-L.\ flow given by Eq.\ (\ref{ele}) is complete in the forward
(resp.\ backward) direction.
\end{theorem}

Thus theorem establishes some sufficient asymptotic  behaviors
for $\partial_q V$, $\partial_t b_t$, and $a_t$. Typically for the
two former quantities a linear dependence on the distance will be
sufficient in order to obtain completeness, as one gets from the
choice $g_r=cnst.>1$, however, faster than linear behaviors are also
possible, as one can infer strengthening the just given conclusions
with the choices $g=\ln(\eta+ x)$, $g=\ln(\eta+ x) \ln(\eta+
\ln(\eta+ x))$ and so on,  where $\eta=e-1$.

\subsubsection{Homogeneous pp-waves}

A metric of the form
\begin{align}
g&=a_{ij}(t) \dd q^i \dd q^j -\dd t \otimes (\dd y-\frac{1}{2}
M_{ij}(t)\, q^i \dd q^j) -(\dd y-\frac{1}{2} M_{ij}(t)\, q^i \dd
q^j) \otimes \dd t \nonumber \\& \qquad -2(h_{ij}(t)\, q^i q^j) \,
\dd t ^2 , \label{eis2}
\end{align}
where $q\in Q=\mathbb{R}^d$, and $a_{ij}(t)$ is positive definite
for each $t$, is clearly homogeneous under replacements $q^i \to
\lambda q^i$. These spacetimes will be called {\em homogeneous
pp-waves}.

According to the next lemma, a global coordinate change simplifies
their form.

\begin{lemma} \label{lem:hom}
Suppose that  $Q=\mathbb{R}^d$, $a_t=a_{ij}(t) \dd q^i \dd q^j$, and
that in a  coordinate system $\{q^i\}$,
\begin{align*}
V(t,q)&=h_{ij}(t)\, q^i q^j, \\
b_{t\, j}&=\frac{1}{2} M_{ij}(t)\, q^i,
\end{align*}
where $h_{ij}(t)$ is  symmetric. Through  a change of gauge and
($C^r$) time-dependent  linear transformation of the coordinates we
can make $a_{ij}(t)=\delta_{ij}$, $M_{ij}=0$, so that in the new
coordinate system the E.-L. equations (\ref{ele}) take the form
\begin{equation} \label{elo}
\ddot q^i=-2 h_{ij}(t)\, q^j,
\end{equation}
while the line element becomes
\begin{equation}
g=\delta_{ij} \dd q^i \dd q^j \!-\!\dd t \otimes \dd y -\!\dd y
\otimes \dd t-2(h_{ij}(t)\, q^i q^j) \dd t ^2 .
\end{equation}
\end{lemma}

\begin{proof}
Let $\{e_i\}$ be the canonical basis of $\mathbb{R}^d$. Through the
Gram-Schmidt process we can find a new basis $\{e_i'(t)\}$ which is
orthonormal with respect to the scalar product  $a_t$. By
construction the components of $e_i'(t)$ are  functions of the
matrix coefficients $a_{ij}(t)$ thus $C^r$ in $t$. The $C^r$ matrix
$U_{ij}(t)$ such that $e_{j}'(t)=U_{ij}(t) e_i$ is also such that
$U_{ks}(t) a_{ki}(t)U_{ij}(t)=\delta_{sj}$. Thus the replacement
$q^i= U_{ij}(t) {q'}^j$ shows that we can indeed assume
$a_{ij}(t)=\delta_{ij}$ (observe that the replacement does indeed
preserve the symmetry of $h$).

Let $M_{ij}(t)=F_{ij}(t)+S_{ij}(t)$ where $F_{ij}$ is antisymmetric
and $S_{ij}$ is symmetric. A redefinition of the $y$ coordinate,
$y=y'+\frac{1}{4} S_{ij}(t)q^i q^j$, shows that we can assume,
without loss of generality, $S_{ij}=0$.

Let $O(t)$ be the orthogonal matrix such that
$\dot{O}=-\frac{1}{2} O F$ (by Wintner's theorem on first order
differential equations \cite{hartman64}, there is a complete
solution which is necessarily $C^{r+1}$) and define $q'(t,q)=O(t)q$
so that $\dd q^j=\frac{1}{2} F_{jk}q^k \dd t+ (O^{-1})_{ji}\dd
q'^i$. We obtain
\begin{align*}
g&=\dd q' \cdot \dd q'-\dd t\otimes [\dd y-\frac{1}{2} F_{jk}q^k
O_{sj} \dd
q'^s-b_t] -[\dd y-\frac{1}{2} F_{jk}q^k  O_{sj} \dd q'^s-b_t]\otimes \dd t\\
&\quad -2[V(t,q) +\frac{1}{8} q^i F_{ij}F_{jk}q^k] \dd t^2,
\end{align*}
thus, working out $b_t$ we obtain
\begin{align*}
V'(t,q')&= h'_{ij}(t) \, q'^i q'^j, \qquad H'=O(H-\frac{1}{8} F^2)O^{-1},  \\
b'_{t\, i}&=0.
\end{align*}
\end{proof}

%

The next result establishes that the homogeneous $pp$-waves  are
complete.

\begin{corollary} \label{ifg}
Any homogeneous pp-wave is stably causal. Furthermore, the E.-L.
flow and the spacetime $(M,g)$ are complete.
\end{corollary}

\begin{proof}
By Lemma \ref{lem:hom} any homogeneous pp-wave admits coordinates for
which $a_{ij}=\delta_{ij}$ and $M_{ij}=0$. As a consequence, these
spacetimes are stably causal because the potential is almost
quadratic. By Theorem \ref{bix} the spacetime and the E.-L. flow are
complete.
\end{proof}

 The next result establishes that for homogeneous pp-waves
causal simplicity and global hyperbolicity coincide.

\begin{corollary}
Let $(M,g)$ be a homogeneous pp-wave. Let $e_0=(0,t_0)$ and
$e_1=(0,t_1)$, $t_0<t_1$ (i.e.\ events whose projections in the
quotient space $Q=\mathbb{R}^d$ stay in the origin), then either
$S(e_0,e_1)=0$ or $S(e_0,e_1)=-\infty$. In particular, the lower
semi-continuity of the least action $S(e_0,e_1)$ implies its
finiteness whenever $t_0<t_1$. As a consequence, the causally simple
homogeneous pp-waves are globally hyperbolic.
\end{corollary}

\begin{proof}
Let $q(t):[t_0,t_1]\to Q$ be a $C^1$ curve which starts and returns
to the origin of the Cartesian coordinate system. If
$\mathcal{S}_{e_0,e_1}[q]<-\epsilon<0$ then defined $q_s(t)=sq(t)$
we have $\mathcal{S}_{e_0,e_1}[q_s]=s^2
\mathcal{S}_{e_0,e_1}[q]<-s^2\epsilon$, from which it follows that
$S(e_0,e_1)=-\infty$.


We recall that we can suppose $M_{ij}=0$, as we can accomplish this
condition with a time dependent linear coordinate transformation.
Suppose that there are events $e'_0,e'_1$ such that
$S(e_0',e_1')=-\infty$, then by the triangle inequality for the
least action (Cor.\ \ref{hgf}) we can find events of the form
$e_0=(0,t_0)$, $e_1=(0,t_1)$, $t_0<t_1$, such that
$S(e_0,e_1)=-\infty$. If we decrease the value of $t_1$ then we must
eventually obtain $S(e_0,e_1')=0$ for some $t_1'>t_0$ since $V$ is
almost quadratic and hence $(M,g)$ is strongly causal, that is, $S$
is lower semi-continuous on the diagonal. The supremum of the
numbers  $t_1'$ identifies a point $\hat{e}=(0,\hat{t})$ for which
$S(e_0,\hat{e})=0$. For, if not, we would have
$S(e_0,\hat{e})=-\infty$, hence $(\hat{e},0)\in I^{+}((e_0,0))$
which implies, as $I$ is open, that for some $t_1'<\hat{t}$,
$(0,t_1',0)\in I^{+}((e_0,0))$ thus $S(e_0,e_1')<0$, a
contradiction. Thus $S(e_0,(0,t,0))=0$ for $t\le [t_0,\hat{t}]$ and
$S(e_0,(0,t,0))=-\infty$ for $t>\hat{t}$, a contradiction with the
lower semi-continuity of $S$.

If $(M,g)$ is causally simple then $S$ is lower semi-continuous and
hence, by the just proved result, finite. As completeness, is
assured by Cor.\ \ref{ifg} we have
 by Cor.\ \ref{dfc} that $(M,g)$ is globally hyperbolic.
\end{proof}

This result suggests that it should not be so simple to construct
causally simple but non-globally hyperbolic spacetimes $(M,g)$.
Nevertheless, in section \ref{yak} we shall give some interesting
examples which solve the Einstein equations.

\subsection{Conjugate points on homogeneous pp-waves} \label{ss:conj}

The quadratic dependence of the metric on $q^i$  implies that the
first variation of the Lagragian (and hence the ELE) is linear on
the same variable. This is pretty clear from the expression
(\ref{elo}) in the simplifying coordinate system. As the system is
linear in $q^i$ the first variation of the ELE, and hence the second
of the classical action, returns an equation which coincides with
the ELE (we give it in a simplifying coordinate system)
\[
\frac{\dd^2 \delta q^i}{\dd t^2}=-2 h_{ij}(t)\, \delta q^j.
\]
As usual, a solution $q(t)$ to ELE is said to admit {\em conjugate
points} $q=q(t_0)$, $q_1=q(t_1)$, $t_0<t_1$, if there is a
non-trivial solution $\delta q(t)$ to the first variation to ELE
evaluated on $q(t)$ which vanishes at $t_0$ and $t_1$.

The linearity of ELE brings several simplifications. For instance,
the notion of conjugate points is independent of the solution $q(t)$
considered, and hence we can speak of {\em conjugate instants}
$t_0$, $t_1$.

\begin{lemma}
If $q(t)$ and $q'(t)$ are solution to the linear ELE, then $q(t_1)$
and $q(t_2)$ are conjugate if and only if so are $q'(t_1)$ and
$q'(t_2)$.
\end{lemma}

\begin{proof}
It follows from the fact that the first variation to ELE does not
need to be evaluated on a solution to ELE, as it is already
independent of $q^i$.
\end{proof}

In the proofs of this section we shall use the simplifying coordinate system of Lemma \ref{lem:hom}, for which
$a_{ij}=\delta_{ij}$, $M_{ij}=0$. This is not restrictive since
in order to obtain it we make time-dependent linear coordinate
changes, which do not change the value of the action functional nor the notion of conjugate instants, and a change of gauge, which changes the action functional by the addition of a continuous function of the endpoints (recall
that up to a constant the action is the coordinate $y$ of
the light lift endpoint, and thus it transforms with the gauge
transformation for the coordinate $y$, Eq.\ (\ref{jiy})).

\begin{lemma} \label{ual}
Let $q(t)$ be a a solution to ELE on the interval $[t_0,t_1]$ such
that $q(t_0)=q(t_1)=0$, that is it starts and ends at the origin of
the quotient space $Q=\mathbb{R}^d$. Then $\int_{t_0}^{t_1}
L(t,q(t),\dot q(t))\,\dd t=0$.
\end{lemma}

\begin{proof}
By integration by parts
\begin{align*}
\mathcal{S}_{e_0,e_1}[ q]&=\int_{t_0}^{t_1}[ \frac{1}{2}
\dot{ q}\cdot \dot{ q}-h_{ij} q^i  q^j] \,
\dd t=\int_{t_0}^{t_1}[ -\frac{1}{2} q\cdot \ddot{
q}-h_{ij} q^i  q^j] \,\dd  t\\
&= \int_{t_0}^{t_1}[ h_{ij} q^i  q^j-h_{ij} q^i
 q^j] \dd  t=0.
\end{align*}
\end{proof}

The next lemma collects the consequences of the classical Jacobi theory for quadratic functionals.

\begin{lemma} \label{jac}
Let $t_0<t_1$, $e_0=(t_0,q_0)$, $e_1=(t_1,q_1)$.
\begin{itemize}
\item[(a)] If no instant in $(t_0,t_1]$ is conjugate to $t_0$ (equivalently, if no instant in $[t_0,t_1)$ is conjugate to $t_1$) then for every $q_0,q_1$ there is one and only one solution to ELE which connects $e_0$ to $e_1$, this solution is the unique minimizer of $\mathcal{S}_{e_0,e_1}$, and $S(e_0,e_1)$ is finite. Moreover, $S(e_0,\cdot)$ is continuous in a neighborhood of $e_1$ and $S(\cdot,e_1)$ is continuous in a neighborhood of $e_0$.
\item[(b)] If there are conjugate instants $t_a<t_b$ such that  $[t_a,t_b]\subset [t_0,t_1]$ and $[t_a,t_b]\ne [t_0,t_1]$, then $S(e_0,e_1)=-\infty$ for every $q_0,q_1$.
\end{itemize}
\end{lemma}

\begin{proof}
Let us consider the quadratic functional
\[
\mathcal{Q}[\eta]=\int_{t_0}^{t_1}[ \frac{1}{2}
\dot{ \eta}\cdot \dot{ \eta}-h_{ij} \eta^i  \eta^j] \,
\dd t
\]
over the piecewise $C^1$ curves $\eta:[t_0,t_1]\to \mathbb{R}^d$ which vanish at the endpoints. If $q$ is a solution to ELE which connects $e_0$ to $e_1$ then, integrating by parts, for every such $\eta$
\begin{equation} \label{qua}
\mathcal{S}_{e_0,e_1}[q+\eta]=\mathcal{S}_{e_0,e_1}[q]+\mathcal{Q}[\eta].
\end{equation}

(a). By Jacobi's theorem  \cite[Chap.\ 5]{gelfand63} the functional $\mathcal{Q}$ is positive definite if and only if no instant in $(t_0,t_1]$ is conjugate to $t_0$. As the positive definiteness does not change under time reversal this condition holds if and only if  no instant in $[t_0,t_1)$ is conjugate to $t_1$. Let $X(t)$, $Y(t)$ be the matrix solutions of ELE such that $X(t_0)=I$, $\dot X(t_0)=0$, $Y(t_0)=0$, $\dot Y(t_0)=I$. The instant $t$ is conjugate to $t_0$ if and only if $\det Y(t)=0$, thus $Y(t)$ is invertible for $t\in (t_0,t_1+\delta)$ for some $\delta>0$. For $t$ in this interval the only solution which connects $e_0$ to $e=(t,q)$ is $s\mapsto X(s)q_0+Y(s)w$ with $w=Y(t)^{-1}[q-X(t)q_0]$ and, by Eq.\ (\ref{qua}) and the positive definiteness of $\mathcal{Q}$, it is the unique minimizer of the action. An integration by parts as in Lemma \ref{ual} gives
\[
S(e_0,e)=\frac{1}{2}\, \{ q\cdot[\dot X(t) q_0+\dot Y(t)w]-q_0\cdot w\},
\]
which is a continuous function of $e=(t,q)$. The statement on $S(\cdot,e_1)$ is proved analogously.

(b). Let $\delta q$ be a non-trivial solution of ELE which vanishes at $t_a$ and $t_b$, and let $\eta=\delta q$ on $[t_a,t_b]$ and $\eta=0$ on the remainder of $[t_0,t_1]$. By Lemma \ref{ual} $\mathcal{Q}[\eta]=0$. If $\mathcal{Q}$ were non-negative $\eta$ would be a minimizer of $\mathcal{Q}$, and hence a $C^1$ solution to ELE, which is not the case since $\eta$ has a corner at an instant of $(t_0,t_1)$ (observe that $\dot{\delta q}$ does not vanish at $t_a$ and $t_b$, otherwise $\delta q$ would be trivial). Thus there is $\eta'$ such that $\mathcal{Q}[\eta']<0$. Let $q$ be any piecewise $C^1$ curve which connects $e_0$ to $e_1$. The action of $q+s\eta'$ is a polynomial of second degree in $s$ with leading coefficient $\mathcal{Q}[\eta']$, thus it goes to $-\infty$ for $s\to +\infty$.
\end{proof}

\begin{lemma} \label{ikd}
If $t_0$ and $t_1$, $t_0<t_1$, are conjugate then for every $q_0$,
$q_1$, defined $e_0=(t_0,q_0)$, $e_1=(t_1,q_1)$, we have
\[
\liminf_{e \to e_1} S(e_0,e)=\liminf_{e \to e_0} S(e,e_1)=-\infty.
\]
Stated in another way (see Lemma \ref{prf}), whenever $S(e_0,e_1)$ is finite the function $S(e_0,
\cdot)$ is not lower semi-continuous at $e_1$, and the function
$S(\cdot, e_1)$ is not lower semi-continuous at $e_0$.
\end{lemma}

\begin{proof}
For every event $e=(t,q)$ with $t>t_1$ we have $S(e_0,e)=-\infty$ by Lemma \ref{jac}, point (b), and analogously, for every event $e=(t,q)$ with $t<t_0$ we have $S(e,e_1)=-\infty$.
\end{proof}

Equation (\ref{qua}) and Lemma \ref{ual} imply the next result which we shall not use but which
is worth mentioning.

\begin{corollary}
Let $q(t)$ be a solution to ELE in the time interval $[t_0,t_1]$.
Then $q(t_0)$, $q(t_1)$, are conjugate  if and only if there is a (finite)
variation $q(t,s)$ of ELE solutions with fixed endpoints, if and only if
the light lifts of the curves $(t, q(t,s))$ starting from a point $x_0\in \pi^{-1}(e_0)$ give a
variation of lightlike geodesics with the same endpoints.
\end{corollary}

\begin{proof}
By linearity the variations of ELE solutions with fixed endpoints are $q(t,s)=q(t)+\delta q_s(t)$ where, for each $s$, $\delta q_s$ is a solution to ELE which vanishes at $t_0$ and $t_1$, which proves the first equivalence. By Eq.\ (\ref{qua}) and Lemma \ref{ual}
\begin{equation}  \label{los}
\mathcal{S}_{e_0,{e}_1}[q+\delta q_s]=\mathcal{S}_{e_0,{e}_1}[q]+\mathcal{Q}[\delta q_s]=\mathcal{S}_{e_0,{e}_1}[q],
\end{equation}
thus the light lifts starting from $x_0=(e_0,y_0)$, which are lightlike geodesics, all end at the point $(e_1,y_0+\mathcal{S}_{e_0,{e}_1}[q])$. Conversely, a variation of lightlike geodesics with fixed endpoints which are not tangent to $n$ projects to a variation of ELE solutions with fixed endpoints.
\end{proof}

Joining Lemmas \ref{jac} and \ref{ikd} with Lemma \ref{prf} we obtain.

\begin{corollary}
If $t_0$ and $t_1$, $t_0<t_1$, are conjugate then for every pair of
events $x_0,x_1$ such that $t(x_0)\le t_0$ and $ t_1\le t(x_1)$, we
have
\[
x_1\in \overline{J^{+}(x_0)} \qquad \textrm{and}  \qquad x_0\in
\overline{J^{-}(x_1)}.
\]
If any of the above inequality is strict then $(x_0,x_1)\in I$.
\end{corollary}

\begin{lemma}
Suppose that $t_0$ and $t_1$, with $t_0<t_1$, are not conjugate. Then for every choice of $q_0, q_1\in Q$ there is a solution to ELE $q: [t_0,t_1]\to Q$, such that $q(t_0)=q_0$ and $q(t_1)=q_1$.
\end{lemma}

\begin{proof}
Let $\delta q:[t_0,t_1] \to Q$ be a
solution to ELE such that $\delta q(t_0)=0$. As the ELE is linear it
coincides with it first variation, so to say that $t_1$ is not
conjugate to $t_0$ means that for no value of $\dot{\delta q} (t_0)\ne 0$
the resulting solution $\delta q(t)$ is such that $\delta q(t_1)=0$.
In other words, the linear map $\dot{\delta q} (t_0) \to \delta
q(t_1)$ is injective and hence surjective. By adding to any solution
$q'(t)$ starting from $q_0$ a suitable solution $\delta q(t)$ we can
therefore find some solution to ELE $q(t)$ such that $q(t_0)=q_0$
and $q(t_1)=q_1$.
\end{proof}

\begin{lemma} \label{clo}
If $x_1\in \overline{J^{+}(x_0)}$ or $x_0\in \overline{J^{-}(x_1)}$,
and $t_0<t_1$,
 then either $(x_0,x_1) \in J$ or $t_0$ and $t_1$ are
conjugate.
\end{lemma}

\begin{proof}
Suppose that  $t_0$
and $t_1$ are not conjugate. If there is an instant in $(t_0,t_1)$ conjugate to $t_0$ then by Lemma \ref{jac}, point (b), $S(e_0,e_1)=-\infty$, and hence $(x_0,x_1)\in I$ by Eq.\ (\ref{cgf1}). Otherwise, by Lemma \ref{jac}, point (a), $S(e_0,e_1)$ is finite, it is attained on a solution $q(t)$ to ELE which connects $e_0$ to $e_1$, and both $S(e_0,\cdot)$ and $S(\cdot,e_1)$ are continuous at $e_1$ and $e_0$, respectively. By Eqs.\ (\ref{prf1}) and (\ref{prf3}) the assumption reads $y_1-y_0\ge S(e_0,e_1)$. The light lift of $(t,q(t))$ starting from $x_0$ ends at the point $x_1'=(e_1, y_0+S(e_0,e_1))$, and $x_1\in J^{+}(x_1')$ since the two points belong to the same integral line of $n$. Thus  $(x_0,x_1) \in J$.
\end{proof}

We are ready to improve Ehrlich and Emch result on gravitational
pp-waves \cite{ehrlich92}.

\begin{theorem} \label{hcc}
Every homogeneous pp-wave is causally continuous.
\end{theorem}

\begin{proof}
By Theorem \ref{xxd} we have to show that for every $e_0,e_1\in E$
\[
\liminf_{e \to e_1} S(e_0,e)=\liminf_{e \to e_0} S(e,e_1),
\]
and that this quantity vanishes for $e_0=e_1$. By Lemma \ref{prf} each of these limits equals either $S(e_0,e_1)$ or $-\infty$. If $t_0>t_1$ they are both equal to $+\infty$.
Let $t_0=t_1$. By Cor.\ \ref{ifg} $(M,g)$ is stably causal, hence distinguishing, thus by Theorem \ref{psf} $\liminf_{e \to e_0} S(e_0,e)=\liminf_{e \to e_0} S(e,e_0)=0$. If we had $\liminf_{e \to e_1} S(e_0,e)=-\infty$ then by Lemma \ref{vtr} we would have $\liminf_{e \to e_0} S(e_0,e)=-\infty$, a contradiction. Thus $\liminf_{e \to e_1} S(e_0,e)=S(e_0,e_1)$ and analogously $\liminf_{e \to e_0} S(e,e_1)=S(e_0,e_1)$.

Let $t_0<t_1$. If no instant in $(t_0,t_1]$ is conjugate to $t_0$ then by Lemma \ref{jac}, point (a), both limits equal $S(e_0,e_1)$. If there is an instant in $(t_0,t_1)$ conjugate to $t_0$ then by Lemma \ref{jac}, point (b), $S(e_0,e_1)=-\infty$ and both limits equal $-\infty$. Finally, in the remaining case $t_0$ and $t_1$ are conjugate and both limits equal $-\infty$ by Lemma \ref{ikd}.
\end{proof}

Observe that, contrary to the result of Theorem \ref{byt}, here the Lagrangian is allowed to depend on time.

\begin{proposition}
Let $(M,g)$ be a homogeneous pp-wave and consider coordinates for
which $a_{ij}=\delta_{ij}$, $M_{ij}=0$. Let $q:\mathbb{R}\to Q$ be
such that $q^i(t)=0$ in Cartesian coordinates, then the light lift
$\gamma(t)=(t,q(t),0)$ of $(t,q(t))$ is a lightlike geodesic.
Furthermore, if there is a non-trivial solution $\tilde{q}(t)$ to
the E.-L.\ equations such that $\tilde q^i(t_0)=\tilde q^i(t_1)=0$,
then $\gamma(t_0)$ and $\gamma(t_1)$ are conjugate points.

Finally, suppose that $h_{ij}(t)$ is diagonal (so that the E.-L.
equations decouple), then $\gamma$ is not achronal in $[\tilde{t}_0,
\tilde{t}_1]$ if and only if we can find an interval
$[t_0,t_1]\subset (\tilde t_0,\tilde t_1)$ and a non-trivial
solution ${q}'$ to the E.-L.\ equations such that for every $i$,
$q'^i(t_0)= q'^i(t_1)=0$. In this case we can find $q'$ with  no
other zeros in $[t_0,t_1]$.
\end{proposition}

\begin{proof}
Let us pass to the coordinate system for which $M_{ij}=0$. The
condition $q^i(t)=0$ over the curve $q$ is preserved.
 If $\tilde q(t)$ is a non-trivial solution to the
E.-L.\ equations which starts and ends at the origin of coordinates,
then the action vanishes because by integration by parts
\begin{align*}
\mathcal{S}_{e_0,e_1}[\tilde q]&=\int_{t_0}^{t_1}[ \frac{1}{2}
\dot{\tilde q}\cdot \dot{\tilde q}-h_{ij}\tilde q^i \tilde q^j] d
t=\int_{t_0}^{t_1}[ -\frac{1}{2}\tilde q\cdot \ddot{\tilde
q}-h_{ij}\tilde q^i \tilde q^j] \dd  t\\
&= \int_{t_0}^{t_1}[ h_{ij}\tilde q^i \tilde q^j-h_{ij}\tilde q^i
\tilde q^j] \dd  t=0.
\end{align*}
By homogeneity, $q_s(t)=s \tilde{q}(t)$ is also a non-trivial
solution to the E.-L.\ equations with vanishing action on
$[t_0,t_1]$, thus the light lifts
$\gamma_s(t)=(t,q_s(t),\mathcal{S}_{e_0,e(t)}[\tilde
q_s\vert_{t_0,t}])$, $\gamma_0=\gamma$, provide a 1-parameter family
of null geodesics connecting $\gamma(t_0)$ and $\gamma(t_1)$, which
are therefore conjugate points.

Suppose that we can find an interval $[t_0,t_1]\subset (\tilde
t_0,\tilde t_1)$ and a non-trivial solution $\tilde{q}$ to the
E.-L.\ equations such that $\tilde q^i( t_0)=\tilde q^i( t_1)=0$,
then $\gamma(t_0)$ and $\gamma(t_1)$ are conjugate thus
$\gamma(\tilde t_0)$ and $\gamma(\tilde t_1)$ are chronologically
related \cite[Theor.\ 10.72]{beem96}. For the converse, suppose that
$\gamma(\tilde t_0)$ and $\gamma(\tilde t_1)$ are chronologically
related then by Eq.\ (\ref{cgf1}) we can find an interval
$[t_0,t_1']\subset (\tilde t_0, \tilde t_1)$ and curve $ q:[t_0,
t_1']\to Q$, ${q}^i(t_0)={q}^i(t_1')=0$, such that
$S_{{e}_0,{e}_1'}[{q}]<0$. Since $h_{ij}$ is diagonal, the
Lagrangian is the sum of $d$ terms $S^i=\int_{t_0}^{t_1'}[
\frac{1}{2} (\dot q^i)^2-h_{ii} {q^i}^2]\dd t$ (no sum over $i$),
one for each coordinate, and one of them must be negative. Let us
assume without loss of generality that this happens for $i=1$, and
let $u(t)$ solve ELE1, $\ddot u=-2 h_{11}(t)u$, $u(t_0)=0$,
$\dot{u}(t_0)\ne 0$, then by \cite[Prop.\ 3, Chap. 1]{coppel71}
either $q^1(t)$ coincides with $u(t)$ and hence is a solution to
ELE1, which is impossible since the action $S^1$ over $[t_0,t_1]$
would vanish, or $u(t_1)=0$ for some $t_1\in (t_0,t'_1)$ which
implies that ${q}'(t)=(u(t),0,\cdots,0)$ is a solution to the E.-L.
equations which satisfies the claimed properties.

\end{proof}

\begin{remark}
A Riemannian metric depending continuously on a parameter $t$ can be
incomplete for some value of the parameter and complete for all the
other values. For instance, consider  on $\mathbb{R}$ the metric
$\dd s^2=(1+t^2x^2)/(1+x^2)^{3/2} \dd x^2$ which is incomplete only
for $t=0$.

Even worse, $(Q,a_t)$ can be complete for $t\in [0,1]$ although
there can be a path $q:[0,1)\to Q$ of finite length going to
infinity. Consider for instance $Q=\mathbb{R}$ with $a_t$ conformal
to the Euclidean metric $\dd x^2$, with conformal factor equal to
$1$ everywhere  but on $[\frac{1}{1-t}-1, \frac{1}{1-t}+1]$ where it
reaches the minimum value $(1-t)^6$ at the center $\frac{1}{1-t}$.
Then $q(t)=\frac{1}{1-t}$ has length $\int_0^1 \sqrt{a_t(\dot q,
\dot q)} \,\dd t=\int_0^1 [(1-t)^6 \frac{1}{(1-t)^4} ]^{1/2} \, \dd
t=1/2$. In particular we cannot find a complete metric $a$ such that
for every $t\in [0,1]$, and for every $v\in TQ$, $a(v,v)\le a_t
(v,v)$.

\end{remark}

\section{Time independence: spacelike dimensional reduction to the Maupertuis spacetime} \label{mau}

In this section we assume that the Lagrangian does not depend on
time, and hence that $(M,g)$ admits the Killing vector
$k=\partial_t$. We shall show that the Maupertuis principle in the
Jacobi metric formulation has also a nice spacetime interpretation
and we shall relate it to the Eisenhart spacetime (Fig.\ \ref{fig:mau}).

Since our classical spacetime $E$ admits the splitting
$\mathbb{R}\times Q$, induced by the flow of $\partial_t$, in the
remainder of the section we shall identify $E$ with such product.

Let us suppose that $V(q)$ is bounded from above and let $e$ be a
constant such that $V(q)<e$ for every $q\in Q$.  Let
\[
a_J=2(e-V)\,a
\]
be the Jacobi metric on $Q$, we define the {\em
Maupertuis spacetime} $(\mathcal{M},g^{\mathcal{M}})$ as
$\mathcal{M}=\mathbb{R}\times Q$, of coordinates $(\tilde y,q)$,
endowed with the Maupertuis spacetime metric
\begin{equation} \label{gmo}
g^{\mathcal{M}}=a_J-(\dd \tilde y-b)^2,
\end{equation}
and time orientation given by $\partial_{\tilde y}$. We remark that the
metric $g^{\mathcal{M}}$ as well as the projection $\pi_\mathcal{M}:
M\to \mathcal{M}$ depend on the chosen value of $e$, thus we have a
Maupertuis spacetime for each $e$. With some abuse of notation we
shall write  for the projection on the second factor
$\pi:\mathcal{M}\to Q$, though we already used $\pi$ for the
projection $\pi: M\to E$.  We have $\pi \circ
\pi_{\mathcal{M}}=\pi_{Q}\circ \pi$.

Let us show the connection between the Eisenhart spacetime and the
Maupertuis spacetime. On $(M,g)$ let us make a simple change of
gauge of the form
\[
y=\tilde y-et,
\]
where $e$ is a constant. The spacetime metric  on $M$ reads
\begin{equation}
g=a \!-\!\dd t \otimes (\dd \tilde y-\!b) -\!(\dd \tilde y-\!b)
\otimes \dd t+2(e-V)\, \dd t ^2 .
\end{equation}
Let $\tilde k=\partial_t)_{\tilde y}$, since ($n=\partial_y$)
\[
\tilde k=k-e n
\]
this is also a Killing vector.  Moreover,
\[
g(\tilde k,\tilde k)=2(e-V)>0,
\]
thus it is spacelike.
 Introduced the spacetime 1-form
\[
\omega^{\mathcal{M}}=\frac{g(\cdot,\tilde k)}{g(\tilde k,\tilde
k)}=\dd t-\frac{1}{2(e-V)}(\dd \tilde y-b),
\]
the spacetime metric  on $M$ reads
\[
g=\frac{1}{2(e-V)}\,g^{\mathcal{M}}+ 2(e-V)\, \omega^{\mathcal{M}}
\otimes \omega^{\mathcal{M}}.
\]
The manifold $M$ is a trivial $(\mathbb{R},+)$- principal fiber
bundle $\pi_{\mathcal{M}}: M\to \mathcal{M}$, where the group
$(\mathbb{R},+)$ is given by the diffeomorphisms generated by
$\tilde k$.  Since $\tilde k$ is spacelike $M$ is a Kaluza-Klein
bundle, i.e.\ a bundle with $\mathbb{R}$ spacelike fiber. We can
call the projection of $M$ onto $\mathcal{M}$ a {\em spacelike
dimensional reduction}. The 1-form $\omega^{\mathcal{M}}$ is
actually an (Abelian) connection \cite[p.\ 64]{kobayashi63} since
\begin{itemize}
 \item[(a)] $\qquad  \omega^{\mathcal{M}}(\tilde k)=1$,
\item[(b)]  $\qquad L_{\tilde k} \omega^{\mathcal{M}}=0$,
\end{itemize}
 thus if we rewrite
\[
g=\frac{1}{2(e-V)}\,\{g^{\mathcal{M}}+ [2(e-V)]^2\, \omega^{\mathcal{M}}
\otimes \omega^{\mathcal{M}} \} ,
\]
we find that $g$ is conformal to a Kaluza-Klein metric. We  need to
recall the following fact on spacelike dimensional reduction.

\begin{proposition} \label{puh}
Let $(M,g)$ and $(\hat M,\hat g)$,  such that $\pi:M\to \hat{M}$ is
a $(\mathbb{R},+)$-principal bundle with spacelike fibers. Let
$\omega$ be a connection and suppose that $g$ takes the Kaluza-Klein
form
\begin{equation} \label{nuk}
g=\hat g+\varphi \,\omega \otimes \omega,
\end{equation}
where $\varphi:\hat M\to (0,+\infty)$. Then if $U\subset \hat M$ is
an open subset and $\sigma: U\to \pi^{-1}(U)$ is a local section,
defined the section independent curvature $F=-\sigma^*\dd \omega$,
any curve $\gamma: I\to M$, $s\to \gamma(s)$, is a geodesic if and
only if  $p:=\varphi(\pi(\gamma(s))) \omega(\gamma'(s))$ is constant
and the projection $q(s)$ satisfies
\begin{equation} \label{usl}
\hat g(\cdot, \frac{\hat D}{\dd s} \,q')=-p\, F(\cdot,
q')+\frac{p^2}{2 \varphi^2} \,\dd \varphi.
\end{equation}
In particular, any horizontal curve is  geodesic if and only if so
is its projection.
\end{proposition}

The unusual sign in front of the curvature $F$ in Eq.\ (\ref{usl})
is due to the signature choice $(-,+,\cdots,+)$ for the spacetime
metrics, see Eq.\ \ref{nuk}.

\begin{proof}
It suffices to recall that the geodesics on $M$ are the stationary
points of the 'energy' functional $\int
\frac{1}{2}g(\gamma',\gamma')\dd s$. Using the expression for $g$ in
a local bundle trivialization and taking the variation it is rather
straightforward to obtain the expressions for $p$ and Eq.\
(\ref{usl}).
\end{proof}

Let us observe that under a spacelike dimensional reduction ($C^1$)
causal (timelike) curves project to causal (resp.\ timelike) curves,
and that a causal curve projects to a lightlike curve if and only if
it is lightlike and horizontal. The projection of a lightlike
geodesic need not be a geodesic though, as just shown, this is so if
the geodesic is horizontal.

We shall need another $(\mathbb{R},+)$-principal bundle given by
$\pi: \mathcal{M}\to Q$, where the fibers are generated by
$\partial_{\tilde{y}}$ and hence are timelike with respect to
$g^\mathcal{M}$. Here $\dd \tilde y-b$ plays the role of connection
1-form. Thus we have three  interrelated dimensional reductions
\begin{itemize}
\item[($\alpha$)] $\pi:M\to E$, lightlike dimensional reduction. Each slice
$\pi_t:N_t\to Q_t$ is endowed with the connection $\omega_t=\dd
y-b_t$.
\item[($\beta$)]$\pi_{\mathcal{M}}:M\to \mathcal{M}$, spacelike dimensional
reduction. The connection is $\omega^\mathcal{M}$.
\item[($\gamma$)] $\pi:\mathcal{M}\to Q$, timelike dimensional reduction. The
connection is $\dd \tilde y-b$.
\end{itemize}

It is natural to ask which causal curves on $M$ project to lightlike
curves on $\mathcal{M}$. We recall that by {\em pregeodesic} we mean
a geodesic up to the parametrization.

\begin{proposition} \label{ijs}
The causal curves on $M$ that project under $\pi^\mathcal{M}$ to
lightlike curves on $\mathcal{M}$ can be parametrized by $t$ and are
exactly the light lifts to curves $(t,q(t))$, $q:I\to Q$, for which
the energy $\frac{1}{2}a(\dot q, \dot q)+V(q)$ is independent of
time and coincides with $e$.

The causal geodesics $(t, q(t),y(t))$ on $M$ that  project under
$\pi^\mathcal{M}$ to lightlike curves  on $\mathcal{M}$ are exactly
the light lift of curves $(t,q(t))$ where $q:I\to Q$  solves the ELE
(\ref{ele}) with energy $e$. In particular they are
$\omega^\mathcal{M}$-horizontal lightlike geodesics parametrized by
$t$, while their projection on $\mathcal{M}$ are lightlike
pregeodesics.
\end{proposition}

The first statement essentially establishes that the light lift from
$\mathcal{M}$ to $M$ of $(\tilde{y}(s),q(s))$ has the effect of
parametrizing the curve $q(s)$ with a time parameter $t$ so as to
give it a constant energy equal to $e$.

\begin{proof}
As just mentioned the causal curve $\gamma$ on $M$ has to be
lightlike. It cannot be tangent to $n$ at any point, otherwise the
projection would be tangent to $\partial_{\tilde y}$ and hence would
have timelike tangent at some point. As a consequence $\gamma$ can
be parametrized by $t$, so we can write $\gamma(t)=(t, q(t), y(t))$.
As $\gamma$ is lightlike, it is the light lift of $(t,q(t))$ thus
$y(t)=y_0+\int_{t_0}^t L(t,q,\dot q) \dd t$, and hence
$\tilde{y}(t)=y_0+\int_{t_0}^t L(t,q,\dot q) \dd t+ e t$. The
projection of $\gamma$ is lightlike if and only if
\begin{align*}
0=\omega^{\mathcal{M}}(\dot \gamma)&=\frac{1}{2(e-V)}\{-L(t,q(t),\dot q(t))-e+b(\dot q)+2(e-V)\}\\
&=\frac{1}{2(e-V)}\{e-[\frac{1}{2}a(\dot q(t), \dot q(t))+V(q(t))]\},
\end{align*}
from which we conclude that the energy must be constant and
coincident with $e$.

Let us come to the second statement. If $(t,q(t),y(t))$ is a causal
curve which projects under $\pi_\mathcal{M}$ to a lightlike curve
then it is a lightlike and horizontal. But an
$\omega^\mathcal{M}$-horizontal lightlike curve on $M$ is a
lightlike geodesic if and only if the projection to $\mathcal{M}$ is
a lightlike $\frac{1}{2(e-V)}g^\mathcal{M}$-geodesic (thus a
$g^\mathcal{M}$-pregeodesic) (Prop.\ \ref{puh}), if and only if
$q(t)$ solves ELE (Theorem \ref{kdf}).
\end{proof}

\subsection{Fermat's principle and Jacobi-metric characterization of dynamics}

Let us consider the bundle $\pi:\mathcal{M}\to Q$. By Eq.\
(\ref{gmo}) given a $C^1$ curve $q:I\to Q$, $s\to q(s)$, on the
base, there are two light lifts to $\mathcal{M}$ given by
\begin{equation} \label{whx}
s\to (\tilde y_0+\int_{s_0}^s [\pm\sqrt{a_J(q',q')}+b(q') ]\dd s, q(s)).
\end{equation}
The presence of two light lift is in fact a feature of timelike
dimensional reduction. We observe that of these lifts only that with
the plus sign has  a tangent vector whose $g^\mathcal{M}$-scalar
product with $\partial_{\tilde y}$ is negative, thus the plus choice
gives a future-directed causal curve, while the minus choice gives a
past-directed causal curve.

\begin{proposition} \label{juo}
Let $q:I\to Q$, $t \to q(t)$,  be a $C^1$ curve for which the energy
is constant and equal to $e$. Up to translation along the fiber, the
light lift of $q(t)$ over the lightlike bundle $\pi:\mathcal{M}\to
Q$  gives a curve $(\tilde{y}(t),q(t))$ which coincides with the
light lift of $(t,q(t))$ on $E$ to $M$ projected  to $\mathcal{M}$.
Moreover,  $\tilde{y}(t)= \tilde y_0+\int_{t_0}^t
[\sqrt{a_J(q',q')}+b(q') ]\dd s$.
\end{proposition}

\begin{proof}
The light lift of $(t,q(t))$ to $M$ projected  to $\mathcal{M}$ is,
by Prop.\ \ref{ijs},  a future-directed lightlike curve
$(\tilde{y}(t),q(t))$ which projects on $q(t)$. Since the
future-directed light lift is unique we get the first desired
conclusion. Moreover, by Eq.\ (\ref{whx}), we have $\tilde{y}(t)=
\tilde y_0+\int_{t_0}^t [\sqrt{a_J(q',q')}+b(q') ]\dd s$.
\end{proof}

In what  follows we shall need Fermat's principle which states that
the lightlike geodesics can be  characterized as those lightlike
paths which minimize the time of arrival (at an, arbitrarily
parametrized, receiver worldline). For static spacetimes the proof,
first given by Weyl, is quite elementary and can be found in many
textbooks \cite{misner73}. In the non-trivial stationary case it was
first proved by Levi-Civita in  \cite{levicivita18,levicivita28}.
Unfortunately, these works, being in Italian, passed largely
unnoticed, and many other more or less rigorous proofs appeared
along the years
\cite{lichnerowicz55,pham57,landau62,brill73,perlick90b}. Finally,
the principle was extended to arbitrary spacetimes by Kovner and
Perlick \cite{kovner90,perlick90}.

We mention that Levi-Civita's proof, and some other variations, use
the square-root-`energy' Lagrangian on spacetime. However, this
Lagrangian works well on the space of timelike curves, but not on
that of lightlike curves, as the variation is singular in the latter
case. This difficulty demands special care, as one needs to take a
suitable limit to the lightlike case. We give a new proof of the
stationary case in a version well suited for our purposes which
avoids these problems.

\begin{proposition} \label{sqo}
Let $q_0,q_1\in Q$, and let us consider the (arrival
$\tilde{y}$-time) functional
\begin{equation}
F_{q_0,q_1}[q]=\int [\sqrt{a_J(q',q')}+b(q') ]\dd s
\end{equation}
on the family of $C^1$ curves $q:I\to Q$ connecting $q_0$ to $q_1$.
Observe that $F_{q_0,q_1}[q]$ does not depend on the parametrization
of $q(s)$. The curve $q$ is a stationary point  if and only if once
parametrized with $a_J$-arc length $s$, the future-directed light
lift
\[
s\to (\tilde y_0+\int_{s_0}^s [\sqrt{a_J(q',q')}+b(q') ]\dd s,\,
q(s)).
\]
is a lightlike geodesic.
\end{proposition}

We remark that the proof does not use the special for of $a_J$, as
the metric $a$ is not mentioned. In other words
$(\mathcal{M},g^\mathcal{M})$  could be any stationary spacetime
where $a_J$ and $b$ are defined through $g^\mathcal{M}=a_J-(\dd
\tilde y-b)^2$.

\begin{proof}
Let $\gamma$ be a lightlike curve which starts from
$p_0=(\tilde{y}_0,q_0)\in \mathcal{M}$,  and reaches
$\pi^{-1}(q_0)$. Let $\gamma(s,r)$, $(s,r)\in [0,b] \times
[0,\epsilon]$, be a $C^1$ variation of lightlike curves of
$\gamma(s):=\gamma(s,0)$, where $s$ is the longitudinal parameter,
$\gamma(0,r)=p_0$ and $\gamma(b,r)\in \pi^{-1}(q_0)$. There are
curves $q(s,r)$ whose future-directed light lift is $\gamma(s,r)$.
As the variation is lightlike we have
\begin{align*}
0&=\frac{\dd }{\dd r}\,\frac{1}{2}\,
g^{\mathcal{M}}(\partial_s,\partial_s)=g^{\mathcal{M}}(\p_s,
\nabla_r^{\mathcal{M}} \p_s)=\frac{\dd }{\dd s}\,
g^{\mathcal{M}}(\p_s, \p_r) -g^{\mathcal{M}}(\gamma'', \p_r).
\end{align*}
Integrating over $\gamma$ we obtain
\[
(\partial_r F_{q_0,q_1}[q(s,r)])\vert_{r=0}\,
g^{\mathcal{M}}(\p_s,\partial_{\tilde{y}})\vert_{(b,0)}
=g^{\mathcal{M}}(\p_s, \p_r)\vert_{(b,0)}=\int_0^b
g^{\mathcal{M}}(\gamma'', \p_r) \,\dd s.
\]
We observe that
$g^{\mathcal{M}}(\p_s,\partial_{\tilde{y}})\vert_{(b,0)}< 0$ as it
is the scalar product between a future directed lightlike vector and
a timelike vector. Thus if $\gamma$ is a geodesic $F$ is indeed
stationary on the projection. Suppose that $F$ is stationary on
$q(s)$ and let $\gamma(s)=(y_0+F_{q_0,q_1}[q(s)], q(s))$ be its
light lift starting from $p_0=(\tilde y_0,q_0)$ and hence reaching
the fiber $\pi^{-1}(q_1)$.  Let $\gamma(s,r)$ be a variation of
$\gamma(s)$ with fixed endpoints, and let $q(s,r)$ be its
projection. Let $\gamma^L(s,r)$ be the future-directed light lift of
$q(s,r)$ and let $\p^L_s$, $\p_r^L$, be the holonomic base on this
last embedded surface. We can write
\[
\gamma(s,r)=( \tilde y_0+F_{q_0,q(s,r)}[q(s',r)\vert_{s'\in
[0,s]}]+\alpha(s,r), q(s,r)),
\]
so that $\p_r=\p_r^L+(\p_r \alpha) \p_{\tilde y}$ and
$\p_s=\p_s^L+(\p_s \alpha) \p_{\tilde y}$. We have using the
expression for ${g}^\mathcal{M}$
\begin{align*}
\p_r\int [\frac{1}{2}\, g^{\mathcal{M}}(\partial_s,\partial_s)] \dd
s&=\int \p^L_r[\frac{1}{2}\, g^{\mathcal{M}}(\partial_s,\partial_s)]
\dd s- \int (\p_r \alpha) \frac{\dd }{\dd s}[\frac{\dd
\tilde{y}}{\dd
s} - b(\frac{\dd q}{\dd s})] \dd s \\
&\quad + (\p_r \alpha)\,
g^\mathcal{M}(\p_s,\p_{\tilde{y}})\vert_{(b,0)}
\end{align*}
Let us show that the right-hand side vanishes. The first term on the
right-hand side vanishes because the variation $\gamma^L(s,r)$ is
lightlike. The second term vanishes because as $q(s)$ is $a_J$-arc
length parametrized
\[
\frac{\dd \tilde{y}}{\dd s} - b(\frac{\dd q}{\dd s})= \frac{\dd
F_{q_0,q(s)}}{\dd s} - b(\frac{\dd q}{\dd s})=\sqrt{a_J(q',q')}=1.
\]
Finally, since the variation is at fixed endopoints the variation
vector at the last point of $\gamma$ must vanish which implies
$(\p_r \alpha)\vert_{(b,0)}=-(\partial_r
F_{q_0,q_1}[q(s,r)])\vert_{r=0}=0$, where the last equality follows
from the assumption of stationarity of the arrival $\tilde y$-time.
Thus we conclude that under fixed endpoints the first variation of
$\int [ g^{\mathcal{M}}(\partial_s,\partial_s)] \dd s$ vanishes,
that is, $\gamma(s)$ is geodesic.
\end{proof}
%

We recall that under a conformal transformation of a Lorentzian
metric $g$, $g \to \Omega^2 g$, lightlike pregeodesics are left
invariant while their affine parameter $\lambda$, changes as
$\lambda\to \lambda'$ where $\dd \lambda'/\dd \lambda=\Omega^2$, see
for  instance \cite[Appendix D]{wald84}.

We are ready to give a Lorentzian geometric proof and interpretation
of Jacobi geometric characterization of dynamics.

\begin{theorem}
Let $q(t)$ be a $C^1$ stationary point of $F_{q_0,q_1}$ parametrized
so as to have energy $e$, then it is also a solution to ELE
(\ref{ele}). Conversely, any solution to  ELE with energy $e$
passing through $q_0$ and then $q_1$ is a stationary point of
$F_{q_0,q_1}$.
\end{theorem}

\begin{proof}
By Prop.\ \ref{juo} the future-directed light lift $(\tilde{y}(t),
q(t))$ to $\mathcal{M}$ coincides with the light lift on $M$,
$(t,q(t),y(t))$ of $(t,q(t))$, projected back to $\mathcal{M}$. By
Prop.\ \ref{sqo} $(\tilde{y}(t), q(t))$ is a lightlike
$g^\mathcal{M}$-pregeodesic as it is a geodesic if parametrized with
respect to $a_J$-arc length $s$. The relation between $s$ and $t$ is
\[
\frac{\dd s}{\dd t}={2(e-V)}.
\]
However, this is the change of affine parametrization passing from
$g^\mathcal{M}$ (affine parameter $s$) to $\frac{1}{2(e-V)}
g^\mathcal{M}$ (affine parameter $t$), thus $(\tilde{y}(t), q(t))$
is actually affinely parametrized in $(\mathcal{M}, \frac{1}{2(e-V)}
g^\mathcal{M})$. Its horizontal lift to $M$, $(t,q(t),y(t))$ is
therefore a lightlike geodesic (Prop.\ \ref{puh}) and so $q(t)$ is a
solution to ELE (\ref{ele}), see Theor.\ \ref{kdf}.

Conversely, if $q(t)$ is a solution to  ELE (\ref{ele}) with energy
$e$, the light lift $(t,q(t), y(t))$ is a lightlike geodesic
(Theor.\ \ref{kdf}). By Prop.\ \ref{ijs} the projection to
$\mathcal{M}$, $(\tilde{y}(t),q(t))$ is lightlike, and actually a
lightlike geodesic on $(\mathcal{M}, \frac{1}{2(e-V)}
g^\mathcal{M})$ (Prop.\ \ref{puh}), thus the projection back to $Q$,
$q(t)$, is a stationary point for $F_{q_0,q_1}$ (Prop.\ \ref{sqo}).
\end{proof}

\begin{figure}[ht]
\psfrag{M}{{\small \parbox{3cm}{Eisenhart's \\ spacetime $M$}}}
\psfrag{S}{{\small $\!\!\!Q$}} \psfrag{X}{{\small
$\!\!\!\!\!\{q^a\}$}} \psfrag{T}{{\small $t$}} \psfrag{P}{{\small
$\pi$}} \psfrag{ST}{{\small $\!\!\!(Q_t,a_t)$}} \psfrag{Q}{{\small \parbox{3cm}{Classical \\
spacetime $E$}}} \psfrag{NT}{{\small $N_t$}}  \psfrag{Y}{{\small
$\!\!\!\!\!\tilde y:= y+et$}} \psfrag{XD}{{\small $y$}}
\psfrag{V}{{\small \hspace{-0.5cm}%
\parbox{3.5cm}{(Stationary) \\ Maupertuis \\ spacetime $\mathcal{M}$}}}
\psfrag{R}{{\small $\pi_{\mathcal{M}}$}} \psfrag{B}{{\small
$\pi_{Q}$}} \psfrag{H}{{\small  $\, \tilde{k}$}} \psfrag{L}{{\small
$n$}} \psfrag{K}{{\small $\p_t$}} \psfrag{U}{{\small $\p_{\tilde
y}$}} \psfrag{A}{{\small $\tilde y$}}
\begin{center}
\includegraphics[width=11cm]{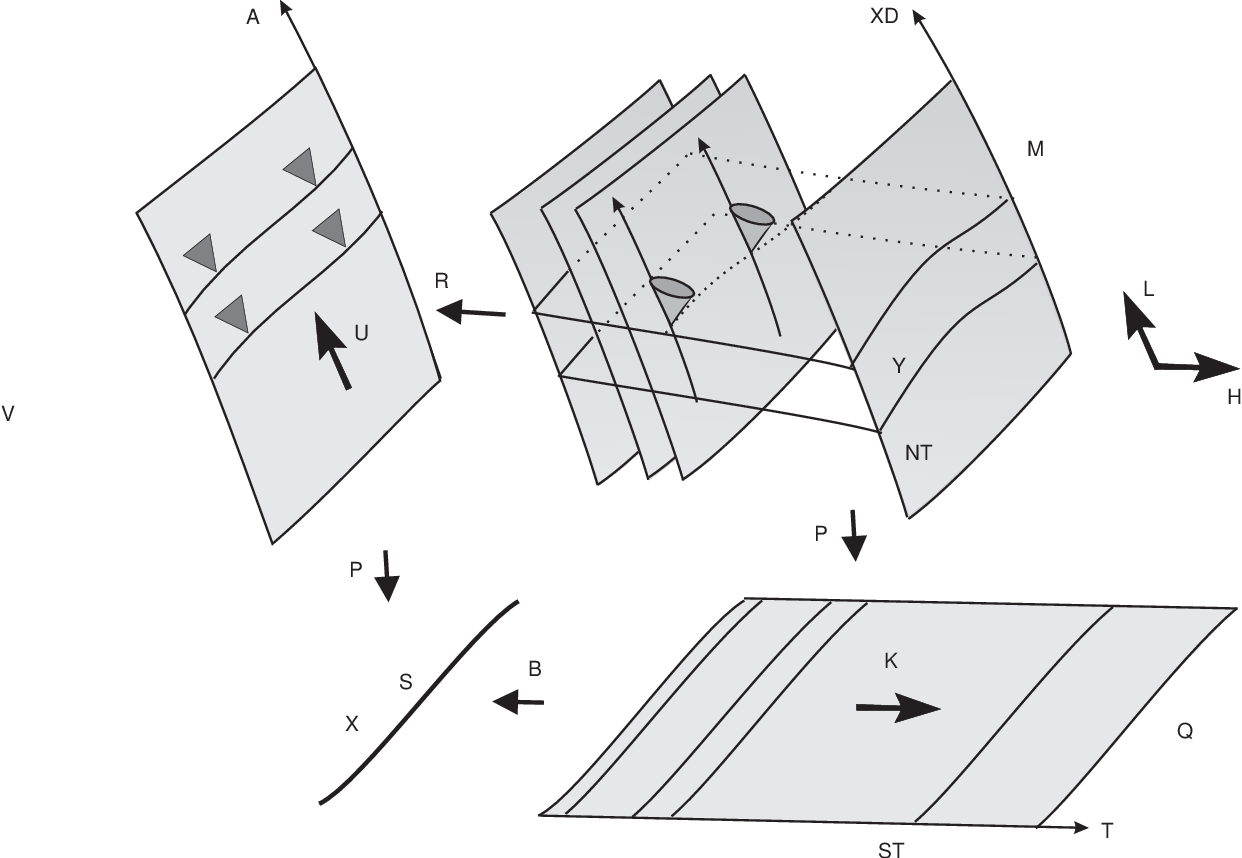}
\end{center}
\caption{If the Lagrangian is independent of time it is possible to
construct the Maupertuis spacetime for any chosen value of $e$ such
that for every $q\in Q$, $V(q)<e$. The figure illustrates its
connection with the Eisenhart's spacetime, the configuration space
$Q$ and the classical spacetime $E$ (extended configuration space).
The Jacobi metric characterization of dynamics is connected with the
projection of lightlike pregeodesics from $M$ to $\mathcal{M}$. The
coordinate $\tilde{y}$ on $\mathcal{M}$ is not always a time
function as displayed.} \label{fig:mau}
\end{figure}

\subsection{Global hyperbolicity and Jacobi space metric
completeness} \label{ss:jac}

Let us consider the causality properties of $(\mathcal{M},
g^\mathcal{M})$ or equivalently, of $(\mathcal{M},
\frac{g^\mathcal{M}}{2(e-V)})$.

We shall need the inverse of $g^\mathcal{M}$
\begin{equation}
(g^\mathcal{M})^{-1}=(a_{J}^{-1})_{cd}(\p_c+b_c\p_{\tilde y})\otimes
(\p_d+b_d\p_{\tilde y})-\p_{\tilde{y}}\otimes \p_{\tilde{y}}.
\end{equation}
For future reference we write also the inverse of $g$
\begin{align}
g^{-1}&=2(e-V)\Big\{(a_{J}^{-1})_{cd}(\p_c+b_c\p_{\tilde
y})\otimes (\p_d+b_d\p_{\tilde
y}) \nonumber \\
&\quad
-(\p_{\tilde{y}}+\frac{1}{2(e-V)}\,\p_t)\otimes(\p_{\tilde{y}}+\frac{1}{2(e-V)}\,\p_t)
 +\frac{1}{[2(e-V)]^2}\, \p_{t}\otimes \p_{t} \Big\} .
\end{align}

We recall that the construction of $(\mathcal{M},g^\mathcal{M})$ was
based on the assumption that the inequality $V(q)<e$ holds for every
$q\in Q$. By suitably strengthening this inequality we obtain stable
causality.

\begin{theorem} \label{stc}
Suppose that the constant $e$ is such that for every $q\in Q$,
$V+\frac{1}{2} a^{-1}(b,b)< e$, then $(\mathcal{M}, g^\mathcal{M})$
is stably causal and $\tilde{y}$ provides a time function with
past-directed timelike gradient. If the weaker inequality
$V+\frac{1}{2} a^{-1}(b,b)\le e$ holds instead, then $\tilde{y}$ is
a semi-time function.
\end{theorem}

We observe that $\tilde{y}$ has the degree of differentiability of
$y:M\to \mathbb{R}$.

\begin{proof}
Let us calculate
\[
g^\mathcal{M}(\nabla^\mathcal{M}\tilde y,\nabla^\mathcal{M}\tilde
y)=(g^\mathcal{M})^{-1}(\dd \tilde y, \dd \tilde y)=
a_{J}^{-1}(b,b)-1=\frac{1}{e-V}[\frac{1}{2}a^{-1}(b,b)-(e-V)],
\]
which is indeed negative under the former assumption and
non-positive under the latter assumption.
\end{proof}

Global hyperbolicity places some conditions on the completeness of
the Jacobi metric obtained from $a$.


\begin{proposition} \label{nap}
Let us consider the time independent case (but possibly $b\ne 0$),
and let us assume that $(M,g)$ is globally hyperbolic (or
equivalently the action $\mathcal{S}_{e_0,e_1}[q]$ is coercive).
Further, let us suppose that the constant $e$ is such that $\sup_Q
V<e$.  Then the Jacobi metric $a_J=2(e-V)a$ is complete.
\end{proposition}

We remark that if $a$ is complete and $e$ is such that $\sup_Q V<e$,
then $a_J=2(e-V)a$ is complete as $e-V$ is positive and bounded from
below. The theorem states only the completeness of $a_J$.

\begin{proof}
If $Q$ is compact the conclusion follows from the Hopf-Rinow
theorem, thus we can assume that $Q$ is non compact.


Let $t_0\in \mathbb{R}$. We will proceed by contradiction. We are
going to identify a point $q_0$ and a time $t_1>t_0$ such that
defined $e_0=(t_0,q_0)$, $e_1=(t_1,q_0)$ there is a sequence of
piecewise $C^1$ curves $q_n$ which escape every compact and are such
that $\mathcal{S}_{e_0,e_1}[q_n]<I<+\infty$ (failure of coercivity).
The sequence could be taken $C^1$ smoothing the corners. We shall
choose the path of $q_n$ in such a way that it can be decomposed
into two parts, the latter part being equal to the former part apart
from the direction taken over it. As a consequence, the action
functional term $\int_{q_n} b$ will cancel out and will play no
role.

Let us come to the details of  the proof. If $2(e-V)a$ is not
complete then there is some Jacobi-incomplete Jacobi-geodesic
$\sigma$ which escapes every compact set. Let us parametrize it by
the $a$-arc-length parameter $s$, so that $\sigma: [0,b)\to Q$ is
such that $\sigma(s)$ escapes every compact set as $s\to b$. It must
be $b<+\infty$, otherwise, as $(e-V)$ is bounded from below,
$\sigma$ would be Jacobi-complete.

Let $q_0=\sigma(0)$. We can suppose that $V(q_0)=0$, otherwise we
can just redefine the potential (an operation which corresponds to a
change of coordinate $y'=y+Ct$ in the spacetime interpretation). Let
$B=\sup_Q V$ and let $e>B$, we define $t_1=t_0+2\int_\sigma
\frac{\dd s}{\sqrt{2(e-V(\sigma(s)))}}$. The integral on the right
hand side is smaller than $b/\sqrt{2(e-B)}$, thus the definition
makes sense.

The increasing real function $f_n: [t_0, \tau_n] \to \mathbb{R}$,
$f_n(t_0)=0$, is chosen so as to be a stationary point of the action
functional $\int[\frac{1}{2}\dot{f}^2-V(\sigma(f(t))]\dd t$ with
energy $e$. The last instant $\tau_n$ is fixed through the condition
$f_n(\tau_n)=b-1/n$. Since by energy conservation
$\frac{1}{2}\dot{f}_n^2+V(\sigma(f(t))=e$, we have
$\tau_n-t_0=\int_{0}^{b-1/n} \frac{\dd
s}{\sqrt{2(e-V(\sigma(s)))}}\le\frac{t_1-t_0}{2}$, that is
$\tau_n\le (t_0+t_1)/2$.

The curves $q_n$ will have the form $q_n(t)=\sigma(f_n(t))$ for
$t_0\le t\le \tau_n$; $q_n(t)=\sigma(f_n(2\tau_n-t))$ for $\tau_n\le
t\le 2\tau_n-t_0$, and $q_n(t)=q_0$ for $2\tau_n-t_0\le t\le t_1$ (note that $2\tau_n-t_0\le t_1$).
It should be noted that
$\dot{f}^2=a(\frac{d}{dt}\sigma(f(t),\frac{d}{dt}\sigma(f(t))$ since
$\sigma$ is $a$-arc-length parametrized.


Over the last stationary segment the action functional gives a
vanishing contribution. The contribution of the other two time
intervals is the same, up to the terms $\int b$ which cancel out,
thus the sum equals
\begin{align*}
&2 \int_{t_0}^{\tau_n} (T-V) \dd
t=-2e(\tau_n-t_0)+4\int_{t_0}^{\tau_n} T\dd
t\\
&=-2e(\tau_n-t_0)+4\int_{t_0}^{\tau_n} (e-V)\dd t
=-2e(\tau_n-t_0)+4\int_0^{b-1/n} \frac{(e-V)}{\sqrt{2(e-V)}} \dd s
\\&\le 2\int_0^b
\sqrt{2(e-V)} \dd s-2e(\tau_n-t_0)<I<+\infty
\end{align*}
by the definition of $\sigma$, for some finite constant $I(e)$
independent of $n$.
\end{proof}

\section{Some examples} \label{s:exa}

The problem as to whether a distinguishing generalized gravitational
wave has to be causally continuous proved to be complex and was left
open by previous studies.

With Theorem \ref{bup} we  proved that for generalized gravitational
wave spacetimes the levels of the  causal ladder of spacetimes
between strong causality and stable causality coincide (Fig.\ \ref{fig:ladder} summarizes the causal ladder of these spacetimes). In this
section we shall show that, nevertheless, distinction and strong
causality differ and so that not all distinguishing generalized
gravitational waves have to be reflective.

\begin{figure}[tp]
\psfrag{c}{{\small (E.-L.) Completeness}} \psfrag{t}{{\small Time
independence}} \psfrag{q}{{\small $Q$ is compact}}
\psfrag{e}{{\small \parbox{2cm}{\small Global subsolution to
H.-J.}}} \psfrag{g}{{\scriptsize  Global hyperbolicity:}}
\psfrag{f}{{\small $S(e_0,e_1)$, $t_0<t_1$, is finite}}
\begin{center}
\includegraphics[width=11cm]{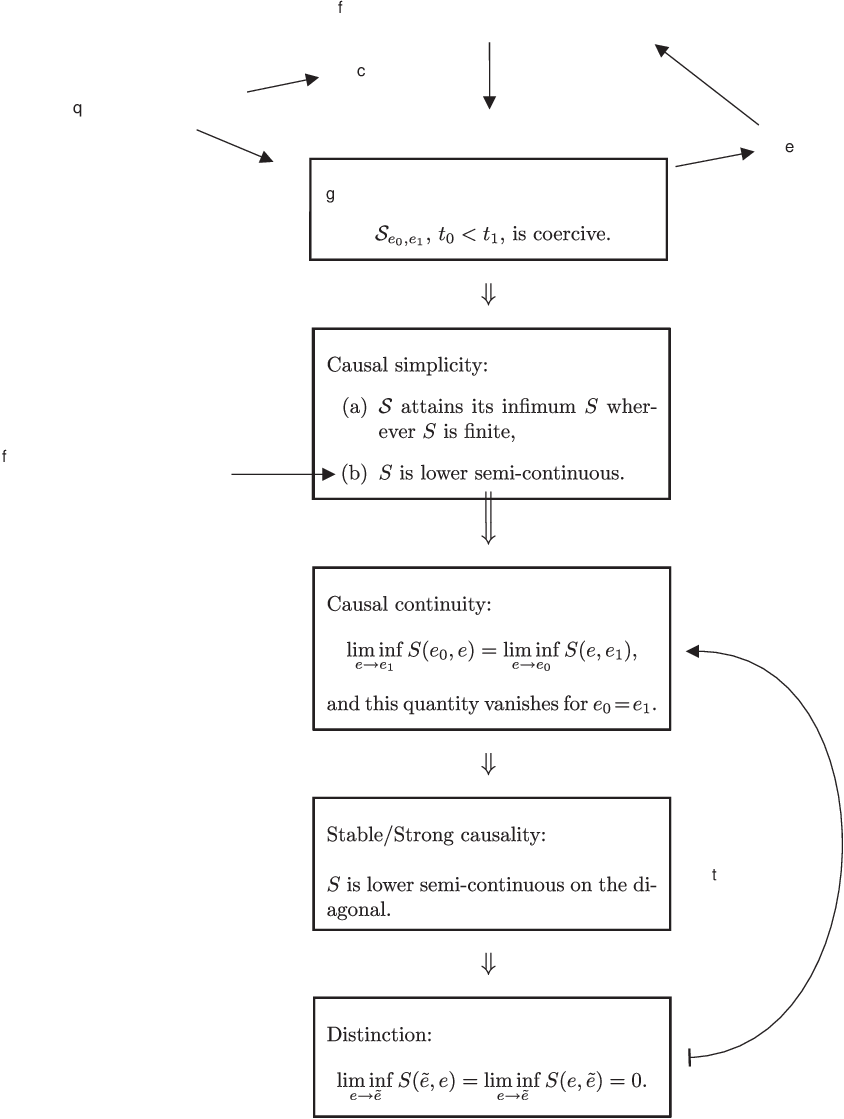}
\end{center}
\caption{The causal ladder for the generalized gravitational wave
spacetimes, and its connection with the properties of the action
potential $S$ and action functional $\mathcal{S}$. The spacetime is
necessarily non-total imprisoning and the levels which have not been
identified are indeed different. If an arrow crosses a property,
the property is used in the implication. Homogeneous $pp$-waves are complete
and causally continuous, also when they depend on time (Theorem
\ref{hcc}), and for them causal simplicity and global hyperbolicity
coincide.} \label{fig:ladder}
\end{figure}

We shall give some remarkable examples which prove that causal
simplicity differs from global hyperbolicity, and we shall also give
new examples that distinction does not necessarily hold. Moreover,
we shall show that gravitational plane waves are causally continuous
but not causally simple, as a consequence of the mechanics of the
classical harmonic oscillator.


We begin by noting that there are examples of generalized
gravitational wave spacetimes which are globally hyperbolic. As
observed in the introduction, Minkowski $(d+1)+1$ spacetime provides
the simplest example. It corresponds to the dynamical system of a
free particle moving in Euclidean space.

\subsection{Global hyperbolicity differs from causal simplicity: \texorpdfstring{\\}{} {Marchal-Chenciner's
theorem}} \label{yak}

The configuration of $N$ point particles in the Euclidean
3-dimensional space $(E,\langle \rangle )$, is given by $E^N$. Each
point on the space is determined by an $N$-tuple $q=(\vec{r}_1,
\vec{r_2},\ldots,\vec{r}_N)$. The non-collision configuration space
$Q\subset E^N$ is given by the $N$-tuples for which no two positions
of the particles coincide, namely for all $i\ne j$, $\vec{r_i}\ne
\vec{r_j}$.

Let us assign to each particle a mass $m_j$. We endow $TE^N$ with
the ``mass scalar product'' $a_t(v,w)=\sum_j m_j \langle\vec{v}_j,
\vec{w}_j\rangle$ where $v=(\vec{v}_1, \vec{v_2},\ldots,\vec{v}_N)$
and analogously for $w$. Let us introduce the Newtonian potential
\[
V=-\sum_{i<j} \frac{m_i m_j}{\Vert \vec{r}_i-\vec{r}_j\Vert}
\]
Finally, we set $b_t=0$  so that the Lagrangian (\ref{clas}) of the
system is that of the Newtonian $N$-body problem.

Since $V<0$ the action functional is non-negative over every curve,
which implies that the least action satisfies $S(e_0,e_1)\ge 0$. By
Prop. \ref{bfj} $S: E\times E\to (-\infty,+\infty]$  is  lower
semi-continuous.

We are going to prove that $\mathcal{S}_{e_0,e_1}$ is not coercive
for some $e_0,e_1\in E$, $t_0<t_1$, (thus $(M,g)$ is not globally
hyperbolic) and that, provided $S(e_0,e_1)$ is finite, there is a
minimizer $q:[t_0,t_1]\to Q$ (necessarily smooth).

As a matter of terminology, we remark that in works on the Newtonian
$N$-body problem one can often find the statement that the action
functional is coercive. The reason is that in this field the
collision points are not regarded as points a `infinity'. In our
terminology the collision points $E^N\backslash Q$ are outside $Q$
and hence at infinity.

If we can find a sequence of curves $q_n:[t_0,t_1]\to Q$ which
escape every compact set (on $Q$) and keep the action bounded from
above then the action functional is not coercive. It is well known
that this possibility is realized in the Newtonian $N$-body problem
because, as it was noted by Poincar\'e, for some $e_0, e_1\in E$,
curves connecting $e_0$ and $e_1$ exist which are made of connected
pieces satisfying Newton's law which present singularities, and
whose total action stay finite (thus the singularities are
collisions otherwise the action would not stay finite, see von
Zeipel theorem \cite{xia92}). For instance, one could divide the
masses in two groups, where the points inside the same group are
placed at negligible distance to each other, and then consider the
head-on collision of the two groups, i.e. the so called
collision-ejection solution of the Kepler problem. It is then easy
to construct the sequence $q_n$ as a sequence converging to these
collision curves.

Another method to prove that $(M,g)$ is not globally hyperbolic
passes through Prop. \ref{nap}. Indeed, it is easy to prove that the
Jacobi metric for a positive energy fails to be complete because of
the collisions.

It remains to prove that $S(e_0,e_1)$ is attained by some
(collision-free) minimizer. Remarkably, this problem has been
recently solved in the affirmative by C. Marchal \cite{marchal02}
and A. Chenciner \cite{chenciner02}  (see also \cite{ferrario04} for
a generalization). We conclude that the spacetime
$M=\mathbb{R}\times Q\times
 \mathbb{R}$ equipped with the metric
\[
g=\sum_j^N m_j(\dd \vec{r}_j)^2  \!-\!\dd t \otimes \dd y -\!\dd y
\otimes \dd t+2 \sum_{i<j} \frac{m_i m_j}{ \Vert
\vec{r}_i-\vec{r}_j\Vert}\, \dd t ^2 ,
\]
is causally simple but not globally hyperbolic. Furthermore, it is
not hard to prove that since $V$ is harmonic with respect to every
variable $\vec{r}_j$, the spacetime satisfies the vacuum Einstein
equations on $M$ (see e.g. \cite{ehrlich92}, for the calculation of
the Einstein tensor for a gravitational wave).

A simpler example can be obtained by considering the 2-body problem.
Introduced the relative position $\vec{r}=\vec{r}_2-\vec{r}_1$, the
center of mass position
$\vec{r}_G=\frac{m_1\vec{r}_1+m_2\vec{r}_2}{m_1+m_2}$, the total
mass $M=m_1+m_2$, and the reduced mass $\mu=m_1m_2/M$, the kinetic
energy can be rewritten
$T=\frac{1}{2}M\dot{\vec{r}}_G^2+\frac{1}{2}\mu \dot{\vec{r}}^2$ and
the potential $V=-\frac{M\mu}{r}$. The mass metric can be rewritten
$m_1 \dd \vec{r}_1^2+ m_2\dd \vec{r}_2^2=M\dd \vec{r}_G^2+\mu \dd
\vec{r}^2$. With no loss of generality, as the motion takes place on
a plane, we can restrict ourselves to the planar case. As the total
momentum $M\dot{\vec{r}}_G$ is conserved, the center of mass
proceeds at constant velocity while the relative motion is
determined by the minimization of the Lagrangian $\frac{1}{2}\mu
\dot{\vec{r}}^2+\frac{M\mu}{r}$ where $\vec{r}=w \vec{i} + z
\vec{j}$ is a planar vector.

We conclude that in this mechanical problem, as the space is
$Q=\mathbb{R}^2\backslash\{(0,0)\}$, the corresponding spacetime $M$
is 4-dimensional and endowed with the metric
\[
g=\mu (\dd w^2+\dd z^2)-\!\dd t \otimes \dd y -\!\dd y \otimes \dd
t+2 \frac{\mu M}{ \sqrt{w^2+z^2}}\, \dd t ^2.
\]
Since by Marchal's theorem the minimizers of the 2-point action
functional exist and are collisionless, we conclude, arguing as
above, that $(M,g)$ is causally simple but not globally hyperbolic.
However, the potential $1/r$ is not harmonic in two dimensions, thus
$g$ does not solve the vacuum Einstein equations. Moreover, the
metric cannot be made of Einstein type through multiplication by a
conformal factor.

Fortunately, in a recent work \cite{barutello08} Barutello, Ferrario
and Terracini proved, elaborating  Marchal's strategy, that for
logarithmic potentials the action minimizing orbits are also
collisionless. As a consequence, for
$Q=\mathbb{R}^2\backslash\{(0,0)\}$, and the same coordinates as
above with $r=\sqrt{w^2+z^2}$, the metric
\[
g=\mu (\dd r^2 +r^2\dd \theta^2)-\!\dd t \otimes \dd y -\!\dd y
\otimes \dd t+2 \mu M \ln (r/r_0)\, \dd t ^2.
\]
defines a very simple causally simple non-globally hyperbolic
spacetime which satisfies the vacuum Einstein equations.

\begin{remark}
The previous metrics clarify  the following point. One could be lead
to suspect that if, say, strong causality holds at an event
$x=(t,q,y)$, then the generic gravitational wave spacetime,
restricted to a sufficiently small time interval $(t_0,t_1)\ni t$,
should be globally hyperbolic. This belief is incorrect as for a
given time $t$, and for any $r=\sqrt{w^2+z^2}$
 we can find a collision-ejection trajectory connecting two
points at the same distance $r$ from the central singularity. As the
time needed in the collision goes to zero with $r$ going to zero, we
conclude that the collision-ejection can happen in a time smaller
than $(t_1-t)/2$ thus that no `sandwich' portion of the above
spacetime can be globally hyperbolic.
\end{remark}

\subsection{Causal simplicity differs from causal continuity: plane waves and harmonic oscillators}

The simplest  causally continuous but not causally simple example is
given by $(M,g)$ where $M$ is Minkowski 3+1 spacetime endowed with
coordinate $y,t,z$ and metric $g=-2\dd y \dd t+\dd w^2+\dd z^2$ and
where we have removed a lightlike plane generated by $\p_y$  and
$\p_t$ (or equivalently, a point of the quotient space
$\mathbb{R}^2$ endowed with coordinates $(w,z)$). It is not causally
simple because condition (a) of Theorem \ref{cga2} is not satisfied
(while (b) holds true). It can be easily checked to be causally
continuous, indeed, it is distinguishing because it is a subset of a
distinguishing spacetime, and it admits a complete timelike Killing
vector, see \cite[Prop.\ 3.1]{clarke88} or Theor.\ \ref{byt}.

We can also find an example of causally continuous spacetime which
is not causally simple because condition (b) of Theorem \ref{cga2}
is not satisfied (while (a) holds true). Consider $Q=\mathbb{R}$,
$a_t=\dd q^2$, $b_t=0$, $V(t,q)=\frac{1}{2} k q^2$. As the potential
is almost quadratic the spacetime is strongly causal (Cor.\
\ref{njx2}). By time independence  the spacetime is causally
continuous (Theor.\ \ref{byt}).

Let $(e_0,e_1)$ be such that $t_1<t_0+\pi/\sqrt{k}$, then an
application of the Poincar\'e-Dirichlet inequality  gives that the
action is bounded and hence lower semi-continuous (Prop. \ref{bfj}),
and furthermore, it is easy to check that a minimum exists. If
$t_1=t_0+\pi/\sqrt{k}$, then the least action is finite only at
$(e_0,e_1)$ with $q_0=-q_1$, where it vanishes, but is not lower
semi-continuous there, and hence $S(e_0,e_1)=-\infty$ whenever
$t_1>t_0+\pi/\sqrt{k}$ (Lemma \ref{vtr}, \ref{vtr2}). Thus in this
case (b) does not hold but (a) holds.

Let us consider the plane fronted gravitational waves. For these
spacetimes $Q=\mathbb{R}^2$ has coordinates $(w,z)$, $a_t$ is the
usual Euclidean metric, $b_t=0$, $M=\mathbb{R}\times Q\times
\mathbb{R}$, and $V=\frac{1}{2}f(t)[w^2-z^2]+ h(t) w z$, where $f,h:
\mathbb{R}\to \mathbb{R}$ are arbitrary $C^2$ functions. As a
consequence, the spacetime metric is
\[
g=[\dd w^2+\dd z^2]-\dd t\otimes \dd y-\dd y \otimes \dd
t-\{f(t)[w^2-z^2]+2 h(t) w z\} \dd t^2.
\]
These spacetimes are vacuum solutions to the Einstein equations and
it has been proved by Ehrlich and Emch \cite[Prop. 5.12,Theor.
6.9]{ehrlich92} that for suitable choices of $f$ and $g$ they are
causally continuous but not causally simple (with $f=h=0$ we get
Minkowski spacetime which is globally hyperbolic). Storically the
first result on the causality properties of these solutions was
obtained by Penrose who proved that for some choices of $f$ and $h$,
they are not globally hyperbolic \cite{penrose65}.

It is easy to visualize the lack of causal simplicity  using the
mechanical behavior on the quotient space. Let $h=0$, let $k>0$ be a
constant and let $f$ be such that $f=k/2$ on the interval $[0,
2\pi/\sqrt{k}]$, and arbitrary elsewhere. The Lagrangian of the
mechanical problem on $Q$ is
\[
L=\frac{1}{2} [\dot{w}^2+\dot{z}^2]-\frac{k}{2}[w^2-z^2].
\]
It is clear that the coordinates $w$ and $z$ are decoupled, and that
the coordinate $w$ represents the amplitude of an harmonic
oscillator with period $2\pi/\sqrt{k}$. Let $q_0=q_1=(0,0)$,
$t_0=0$, $t_1=\pi/\sqrt{k}$. Let us consider the solutions
$q_A(t)=(A\sin (\sqrt{k} t),0)$ of the E-L equation of the harmonic
oscillator connecting $e_0$ with $e_1$. They are parametrized by the
amplitude $A\ge 0$ and, as it is easy to check, the value of the
action functional $\mathcal{S}_{e_0,e_1}[q_A]$ over such paths is
zero and hence independent of the amplitude. On the spacetime this
leads to the phenomenon of refocusing discovered by Penrose
\cite{penrose65}, namely the lightlike geodesics that correspond to
the light lift of those solutions leave $x_0=(e_0,0)$ and meet again
at the event $x_1=(e_1,0)$. These lightlike geodesics accumulate on
the past lightlike ray generated by $n$ ending at $x_1$ while no
point of this ray different from $x_1$ is causally related to $x_0$.
This fact implies that $J^+(x_0)$ is not closed and hence that
$(M,g)$ is not causally simple. Another way to put it is that $S$ is
not lower semi-continuous at $(e_0,e_1)$. Indeed,
\[
\mathcal{S}_{e_0,(q_0,\frac{1}{\sqrt{k}}
(\pi-\epsilon))}[q_A]=\int_0^{\frac{1}{\sqrt{k}} (\pi-\epsilon)}
\frac{1}{2}(\dot{w}^2-w^2)\dd t=-\frac{A^2}{4} \sqrt{k}
\sin(2\epsilon),
\]
thus for $A=1/\epsilon$, we see that the limit $\epsilon\to 0$ does
not  give  $\mathcal{S}_{e_0,e_1}[q_A](=0)$.
%
%
%

Coming to the positive results, let $f=k$ everywhere and let $h=0$.
The spacetime is reflective as a consequence of the independence of
time (see Theorem \ref{byt}). In order to prove casual continuity it
suffice to prove strong causality, which, according to our
characterization, follows from the lower semi-continuity of the
action functional (Theor. \ref{psf2}) at $(e_0,e_0)$ (recall that by
time independence strong casuality at a time implies strong
causality at any other time). It suffices to show that
$S(e_0',e_1')$ is bounded from below for some events such that
$t_0'<t_0$, $t_0=t_1<t_1'$. Thus let these events be
$e_0'=(t_0-\epsilon,0,0)$ and $e_1'= (t_0+\epsilon,0,0)$, with
$2\epsilon <1/\sqrt{k}$. Taking the integral over a curve $t\to
(t,w(t),z(t))$ which starts from $e_0'$ and returns to $e_1'$,
 we have
\[
\int L\, \dd t=\frac{1}{2}\int [ \dot{z}^2+ k z^2] \dd
t+\frac{1}{2}\int [\dot{w}^2- kw^2] \dd t\ge 0
\]
because  the former integral on the right-hand side is clearly
non-negative, while the last integral is also non-negative because
of the  Dirichlet-Poincar\'e inequality, $\int^{t_1'}_{t_0'} {w}^2
\dd t\le (t_1'-t_0')^2 \int^{t_1'}_{t_0'} \dot{w}^2 \dd t$.

\subsection{Strong causality and distinction differ}

Let $Q=\mathbb{R}$, $a_{t}=\dd q^2$, $b_t=0$, and $V(t,q)=q^4
e^{-aq^2 t^2}$. We wish to show that for sufficiently large $a$,
distinction holds everywhere on $M$, but strong causality does not
hold at any point of
$\pi^{-1}(e)$, $e=(0,0)$.\\

{\bf Step 1. Strong causality does not hold everywhere}

For any positive integer $n$, let us consider the piecewise $C^1$
curve $q_{n}: [-\frac{1}{n},\frac{1}{n}] \to Q$ given by
\begin{align*}
q_{n}(t)&=\frac{8n^2}{a} (t+\frac{1}{n}), &  t\in
[-\frac{1}{n},-\frac{1}{2n}], \\
q_{n}(t)&=- \frac{2}{at} , &  t\in [-\frac{1}{2n},-\frac{1}{3n}],\\
q_{n}(t)&=\frac{6n}{a} , &   t\in
[-\frac{1}{3n},+\frac{1}{3n}], \\
q_{n}(t)&=\frac{2}{at} , &  t\in [\frac{1}{3n},\frac{1}{2n}],\\
q_{n}(t)&=\frac{8n^2}{a} (\frac{1}{n}-t), &  t\in
[\frac{1}{2n},\frac{1}{n}].
\end{align*}
The curve starts from $q=0$ at time $-1/n$ and returns to this same
point at time $1/n$. Let $\check{e}_{n}=(-1/n, 0)$ and
$\hat{e}_{n}=(1/n, 0)$ be the starting and ending points of
$e(t)=(t,q_{n}(t))$. We wish to prove that $\liminf_{n\to
+\infty}S({\check{e}_n,\hat{e}_n})=-\infty$, that is that $S$ is not
lower semi-continuous at $(e,e)$.

Let us evaluate an upper bound for
$\mathcal{S}_{\check{e}_{n},\hat{e}_{n}}[q_{n}]$. Since $q_{n}(t)$
starts from the origin, reaches the maximum distance $\frac{6n}{a}$,
and comes back to the origin a total time $2/n$, the kinetic term
is bounded from above by $\frac{2}{n}(\frac{6n}{a})^2$.

The potential $V$ is positive, thus $\int V \dd t$ is bounded from
below by its integral on a subset of its original domain, in
particular by
\[
\int_{-\frac{1}{2n}}^{-\frac{1}{3n}}V(t,q(t))\dd
t+\int_{\frac{1}{3n}}^{\frac{1}{2n}}V(t,q(t))\dd
t=2\int_{\frac{1}{3n}}^{\frac{1}{2n}}(\frac{2}{a})^4 \frac{1}{t^4}
e^{-2} \dd t=\frac{2}{3}(\frac{2}{a})^4 e^{-2}[(3n)^3-(2n)^3].
\]
We conclude that
\[
\mathcal{S}_{\check{e}_{n},\hat{e}_{n}}[q_{n}]\le
{2n}(\frac{6}{a})^2- \frac{10}{3}(\frac{2}{a})^4 e^{-2} n^3.
\]
As the right-hand goes to $-\infty$ for $n\to +\infty$, the desired
conclusion follows.\\

{\bf Step 2. Distinction holds}

We observe that $V(t,q)$, for fixed $t$ is given by $q^4$ for  $t=
0$, while for $t\ne 0$ it reaches a maximum
$V^*(t)=(\frac{2}{a})^4e^{-2} \frac{1}{t}^4$ at $q^*(t):=
\frac{2}{a}\frac{1}{\vert t\vert}$ and $-q^*(t)$. Moreover, for
given $q$, $V(t,q)$ decreases on $[0,+\infty)$ and increases on
$(-\infty,0]$.

Let  $t\ne 0$, and let us consider the point $e_0=(t_0,q_0)$. There
is some small compact interval $I$, $t_0\in \textrm{I}$, $0\notin
I$, such that there are $m,\epsilon>0$, with the property that for
every $t\in {I}$ and $q\in \mathbb{R}$, $V(t,q)\le
m+\frac{\epsilon}{2} q^2$. By Theorem \ref{psf2}  and \ref{ber}
point (a), the spacetime is strongly causal at $x_0$.

Let us consider the event $\hat{e}=(0,\hat{q})$ in the time slice
$Q_0$. Let
 $q_n :[0,1/n]\to Q$, $q_n(0)=\hat{q}$,
$q_n(1/n)\to \hat{q}$ for $n \to +\infty$. We wish to prove that
$\mathcal{S}_{\hat{e},(1/n, q_n(1/n))}[q_n]$ is bounded from below
by a constant $C$ independent of $n$ and $q_n(t)$, for that would
imply by Lemma \ref{prf} Eq. (\ref{bpo1}) that  $S(\hat{e}, \cdot)$
is lower semi-continuous at $\hat{e}$.


Let us set $\theta_n=1$ if $q_n(t)$ intersects $q^*(t)$ at a certain
time $\bar{t}_n\in [0,1/n]$, and $\theta_n=0$, $\bar{t}_n=1/n$, if
they do not intersect.
Moreover, let $\tilde{t}_n\in [0,\bar{t}_n]$ be the time of the
maximum distance of $q_n(t)$ from $q=0$, as it is reached in the
interval $[0,\bar{t}_n]$, and set $\tilde{q}_n=q_n(\tilde{t}_n)$.
We have
\[
\frac{1}{2} \int_0^{1/n} \dot{q}_n^2 \dd t\ge
\frac{1}{2}\int_0^{\tilde{t}_n} \dot{q}_n^2 \dd t\ge
\frac{(\tilde{q}_n-\hat{q})^2}{2\tilde{t}_n}.
\]
Observe that for $t\in [0,\bar{t}_n]$ we have $q_n(t)\le
\tilde{q}_n=q(\tilde{t}_n)\le q^*(\tilde{t}_n)$, and $V(t,\cdot)$
increases on $[0,q^*(t)]$, thus
\begin{align*}
\int_0^{1/n} V(t,q_n(t)) \dd t&=\int_0^{\bar{t}_n} V(t,q_n(t)) \dd
t+\theta_n\int_{\bar{t}_n}^{1/n} V(t,q_n(t)) \dd t\\
&\le \int_0^{\bar{t}_n} V(t,\tilde{q}_n) \dd t+\theta_n
\int_{\bar{t}_n}^{1/n}
V^*(t) \dd t\\
&\le \int_0^{+\infty} V(t,\tilde{q}_n)\dd t+\frac{e^{-2}}{3}
(\frac{2}{a})^4 \theta_n[\frac{1}{\bar{t}_n^3}-n^3]\\
&\le \frac{1}{2}\sqrt{\frac{\pi}{a}} \,\vert
\tilde{q}_n\vert^3+\frac{e^{-2}}{3} (\frac{2}{a})^4\theta_n
[\frac{1}{\tilde{t}_n^3}-n^3].
\end{align*}
The action functional is bounded by
\begin{align*}
\mathcal{S}_{e_0,e_{n}(1/n)}[q_n]&\ge\frac{(\tilde{q}_n-\hat{q})^2}{2\tilde{t}_n}-
\frac{1}{2}\sqrt{\frac{\pi}{a}}
\,\vert\tilde{q}_n\vert^3-\frac{e^{-2}}{3} (\frac{2}{a})^4 \theta_n
\frac{1}{\tilde{t}_n^3}+ n^3 \theta_n \frac{e^{-2}}{3}
(\frac{2}{a})^4\\
&\ge \left[\frac{a}{4}- \frac{1}{2}\sqrt{\frac{\pi}{a}} -\frac{2
e^{-2}}{3a}\theta_n\right] \vert\tilde{q}_n\vert^3-\frac{a}{2}
\hat{q}\,\tilde{q}_n^2 +\frac{a}{4}\hat{q}^2\,\tilde{q}_n
+n^3\theta_n\frac{e^{-2}}{3} (\frac{2}{a})^4
\end{align*}
where we used $\vert \tilde{q}_n\vert\le
q^*(\tilde{t}_n)=\frac{2}{a} \frac{1}{\tilde{t}_n}$. Let us take $a$
sufficiently large in such a way that, $\frac{a}{4}-
\frac{1}{2}\sqrt{\frac{\pi}{a}} -\frac{2 e^{-2}}{3a}>0$, then the
right-hand side, which we regard as a function of $(\tilde{q}_n,
\theta_n)$ and hence of $q(t)$, is bounded from below.

%
%
%

\subsection{Distinction does not necessarily hold}

Flores and S\'anchez \cite{flores03} and Hubeny,  Rangamani and Ross
\cite{hubeny03,hubeny04,hubeny05} have observed that spacetimes of
Brinkmann type do not need to be distinguishing. The simplest
example, according to Theorem \ref{dis},  is obtained for a time independent Euclidean metric $a_t$
($Q=\mathbb{R}$), $b_t=0$, and a potential $V(t,q)$ which grows
faster than quadratically, e.g.  $V(t,q)=k q^4$, where $k$ is a constant.

There are other interesting examples of spacetimes belonging to the
class considered in this work which are not distinguishing. Using our characterization in terms of the lower-semicontinuity of the partial least action, it is not difficult to check if a
spacetime is distinguishing or not. For instance, the choice $a_{t\,
ij}=\delta_{ij}$, $i=1,2$, $b_t=f(\sqrt{q_1^2+q_2^2}\,)\,(q_1\dd
q_2-q_2\dd q_1)$, $V=0$, is distinguishing if $f(r)$ stays bounded
and non-distinguishing if $f(r) \to +\infty$ as $r\to+\infty$.

\section{Future sets and Lax-Oleinik semigroup}\label{s:future_sets}

In this section we study the future sets, their representation
through time slices and the relation between the causal relation and
the Lax-Oleinik semigroup.

\subsection{Future sets}

We recall that a subset $F\subset M$ is a future set if
$I^{+}(F)\subset F$. Past sets are defined analogously and share
time dual properties. For a future set we have $\bar{F}=\{q \in M:
I^{+}(q) \subset F\}$ and thus $\bar{F}$ is also a future set
\cite{beem96}. For instance given $x\in M$, $I^{+}(x)$, $J^{+}(x)$,
$\overline{J^{+}(x)}$ are all future sets. If $U$ is a closed future
set then if $x\in U$, we have $I^{+}(x)\subset U$ and thus
$J^+(x)\subset \overline{I^{+}(x)}\subset \bar{U}=U$, that is the
closed future set is sent into itself under the causal relation.

\begin{proposition} \label{nts}
Let $U$ be a closed future set on $(M,g)$, then there is a lower
semi-continuous function $u: E \to [-\infty,+\infty]$ such that
\begin{equation} \label{zgq}
U=\{(e,y):  u(e)\le y\}.
\end{equation}
Furthermore, for every $e_0,e_1\in E$ with $t_0<t_1$ and every
$C^1$ curve $e(t)=(t,q(t))$ connecting $e_0$ to $e_1$, the function
$u$ satisfies
\begin{equation} \label{zgw}
u(e_1)\le u(e_0)+\mathcal{S}_{e_0,e_1}[q],
\end{equation}
where the right-hand side is well defined in $[-\infty,+\infty]$ since the action of the curve is finite. Where $u(e_0)$ and $u(e_1)$ are finite, Eq.\ (\ref{zgw}) is equivalent to $u(e_1)-u(e_0)\le S(e_0,e_1)$, the form in which it is mostly used below.
 Conversely, suppose you are given a lower semi-continuous function $u: E
\to [-\infty,+\infty]$ which satisfies Eq.\ (\ref{zgw}), then the set given by
Eq.\ (\ref{zgq}) is a closed future set.

Finally, there are constants $t_{min}, t_{max}\in
[-\infty,+\infty]$, $t_{min}[u]\le t_{max}[u]$, such that for
$t<t_{min}[u]$, $u=+\infty$; for $t=t_{min}[u]$, either $u>-\infty$
or $t_{min}[u]=t_{max}[u]$; for $t\in (t_{min}[u],t_{max}[u])$, $u$
is finite; for $t\ge t_{max}[u]$, $u=-\infty$.


Analogous statements hold for closed past sets, this time the
function is upper semi-continuous,  and the inequality is inverted
in Eq. (\ref{zgq}) but not in Eq. (\ref{zgw}).
\end{proposition}

\begin{proof}
Let $e\in E$, if $e\notin \pi(U)$ we define $u(e)=+\infty$ whereas
if $e\in \pi(U)$ we define
\[
u(e)=\inf_{y:\, (e,y) \in U} \, y .
\]
so that $u(e)\in[-\infty,+\infty)$. Let us prove Eq.\ (\ref{zgq}). The direction
$U\subset \{(e,y):  u(e)\le y\}$ follows from the fact that if
$(e,y)\in U$ then $e \in \pi(U)$, thus $u(e)\le y$ by definition of
$u(e)$. Conversely, let $e\in \pi(U)$. Since $U$ is closed and future, with every point $(e,y)$ it contains its causal future, and in particular the future lightlike ray generated by $n$, made of the points $(e,y')$ with $y'\ge y$. Thus $\{y: (e,y)\in U\}$ is a closed interval unbounded from above, namely $[u(e),+\infty)$ if $u(e)$ is finite, and $\mathbb{R}$ if $u(e)=-\infty$. This proves $\{(e,y):  u(e)\le y\} \subset U$.

Let us prove that $u$ is lower semi-continuous. Indeed, if $e_n \to
e$ but
\[
\liminf_{n \to +\infty}u(e_n)<u(e),
\]
let $C\in\mathbb{R}$ be such that $\liminf_{n \to +\infty}u(e_n)<C<u(e)$. Passing to a
subsequence we can assume $u(e_n)<C$, so that $(e_n,C)\in U$ by Eq.\ (\ref{zgq}). As $U$ is closed, $(e,C)\in U$, a contradiction with $C<u(e)$.

Now let $t_0<t_1$ and let $e(t)=(t,q(t))$ be a
$C^1$ curve joining $e_0$ to $e_1$. If $u(e_0)=+\infty$ there is nothing to prove. Otherwise, let $y_0\in\mathbb{R}$, $y_0\ge u(e_0)$, so that $x_0:=(e_0,y_0)\in U$ by Eq.\ (\ref{zgq}). The light lift
$\gamma(t)=(t,q(t),y_0+\mathcal{S}_{e_0,e(t)}[q\vert_{[t_0,t]}])$
is a causal curve connecting $x_0$ to $x_1:=
(e_1,y_0+\mathcal{S}_{e_0,e_1}[q])$, which therefore belongs
to $U$ as this set is closed and future. By the definition of
$u(e_1)$ we get $ u(e_1)\le y_0+\mathcal{S}_{e_0,e_1}[q]$. Choosing $y_0=u(e_0)$ if $u(e_0)$ is finite, and letting $y_0\to-\infty$ if $u(e_0)=-\infty$, we get Eq.\ (\ref{zgw}).

For  the converse, let $x_0=(e_0,y_0)\in U$ and $x_1=(e_1,y_1)\in
I^{+}(x_0)$ so that $t_0<t_1$ and $y_0\ge u(e_0)$.  Equation
(\ref{cgf1}) implies $y_1-y_0>S(e_0,e_1)$, thus there is a $C^1$ curve $e(t)=(t,q(t))$ connecting $e_0$ to $e_1$ such that $y_1-y_0>\mathcal{S}_{e_0,e_1}[q]$. By Eq.\ (\ref{zgw})
\[
y_1>y_0+\mathcal{S}_{e_0,e_1}[q]\ge u(e_0)+\mathcal{S}_{e_0,e_1}[q]\ge u(e_1),
\]
and thus $x_1\in U$, that is, because $x_0$ is arbitrary, $U$ is
future.

Finally, if $x_n=(e_n,y_n) \in U$, $x_n \to x=(e,y)$, then
$u(e_n)\le y_n$ which taking the limit, by lower semi-continuity of
$u$, gives $u(e)\le \liminf u(e_n)\le \liminf y_n=y$ thus $x\in U$,
and hence $U$ is closed.

Note that if $u(t,q)=-\infty$ then for all $t'>t$,
$u(t',q')=-\infty$ because $u(t',q')\le u(t,q)+\mathcal{S}_{e,e'}[\sigma]$ for any $C^1$ curve $\sigma$ connecting $e=(t,q)$ to $e'=(t',q')$. The lower semi-continuity of $u$ implies that the
same implication holds for $t'=t$. Also note that if $u(t,q)\ne
+\infty$  then for $t'>t$ the same holds for $u(t,q')$ again because
$u(t',q')\le u(t,q)+\mathcal{S}_{e,e'}[\sigma]$ thus if $u(t,q)=+\infty$ then the same
is true for $u(\tilde{t},\tilde{q})$ with $\tilde{t}<t$. From these
observations the desired conclusion follows.
\end{proof}

Let $F$ be a future set, the set $B=\dot{F}$ (with a dot we denote
the boundary) is an {\em achronal boundary}. It is indeed achronal
in the sense that no two points are connected by a timelike curve.
The achronal boundary is also the boundary of a past set $B=\dot{P}$
where $P=M\backslash F$.


Given an achronal boundary $B$ for any point $q\in B$ there may not
exist points $p,r \in B$, with $p \in E^{-}(q)\backslash q$, $r \in
E^{+}(q)\backslash q$. As usual in Lorentzian geometry
\cite{hawking73}, the achronal boundary $B$  can be partitioned into
four disjoint sets $B_0$, $B_+$, $B_-$ and $B_N$ according to the
following prescription: $\exists p,r \Rightarrow q \in B_N$;
$\exists p, \nexists r \Rightarrow q \in B_+$; $\nexists p, \exists
r \Rightarrow q \in B_-$ and $\nexists p, \nexists r \Rightarrow q
\in B_0$.

The following result is standard in Lorentzian geometry.

\begin{lemma} \label{psh}
Let $B=\dot{F}$ be an achronal boundary.  If $p \in B_{N}\cup
B_{+}\cup B_{-}$ then through $p$ there passes an achronal lightlike
geodesic (not necessarily inextendible) entirely contained in $B$
which cannot be further extended without escaping $B$. If $p\in B_N$
there is only one lightlike geodesic with the above properties. In
any case, if this geodesic has future endpoint then the endpoint
belongs to $B_{+}$, if it has past endpoint then the endpoint
belongs to $B_{-}$. In particular if $p\in B_{+}$ then the future
endpoint exists and is $p$ itself. Analogously, if  $p\in B_{-}$
then the past endpoint exists and is $p$ itself. Finally, if $p\in
B_0\cup B_{+}$ then $J^{+}(p)\backslash\{p\} \subset I^{+}(F)$ and
analogously, if $p \in B_0 \cup B_{-}$ then $J^{-}(p)\backslash\{p\}
\subset I^{-}(M\backslash F)$.
\end{lemma}

\begin{lemma} \label{pluf}
On the spacetime $(M,g)$, let $A$  be a closed and achronal set. Let
us consider the future set  $F=\overline{J^{+}(A)}$ which has
achronal boundary $B=\dot{F}$. We have $A\subset B$ and $B\backslash
A\subset B_+\cup B_N$. In particular given $p\in B\backslash A$
either there is an achronal lightlike geodesic segment contained in
$B$ connecting a point of $A$ to $p$ or there is an achronal past
inextendible lightlike geodesic contained in $B$ ending at $p$.
\end{lemma}

\begin{proof}
Let us show that $A\subset B$. Clearly, $A\subset F$. If $x\in A$,
then taken $y\ll x$ it cannot be $y\in F$, otherwise since $I$
is open one would find a timelike curve from $A$ to $x$ in
contradiction with the achronality of $A$. Thus $y \notin F$ and
since $y$ can be taken arbitrarily close to $x$, $x\in \dot{F}$, and
hence $A\subset B$.

Let $p\in B\backslash A$ and let $p_n\in J^{+}(A)$ such that $p_n\to
p$. There are $q_n \in A$, and causal curves  $\sigma_n\subset F$
connecting $q_n$ to $p_n$. Without loss of generality we can assume
that either $q_n \to q\in A$ or $q_n$ escapes every compact set. In
the latter case the limit curve theorem \cite{minguzzi07c}
establishes the existence of a past inextendible continuous causal
curve $\sigma^p\subset F$ ending at $p$. This curve must actually be
 achronal and not contained in $I^{+}(A)$ otherwise $p\in I^{+}(A)$
a contradiction with the achronality of the boundary. Thus
$\sigma^p$ is a past inextendible achronal lightlike geodesic
contained in $B$ (observe that
$\overline{J^{+}(A)}=\overline{I^{+}(A)}$) and hence taking $z\in
\sigma^p\backslash\{p\}$ we have $z \in E^{-}(p)\backslash\{p\}$. In
the former case by applying the limit curve theorem  either one has
the same case already considered or the limit curve joins $q$ to
$p$. Again by achronality of $B$ this limit curve must be an
achronal lightlike geodesic and hence $q\in E^{-}(p)\backslash\{p\}$.
\end{proof}

\begin{lemma} \label{lem:graph}
On the spacetime $(M,g)$, let $u_{t_0}: Q\to
\mathbb{R}$ be a continuous function, and let $A$ be the closed and acausal set defined by
\[
A=\{(t_0,q,u_{t_0}(q)): q\in Q  \}.
\]
Let us consider the future set $F=\overline{J^{+}(A)}$ which has
achronal boundary $B=\dot{F}$. The function
$u$ obtained from  Proposition
\ref{nts} is such that $t_{min}[u]=t_0$, and it is finite and continuous on $(t_0,t_{max})\times Q$, where $t_{max}=t_{max}[u]$.
Furthermore,
\begin{equation} \label{cor}
B\cap \{x: t(x)<t_{max}\}=\{(t_0,q,y): \ y \ge u(t_0,q) \} \cup \{(t,q,u(t,q)): \ t\in (t_0,t_{max})\}.
\end{equation}
The last set on the right-hand side is  empty if $t_{max}=t_0$. If
this is not the case then $u(t_0,q)=u_{t_0}(q)$, for any $q\in Q$, and $u$ is continuous on $[t_0,t_{max})\times Q$. Finally, if  $Q$ is compact then $t_{max}=+\infty$.
(If $t_{max}<+\infty$ the remaining part of $B$ is contained in the slice $N_{t_{max}}$.)

\end{lemma}

\begin{proof}
Since $F$ is a closed future set there is a lower semi-continuous
function $u: E \to [-\infty,+\infty]$ such that $F=\{(e,y):  u(e)\le
y\}$. Clearly, no point of $F$ is contained in the region $t<t_0$,
thus there $u=+\infty$. Since $A\subset F$ necessarily $u<+\infty$
at some point $(t_0,q)$ thus $t_{min}=t_0$. If $t_{max}=t_{min}$
there is nothing else to prove. It remains to
consider the case $t_{min}=t_0<t_{max}$.

Let us preliminarily prove Eq. (\ref{cor}). Let $x_1=(e_1,y_1)\in B$
be such that $t_0<t_1<t_{max}$ so that $x_1 \in B\backslash A$ and
$u(e_1)\le y_1$, where $u(e_1)$ is finite. Suppose that $u(e_1)< y_1$ and define
$x_1'=(e_1,y_1')$ with $u(e_1)<y_1'< y_1$. By construction $x_1'\in
F$ and it cannot be $x_1'\in I^{+}(F)$ otherwise $x_1\in I^{+}(F)$
in contradiction with $x_1\in B$. Thus $x_1'\in \p I^{+}(F)=B$.

By Lemma \ref{pluf} there is either a past inextendible achronal
geodesic ending at $x_1'$ contained in $B$ or a geodesic segment
contained in $B$ connecting a point of $A$ to $x_1'$. In the former
case it cannot be that the past inextendible geodesic is an integral
line of $n$, indeed, as it belongs to $B\subset F$, it would imply
that $u(e_1)=-\infty$, a contradiction. In the remaining cases there
is a point $z\in (E^{-}(x_1')\backslash\{x_1'\})\cap B$ necessarily
not contained in the integral line of $n$ passing through $x_1'$. As
a consequence, the portion of the integral line of $n$ in the causal
future of $x_1'$ belongs to the chronological future of $z$ and
hence of $F$. In particular $x_1\in I^{+}(F)$, which is a
contradiction which implies that $u(e_1)=y_1$.

The equation (\ref{cor}) is now proved (the points of $B$ with $t=t_0$ are easily characterized) and it also implies the
continuity of $u(t,q)$ for $t\in (t_0,t_{max})$. Indeed, let
$\hat{e}=(\hat{t},\hat{q})$, $\hat{t}\in (t_0,t_{max})$ and suppose
that upper semi-continuity does not hold, namely, there is a
sequence $(t_n,q_n)\to \hat{e}$ such that $u(t_n,q_n)\to
u^*>u(\hat{t},\hat{q})$, $u^*\in (-\infty,+\infty]$. We can assume that $t_0<t_n<t_{max}$. Let $y^*$ be such that $u(\hat{e})<y^*<u^*$. For sufficiently large $n$ the events $(t_n,q_n,y^*)$ do not belong to
$F$ since $y^*<u(t_n,q_n)$, while their limit $(\hat{e},y^*)$ belongs to $F\backslash B=\textrm{Int} F$ by Eq.\ (\ref{cor}), a contradiction.

If $t_{max}>t_0$ then taken $x\in A$ the past lightlike ray $\sigma$
ending at $x=(t_0,q,y)$ generated by $n$ is not entirely contained
in $F$, for otherwise by Prop. \ref{nts} it would be
$u(t,q)=-\infty$ and hence $t_{max}=t_0$. Let $\hat{x}$ be the first
point at which $\sigma$ enters $F$, so that, by Prop. \ref{nts},
$\hat{x}=(t_0,q,u(t_0,q))$. If $\hat{x}\notin A$, as $\hat{x}\in
\overline{J^{+}(A)}$, there must be a sequence of causal curves
starting from $A$ whose endpoints converge to $\hat{x}$. As the
slice $t=t_0$ is achronal and $\hat{x}<x$ it cannot be $\hat{x}\in
J^{+}(A)$, thus by the limit curve theorem there is a past
inextendible causal curve ending at $\hat{x}$ and contained in
$\overline{J^{+}(A)}=F$. Since $A$ is achronal, the limit curve must
be a past lightlike ray necessarily contained in the slice $t=t_0$
and hence coincident with $\sigma$, which gives contradiction since
$\hat{x}$ is the first point of $\sigma$ in $F$. The contradiction
proves that $\hat{x}\in A$ and since $A$ is acausal
$x=\hat{x}=(t_0,q,u(t_0,q))$ which proves that $A$ is the graph of
$u(t_0,\cdot)$.

Let us prove the upper semi-continuity of
$u\vert_{[t_0,+\infty)\times Q}$ at $(t_0,q)$. Suppose that upper
semi-continuity does not hold, namely, there is a sequence
$(t_n,q_n)\to (t_0,q)$ such that $u(t_n,q_n)\to u^*>u(t_0,q)$. Since
$t_0<t_{max}$ we can assume $t_n<t_{max}$. Without loss of
generality we can assume that either $t_n=t_0$ for all $n$, or
$t_n>t_0$ for all $n$. In the former case, by the just proved
interpretation of $A$ as a graph, $(t_n,q_n,u(t_n,q_n))\in A$ and
since $A$ is closed $(t_0,q,u^*)\in A$, thus $u^*=u(t_0,q)$ because
$A$ is the graph of $u(t_0,\cdot)$. The contradiction proves that
the former  case does not apply, and, as a consequence, that at
least the function $u(t_0,\cdot)$ is continuous. In the latter case,
$x_n:=(t_n,q_n,u(t_n,q_n))\in B$ by Eq. (\ref{cor}). Let $U\subset
Q$ be an open set of $(t_0,q)$, such that for $q'\in U$,
$u(t_0,q')-u(t_0,q)<\epsilon<2\epsilon<u^*-u(t_0,q)$. We can assume
$q_n\in U$. As the stationary motion at $q_n$ proves, for
sufficiently large $n$, the least action between $(t_0,q_n)$ and
$(t_n,q_n)$, is smaller than any positive value and hence of
$u^*-u(t_0,q)-2\epsilon<u^*-u(t_0,q_n)-\epsilon$, and furthermore
since $n$ is large we can assume $u^*-u(t_n,q_n)<+\epsilon$, thus
the least action between $(t_0,q_n)$ and $(t_n,q_n)$, is smaller
than $u(t_n,q_n)-u(t_0,q_n)$ which implies $x_n\in
I^{+}((t_0,q_n,u(t_0,q_n)))\subset I^{+}(A)$, in contradiction with
$x_n\in B$.

For the last statement, let $t_1>t_0$. As $Q$ is compact there is a constant $C$ such that $V+\frac{1}{2}a_t^{-1}(b_t,b_t)\le C$ on $[t_0,t_1]\times Q$, thus $L\ge -C$ on $[t_0,t_1]\times TQ$ and hence $S(e_0,e)\ge -C(t-t_0)$ for every $e_0=(t_0,q_0)$ and $e=(t,q)$ with $t_0\le t\le t_1$. By Eq.\ (\ref{cgf2}) if
$x=(t,q,y)\in {J^{+}(A)}$, $t\le t_1$, then  $y\ge \min_Q u_{t_0}-C (t-t_0)$, and the same inequality holds on $\overline{J^{+}(A)}$. As a consequence, $u(t,q)>-\infty$ for $t\le
t_1$, and since $t_1$ is arbitrary $t_{max}=+\infty$.

\end{proof}

\subsection{The Lax-Oleinik evolution map} \label{s:lo}

Let $u_{t_0}:  Q \to \mathbb{R}$ be a continuous function which we
regard as initial data at time $t_0$. We shall write
$u(t_0,\cdot):=u_{t_0}$.

\begin{definition}
The Lax-Oleinik evolution map
$T^{-}_{t_0,t_1}$ for $t_1 > t_0$ is defined by the expression
\begin{align*}
(T^{-}_{t_0,t_1}u_{t_0})(q_1)&=\!\!\!\inf_{e_0, \, t(e_0)=t_0}\!\!
[u(e_0)+S(e_0,e_1)]\\
&=\inf_{q}\,\Big[u_{t_0}(q(t_0))+\int_{t_0}^{t_1}
L(s,q(s),\dot{q}(s) ) \dd s\Big],
\end{align*}
where $e_1=(t_1,q_1)$ and the last infimum is with respect to all the $C^1$ curves
$q:[t_0,t_1]\to Q$ which end at $q_1$ starting from any point of $Q$. We also set
$T^{-}_{t_0,t_0}=Id$.
\end{definition}

It extends the function $u_{t_0}$ to the function
$u(t,q):=(T^{-}_{t_0,t}u_{t_0})(q)$ which is a map $u:
[t_0,+\infty)\times Q\to [-\infty, +\infty)$. In the time independent
(autonomous) case we denote $T^{-}_{t_0,t}=T^-_{t-t_0}$ and, as we shall see in a moment, the
evolution maps give a semigroup.

The next theorem proves that the Lax-Oleinik
evolution map describes the boundary of the future set generated by the
 graph of the initial condition. In other words, the Lax-Oleinik evolution is nothing but the causal propagation of an
 achronal set contained in a slice $N_{t_0}$ (Fig.\ \ref{fig:lax}).

\begin{figure}[p]
\begin{center}
\includegraphics[width=0.64\textwidth,trim=0 108 0 0,clip]{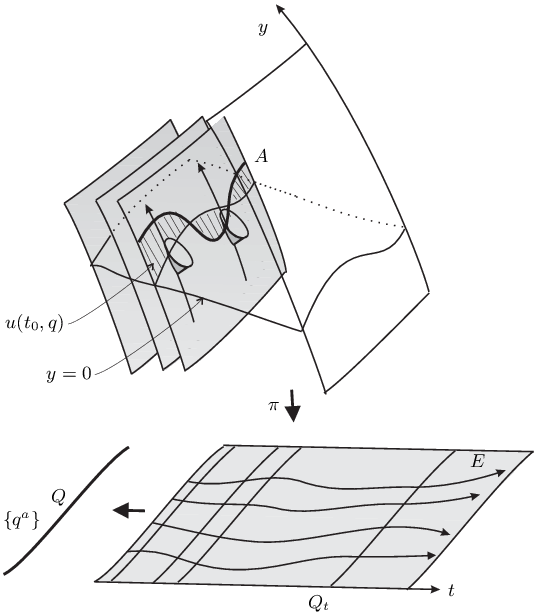}

\vspace{0.3cm}
\includegraphics[width=0.64\textwidth]{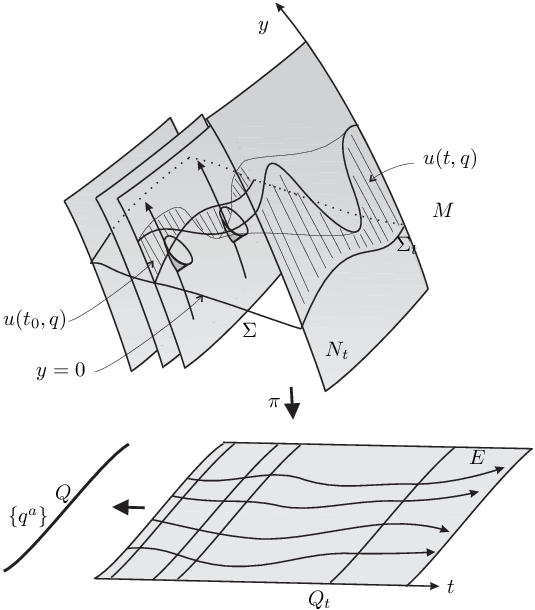}
\end{center}
\caption{The Lax-Oleinik evolution as the causal propagation of the initial condition (Theorem \ref{lo1}). Top: the initial condition $u(t_0,\cdot)$ is represented by its graph $A$, an achronal set contained in the lightlike slice $N_{t_0}$ (the light cones are tangent to the slices $N_t$). Bottom: the boundary of the causal future of $A$ intersects each subsequent slice $N_t$ in the graph of $u(t,\cdot)=T^-_{t_0,t}u_{t_0}$ (the hatched region lies between this graph and $\Sigma_t$). As in Fig.\ \ref{tw}, $\Sigma$ is the section $y=0$ of $\pi:M\to E$, with slices $\Sigma_t\subset N_t$, and $E$ projects on $Q$ along the Newtonian flow.} \label{fig:lax}
\end{figure}

\begin{theorem} \label{lo1}
On the spacetime $(M,g)$ let $A=\{(t_0,q,u_{t_0}(q)): q\in Q\}$
be the graph of the continuous initial condition $u_{t_0}$. Let us consider
the future set $F=\overline{J^{+}(A)}$, its achronal
boundary $B=\dot{F}$, and the  function
$u:E \to [-\infty,+\infty]$, $F=\{(e,y):u(e)\le y\}$, and the constant $t_{max}\in [t_0,+\infty]$  provided by Prop.\ \ref{nts} and Lemma \ref{lem:graph}. Then
\begin{itemize}
\item[(a)] $I^{+}(A)=\{(t,q,y): \ t>t_0 \textrm{ and } y> (T^{-}_{t_0,t}u_{t_0})(q)\}$.
\item[(b)] Suppose that $t_{max}>t_0$ (this is the case if $Q$ is compact, where actually $t_{max}=+\infty$). For $t_0\le t<t_{max}$ we have
$u(t,\cdot)=T^{-}_{t_0,t}u_{t_0}$, and this function is finite and continuous on $[t_0,t_{max})\times Q$, while $T^{-}_{t_0,t}u_{t_0}\equiv -\infty$ for $t>t_{max}$. Thus
\begin{align*}
B\cap \{x: t(x)<t_{max}\}=\ &\{(t_0,q, y):  y\ge u_{t_0}(q) \}\\
&\cup \{(t,q, (T^{-}_{t_0,t}u_{t_0})(q)):  t_0< t< t_{max} \},
\end{align*}
while, if $t_{max}<+\infty$, the remaining part of $B$ is the subset $\{(t_{max},q,y): y\le (T^{-}_{t_0,t_{max}}u_{t_0})(q)\}$ of the totally geodesic slice $N_{t_{max}}$.
\item[(c)] Under the assumption of (b), the Lorentzian distance from the achronal set $A$, $d(A,x)=\sup_{z\in A} d(z,x)$,   is given for $x=(t,q,y)$, $t_0<t<t_{max}$, by
\begin{equation} \label{dax}
d(A, x)=\sqrt{2\,[y-u(t,q)]^{+}\,(t-t_0)}\, ,
\end{equation}
where $[s]^+=\max(s,0)$.
\end{itemize}
\end{theorem}

\begin{proof}
(a). By Eq.\ (\ref{cgf1}) a point $x=(e,y)$ belongs to $I^{+}(A)$ if and only if $t>t_0$ and there is $q_0\in Q$ such that, defined $e_0=(t_0,q_0)$,
$y-u_{t_0}(q_0)>S(e_0,e)$, which is equivalent to $y> \inf_{q_0}[u_{t_0}(q_0)+S(e_0,e)]$.

(b). Let us denote for short $w(t,q)=(T^{-}_{t_0,t}u_{t_0})(q)$, $t>t_0$. Let us recall \cite{hawking73,beem96} that for any set $C$, $I^{+}(\overline{C})=I^{+}(C)$,  and that for any future set $F$, $\textrm{Int} F=I^{+}(F)$. Since $F=\overline{J^{+}(A)}=\overline{I^{+}(A)}$ we have $\textrm{Int} F=I^{+}(A)$, thus $F\backslash B=I^{+}(A)$. Let $e=(t,q)$ with $t_0<t<t_{max}$. By Lemma \ref{lem:graph} the only point of $B$ in the fiber $\pi^{-1}(e)$ is $(e,u(e))$, thus the points $(e,y)$ with $y>u(e)$ belong to $F\backslash B=I^{+}(A)$, which by (a) reads $y>w(e)$, while the points $(e,y)$ with $y<u(e)$ do not belong to $F$, and hence do not belong to $I^+(A)$, that is $y\le w(e)$. We conclude that $w(e)=u(e)$. If $t>t_{max}$ every point $(e,y)\in \pi^{-1}(e)$ belongs to $\textrm{Int} F=I^{+}(A)$, thus by (a) $w(e)=-\infty$. The continuity of $u$ on $[t_0,t_{max})\times Q$ and the expression for $B\cap \{t<t_{max}\}$ have been proved in Lemma \ref{lem:graph}. Finally, the points of the slice $N_{t_{max}}$ belong to $F$, thus they belong to $B=F\backslash I^{+}(A)$ if and only if, by (a), $y\le w(e)$.

(c). Let $z=(e_0,u_{t_0}(q_0))\in A$. By Lemma \ref{pcf} and Theorem \ref{vfg}, $d(z,x)^2=2[y-u_{t_0}(q_0)-S(e_0,e)]^{+}(t-t_0)$, thus taking the supremum over $q_0\in Q$
\[d(A,x)^2={2\,[y-\!\!\!\!\inf_{e_0,\,
t(e_0)=t_0}(u(e_0)+S(e_0,e))]^{+}\,(t-t_0)}={2\,[y-u(t,q)]^{+}\,(t-t_0)}.\]
\end{proof}

The next theorem collects the main properties of the Lax-Oleinik evolution map. They are well known for Tonelli Lagrangians on compact manifolds \cite{fathi08,cannarsa04}; here they are obtained, for possibly non-compact $Q$ and with no completeness assumption, using tools from
Lorentzian geometry. In this approach the minimizers are  generators of an achronal boundary,  the
points of non-differentiability of the solution are the points where several generators of the boundary meet, and the semi-concavity
of the solution is the semi-convexity of horizons \cite{chrusciel01,minguzzi14d}.

\begin{theorem} \label{lo2}
Let us keep the notations of Theorem \ref{lo1} and let us suppose that $t_{max}>t_0$. Let $\Omega=(t_0,t_{max})\times Q$. The Lax-Oleinik evolution map has the following properties.
\begin{itemize}
\item[(i)] {\em (Calibrated curves).} For every  $e_1=(t_1,q_1)\in \Omega$
there exists a $C^{r+1}$ solution $q:I\to Q$ to the E.-L.\ equations, $q(t_1)=q_1$, where either $I=[t_0,t_1]$ or $I=(t_0', t_1]$ with $t_0\le t_0'<t_1$, which
minimizes the functional $\mathcal{S}$ between any pair of its points, and such that for every $t\in I$
\begin{equation} \label{cal1}
u(t_1,q_1)=u(t,q(t))+\int_t^{t_1}L(s,q(s),\dot{q}(s)) \,\dd s.
\end{equation}
In the former case the infimum which defines $(T^-_{t_0,t_1}u_{t_0})(q_1)$ is attained on $q$. In the latter case $q$ is a past inextendible solution (hence it is not imprisoned in any compact set of $Q$ for $t\to t_0'$, Theor.\ \ref{jhy}); this case does not occur if the E.-L.\ flow is complete on $[t_0,t_1]$, e.g.\ if $Q$ is compact. The light lifts of these curves ending at $(e_1,u(e_1))$ are precisely the generators of the achronal boundary $B$ which reach that point.
\item[(ii)] {\em (Monotonicity and non-expansiveness).} If $u_{t_0}\le v_{t_0}$ then $T^{-}_{t_0,t}u_{t_0}\le T^{-}_{t_0,t}v_{t_0}$. Moreover, $\Vert
T^{-}_{t_0,t}u_{t_0}-T^{-}_{t_0,t}v_{t_0}\Vert_\infty \le \Vert
u_{t_0}-v_{t_0}\Vert_\infty$ wherever the left-hand side makes sense.
\item[(iii)] For every constant $k\in \mathbb{R}$, $T^{-}_{t_0,t}(u_{t_0}+k)=T^{-}_{t_0,t}u_{t_0}+k$.
\item[(iv)] {\em (Evolution property).} For $t_0\le t_1\le t_2<t_{max}$,
$T^{-}_{t_1,t_2}\circ T^{-}_{t_0,t_1}=T^{-}_{t_0,t_2}$.
In the autonomous case  $T^-_{s}\circ T^-_{t}=T^-_{s+t}$.
\item[(v)] {\em (Differentiability).} The function $u$ is
locally Lipschitz and locally semi-concave on $\Omega$, hence differentiable almost everywhere. It is differentiable at $e_1\in \Omega$ if and only if the
curve  mentioned in point (i) is unique, and it is differentiable at the points $(t,q(t))$, $t\in I\backslash\{t_0, t_1\}$, of any such curve.
At each point where $u$ is differentiable it is a
true solution of the Hamilton-Jacobi equation and $\p_q u=\p L/\p v\,(t,q(t),\dot q(t))$.
\end{itemize}
\end{theorem}

\begin{proof}
Proof of (i). The point $x_1=(e_1,u(e_1))$
belongs to $B\backslash A$, thus by Lemma~\ref{pluf} there is an achronal lightlike geodesic $\gamma\subset B$ ending at $x_1$ which either starts from a point of $A$ or is past inextendible. This geodesic is not an integral line of $n$ since by Theorem \ref{lo1} the fiber of $e_1$ intersects $B$ only at $x_1$. By Theorem~\ref{kdf} it is the light lift of its projection $(t,q(t))$, where $q:I\to Q$ solves the E.-L.\ equations. Since $\gamma\subset F\subset \{x: t(x)\ge t_0\}$ and $t$ is an affine parameter for $\gamma$, either $\gamma$ starts from $A$ and $I=[t_0,t_1]$, or $\gamma$ is past inextendible and $I=(t_0',t_1]$ with $t_0'\ge t_0$. As $\gamma\subset B$, for $t\in I\backslash\{t_0\}$ we have by Theorem \ref{lo1} $y(t)=u(t,q(t))$ (this equation holds also for $t=t_0$ since in the former case $\gamma$ starts from $A$), which jointly with the light lift condition $\dot y=L$ gives Eq.\ (\ref{cal1}). The curve $q$ is action minimizing because $\gamma$ is achronal (Cor.\ \ref{vqz}). In the former case Eq.\ (\ref{cal1}) with $t=t_0$ shows that the infimum is attained. In the latter case $q$ is a past inextendible solution, which is impossible if the E.-L.\ flow is complete on $[t_0,t_1]$. Finally, every generator of $B$ reaching $x_1$ is obtained in this way, and conversely the light lift of a curve which satisfies Eq.\ (\ref{cal1}) is a lightlike curve contained in the achronal boundary $B$, and hence a generator.

Proof of (ii) and (iii). They are immediate from the definition of $T^-_{t_0,t}$.

Proof of (iv). Let $t_0<t_1<t_2$. Every $C^1$ curve connecting $e_0=(t_0,q_0)$ to $e_2=(t_2,q_2)$ passes through some point $e_1$ of the slice $Q_{t_1}$, thus using also the triangle inequality (Cor.\ \ref{hgf}), $S(e_0,e_2)=\inf_{e_1,\, t(e_1)=t_1}[S(e_0,e_1)+S(e_1,e_2)]$. As a consequence
\begin{align*}
(T^-_{t_0,t_2}u_{t_0})(q_2)&=\inf_{e_0}\inf_{e_1}\, [u(e_0)+S(e_0,e_1)+S(e_1,e_2)]\\&=\inf_{e_1} \,[(T^-_{t_0,t_1}u_{t_0})(q_1)+S(e_1,e_2)]=(T^-_{t_1,t_2}T^-_{t_0,t_1}u_{t_0})(q_2),
\end{align*}
where we used the fact that, by Theorem \ref{lo1}, $T^-_{t_0,t_1}u_{t_0}$ is finite and continuous. In spacetime terms this identity just states that
$\overline{J^+(A)}\cap\{x: t(x)\ge t_1\}=\overline{J^+(A_1)}$, where $A_1$ is the graph of $u(t_1,\cdot)$.

Proof of (v). By Theorem \ref{lo1} and point (i), $B\cap \pi^{-1}(\Omega)$ is the graph of the continuous function $u$, it is an achronal topological hypersurface, and each of its points is the future endpoint of a lightlike geodesic segment contained in it. Thus it is  a horizon in the sense of  \cite{chrusciel01,minguzzi14d} (in the time dual version). It is known that a horizon is differentiable at a point if and only if the point belongs to just one generator \cite{beem98,chrusciel98} \cite[Theor.\ 2]{minguzzi14d}, in particular it is differentiable at the interior points of its generators. The tangent hyperplane at a differentiability point is lightlike and  contains the tangent to the generator, which is not proportional to $n$,  thus it is transverse to $n$. We conclude that $u$ is differentiable at $e_1$ if and only if $B$ is differentiable at $(e_1,u(e_1))$, namely if and only if the curve of point (i) is unique. At such points the 1-form $\dd (y-u)$ is lightlike thus, by Eq.\ (\ref{glg}), the Hamilton-Jacobi equation holds, while the equation $\p_q u=\p L/\p v$ follows from the fact that the kernel of $\dd (y-u)$ contains the tangent to the generator, as in the proof of Theorem \ref{lkp}.

Let us prove the semi-concavity. Let $K\subset \Omega$ be a compact neighborhood of $e_1$. There are constants $\delta>0$ and $R>0$ such that all the curves of point (i) which end at the points of $K$ are defined at least on an interval of duration $2\delta$ and have there  velocities bounded by $R$ (with respect to an auxiliary Riemannian metric).   Indeed, otherwise by the limit curve theorem \cite{minguzzi07c} we could find a sequence of generators of $B$ which converges to a past inextendible lightlike geodesic contained in $B$ and in a slice $N_t$, hence to an integral line of $n$, which is impossible since $u$ is finite on $\Omega$. Let $\bar{e}=(\bar t,\bar q)\in K$, let $q(t)$ be a curve as in (i) ending at $\bar e$, and let $e'=(\bar t-\delta, q(\bar t-\delta))$. Taking a smaller $\delta$ if necessary (depending only on $K$ and $R$) the E.-L.\ solutions issued from $e'$ with velocities close to $\dot q(\bar t-\delta)$ cover simply a neighborhood of $\bar e$, and the action calculated over them defines there a $C^r$ function $\mathcal{A}_{e'}$ whose second derivatives are bounded by a constant which depends only on $K$ and $R$ (this is the usual construction of a field of extremals, i.e.\ of a portion of the smooth null cone $E^{+}((e',u(e')))$ near its vertex). By Eqs.\ (\ref{zgw}) and (\ref{cal1})
\[
u(e)\le u(e')+S(e',e)\le u(e')+\mathcal{A}_{e'}(e), \qquad u(\bar e)=u(e')+\mathcal{A}_{e'}(\bar e).
\]
Thus at every point of $K$ the function $u$ admits a $C^2$ upper support function with second derivatives bounded independently of the point. This property is equivalent to the semi-concavity of $u$ on $K$, and implies that $u$ is Lipschitz on $K$ \cite[Prop.\ 1.1.3, Theor.\ 2.1.7]{cannarsa04} \cite[Prop.\ 3]{minguzzi14d}.
\end{proof}

\begin{remark}
The function $u$ will be shown to be a viscosity solution of the Hamilton-Jacobi equation on $\Omega$, see Theorem \ref{cau}.
Due to  Eq.\ (\ref{dax}), and to the calculation at the end of Sect.\ \ref{pqz}, the function $x\mapsto d(A,x)$ satisfies the timelike eikonal equation on $I^+(A)$ at the points where $u$ is differentiable. More precisely, for $x=(t,q,y)\in I^{+}(A)$, $t<t_{max}$,
\begin{equation} \label{eik2}
g(\nabla d(A,x),\!\nabla d(A,x))+1=\!\frac{2(t-t_0)^2}{d(A,x)^2}\,[\,\frac{\p u}{\p
t}+ H(t,q,\frac{\p u}{\p q})].
\end{equation}
\end{remark}

\section{Subsolutions and weak KAM solutions of the Hamilton-Jacobi equation}\label{s:subsolutions}
In this section we shall be interested in viscosity solutions
$u(q,t)$ of the Hamilton-Jacobi (HJ) equation
\[
\textrm{(HJ)} \qquad \qquad \p_t u+H(t,q,D_q u)=0, \]
where the Hamiltonian $H(t,q,p)=\max_{v\in TQ_q}(p[v]-L(t,q,v))$ is
the Legendre transform of the Lagrangian. For the general theory of
viscosity solutions to the equations of the form $F(e,u(e),D_e u)=0$
the reader is referred to \cite{bardi97}. The equation (HJ) has this
form. It must be remarked that in the viscosity solutions theory the
sign in (HJ) matters. Indeed, $u$ is a viscosity supersolution of
$F(e,u(e),D_e u)=0$ if and only if $v=-u$ is a viscosity subsolution of
$-F(e,-u(e),-D_e u)=0$, see \cite{bardi97} (thus, more rigorously
speaking, to give a viscosity solution problem means to give $F$
rather than the equation $F=0$.). Therefore with (-HJ) we denote the
equation
\[
\textrm{(-HJ)} \qquad \quad -\{\p_t u+H(t,q,D_q u)\}=0. \]


 Let $u:E\to \mathbb{R}$ be continuous at $e=(t,q)$,
we define\footnote{Here $\exp_q$ is the exponential map at $q$ on
the Riemannian space $(Q,h)$ where $h$ is any Riemannian metric, its
modulus being $\parallel v\parallel=\sqrt{h(v,v)}$. One can for
instance choose $h$ to be the canonical Euclidean metric associated
to a coordinate chart around $q$.}
\begin{align}
&D^{+}u(t,q):=\{(b,p)\in \mathbb{R}\times T^{*}_qQ:
\limsup_{(a,v)\to 0} \frac{u(t+a,\exp_q v)-u(t,q)-ba-p[v]}{\vert
a\vert+\parallel
v\parallel}\le 0 \}, \\
& D^{-}u(t,q):=\{(b,p)\in \mathbb{R}\times T^{*}_qQ:
\liminf_{(a,v)\to 0} \frac{u(t+a,\exp_q v)-u(t,q)-ba -p[v]}{\vert
a\vert+\parallel v\parallel} \ge 0\}.
\end{align}
It is perhaps worth recalling here that  $D^+u\cap D^-u \ne
\emptyset$ at $e\in E$ if and only if $u$ is differentiable at $e$ in which
case $D^+u(e)=D^-u(e)=\{Du(e)\}$ (see \cite[Lemma 1.8]{bardi97}).

A continuous function $u:\Omega\to \mathbb{R}$, $\Omega\subset E$,
is a viscosity subsolution of (HJ) if for all $(t,q)\in \Omega$ and
$(b,p)\in D^{+}u(t,q)$ we have
\[ b+H(t,q,p) \le 0. \]

Analogously, $u$ is a viscosity supersolution of (HJ) if for all
$(t,q)\in \Omega$ and  $(b,p)\in D^{-}u(t,q)$ we have
\[  b+H(t,q,p)\ge  0. \]
A continuous function which is both a subsolution and a
supersolution is a {\em viscosity solution} of (HJ).

There is an equivalent approach which uses test $C^1$ functions and
it is based on the observation that $(b,p)\in D^{+}u(t,q)$ if there
is a $C^1$ function $\varphi:E\to \mathbb{R}$  such that $u-\phi$
has  a local maximum at $(t,q)$ and $(b,p)=(\p_t \varphi,D_q
\varphi)$ (see \cite[Lemma 1.7]{bardi97}). Thus $u$ is a subsolution
of (HJ) at $(t,q)$ if and only if for every $C^1$ function $\varphi:E\to
\mathbb{R}$ such that $u-\varphi$ has a local maximum at $(t,q)$, we
have at $(t,q)$
\[
\p_t \varphi+H(t,q,D_q\varphi)\le 0.
\]
Similar results hold in the supersolution case.

An alternative way of defining non smooth solutions of (HJ) has been
proposed by A. Fathi through the concept of weak KAM solution
\cite[Def. 4.1.11]{fathi08}. Here we give anew the definition since
we do not assume that the Lagrangian is Tonelli nor that it is
independent of time. However, we need first to define the concept of
$u$-calibrated curve (contrary to what is done in weak KAM theorem
we assume $L$ as fixed, and the constant $c$ does not appear since
we are not in the autonomous case, see \cite[Def. 4.1.7]{fathi08})

\begin{definition}
Let $u:\Omega \to \mathbb{R}$ be a  function. A piecewise $C^1$
curve $e: I \to \Omega$, of the form  $e(t)=(t,q(t))$, defined on
the interval $I\subset \mathbb{R}$ is $u$-calibrated if for every
$t,t'\in I$, with $t\le t'$ we have
\[
u(e(t'))-u(e(t))=\int_t^{t'}L(s,q,\dot{q}) \, \dd s.
\]
\end{definition}

\begin{definition}
A continuous function $u:\Omega\to\mathbb{R}$, $\Omega\subset E$, is
a weak KAM solution of negative type (resp. positive) if Eq.
(\ref{zgw}) holds  (called domination) and for every $e=(t,q)\in
\Omega$ we can find an $u$-calibrated curve $\tilde{e}:
(t-\epsilon,t]\to E$ (resp. $\tilde{e}: [t,t+\epsilon)\to E$) with
$\tilde{e}(t)=e$.
\end{definition}

The relation of weak KAM solutions with viscosity solutions and in
particular with future sets will be clarified in a moment. A first
hint is given by the following which relates the concept of
$u$-calibrated curve with that of lightlike generator of the
boundary of some set (typically future).

\begin{proposition} \label{nze}
Let $u:\Omega \to \mathbb{R}$ be a  continuous function defined on an open set
$\Omega$. A piecewise $C^1$ curve $e: I \to \Omega$, of the form
$e(t)=(t,q(t))$, defined on the interval $I\subset \mathbb{R}$ is
$u$-calibrated if and only if the light lift of $e(t)$ (here $t_0\in I$)
\begin{equation} \label{knx}
(\,t,\,q(t),\, u(e(t_0))+\int_{t_0}^{t}L(s,q,\dot{q}) \, \dd s\,)
\end{equation}
stays entirely in the  boundary $\dot{U}$ where
\[ U=\{(e,y):e\in\Omega, \ u(e)\le y\} .\]
(Note that if $\Omega=E$ and $u$ satisfies Eq. (\ref{zgw}) then $U$
is a future set, and $\dot{U}$ is an achronal boundary).
\end{proposition}

\begin{proof}
Assume $e$ is $u$-calibrated. The point $x(t)=(e(t), u(e(t)))$
belongs to $\dot{U}$ by definition of $U$. Since $u(t)=
u(e(t_0))+\int_{t_0}^{t}L(s,q,\dot{q}) \, \dd s$ we get that the
light lift stays entirely in $\dot{U}$. Conversely, if the light
lift of $e(t)$ given by Eq. (\ref{knx}) belongs to $\dot{U}$ then
$u(e(t_0))+\int_{t_0}^{t}L(s,q,\dot{q})=u(t)$ which implies that
$e(t)$ is $u$-calibrated.
\end{proof}

The next result show the correspondence between weak KAM solutions
with special types of  achronal boundaries.

\begin{proposition} \label{nzf}
Let $u:\Omega\to\mathbb{R}$, $\Omega\subset E$, be a continuous function and let
\[
U=\overline{J^{+}(\{(e,u(e)):\, e \in\Omega\})}
\]
be the closed causal future of its graph. Assume that $U\cap \pi^{-1}(\Omega)=\{(e,y):\,e\in
\Omega,\ u(e)\le y\}$, as it happens if $\Omega=E$ and $u$ is dominated, namely satisfies Eq.\ (\ref{zgw}). Then $u$ is
a weak KAM solution of negative type (resp. positive) if and only if $U$ is such that
$\dot{U}\cap \pi^{-1}(\Omega)=[\dot{U}_N\cup\dot{U}_{+}]\cap
\pi^{-1}(\Omega)$ (resp. $\dot{U}\cap
\pi^{-1}(\Omega)=[\dot{U}_N\cup\dot{U}_{-}]\cap \pi^{-1}(\Omega)$).
\end{proposition}

\begin{proof}
By assumption and by the continuity of $u$, $\dot{U}\cap
\pi^{-1}(\Omega)=\{(e,u(e)):e \in\Omega\}$. (If $\Omega=E$ and $u$ is dominated, along a future directed causal curve from a point of the graph the coordinate $y$ increases at least by the action of the projection, which by domination is not smaller than the increase of $u$, so the curve stays in $\{y\ge u\}$, a closed set by continuity.)

$\Rightarrow$. Let $q=(e,u(e))$ be a
point of $\dot{U}\cap \pi^{-1}(\Omega)$, since $u$ is a weak KAM
solution of negative type there is a piecewise $C^1$ curve $e:
(t-\epsilon,t]\to E$ of the form $e(\tau)=(\tau,q(\tau))$ with
$e(t)=e$. By Prop. \ref{nze} the light lift
\[
(e(\tau),\, u(e(t))-\int_{\tau}^{t}L(s,q,\dot{q}) \, \dd s\,)
\]
stays entirely in $\{(e,u(e)):e \in\Omega\}\subset \dot{U}$. But
$\dot{U}$ is an achronal boundary thus the light lift is an achronal
lightlike geodesic which ends at $q$. Taking any point $p$ of the
light lift different from $q$ we have $p\in \dot{U}\cap
E^{-}(q)\backslash\{q\}$ and hence $q\in \dot{U}_N\cup\dot{U}_+$.

$\Leftarrow$. The validity of Eq.
(\ref{zgw}) on $\Omega$ follows from Prop.\ \ref{nts} and the fact
that $U$ is closed and future. It remains to prove that given $e\in
\Omega$ we can find a piecewise $C^1$ path $e:(t-\epsilon,t]\to E$
of the form $e(\tau)=(\tau,q(\tau))$ with $e(t)=e$ which is
$u$-calibrated. The point $q=(e,u(e))$ belongs to $\dot{U}$ and
hence to $\dot{U}_N\cup\dot{U}_+$ which implies that there is an
achronal lightlike geodesic $\gamma(t)=(t,q(t),y(t))$ on
$\dot{U}\cap \pi^{-1}(\Omega)=\{(e,u(e)):\,e\in \Omega\}$ ending at
$q$. By corollary \ref{vqz} its projection $(t,q(t))$ is
$u$-calibrated (note that $c=0$ in that corollary).

\end{proof}

\begin{remark}
The assumption on $U\cap\pi^{-1}(\Omega)$ cannot be dropped, since forming the closed causal future of the graph can discard the raised portions of the graph. Let $Q=S^1$ with its standard metric, $L=\frac{1}{2}\vert v\vert^2$, and $u(t,q)=\varphi(t)$, with $\varphi\ge 0$ a smooth non-vanishing bump function supported in $(-1,1)$. Every point with $y\ge 0$ is reached from a point $(s,q,0)$ of the graph with $s<-1$, while $y$ does not decrease along future directed causal curves for this Lagrangian, thus $U=\{y\ge 0\}$. Its boundary $\{y=0\}$ is generated by the complete lightlike lines $t\mapsto(t,q,0)$, so the condition on $\dot U$ holds, but $u$ is not a weak KAM solution, since domination fails along the constant curves wherever $\varphi$ increases.
\end{remark}

We have seen in Prop.\ \ref{nts} that  future sets $U$ are in
one-to-one correspondence with lower semi-continuous functions
$u:E\to [-\infty,+\infty]$ with the property given by Eq.
(\ref{nts}). The relation of these functions and hence of future
sets with viscosity subsolutions is given by the following results
which basically says that the boundary of a future set is described
by a viscosity subsolution. Algebrically, the idea of the proof is
the same as \cite[Prop. 7.2.7]{fathi08}.

\begin{theorem} \label{vis}
Let $u:E\to [-\infty,+\infty]$  be a lower semi-continuous function
as in Prop.\ \ref{nts} i.e. satisfying Eq.\ (\ref{zgw}) then, defined
$\Omega=(t_{min}[u],t_{max}[u]) \times Q$,  $u$ is a viscosity
subsolution of (HJ) on $\Omega_C\subset \Omega$ where $\Omega_C$ is
the subset of $\Omega$ where $u$ is continuous. Moreover, $u$ is a
viscosity supersolution of (-HJ) on $\Omega_C$.

More strongly, let $u:\Omega \to \mathbb{R}$, $\Omega\subset E$ open
subset, be a weak KAM solution of negative type (hence it is
continuous and it still satisfies Eq.\ (\ref{zgw})) then $u$ is a
viscosity solution of (HJ).
\end{theorem}

\begin{proof}
Let $e=(t,q)\in \Omega_C$. Since $t \in (t_{min}[u],t_{max}[u])$,
$u(t,q)$ is finite. Let $v\in T_qQ$ and let $e:[t-\epsilon,t] \to
E$, $e(\tau)=(\tau,q(\tau))$, be any path with tangent vector
$(1,v)\in\mathbb{R}\times T_qQ$
at $e(t)=e$ and with $\epsilon>0$ so small that
$t-\epsilon>t_{min}[u]$. Let $\varphi:E\to \mathbb{R}$ be a $C^1$
function such that $u-\varphi$ has a local maximum at $e$.

From Eq.\ (\ref{zgw}) we get, considering the path
$e\vert_{[\tau,t]}$, for sufficiently small $t-\tau$,
\[
\varphi(e(t))-\varphi(e(\tau)) \le u(e(t))-u(e(\tau))\le
\int_{\tau}^{t}L(t',q,\dot{q})\,\dd t'.
\]
Dividing by $t-\tau>0$
\[
\frac{\varphi(e(t))-\varphi(e(\tau))}{t-\tau}\le
\frac{1}{t-\tau}\int_{\tau}^{t}L(t',q,\dot{q})\,\dd t',
\]
and taking the limit $\tau\to t$
\[
\p_t \varphi+[D_q\varphi[v]- L(t,q,v)]\le 0.
\]
Since this inequality holds  for every $v\in T_qQ$ and
$H(t,q,p)=\max_{v\in TQ_q}(p[v]-L(t,q,v))$ we finally get
\[
\p_t \varphi+H(t,q,D_q\varphi)\le 0.
\]

Let $e=(t,q)\in \Omega_C$. Since $t \in (t_{min}[u],t_{max}[u])$,
$u(t,q)$ is finite. Let $v\in T_qQ$ and let $e:[t,t+\epsilon] \to
E$, $e(\tau)=(\tau,q(\tau))$, be any path with tangent vector
$(1,v)\in\mathbb{R}\times T_qQ$
at $e(t)=e$ and with $\epsilon>0$ so small that
$t+\epsilon<t_{max}[u]$. Let $\varphi:E\to \mathbb{R}$ be a $C^1$
function such that $u-\varphi$ has a local minimum at $e$.

From Eq.\ (\ref{zgw}) we get, considering the path
$e\vert_{[t,\tau]}$ for sufficiently small $\tau-t$,
\[
\varphi(e(\tau))-\varphi(e(t)) \le u(e(\tau))-u(e(t))\le
\int_{t}^{\tau}L(t',q,\dot{q})\,\dd t'.
\]
Dividing by $\tau-t>0$
\[
\frac{\varphi(e(\tau))-\varphi(e(t))}{\tau-t}\le
\frac{1}{\tau-t}\int_{t}^{\tau}L(t',q,\dot{q})\,\dd t',
\]
and taking the limit $\tau\to t$
\[
\p_t \varphi+[D_q\varphi[v]- L(t,q,v)]\le 0.
\]
Since this inequality holds  for every $v\in T_qQ$ and
$H(t,q,p)=\max_{v\in TQ_q}(p[v]-L(t,q,v))$ we finally get
\[
-\{\p_t \varphi+H(t,q,D_q\varphi)\}\ge 0.
\]

In order to prove the last statement we have only to prove that $u$
is a viscosity supersolution. Let $e=(t,q) \in \Omega$ and let
$e:[t-\epsilon,t] \to \Omega$, $e(\tau)=(\tau,q(\tau))$, be an
$u$-calibrated curve. Its tangent vector at $e(t)=e$ is
$(1,v)\in\mathbb{R}\times T_qQ$. Let $\psi:E\to \mathbb{R}$ be a
$C^1$ function  such that $u-\psi$ has a local minimum at $e$. From
Eq.\ (\ref{zgw}) we get, considering the path $e\vert_{[\tau,t]}$,
for sufficiently small $t-\tau$,
\[
\psi(e(t))-\psi(e(\tau)) \ge u(e(t))-u(e(\tau))=
\int_{\tau}^{t}L(t',q,\dot{q})\,\dd t'.
\]
Dividing by $t-\tau$
\[
\frac{\psi(e(t))-\psi(e(\tau))}{t-\tau}\ge
\frac{1}{t-\tau}\int_{\tau}^{t}L(t',q,\dot{q})\,\dd t',
\]
and taking the limit $\tau\to t$
\[
\p_t \psi+H(t,q,D_q\psi)\ge \p_t \psi+[D_q\psi[v]- L(t,q,v)]\ge 0.
\]

%

\end{proof}

The relation between the Cauchy problem for (HJ) and the causal
relation is clarified by the following result which in short states that
the viscosity solution of a Cauchy problem is obtained by slicing a
future set obtained  from the causal future of an initial slice
constructed from the initial conditions.

\begin{theorem} \label{cau}
Let $t_0\in \mathbb{R}$, let $u_{t_0}:Q\to \mathbb{R}$ be a
continuous function and let
\[
F=\overline{J^{+}(\{(t_0,q,u_{t_0}(q)): q\in Q\})} .
\]
Let $u:E\to [-\infty,+\infty]$  be the function defined by $F=\{(e,y): u(e)\le y\}$, so that $t_{min}[u]=t_0$, and let us suppose that $t_{max}:=t_{max}[u]>t_0$ (which is the case if $Q$ is compact, where $t_{max}=+\infty$). Then $u$ is continuous on $[t_0,t_{max})\times Q$, $u(t_0,\cdot)=u_{t_0}$, and on $\Omega=(t_0,t_{max})\times Q$ it  is a weak KAM solution of negative type, hence a
viscosity solution of (HJ), which is given by the Lax-Oleinik formula
$u(t,\cdot)=T^-_{t_0,t}u_{t_0}$.
\end{theorem}

\begin{proof}
The continuity and the Lax-Oleinik representation have been proved in Theorem \ref{lo1}. The function $u$ satisfies Eq.\ (\ref{zgw}) by Prop.\ \ref{nts}, and by Theorem \ref{lo2}, point (i), every $e\in \Omega$ is the final endpoint of a $u$-calibrated curve, thus $u$ is a weak KAM solution of negative type on $\Omega$. By Theorem \ref{vis} it is a viscosity solution of (HJ).
\end{proof}

We do not enter here into the problem of the uniqueness of the viscosity solution, for which we refer to \cite{crandall83,lions82,bardi97}.

\subsection{Global solutions, horizons and the timelike eikonal equation} \label{s:glob}

In this subsection we assume $T=\mathbb{R}$. By {\em global weak KAM solution} we mean a weak KAM solution of negative type defined on the whole classical spacetime, $\Omega=E$. Reversing the time orientation one obtains analogous results for the solutions of positive type.

\begin{proposition} \label{glob2}
Let $u:E\to \mathbb{R}$ be a global weak KAM solution. Then
\begin{itemize}
\item[(a)] The graph $W=\{(e,u(e)): e\in E\}\subset M$ is the achronal boundary of the closed future set $\{(e,y): u(e)\le y\}$, and every point of $W$ is the future endpoint of a past inextendible  lightlike geodesic contained in $W$ and transverse to $n$. Thus $W$ is a (future) horizon whose generators are the light lifts of the maximal $u$-calibrated curves.
\item[(b)] The function $u$ is locally Lipschitz and locally semi-concave. It is differentiable at $e$ if and only if there is just one maximal $u$-calibrated curve ending at $e$, and at the differentiability points the Hamilton-Jacobi equation holds.
\item[(c)] If the E.-L.\ flow is complete the maximal $u$-calibrated curve ending at $(t,q)$ is defined on $(-\infty,t]$, and for every $s<t$, $u(t,\cdot)=T^-_{s,t}u(s,\cdot)$, namely $u$ is invariant under the Lax-Oleinik evolution.
\end{itemize}
\end{proposition}

\begin{proof}
(a). The set $U=\{(e,y): u(e)\le y\}$ is a closed future set by Prop.\ \ref{nts}, and by Prop.\ \ref{nzf} $\dot{U}=\dot{U}_N\cup \dot{U}_+$. Let $x\in W=\dot{U}$. By Lemma \ref{psh} there is an achronal lightlike geodesic contained in $W$ and ending at $x$ which cannot be extended to the past without escaping $W$. If it had a past endpoint this endpoint would belong to $\dot{U}_-$, which is empty. Thus the geodesic is past inextendible. It is not an integral line of $n$ since $W$ is a graph.

(b). The arguments in the proof of point (v) of Theorem \ref{lo2} use only  the horizon property of the graph and the finiteness of $u$, thus they apply to the present case.

(c). By Theorem \ref{kdf} the projection of a generator of $W$ is a maximal solution to the E.-L.\ equations, thus it is defined on  $(-\infty,t]$. Let $s<t$ and $e=(t,q)$. By Eq.\ (\ref{zgw}), $u(e)\le \inf_{e', t(e')=s}[u(e')+S(e',e)]=(T^-_{s,t}u(s,\cdot))(q)$, and the equality holds because the calibrated curve ending at $e$ passes through a point $e'$ of the slice $Q_s$, for which $u(e)-u(e')=\int_s^t L\,\dd t\ge S(e',e)$.
\end{proof}

We are going to prove that global solutions exist whenever $Q$ is compact, with no assumption on the time dependence of the Lagrangian. We need  a simple consequence of the results of the previous sections.

\begin{lemma} \label{sco}
If $Q$ is compact then $S$ is finite and continuous on the open set $\{(e_0,e_1)\in E\times E: t_0<t_1\}$, where the infimum which defines $S$ is attained. Moreover, the E.-L.\ flow is complete.
\end{lemma}

\begin{proof}
The action functional is coercive since $Q$ is compact, thus $(M,g)$ is globally hyperbolic (Theor.\ \ref{glo}) and hence causally simple, $S$ is finite for $t_0<t_1$ (Prop.\ \ref{cor2}), lower semi-continuous (Prop.\ \ref{bfj}) and  the infimum is attained (Theor.\ \ref{cga2}). The upper semi-continuity is proved in  Cor.\ \ref{hgf}, and the completeness in Theorem \ref{jhy}.
\end{proof}

\begin{theorem} \label{glob}
If $Q$ is compact there is a global weak KAM solution $u:E\to \mathbb{R}$. In particular, the Hamilton-Jacobi equation admits a  viscosity solution defined on the whole classical spacetime $E=\mathbb{R}\times Q$, and $(M,g)$ admits a horizon which intersects exactly once every integral line of the parallel null vector $n$.
\end{theorem}

The solution is obtained as a limit of solutions to Cauchy problems with initial data assigned on slices which recede to the infinite past. It can be regarded as a Busemann function (horofunction) for the classical spacetime. In Sect.\ \ref{s:per} we shall see that in the periodic case the global solution can be chosen to be compatible with the periodicity.

\begin{proof}
Let $q_*\in Q$. For every integer $n\ge 1$ let $u_n:[-n,+\infty)\times Q\to \mathbb{R}$ be given by $u_n(t,\cdot)=T^-_{-n,t} k_n$, where $k_n$ is the constant initial condition chosen so that $u_n(0,q_*)=0$ (it exists by Theorem \ref{lo2}, point (iii)). By Theorem \ref{cau} $u_n$ is a continuous weak KAM solution on $(-n,+\infty)\times Q$.

Let $m\ge 1$ be an integer, let $K_m=[-m,m]\times Q$ and let $C_m=Q_{-m-1}\times K_m$. By Lemma \ref{sco} $S$ is continuous and hence bounded and uniformly continuous on the compact set $C_m$. For $n>m+1$ the evolution property of Theorem \ref{lo2} gives for every $e\in K_m$
\[
u_n(e)=\min_{e', \, t(e')=-m-1} [u_n(e')+S(e',e)].
\]
This expression shows that the functions $u_n$, $n>m+1$, share on $K_m$ the  modulus of continuity of the family $\{S(e',\cdot)\vert_{K_m}, e'\in Q_{-m-1}\}$, and that their oscillation on $K_m$ is bounded by $\max_{C_m} S-\min_{C_m} S$. Since $u_n(0,q_*)=0$ the functions $u_n$, $n>m+1$, are equicontinuous and equibounded on $K_m$. By the Ascoli-Arzel\`a theorem and a diagonal argument we can pass to a subsequence, denoted in the same way, which converges uniformly on every compact subset of $E$ to a continuous function $u:E\to \mathbb{R}$.

By Prop.\ \ref{nts} each $u_n$ satisfies Eq.\ (\ref{zgw}), thus so does $u$. It remains to find the calibrated curves. Let $e_1\in E$. For sufficiently large $n$ there is, by Theorem \ref{lo2} point (i) and by the completeness of the E.-L.\ flow, a $u_n$-calibrated curve $q_n:[-n,t_1]\to Q$ ending at $e_1$, whose light lift $\gamma_n$ ending at $(e_1,u_n(e_1))$ is contained in the graph of $u_n$. Since the starting points of the lightlike geodesics $\gamma_n$ escape every compact set, by the limit curve theorem \cite{minguzzi07c} there is a past inextendible continuous causal curve $\gamma$ ending at $(e_1,u(e_1))$ to which a subsequence of $\gamma_n$ converges uniformly on compact subsets. By the locally uniform convergence of $u_n$ to $u$, $\gamma$ is contained in the graph of $u$, which is the boundary of the future set $\{(e,y):u(e)\le y\}$. Thus $\gamma$ is achronal and hence a lightlike geodesic,  it is not an integral line of $n$, and by Theorem \ref{kdf} it is the light lift of its projection $(t,q(t))$. As $\gamma$ is contained in the graph of $u$, the curve $(t,q(t))$ is $u$-calibrated.
\end{proof}

The global solutions of the Hamilton-Jacobi equation are related to the global solutions of the {\em timelike} eikonal equation on $(M,g)$. We have already observed that  the Lorentzian distance from a point or from a slice satisfies the timelike eikonal equation on a chronological future. The next result provides solutions defined on the whole spacetime.

\begin{proposition} \label{eik}
Let $\Omega\subset E$ be an open set, let $u:\Omega\to \mathbb{R}$ be a continuous function and let us define  $\tau: \pi^{-1}(\Omega)\to \mathbb{R}$ by
\begin{equation}
\tau(t,q,y)=y-u(t,q)+\frac{t}{2}\, .
\end{equation}
Then
\begin{itemize}
\item[(a)] $\tau$ is differentiable at $x=(e,y)$ if and only if $u$ is differentiable at $e$, and in this case
\begin{equation} \label{eik3}
g(\nabla \tau,\nabla \tau)+1=2\,[\,\frac{\p u}{\p
t}+ H(t,q,\frac{\p u}{\p q})].
\end{equation}
\item[(b)] $u$ is a viscosity solution of (HJ) on $\Omega$ if and only if $\tau$ is a viscosity solution of the timelike eikonal equation $-\{g(\nabla \tau,\nabla \tau)+1\}=0$ on $\pi^{-1}(\Omega)$.
\item[(c)] If $\Omega=E$ and $u$ satisfies Eq.\ (\ref{zgw}) then $\tau$ is a time function on $(M,g)$.
\end{itemize}
In particular, if $u$ is a global weak KAM solution then $\tau$ is a locally Lipschitz time function which is a viscosity solution of the timelike eikonal equation on the whole spacetime $M$, and whose gradient is almost everywhere a past-directed unit timelike vector.
\end{proposition}

\begin{proof}
(a). The first statement is clear. For every 1-form $\omega=\omega_y \dd y+\omega_q+\omega_t \dd t$, $\omega_q\in T^*Q$, the expression of the contravariant metric given in Sect.\ \ref{ndx} leads to
\begin{equation} \label{ome}
g^{-1}(\omega,\omega)=a_t^{-1}(\omega_q+\omega_y b_t,\omega_q+\omega_y b_t)-2\omega_y\omega_t+2V\omega_y^2 .
\end{equation}
For $\omega=\dd \tau$, that is $\omega_y=1$, $\omega_q=-\p_q u$, $\omega_t=-\p_t u+\frac{1}{2}$, we obtain Eq.\ (\ref{eik3}).

(b). Let $\varphi$ be a $C^1$ function defined in a neighborhood of $\bar x=(\bar e,\bar y)$ and such that $\tau-\varphi$ has a local maximum (minimum) at $\bar x$. Since $\tau$ is affine with unit slope on the fiber of $\bar e$, we have $\p_y\varphi(\bar x)=1$. The function $\psi(t,q)=\bar y+\frac{t}{2}-\varphi(t,q,\bar y)$ is such that $u-\psi$ has a local minimum (resp.\ maximum) at $\bar e$ and, by Eq.\ (\ref{ome}), $g^{-1}(\dd \varphi,\dd \varphi)+1=2[\p_t\psi+H(t,q,\p_q\psi)]$ at $\bar x$. Conversely, if $\psi$ is a $C^1$ function such that $u-\psi$ has a local minimum (maximum) at $\bar e$ then, for any $\bar y$, the function $\varphi=y+\frac{t}{2}-\psi(t,q)$ is such that $\tau-\varphi$ has a local maximum (resp.\ minimum) at $(\bar e,\bar y)$. Thus $u$ is a viscosity supersolution (subsolution) of (HJ) if and only if $\tau$ is a viscosity subsolution (resp.\ supersolution) of $-\{g^{-1}(\dd \tau,\dd \tau)+1\}=0$.

(c). Let $x_0\le x_1$, $x_0\ne x_1$. By Eq.\ (\ref{cgf2}) we have $t_0\le t_1$ and $y_1-y_0\ge S(e_0,e_1)$. If $t_0<t_1$ Eq.\ (\ref{zgw}) gives $\tau(x_1)-\tau(x_0)\ge S(e_0,e_1)-[u(e_1)-u(e_0)]+\frac{1}{2}(t_1-t_0)>0$. If $t_0=t_1$ then $e_0=e_1$ and $y_1>y_0$, thus again $\tau(x_1)>\tau(x_0)$.

The last statement follows from (a)-(c), Prop.\ \ref{glob2} and Theorem \ref{vis}. The gradient is past-directed because $\dd \tau(n)=1>0$.
\end{proof}

\begin{corollary} \label{eikc}
If $Q$ is compact, then $(M,g)$ admits a locally Lipschitz time function $\tau$ which, on the whole of $M$, is a viscosity solution of the timelike eikonal equation \[-\{g(\nabla \tau,\nabla\tau)+1\}=0,\] and whose gradient is almost everywhere a past-directed unit timelike vector. It can be chosen of the form $\tau=y-u+\frac{t}{2}$, where $u$ is a global weak KAM solution, so that its level sets are the graphs of the functions $u-\frac{t}{2}+const$.
\end{corollary}

\begin{proof}
By Theorem \ref{glob} there is a global weak KAM solution $u:E\to\mathbb{R}$, and the claim follows from Prop.\ \ref{eik}.
\end{proof}

\begin{remark}
Viscosity solutions of the Lorentzian eikonal equation have been studied by Cui and Jin, who showed that the negative of the regular cosmological time is a viscosity solution \cite{cui14}, and by Jin and Cui on class A Lorentzian 2-tori \cite{jin18}. The existence of global viscosity solutions on general globally hyperbolic spacetimes was claimed in \cite{zhu25}, but Galloway, McCann and Ohanyan \cite{galloway26} have constructed examples which contradict a key lemma of that work. Corollary \ref{eikc} shows that for the spacetimes considered in this work the existence problem is answered by the classical Lax-Oleinik formula.
\end{remark}

\section{Time-periodic systems: critical value and compact horizons} \label{s:per}

In this section we study the mechanical systems which depend periodically on time (the autonomous systems being a special case), and we show that the basic objects of weak KAM theory and Aubry-Mather theory, namely Ma\~n\'e's critical value, the weak KAM solutions, Ma\~n\'e's potential and the Aubry set \cite{mather91,mane97,contreras99,fathi08}, have a neat interpretation in terms of the causality of the spacetimes obtained from $(M,g)$ by compactification of the time direction. In the opposite direction, an Aubry-Mather theory for Lorentzian manifolds, with the Lorentzian length in place of the action, has been developed by Suhr \cite{suhr19}. The time-periodic weak KAM solutions will turn out to be compact horizons. This result will be used in the next section to study their regularity. We stress that no curvature condition is used in this section, nor in Sections \ref{s:future_sets} and \ref{s:subsolutions}: the existence of the global and of the periodic weak KAM solutions, and hence of the corresponding horizons, holds for every Lagrangian of the class considered, while the null convergence condition enters only in Section \ref{s:ncc}, where it is used to prove their smoothness.

Throughout this section, $T=\mathbb{R}$, $Q$ is compact and connected, and there is a constant $P>0$ (the period) such that for every $t\in \mathbb{R}$
\[
a_{t+P}=a_t, \qquad b_{t+P}=b_t, \qquad V(t+P,q)=V(t,q).
\]
For $e=(t,q)\in E$ and $k\in \mathbb{Z}$ we write $e+kP:=(t+kP,q)$. Clearly, the action functional and the least action are invariant under time translations by multiples of the period, in particular
\begin{equation}
S(e_0+kP,e_1+kP)=S(e_0,e_1).
\end{equation}
We recall that since $Q$ is compact $(M,g)$ is globally hyperbolic, the E.-L.\ flow is complete, and $S$ is finite and continuous for $t_0<t_1$, where the infimum of the action is attained (Lemma \ref{sco}).

\subsection{The compactified spacetimes}

For every $c\in \mathbb{R}$ the map
\[
\sigma_c: M\to M, \qquad \sigma_c(t,q,y)=(t+P,q,y-cP),
\]
is an isometry of $(M,g)$ which preserves the time orientation and the parallel null vector $n=\p_y$. The group $\{\sigma_c^k, k\in \mathbb{Z}\}$ acts freely and properly discontinuously on $M$, thus the quotient
\[
M_c:=M/\langle\sigma_c\rangle
\]
endowed with the induced metric (still denoted $g$) and time orientation is a spacetime, and the projection $\pi_c:M\to M_c$ is a Lorentzian covering. The vector field $n$ projects to a parallel null vector field on $M_c$,  and the 1-form $\dd t=-g(\cdot,n)$ projects to a closed but non-exact 1-form. The space of orbits of $n$ on $M_c$ is the compact manifold $E_P:=E/\!\!\sim\,=(\mathbb{R}/P\mathbb{Z})\times Q$, where $e\sim e+P$.  Thus $M_c$ is a spacetime admitting a parallel null vector, which is  locally (but not globally) of the form considered in the previous sections.

The dependence of $M_c$ on the constant $c$ is clarified introducing the function $\tilde y=y+ct$, which is invariant under $\sigma_c$. The map $(t,q,y)\mapsto ([t],q,\tilde y)$ induces a diffeomorphism between $M_c$ and $(\mathbb{R}/P\mathbb{Z})\times Q\times \mathbb{R}$, and the metric reads
\begin{equation} \label{gc}
g=a_t \!-\!\dd t \otimes (\dd \tilde y-\!b_t) -\!(\dd \tilde y-\!b_t) \otimes \dd
t-2(V-c)\, \dd t ^2 .
\end{equation}
In other words, $M_c$ is obtained from the Eisenhart spacetime of the Lagrangian $L+c$ through the identification of the slices  $\tilde N_0$ and $\tilde N_P$, $(0,q,\tilde y)\sim (P,q,\tilde y)$. The manifold does not depend on $c$ while the metrics for the constants $c<c'$ differ by the term $2(c'-c)\,\dd t^2$, which is positive semi-definite. Thus {\em the constant $c$ controls the opening of the light cones on $(\mathbb{R}/P\mathbb{Z})\times Q\times \mathbb{R}$: the larger $c$ the narrower the cones}. Since for $c=0$ the curves $t\mapsto ([t],q,\tilde y)$ are closed and their causal character is determined by the sign of $V$, it is reasonable to expect that  $M_c$ has closed timelike curves for small $c$ and good causality properties for large $c$. We are going to show that there is exactly one value of $c$ which separates the two behaviors, and that it coincides with Ma\~n\'e's critical value.

The causal relations of $M_c$ are easily obtained from those of $(M,g)$ and hence from the least action.

\begin{lemma} \label{per1}
Let $x_0=(e_0,y_0)$ and $x_1=(e_1,y_1)$ be points of $M$. Then
\begin{itemize}
\item[(a)] $\pi_c(x_1)\in I^{+}(\pi_c(x_0))$ if and only if there is $k\in\mathbb{Z}$ such that $t_1+kP>t_0$ and
\[y_1-y_0>S(e_0,e_1+kP)+ckP.\]
\item[(b)] $\pi_c(x_1)\in J^{+}(\pi_c(x_0))$ if and only if there is $k\in\mathbb{Z}$ such that
$y_1-y_0\ge S(e_0,e_1+kP)+ckP$.
\item[(c)] $M_c$ is non-chronological (non-causal) if and only if there are $e\in E$ and an integer $k\ge 1$ such that $S(e,e+kP)+ckP<0$ (resp.\ $\le 0$).
\end{itemize}
\end{lemma}

\begin{proof}
Every causal (timelike) curve of $M_c$ lifts to a causal (resp.\ timelike) curve of $M$, and conversely. Thus $\pi_c(x_1)\in I^{+}(\pi_c(x_0))$ if and only if there is some $k\in \mathbb{Z}$ such that $\sigma_c^k(x_1)=(e_1+kP,y_1-ckP)$ belongs to $I^{+}(x_0)$, and analogously for the causal relation. Statement (a) follows from Eq.\ (\ref{cgf1}). As for (b), since $(M,g)$ is globally hyperbolic and hence causally simple, and $S$ is lower semi-continuous (Lemma \ref{sco}), by Eqs.\ (\ref{cgf2}) and (\ref{prf1}), $J^{+}(x_0)=\{x_1: y_1-y_0\ge S(e_0,e_1)\}$. Finally, a closed timelike (causal) curve of $M_c$ passing through $\pi_c(x_0)$ lifts to a timelike (resp.\ causal) curve which connects $x_0$ to $\sigma_c^k(x_0)$ for some $k\in \mathbb{Z}$, where $k\ge 0$ since $t$ is a semi-time function, and $k\ne 0$ since $(M,g)$ is causal. Thus (c) follows from (a) and (b).
\end{proof}

\begin{remark} \label{mpo}
Point (a) can be rephrased as follows: $\pi_c(x_1)\in I^{+}(\pi_c(x_0))$ if and only if $y_1-y_0>\Phi_c(e_0,e_1)$ where
\[
\Phi_c(e_0,e_1):=\inf_{k\in \mathbb{Z}:\, t_1+kP>t_0}\, [S(e_0,e_1+kP)+ckP]
\]
differs by the endpoint term $c(t_1-t_0)$ from the infimum of the action of the Lagrangian $L+c$ over the curves which connect the projections of $e_0$ and $e_1$ on the compact classical spacetime $E_P$. The latter, $\tilde\Phi_c(e_0,e_1)=\Phi_c(e_0,e_1)+c(t_1-t_0)$, is Ma\~n\'e's {\em action potential} \cite{mane97,contreras99,contreras13}, and in the coordinate $\tilde y=y+ct$ the condition reads $\tilde y_1-\tilde y_0>\tilde \Phi_c(e_0,e_1)$. The two functions coincide on the diagonal, $e_0=e_1$, so the statements involving closed curves or the Aubry set are not affected by the choice. Therefore, Ma\~n\'e's potential plays for the compactified spacetime $M_c$ the same role played by the least action $S$ for $(M,g)$ (cf.\ Lemma \ref{pcf}).
\end{remark}

\subsection{The critical value and the periodic weak KAM solutions}

\begin{definition}
Let $c\in \mathbb{R}$. A function $u:E\to \mathbb{R}$ is {\em $c$-equivariant} if for every $e\in E$
\[
u(e+P)=u(e)-cP,
\]
namely if the function $(t,q)\mapsto u(t,q)+ct$ has period $P$ or, equivalently, if the graph of $u$ is invariant under $\sigma_c$. A {\em periodic weak KAM solution at level $c$} is a $c$-equivariant global weak KAM solution.
\end{definition}

If $u$ is a periodic weak KAM solution at level $c$ then $\tilde u=u+ct$ is a  viscosity solution of period $P$ of the equation $\p_t \tilde u+H(t,q,\p_q \tilde u)=c$, which is the form in which the Hamilton-Jacobi equation is usually written in weak KAM theory.
The projection on $M_c$ of the graph of a $c$-equivariant function is a compact set homeomorphic to $E_P$.

The next result contains the weak KAM theorem for time-periodic Lagrangians. The existence of the solution is obtained through Fathi's fixed point argument \cite{fathi08} (see \cite{contreras13} for the time-periodic case), which here is applied to the map that sends an achronal graph on the slice $N_0$ to the intersection of the boundary of its causal future with the slice $N_P$.

\begin{theorem} \label{crit}
Let $Q$ be compact and let the Lagrangian be periodic in time with period $P$.
\begin{itemize}
\item[(a)] There is one and only one constant $c\in \mathbb{R}$, denoted $c(L)$, for which there exists a periodic weak KAM solution at level $c$.
\item[(b)] The constant $c(L)$ is given by
\begin{equation} \label{crv}
c(L)=\sup_{e\in E,\, k\ge 1} \frac{-S(e,e+kP)}{kP}\, ,
\end{equation}
that is, $-c(L)$ is the infimum of the average action of the closed curves of the compact classical spacetime $E_P$.
\item[(c)] There is a $c$-equivariant function which satisfies Eq.\ (\ref{zgw}) if and only if $c\ge c(L)$.
\end{itemize}
\end{theorem}

The constant $c(L)$ is  Ma\~n\'e's critical value of the Lagrangian. It coincides with the value $\alpha(0)$ of Mather's $\alpha$ function, namely with the opposite of the minimum of the average action over the invariant probability measures of the E.-L.\ flow \cite{mather91,mane97,contreras98,contreras99,bernard02,contreras13}. By point (c) and Prop.\ \ref{nts} it is also the smallest constant $c$ for which there is a $\sigma_c$-invariant closed future set of $(M,g)$ bounded by a graph over $E$, namely for which $M_c$ contains a compact achronal boundary which is intersected exactly once by every orbit of $n$.

\begin{proof}
Let $C(Q)$ be the space of continuous real functions on $Q$ endowed with the supremum norm, and let $\mathcal{T}:=T^{-}_{0,P}: C(Q)\to C(Q)$,
\[
(\mathcal{T}w)(q)=\min_{q_0\in Q} \,[w(q_0)+S((0,q_0),(P,q))],
\]
which is well defined by Theorem \ref{lo1}. By Lemma \ref{sco} the function $S$ is uniformly continuous on the compact set $Q_0\times Q_P$. Let $\rho$ be a distance on $Q$ which induces its topology and let
\[
\omega(r)=\sup \{\vert S((0,q_0),(P,q))-S((0,q_0),(P,q'))\vert: q_0,q,q'\in Q, \ \rho(q,q')\le r\},
\]
so that $\omega$ is bounded and $\omega(r)\to 0$ for $r\to 0$. For every $w\in C(Q)$ the function $\mathcal{T}w$ belongs to
\[
\mathcal{K}=\{f\in C(Q): \vert f(q)-f(q')\vert \le \omega(\rho(q,q')) \ \textrm{for every } q,q'\in Q\}.
\]
Let $X=C(Q)/\mathbb{R}$ be the Banach space of the continuous functions modulo the additive constants (the norm being half the oscillation). By Theorem \ref{lo2}, points (ii) and (iii), $\mathcal{T}$ commutes with the addition of constants and is non-expansive, thus it induces a continuous map $\hat{\mathcal{T}}: X\to X$. The image $\hat{\mathcal{K}}$ of $\mathcal{K}$ in $X$ is convex, closed and, by the Ascoli-Arzel\`a theorem, compact, and $\hat{\mathcal{T}}(\hat{\mathcal{K}})\subset \hat{\mathcal{K}}$. By Schauder's fixed point theorem $\hat{\mathcal{T}}$ has a fixed point, that is, there are $w\in C(Q)$ and a constant $c\in \mathbb{R}$ such that
\begin{equation} \label{fix}
T^-_{0,P}w=w-cP.
\end{equation}
Let $u(t,\cdot)=T^-_{0,t}w$ for $t\ge 0$. By the periodicity of the Lagrangian $T^-_{P,t+P}=T^-_{0,t}$, thus by the evolution property and Eq.\ (\ref{fix}), $u(t+P,\cdot)=T^-_{P,t+P}T^-_{0,P}w=T^-_{0,t}(w-cP)=u(t,\cdot)-cP$. We extend $u$ to the region $t<0$ in the only way which makes it $c$-equivariant. By construction, for every integer $k\ge 0$ we have $u(t,\cdot)=T^-_{-kP,t}\,u(-kP,\cdot)$ for $t\ge -kP$, thus by Theorem \ref{cau} $u$ is continuous and it is a weak KAM solution of negative type on $(-kP,+\infty)\times Q$. As $k$ is arbitrary, $u$ is a global weak KAM solution, and hence a periodic weak KAM solution at level $c$.

Let us prove the uniqueness of $c$ and Eq.\ (\ref{crv}). Let us denote with $c_*$ the right-hand side of Eq.\ (\ref{crv}). Let $u$ be a $c$-equivariant function which satisfies Eq.\ (\ref{zgw}). For every $e\in E$ and $k\ge 1$, $-ckP=u(e+kP)-u(e)\le S(e,e+kP)$, thus $c\ge c_*$. In particular, $c_*<+\infty$.

Conversely, let $u$ be a periodic weak KAM solution at level $c$, and let $\tilde u=u+ct$. Let $D$ be the oscillation of the continuous periodic function $\tilde u$, and let $D'$ be the maximum of the continuous function $(e,q')\mapsto S(e,(t(e)+P,q'))$, which is finite since by periodicity we can restrict $e$ to the compact set $[0,P]\times Q$. Let $q: (-\infty,t_1]\to Q$ be a $u$-calibrated curve ending at $e_1=(t_1,q_1)$ (Prop.\ \ref{glob2}), let $k\ge 1$ and let $e'=(t_1-kP, q')$ where $q'=q(t_1-kP)$. By calibration $S(e',e_1)\le u(e_1)-u(e')=\tilde u(e_1)-\tilde u(e')-ckP\le D-ckP$, thus by the triangle inequality
\[
S(e',e'+(k+1)P)\le S(e',e_1)+S(e_1,(t_1+P,q'))\le D+D'-ckP.
\]
As a consequence, $c_*\ge  -S(e',e'+(k+1)P)/[(k+1)P]\ge c-(D+D'+\vert c\vert P)/[(k+1)P]$, and taking the limit $k\to +\infty$ we obtain $c_*\ge c$. We conclude that $c=c_*$, which proves (a), (b) and the `only if' part of (c).

Finally, let $u_0$ be a periodic weak KAM solution at level $c(L)$ and let $c\ge c(L)$. The function $u(t,q)=u_0(t,q)-(c-c(L))\,t$ is $c$-equivariant and for $t_0<t_1$
\begin{equation} \label{str}
u(e_1)-u(e_0)\le S(e_0,e_1)-(c-c(L))(t_1-t_0),
\end{equation}
which concludes the proof of (c).
\end{proof}

\subsection{The causal ladder of the compactified spacetimes}

The next result shows that for the compactified spacetimes the causal ladder reduces to just three levels, which are selected by the comparison of the constant $c$ with the critical value.

\begin{theorem} \label{cau2}
Let $Q$ be compact and let the Lagrangian be periodic in time with period $P$.
\begin{itemize}
\item[(i)] If $c<c(L)$ then $M_c$ is totally vicious, namely through every point of $M_c$ there passes a closed timelike curve.
\item[(ii)] If $c=c(L)$ then $M_c$ is chronological but not strongly causal. It is causal if and only if there are no $e\in E$ and $k\ge 1$ such that $S(e,e+kP)+c(L)kP=0$.
\item[(iii)] If $c>c(L)$ then $M_c$ is globally hyperbolic and admits a compact Cauchy hypersurface homeomorphic to $E_P=(\mathbb{R}/P\mathbb{Z})\times Q$.
\end{itemize}
In particular, $M_c$ is chronological if and only if $c\ge c(L)$, and $M_c$ is strongly causal if and only if it is globally hyperbolic, if and only if $c>c(L)$.
\end{theorem}

\begin{proof}
(i). By Eq.\ (\ref{crv}) there are $e_*\in E$ and $k_*\ge 1$ such that $S(e_*,e_*+k_*P)+ck_*P=-\delta<0$. Let $e\in E$, let $m$ be an integer such that $t(e_*)+mP>t(e)$ and let $j$ be an integer such that $t(e)+jP>t(e_*)$, so that $\alpha:=S(e,e_*+mP)$ and $\beta:=S(e_*,e+jP)$ are finite. By the triangle inequality and the periodicity of $S$, for every integer $N\ge 1$
\[
S(e,e+(m+Nk_*+j)P)\le \alpha+N S(e_*,e_*+k_*P)+\beta ,
\]
thus $S(e,e+(m+Nk_*+j)P)+c(m+Nk_*+j)P\le \alpha+\beta+c(m+j)P-N\delta$, which is negative for sufficiently large $N$. By Lemma \ref{per1} there is a closed timelike curve passing through every point of the fiber of $e$.

(ii). Let $c=c(L)$. By Eq.\ (\ref{crv}) for every $e\in E$ and $k\ge 1$ we have $S(e,e+kP)+ckP\ge 0$, thus by Lemma \ref{per1} $M_c$ is chronological, and the last statement of (ii) holds. Let $u$ be a periodic weak KAM solution at level $c$, let $\tilde u=u+ct$, and let $q:(-\infty,t_1]\to Q$ be a $u$-calibrated curve. Since $Q$ is compact there are a point $q_\infty\in Q$ and an increasing sequence of positive integers $k_j$ such that $q'_j:=q(t_1-k_jP)\to q_\infty$. Let $i<j$. By calibration and by the periodicity of $\tilde u$ the restriction of $q$ to the interval $[t_1-k_jP,t_1-k_iP]$ has an action which satisfies
\[
\epsilon_{ij}:=\int_{t_1-k_jP}^{t_1-k_iP} L\,\dd t+c(k_j-k_i)P=\tilde u(t_1,q'_i)-\tilde u(t_1,q'_j),
\]
thus $\epsilon_{ij}\to 0$ for $i,j\to +\infty$. The light lift of this restriction starting from $(t_1-k_jP,q'_j,ck_jP)=\sigma_c^{-k_j}((t_1,q_j',0))$ ends at the point $(t_1-k_iP,q'_i,ck_iP+\epsilon_{ij})=\sigma_c^{-k_i}((t_1,q_i',\epsilon_{ij}))$. Its projection on $M_c$ is a causal curve which connects $\pi_c((t_1,q'_j,0))$ to $\pi_c((t_1,q'_i,\epsilon_{ij}))$ and along which the function $[t]\in \mathbb{R}/P\mathbb{Z}$ makes at least one complete turn. Since for $i,j\to +\infty$ both endpoints converge to $z=\pi_c((t_1,q_\infty,0))$, there are causal curves which start and end arbitrarily close to $z$ and which are not contained in a fixed small neighborhood of $z$. Thus $M_c$ is not strongly causal at $z$.

(iii). Let $c>c(L)$, and let $u$ be the $c$-equivariant function which satisfies Eq.\ (\ref{str}). The function $f:M\to \mathbb{R}$, $f(t,q,y)=y-u(t,q)$, is continuous and invariant under $\sigma_c$, thus it defines a continuous function $\bar f: M_c\to \mathbb{R}$. Let $x_0\le x_1$ with $x_0\ne x_1$. By Eq.\ (\ref{cgf2}) and Eq.\ (\ref{str}) either $t_0<t_1$ and
\begin{equation} \label{stf}
f(x_1)-f(x_0)\ge S(e_0,e_1)-[u(e_1)-u(e_0)]\ge (c-c(L))(t_1-t_0)>0,
\end{equation}
or $e_0=e_1$ and $y_0<y_1$, thus in any case $f(x_0)<f(x_1)$. Since the causal curves of $M_c$ lift to causal curves of $M$, $\bar f$ is a time function on $M_c$. Its level set $\Sigma=\bar f^{-1}(0)$ is the projection of the graph of $u$, hence it is a compact topological hypersurface homeomorphic to $E_P$. Let $\bar \gamma$ be an inextendible causal curve of $M_c$ and let $\gamma$ be a lift, which is an inextendible causal curve of $M$. Since $f$ increases over $\gamma$ it has a limit $f_+\in (-\infty,+\infty]$ in the future direction. If $t$ is unbounded from above on $\gamma$ then $f_+=+\infty$ by Eq.\ (\ref{stf}). If $t$ is bounded from above on $\gamma$ and $f_+<+\infty$, then the coordinate $y=f+u$ is bounded on the future of some point of $\gamma$, thus $\gamma$ is future imprisoned in a compact set of $M$, which contradicts Theorem \ref{blf}. Thus $f_+=+\infty$ and analogously $f$ converges to $-\infty$ in the past direction. We conclude that every inextendible causal curve of $M_c$ intersects $\Sigma$ exactly once, namely $\Sigma$ is a Cauchy hypersurface and $M_c$ is globally hyperbolic \cite{hawking73}.
\end{proof}

As recalled in Remark \ref{mpo}, for $c\ge c(L)$ the chronological relation of $M_c$ is described by $\Phi_c$, which equals Ma\~n\'e's potential up to the endpoint term $c(t_1-t_0)$. Let us show that the locus where the critical spacetime fails to be strongly causal is described by the  Aubry set. For time-periodic Lagrangians the {\em projected Aubry set} $\mathcal{A}\subset E$ is the set of points $e$ such that $h(e,e)=0$, where $h(e_0,e_1)=\liminf_{k\to +\infty} [S(e_0,e_1+kP)+c(L)kP]$ is the Peierls barrier \cite{mather91,bernard02,contreras13,fathi08}. It is invariant under the translations by multiples of $P$.

\begin{proposition} \label{aub}
The projected Aubry set is the set
\begin{equation} \label{aus}
\mathcal{A}=\{e\in E:  \inf_{k\ge 1} \, [S(e,e+kP)+c(L)kP]=0\}.
\end{equation}
It is closed and non-empty, and the spacetime $M_{c(L)}$ fails to be strongly causal precisely at the points of the fibers of $\mathcal{A}$, namely on the set $\pi_{c(L)}(\pi^{-1}(\mathcal{A}))$.
\end{proposition}

\begin{proof}
Let $c=c(L)$. For every $e$ and $k\ge 1$ we have $S(e,e+kP)+ckP\ge 0$, thus if $h(e,e)=0$ then $e$ belongs to the right-hand side of Eq.\ (\ref{aus}). For the converse, if the infimum in Eq.\ (\ref{aus}) is not attained then necessarily $h(e,e)=0$, while if it is attained for the integer $k$, by the triangle inequality for every $j\ge 1$, $0\le S(e,e+jkP)+cjkP\le j[S(e,e+kP)+ckP]=0$, thus again $h(e,e)=0$.

Let $\bar e\in E$, $\bar e=(\bar t,\bar q)$, and let us consider the functions $\phi_k(e_0,e_1)=S(e_0,e_1+kP)+ckP$, $k\ge 1$, for $(e_0,e_1)$ in the compact neighborhood $G=([\bar t-\frac{P}{4},\bar t+\frac{P}{4}]\times Q)^2$ of $(\bar e,\bar e)$. As every curve which connects $e_0$ to $e_1+kP$ intersects the slices of time $\bar t+\frac{P}{2}$ and  $\bar t+kP-\frac{P}{2}$, we have, using the periodicity of $S$,
\[
\phi_k(e_0,e_1)=\inf_{e',e''} \,[S(e_0,e')+S(e',e''+kP)+S(e'',e_1)]+ckP,
\]
where $e'$ runs over the slice $Q_{\bar t+P/2}$ and $e''$ runs over the slice $Q_{\bar t-P/2}$ (for $k=1$ recall that $S(e',e''+P)$ vanishes if $e'=e''+P$ and equals $+\infty$ otherwise). This expression shows that the functions $\phi_k$ share on $G$ the modulus of continuity determined by the uniformly continuous functions $S\vert_{([\bar t-P/4,\bar t+P/4]\times Q)\times Q_{\bar t+P/2}}$ and  $S\vert_{Q_{\bar t-P/2}\times ([\bar t-P/4,\bar t+P/4]\times Q)}$. Thus the function $\phi=\inf_{k\ge 1} \phi_k$, which is non-negative on the diagonal, is finite and continuous on $G$. In particular $\mathcal{A}=\{e:\phi(e,e)=0\}$ is closed.

Suppose that $M_c$ is not strongly causal at $z=\pi_c(x)$, $x=(\bar e,\bar y)$. There are a neighborhood $U$ of $z$ and causal curves $\bar\gamma_j$ not entirely contained in $U$ whose endpoints converge to $z$. Their lifts $\gamma_j$ connect $x_j$ to $\sigma_c^{k_j}(x_j')$, where $x_j\to x$, $x'_j\to x$ and $k_j\in \mathbb{Z}$. Since $t$ is a semi-time function $k_j\ge 0$ for sufficiently large $j$, and since $(M,g)$ is strongly causal at $x$, $k_j\ne 0$ for sufficiently large $j$. By Lemma \ref{per1}, $y'_j-y_j\ge \phi_{k_j}(e_j,e_j')\ge \phi(e_j,e_j')$, thus taking the limit, $\phi(\bar e,\bar e)\le 0$ and hence $\bar e\in \mathcal{A}$. Conversely, let $\bar e\in \mathcal{A}$, so that there are integers $k_j\ge 1$ such that $\epsilon_j:=\phi_{k_j}(\bar e,\bar e)\to 0$. By Lemma \ref{per1} there is a causal curve of $M_c$ which connects $z=\pi_c((\bar e,\bar y))$ to $\pi_c((\bar e,\bar y+\epsilon_j))$ and along which $[t]$ makes $k_j$ complete turns. Thus $M_c$ is not strongly causal at $z$. Finally, $\mathcal{A}$ is non-empty by Theorem \ref{cau2}.
\end{proof}

Summarizing, in the passage from $(M,g)$ to the compactified spacetime $M_c$ the least action $S$ is replaced by Ma\~n\'e's potential, the critical value is the threshold for chronology, and the Aubry set is the locus of the violation of strong causality at the threshold. Observe that by Theorem \ref{cau2} the critical spacetime is causal, and hence provides an example of causal but non-strongly causal spacetime, whenever the E.-L.\ flow has no closed orbit of vanishing $(L+c(L))$-action (e.g.\ for $Q=\mathbb{R}^d/\mathbb{Z}^d$, $L=\frac{1}{2}\vert v-\nu\vert^2$ where $\nu\in \mathbb{R}^d$ is such that $P\nu$ has an irrational component, in which case $c(L)=0$ and $\mathcal{A}=E$).

\subsection{Periodic weak KAM solutions as compact Cauchy horizons} \label{s:hor}

We recall that given a partial Cauchy hypersurface $\Sigma$ (an acausal edgeless set) its future domain of dependence $D^{+}(\Sigma)$ is the set of points $p$ such that every past inextendible causal curve passing through $p$ intersects $\Sigma$, and that the future Cauchy horizon is $H^{+}(\Sigma)=\overline{D^{+}(\Sigma)}\backslash I^{-}(D^{+}(\Sigma))$ \cite{hawking73}. The future Cauchy horizons are achronal and generated by past inextendible lightlike geodesics.

Let $c=c(L)$ and let $u$ be a periodic weak KAM solution at level $c(L)$, which exists by Theorem \ref{crit}. Let $W\subset M$ be its graph and let
\[
\bar W:=\pi_{c}(W)\subset M_{c}.
\]
The function $f=y-u(t,q)$ is $\sigma_c$-invariant, thus it defines a continuous function $\bar f:M_c\to \mathbb{R}$ such that $\bar W=\bar f^{-1}(0)$. By the argument which led to Eq.\ (\ref{stf}), $\bar f$ is non-decreasing over the causal curves of $M_c$, and since $M_c$ is not strongly causal $\bar f$ cannot be a time function.

\begin{theorem} \label{hor}
Let $Q$ be compact, let the Lagrangian be periodic in time with period $P$,  let $u$ be a periodic weak KAM solution at the critical level $c=c(L)$, and let $\bar W\subset M_{c}$ be the projection of its graph.
\begin{itemize}
\item[(a)] $\bar W$ is a compact topological hypersurface homeomorphic to $E_P$ which is achronal in $M_c$. Every point of $\bar W$ is the future endpoint of a lightlike geodesic contained in $\bar W$, which is past inextendible, past complete and transverse to $n$. Thus $\bar W$ is a compact (future) horizon with past complete generators. The generators are the projections of the light lifts of the $u$-calibrated curves.
\item[(b)] The translates of $\bar W$ under the flow of the parallel null vector $n$, namely the level sets of $\bar f$, share the same properties. Thus $M_c$ is foliated by compact horizons, and through every point of $M_c$ there passes a past inextendible lightlike geodesic over which $\bar f$ is constant. In particular, with respect to $g$ the level sets of $\bar f$ have trivial domains of dependence, $D^{+}(\bar f^{-1}(-s))=\bar f^{-1}(-s)$ for every $s\in \mathbb{R}$, and $\bar W$ is not the Cauchy horizon of any of them.
\item[(c)] Let $\chi: M_c\to [0,+\infty)$ be a function of class $C^{r}$ such that $\chi^{-1}(0)=\bar f^{-1}([0,+\infty))$,  that is, which vanishes precisely on $\bar W$ and on its future side, and let
\begin{equation} \label{gpr}
g'=g+2\chi\, \dd t^2.
\end{equation}
Then $g'$ is a $C^r$ Lorentzian metric on $M_c$ which coincides with $g$ on $\{\bar f\ge 0\}$ and whose light cones are contained in those of $g$. With respect to $g'$, $\bar W$ is still a compact horizon with the same generators, for every $s>0$ the level set $\Sigma_s=\bar f^{-1}(-s)$ is a compact partial Cauchy hypersurface, which is spacelike in the sense that the generalized (Clarke) differential of $\bar f$ at each of its points consists of $g'$-timelike covectors, and which is a $C^1$ spacelike hypersurface if $u$ is $C^1$, and
\[
\bar W=H^{+}(\Sigma_s).
\]
That is, $\bar W$ is a compact Cauchy horizon of the spacetime $(M_c,g')$.
\end{itemize}
\end{theorem}

In coordinates $([t],q,\tilde y)$ the metric $g'$ is given by Eq.\ (\ref{gc}) with $V-c$ replaced by $V-c-\chi$, where now $\chi$ depends also on $\tilde y$, so that $n$ is no more a Killing field in the past of $\bar W$. Thus the Cauchy horizon is obtained through a two step control of the opening of the light cones on $(\mathbb{R}/P\mathbb{Z})\times Q\times \mathbb{R}$: first one fixes the constant $c$ to the critical value, so that the weak KAM solution gives a compact horizon, and then one narrows the cones in the past of the horizon, so as to remove the causality violations and to obtain a  globally hyperbolic region which has the horizon as its future boundary. A function $\chi$ with the required properties exists because every closed subset of a manifold is the zero level set of a non-negative function as differentiable as the manifold.

\begin{proof}
(a). By Prop.\ \ref{glob2} $W$ is an achronal boundary of $(M,g)$, and each of its points is the future endpoint of a past inextendible lightlike geodesic contained in $W$, transverse to $n$, which is the light lift of a $u$-calibrated curve. As the function $t$ is an affine parameter and the E.-L.\ flow is complete, these generators are past complete. All these properties are preserved by the projection $\pi_c$, but for achronality which becomes local achronality. However, a timelike curve of $M_c$ with endpoints in $\bar W$ would have a lift over which the non-decreasing function $f$ vanishes identically, hence a lift entirely contained in $W$, which is impossible since $W$ is achronal. Thus $\bar W$ is achronal. It is compact and homeomorphic to $E_P$ because $u$ is continuous and $c$-equivariant.

(b). The translates are the projections of the graphs of the functions $u+const.$, which are still periodic weak KAM solutions at the critical level, thus (a) applies to them. If $\bar f(p)>-s$ the generator of the level set of $\bar f$ passing through $p$ is a past inextendible causal curve which does not intersect $\bar f^{-1}(-s)$, thus $p\notin D^{+}(\bar f^{-1}(-s))$.

(c). In coordinates the determinant of $g'$ equals that of $g$, and $n$ is lightlike for $g'$, thus $g'$ is Lorentzian and it is time oriented by the condition that $n$ is future directed. For every vector $v$, $g'(v,v)\ge g(v,v)$, thus the $g'$-causal vectors are $g$-causal. Since $\chi$ vanishes on the closed set $\{\bar f\ge 0\}$, whose boundary is $\bar W$, by continuity all its derivatives up to order $r$ vanish on $\bar W$, thus $g$ and $g'$ have on $\bar W$ the same Christoffel symbols (and the same curvature) and the generators of $\bar W$ are lightlike geodesics for both metrics. As the $g'$-timelike curves are $g$-timelike, $\bar W$ is $g'$-achronal. Thus $\bar W$ is a compact horizon of $(M_c,g')$.

In the remainder of the proof the causality notions refer to the metric $g'$. Let us work on the covering $M$, where we denote in the same way the lifts of $g'$ and $\chi$. Let $x_0\le x_1$, $x_0\ne x_1$, be two points connected by a  causal curve $\gamma$ contained in a region where $\chi\ge \chi_0$ for some constant $\chi_0\ge 0$. Then $\gamma$ is causal for the metric $g+2\chi_0\dd t^2$, which is the Eisenhart metric of the Lagrangian $L+\chi_0$. Applying to this metric Eq.\ (\ref{cgf2}), and using  Eq.\ (\ref{zgw}), we obtain that either $e_0=e_1$ and $y_0<y_1$ or $t_0<t_1$ and
\begin{equation} \label{stg}
f(x_1)-f(x_0)\ge S(e_0,e_1)+\chi_0(t_1-t_0)-[u(e_1)-u(e_0)]\ge \chi_0(t_1-t_0).
\end{equation}
Thus $\bar f$ is non-decreasing over the causal curves of $(M_c,g')$, and it is increasing over those contained in the open set $\{\bar f<0\}=\{\chi>0\}$. As a consequence, $\Sigma_s$, which is a compact topological hypersurface, is acausal and edgeless. Moreover, since $n^\flat=-\dd t$ and $g(n,n)=0$, we have $g'^{-1}=g^{-1}-2\chi\, n\otimes n$. At the points where $u$ is differentiable, $g^{-1}(\dd f,\dd f)=2[\p_t u+H(t,q,\p_q u)]=0$ and $\dd f(n)=1$, thus $g'^{-1}(\dd f,\dd f)=-2\chi$, which is negative on $\Sigma_s$. Hence the limits of these differentials at the points of $\Sigma_s$ are $g'$-timelike covectors which take the value $1$ on $n$, and so is every convex combination of them, which proves the spacelike character of $\Sigma_s$.

Let $p$ be such that $-s< \bar f(p)<0$ and let $\bar \gamma$ be a past inextendible causal curve ending at $p$. Suppose that it does not intersect $\Sigma_s$, then it is contained in the compact set $C=\bar f^{-1}([-s,\bar f(p)])$ over which $\chi$ has a positive minimum $\chi_0$. Let $\gamma$ be a lift of $\bar\gamma$. By Eq.\ (\ref{stg}) the function $t$ is bounded from below on $\gamma$ by $t(p)-s/\chi_0$. Since also $f$ is bounded on $\gamma$, the curve $\gamma$, which is causal for $g$, is past imprisoned in a compact set of $M$, in contradiction with Theorem \ref{blf}. Thus $\bar f^{-1}([-s,0))\subset D^{+}(\Sigma_s)$. Let $p$ be such that $\bar f(p)\ge 0$. By (b) there is a past inextendible $g$-lightlike geodesic ending at $p$ over which $\bar f$ is constant. It is contained in the set $\{\bar f\ge 0\}$, where $g'=g$, hence it is a past inextendible $g'$-causal curve which does not intersect $\Sigma_s$. Thus $p\notin D^{+}(\Sigma_s)$. Since $D^{+}(\Sigma_s)\subset J^{+}(\Sigma_s)\subset \bar f^{-1}([-s,+\infty))$, we conclude that $D^{+}(\Sigma_s)=\bar f^{-1}([-s,0))$. The closure of this set is $\bar f^{-1}([-s,0])$, and since $\bar f$ is non-decreasing over the causal curves and $\{\bar f<0\}$ is open, $I^{-}(D^{+}(\Sigma_s))\cap \bar f^{-1}([-s,0])=\bar f^{-1}([-s,0))$. We conclude that $H^{+}(\Sigma_s)=\bar f^{-1}(0)=\bar W$.
\end{proof}

\subsection{The autonomous case} \label{s:aut}

Let us suppose that $Q$ is compact and that the Lagrangian does not depend on time, so that the results of this section hold for every $P>0$. In this case the weak KAM theorem \cite{fathi08} establishes the existence of a function $u_-:Q\to \mathbb{R}$ and of a constant $c$ such that $u(t,q)=u_-(q)-ct$ is a global weak KAM solution. This solution is $c$-equivariant for every period $P$, hence by Theorem \ref{crit} the critical value $c(L)$ does not depend on $P$ and coincides with the critical value of the autonomous theory. Since several equivalent characterizations of this constant are used in the following, mostly in Section \ref{s:sync}, we collect them here.

\begin{proposition}[Characterizations of the critical value] \label{critchar}
Let $Q$ be compact and connected and let the Lagrangian $L$ be independent of time, with Hamiltonian $H$. Then the critical value $c(L)$ is:
\begin{itemize}
\item[(a)] the unique constant $k$ for which the Hamilton-Jacobi equation $H(q,\dd u)=k$ admits a viscosity solution on $Q$; these solutions are the weak KAM solutions $u_-$, and $u_-(q)-c(L)\,t$ is a global weak KAM solution. In particular, if the equation $H(q,\dd u)=k$ has a $C^1$ solution, then $k=c(L)$;
\item[(b)] the smallest constant $k$ for which there is a Lipschitz function $w$ on $Q$ with $H(q,\dd w)\le k$ almost everywhere, a condition which is equivalent to the domination $w(\gamma(T))-w(\gamma(0))\le\int_0^T[L(\gamma,\dot\gamma)+k]\,\dd s$ along every absolutely continuous curve $\gamma:[0,T]\to Q$; the minimum is attained, for instance, by the weak KAM solutions;
\item[(c)] the infimum of the constants $k$ such that $\int_0^T[L(\gamma,\dot\gamma)+k]\,\dd s\ge 0$ for every closed absolutely continuous curve $\gamma:[0,T]\to Q$, $T>0$ ({\em Ma\~n\'e's strict critical value});
\item[(d)] the number $\displaystyle\inf_{w}\,\max_{q\in Q}H(q,\dd w(q))$, the infimum being over the smooth functions on $Q$, or equivalently over the $C^1$ functions; in general it is not attained;
\item[(e)] Mather's $\alpha(0)$, namely the opposite of the minimum of $\int L\,\dd\mu$ over the probability measures on $TQ$ invariant under the Euler-Lagrange flow.
\end{itemize}
\end{proposition}

\begin{proof}
These facts are classical, and we only indicate the arguments and references. Point (a), the existence of weak KAM solutions and the equivalence between domination and the almost everywhere inequality in (b) are proved in Fathi's book \cite{fathi08}. (b)$\Leftrightarrow$(c): if $w$ is dominated by $L+k$, then integrating along a closed curve gives $\int_0^T(L+k)\ge 0$; conversely, if the action of $L+k$ is non-negative on every closed curve, Ma\~n\'e's action potential $\Phi_k(q_0,q)=\inf\int_0^T(L+k)$, the infimum being over the curves from $q_0$ to $q$ of arbitrary duration, is finite, Lipschitz and dominated by $L+k$ \cite{mane97,contreras98,contreras99}. (b)$\Leftrightarrow$(d): a $C^1$ function with $H(q,\dd w)\le k$ is a Lipschitz subsolution, and conversely, by the convexity of $H$ in the momentum, a Lipschitz subsolution at the level $k$ can be regularized into smooth subsolutions at the level $k+\epsilon$ for every $\epsilon>0$ \cite{contreras98,fathi04,fathi08}; critical subsolutions of class $C^{1,1}$ exist \cite{bernard07}, while smooth critical solutions do not exist in general. (e) is due to Ma\~n\'e and to Contreras, Iturriaga, G.\ and M.\ Paternain \cite{mane97,contreras98,contreras99}, see also \cite{mather91,fathi08}. The identification with the constant of Theorem \ref{crit} has been shown above.
\end{proof}

For the Lagrangians $\frac{1}{2}a(v,v)+b(v)-V(q)$ a further characterization, in terms of the Jacobi-Maupertuis length of the closed curves, is given in Lemma \ref{fer}(c).
The spacetime  interpretation of the weak KAM solutions of the autonomous theory is as follows:  $u_-$ is a weak KAM solution of negative type if and only if the hypersurface
\[
W=\{(t,q,y): y=u_-(q)-c(L)\,t\}
\]
is a  horizon of $(M,g)$, in the sense of Prop.\ \ref{glob2}, which is invariant under the flow of the Killing vector field $\tilde k=\p_t-c(L)\p_y$ (this is the vector field used in Sect.\ \ref{mau} for the dimensional reduction to the Maupertuis spacetime, with the energy $e$ fixed to the critical value). The identity $T^-_t u_-=u_--c(L)t$, which characterizes the weak KAM solutions as the fixed points of the Lax-Oleinik semigroup modulo constants \cite{fathi08}, is the invariance of the horizon under the Lax-Oleinik evolution established in Prop.\ \ref{glob2}. Analogously, the weak KAM solutions of positive type $u_+$ correspond to the $\tilde k$-invariant past horizons, namely to the achronal boundaries whose points are the past endpoints of future inextendible lightlike generators.

\begin{example}[Natural mechanical systems] \label{nms}
Let $b=0$, so that $H(q,p)=\frac{1}{2}a^{-1}(p,p)+V(q)$, and let us suppose that $V$ attains its maximum only at the point $q_*$. Then $c(L)=V(q_*)$, the projected Aubry set of the autonomous theory is $\{q_*\}$, and up to additive constants there is only one weak KAM solution of negative type, $u_-(q)=d_J(q_*,q)$, and only one weak KAM solution of positive type, $u_+(q)=-d_J(q,q_*)$,  where $d_J$ is the distance of the Jacobi metric $a_J=2(c(L)-V)\,a$ at the critical energy (this metric degenerates at $q_*$) \cite{mane97,contreras99,fathi08}. The rest at the unstable equilibrium is an action minimizing motion, thus its light lift
\[
\eta(t)=(t,q_*,-c(L)\,t)
\]
is a lightlike line of $(M,g)$, namely an achronal inextendible lightlike geodesic. By Lemma \ref{pcf}, $I^{+}(\eta)$ is the set of events such that
\[
y+c(L)\,t> \inf_{\tau>0}\,[S((0,q_*),(\tau,q))+c(L)\,\tau],
\]
and the right-hand side is Ma\~n\'e's potential of the autonomous theory, which for natural systems coincides with $d_J(q_*,q)$. Therefore, the horizons determined by the weak KAM solutions are
\[
W_-=\{y=u_-(q)-c(L)t\}=\p I^{+}(\eta), \qquad W_+=\{y=u_+(q)-c(L)t\}=\p I^{-}(\eta),
\]
that is, {\em the weak KAM solutions are the future and past horizons of the lightlike line determined by the Aubry set}, and the two horizons intersect precisely on the line. They are not differentiable  at the points which project on the points of $Q$ that are joined to $q_*$ by more than one minimizing geodesic of the Jacobi metric.
\end{example}

The pair $(u_-,u_+)$ of the example is a pair of conjugate weak KAM solutions in the sense of Fathi \cite{fathi08}. The single-line description does not extend to all conjugate pairs. For instance, let $Q=S^1$ and $L=\frac{1}{2}v^2+\sin^2\theta$, so that $c(L)=0$ and the projected Aubry set consists of the two equilibria $\theta=0$ and $\theta=\pi$. With the Jacobi distance $d_J$ of $2\sin^2\theta\, \dd\theta^2$, the function $D(\theta)=\min\{d_J(0,\theta),d_J(\pi,\theta)\}=\sqrt{2}\,(1-\vert\cos\theta\vert)$ gives the conjugate pair $u_-=D$, $u_+=-D$. If $\eta_0$, $\eta_\pi$ are the light lifts of the two equilibria at $y=0$, then $I^{+}(\eta_0\cup\eta_\pi)=\{y>D(\theta)\}$ and $I^{-}(\eta_0\cup\eta_\pi)=\{y<-D(\theta)\}$, so that the horizons of $u_\mp$ are $\p I^{\pm}(\eta_0\cup\eta_\pi)$. However, they are not the boundaries of the future and past of a single lightlike line: such a line would lie in their intersection, namely on $\eta_0$ or on $\eta_\pi$, while $\p I^{+}(\eta_0)$ is the graph of $d_J(0,\cdot)$, which takes the value $2\sqrt{2}$ at $\pi$, where $u_-$ vanishes, and similarly for $\eta_\pi$. Thus Example \ref{nms} suggests that Fathi's conjugate pairs should be regarded as the classical counterpart of the pairs of achronal boundaries $\p I^{+}$, $\p I^{-}$ generated by a lightlike line, or more generally by a family of lightlike lines, but in general this is an analogy rather than an exact correspondence. The pairs generated by a lightlike line play a central role in Galloway's null splitting theorem \cite{galloway00}, which states that in a null geodesically complete spacetime which satisfies the null convergence condition  the components of the two boundaries which contain $\eta$ coincide and form a smooth totally geodesic null hypersurface. In the next section we shall obtain a result of this type for the periodic weak KAM solutions.

\section{Null convergence condition and smooth weak KAM solutions} \label{s:ncc}

The weak KAM solutions are in general just Lipschitz and semi-concave functions. Their existence, and their interpretation as horizons, did not require any curvature assumption (Sections \ref{s:subsolutions} and \ref{s:per}). In this section we show that their regularity is controlled by an energy condition on the spacetime $(M,g)$. We first express the Ricci tensor, and hence the null convergence condition, in terms of quantities living on the classical spacetime $E$. Then we use the identification of the global weak KAM solutions with horizons (Sections \ref{s:subsolutions} and \ref{s:per}) and the area theorem for horizons in the low regularity version of \cite{minguzzi14d} to prove that, on a compact configuration space and under the null convergence condition, the global weak KAM solutions are as differentiable as the Lagrangian, for arbitrary time dependence. Actually, they satisfy a rigidity property which determines the dynamics completely. In the periodic case the result can be regarded as the weak KAM counterpart of the smoothness of compact Cauchy horizons \cite{larsson14,minguzzi14d}.

In this section we assume that the fields $a_t$, $b_t$ and $V$ are $C^r$ with $r\ge 3$, so that the spacetime metric is $C^3$.

\subsection{The Ricci tensor on the classical spacetime}

We recall that the spacetime $(M,g)$ satisfies the {\em null convergence condition} if $Ric(X,X)\ge 0$ for every lightlike vector $X$, and the {\em timelike convergence condition} if the same inequality holds for every timelike vector \cite{hawking73}. Since $g(X,X)=0$ for a lightlike vector $X$, the Einstein equations $Ric-\frac{1}{2}R\,g=8\pi G\, T$, possibly with a cosmological term, give $Ric(X,X)=8\pi G\, T(X,X)$ for every lightlike $X$. Thus, for the spacetimes which satisfy the Einstein equations the null convergence condition is equivalent to the {\em null energy condition}, $T(X,X)\ge 0$ for every lightlike $X$, which is a condition on the source \cite{hawking73}. The two conditions are therefore the same requirement, expressed on the curvature side and on the matter side of the Einstein equations, respectively.

Let $D$ be the Levi-Civita connection of $(Q,a_t)$, let $Ric_{a_t}$ be its Ricci tensor, and let $\Delta_{a_t}$ and $\textrm{div}_{a_t}$ be the corresponding Laplacian and divergence. The indices $a,b,c$ are raised and lowered with $a_t$. Let us denote with a dot the partial derivative with respect to $t$ of the time dependent tensor fields on $Q$, so that $\dot a=\p_t a_t$ and $\dot b=\p_t b_t$, and let
\[
F=\dd b_t, \qquad F_{ab}=\p_a b_b-\p_b b_a,
\]
be the exterior differential of $b_t$ on $Q$. For a mechanical system observed from a non-inertial frame $F$ is twice the vorticity of the frame, while for a charged particle it is proportional to the magnetic field.

\begin{proposition} \label{ric}
In the coordinates $(t,q^a,y)$ the Ricci tensor of the metric (\ref{eis}) has components $R_{yy}=R_{yt}=R_{ya}=0$ and
\begin{align}
R_{ab}&=(Ric_{a_t})_{ab}, \label{ri1}\\
R_{ta}&=\rho_a:=\frac{1}{2}\, [D^b \dot a_{ba}-\p_a (a^{bc}\dot a_{bc})-D^b F_{ba}],  \label{ri2}\\
R_{tt}&=\varrho:=\Delta_{a_t} V+ \textrm{div}_{a_t} \dot b +\frac{1}{4} \, F_{ab}F^{ab}-\frac{1}{2} \,a^{ab}\ddot a_{ab}+\frac{1}{4}\, \dot a_{ab} \dot a^{ab}. \label{ri3}
\end{align}
Thus $Ric=\pi^*\mathcal{R}$, where $\mathcal{R}$ is the symmetric tensor field on the classical spacetime $E$ given by
\begin{equation}
\mathcal{R}= Ric_{a_t}+\rho_t\otimes \dd t+\dd t\otimes \rho_t+\varrho\, \dd t^2.
\end{equation}
The tensor $\mathcal{R}$ is the Ricci tensor of the Galilei (Newton-Cartan) connection induced by $g$ on $E$ (Appendix B). It does not depend on the gauge, namely it is invariant under the replacement $b_t\to b_t+\p_q \alpha$, $V\to V-\p_t \alpha$.
\end{proposition}

\begin{proof}
The identities $Ric(n,\cdot)=0$ and $L_n Ric=0$ follow from the fact that $n$ is covariantly constant, and they imply that $Ric$ is the pullback of a tensor on $E$. The expressions (\ref{ri1})-(\ref{ri3}) are proved in Appendix B (we also checked them with a computer algebra system). The gauge invariance follows from the fact that the change of gauge is the coordinate change $y'=y+\alpha$ which does not affect the coordinates of $E$ (Sect.\ \ref{gau}); it can also be checked directly, as $F$ is invariant and the variation of the first two terms on the right-hand side of Eq.\ (\ref{ri3}) is $-\Delta_{a_t}\p_t\alpha+\textrm{div}_{a_t} (\p_q\p_t \alpha)=0$.
\end{proof}

For every $(t,q)\in E$ and $v\in T_qQ$ let us define
\begin{equation} \label{rv}
\mathcal{R}(t,q,v):=\mathcal{R}(\p_t+v,\p_t+v)=Ric_{a_t}(v,v)+2\rho_t(v)+\varrho .
\end{equation}
This function on $T\times TQ$ is a polynomial of second degree in the velocity. By Theorem \ref{kdf} the lightlike vectors of $(M,g)$ which are not proportional to $n$ are proportional to the tangents to the light lifts, $X=\p_t+v+L(t,q,v)\p_y$, and for them $Ric(X,X)=\mathcal{R}(t,q,v)$. As for the timelike vectors, they are proportional to the vectors $X=\p_t+v+\lambda\p_y$ with $\lambda> L(t,q,v)$, for which again $Ric(X,X)=\mathcal{R}(t,q,v)$. Since $\mathcal{R}$ is the Ricci tensor of the Galilei connection of the classical spacetime $E$, and the vectors $\p_t+v$ are its future-directed timelike vectors normalized by $\dd t(\p_t+v)=1$, we say that $E$, or the Lagrangian $L$, satisfies the {\em Galilean convergence condition} if $\mathcal{R}(t,q,v)\ge 0$ for every $(t,q)\in E$ and $v\in T_qQ$. It is a condition on the classical data $(a_t,b_t,V)$ alone. Since $Ric(n,n)=0$, we conclude that

\begin{proposition} \label{ncc}
For the spacetime $(M,g)$ the following conditions are equivalent:
\begin{itemize}
\item[(i)] the null convergence condition,
\item[(ii)] the timelike convergence condition,
\item[(iii)] the Galilean convergence condition, $\mathcal{R}(t,q,v)\ge 0$ for every $(t,q)\in E$ and $v\in T_qQ$,
\item[(iv)] $\mathcal{R}$ is positive semi-definite on $TE$,
\item[(v)] for every $(t,q)\in E$, as quadratic forms on $T_qQ$,
\begin{equation} \label{nec}
Ric_{a_t}\ge 0, \qquad \varrho\ge 0, \qquad \varrho\, Ric_{a_t}-\rho_t\otimes \rho_t\ge 0 ,
\end{equation}
\item[(vi)] the dominant energy condition, namely $G(X,Y)\ge 0$ for every pair of future-directed causal vectors $X,Y$, where $G=Ric-\frac{1}{2}R\, g$ is the Einstein tensor.
\end{itemize}
\end{proposition}

\begin{proof}
The equivalence between (i), (ii) and (iii) has just been proved. Every vector of $TE$ which is not tangent to the slices $Q_t$ is proportional to a vector of the form $\p_t+v$, and the vectors tangent to the slices are limits of vectors of this type, thus (iii) and (iv) are equivalent. Let $(t,q)\in E$ and $X\in T_qQ$. For $\lambda\in\mathbb{R}$
\[
\mathcal{R}(t,q,\lambda X)=Ric_{a_t}(X,X)\,\lambda^2+2\rho_t(X)\,\lambda+\varrho .
\]
If (iii) holds this polynomial is non-negative for every $\lambda$, thus its value at $\lambda=0$ is non-negative, $\varrho\ge 0$, its leading coefficient is non-negative, $Ric_{a_t}(X,X)\ge 0$ (divide by $\lambda^2$ and let $\lambda\to\infty$), and its discriminant is non-positive, $\rho_t(X)^2\le \varrho\, Ric_{a_t}(X,X)$, which also holds if $Ric_{a_t}(X,X)=0$, as in this case a non-negative polynomial of degree at most one is constant. This is (v). Conversely, let (v) hold and let $v\in T_qQ$. If $Ric_{a_t}(v,v)>0$, the minimum of $\lambda\mapsto \mathcal{R}(t,q,\lambda v)$ is $\varrho-\rho_t(v)^2/Ric_{a_t}(v,v)\ge 0$, thus $\mathcal{R}(t,q,v)\ge 0$. If $Ric_{a_t}(v,v)=0$, then $\rho_t(v)=0$ by (\ref{nec}), and $\mathcal{R}(t,q,v)=\varrho\ge 0$. This is (iii).

Finally, we prove the equivalence between (iii) and (vi). Since $g(n,\cdot)=-\dd t$ we have $g^{-1}(\dd t,\cdot)=-n$, and a direct computation of the inverse metric gives $g^{-1}(\dd q^a,\dd q^b)=a_t^{ab}$. As $Ric=\pi^*\mathcal{R}$ has no component along $\dd y$, the scalar curvature is $R=\textrm{Scal}_{a_t}$, which in general does not vanish. Since $G$ is bilinear and every future-directed causal vector is a limit of positive combinations of future-directed lightlike vectors, it suffices to consider lightlike $X,Y$, namely positive multiples of $n$ and of the tangents to the light lifts $X=\p_t+v+L(t,q,v)\p_y$, $Y=\p_t+w+L(t,q,w)\p_y$ (Theorem \ref{kdf}). We have $G(n,n)=0$ and $G(n,X)=-\frac12\textrm{Scal}_{a_t}\, g(n,X)=\frac{1}{2}\textrm{Scal}_{a_t}$. Moreover $g(X,Y)=-\frac{1}{2}g(X-Y,X-Y)=-\frac{1}{2}\,a_t(v-w,v-w)$, as $X-Y=(v-w)+[L(t,q,v)-L(t,q,w)]\p_y$, thus with $p=(v+w)/2$, $r=(v-w)/2$,
\begin{align*}
G(X,Y)&=\mathcal{R}(\p_t+p+r,\p_t+p-r)+\textrm{Scal}_{a_t}\,a_t(r,r)\\
&=\mathcal{R}(t,q,p)+\textrm{Scal}_{a_t}\,a_t(r,r)-Ric_{a_t}(r,r).
\end{align*}
If (iii) holds, then $Ric_{a_t}\ge 0$ by (v), hence $\textrm{Scal}_{a_t}\ge 0$ and $Ric_{a_t}(r,r)\le \textrm{Scal}_{a_t}\, a_t(r,r)$, so that $G(X,Y)\ge 0$ and $G(n,X)\ge 0$, which is (vi). Conversely, (vi) with $X=Y$ gives (iii).
\end{proof}

Thus, for the spacetimes admitting a parallel null vector, the null and timelike convergence conditions and the dominant energy condition coincide, and they are equivalent to the Galilean convergence condition of the classical spacetime: whenever one of them holds, the Einstein tensor is that of a physically reasonable source.

Condition (\ref{nec}) is the formulation of the null convergence condition which we shall use for the statements of classical mechanics. It involves only the tensors $Ric_{a_t}$, $\rho_t$ and $\varrho$ of Prop.\ \ref{ric} at the given point, and no quantifier over the velocities. The condition $\varrho\ge 0$ cannot be dropped where $Ric_{a_t}$ vanishes, as the natural systems on flat manifolds show (Example \ref{nat}). The last condition forces $\rho_t$ to vanish on the kernel of $Ric_{a_t}$, and where $Ric_{a_t}$ is positive definite (\ref{nec}) reduces to the single scalar inequality
\[
\varrho\ge Ric_{a_t}^{-1}(\rho_t,\rho_t),
\]
by the Cauchy-Schwarz inequality for the scalar product $Ric_{a_t}$.

Thus {\em the null convergence condition is a  condition of classical mechanics}, which we shall still call null convergence condition. Let us illustrate it with some examples.

\begin{example}[Natural systems and Poisson's equation] \label{nat}
For a natural mechanical system, $\p_t a_t=0$, $b_t=0$, we have
\[
\mathcal{R}(t,q,v)=Ric_a(v,v)+\Delta_a V,
\]
thus the null convergence condition holds if and only if $(Q,a)$ has non-negative Ricci curvature and $V(t,\cdot)$ is subharmonic for every $t$. In particular, consider a particle moving in the Euclidean space under a Newtonian gravitational field, where $V$ is the gravitational potential (potential energy per unit mass) and Poisson's equation reads $\Delta V=4\pi G \mu$, $\mu$ being the mass density of the sources. Then $Ric=4\pi G\mu\, \dd t\otimes \dd t$, which is the geometrized form of Poisson's equation in Newton-Cartan theory \cite{cartan23,trautman65,kunzle72,kunzle76,ehlers81,malament12}; indeed, the connection induced by $g$ on $E$ is the Newton-Cartan connection of the potential $V$ (Appendix B, see also \cite{duval85,minguzzi06d}). As the Euclidean metric is flat, the scalar curvature, which in general equals $\textrm{Scal}_{a_t}$ (see the proof of Prop.\ \ref{ncc}), vanishes and $Ric$ coincides with the Einstein tensor, thus $(M,g)$ solves the Einstein equations with the null dust source $T=\frac{\mu}{2}\, \dd t\otimes \dd t$. The null convergence condition, which here coincides with the null energy condition for $T$, is the non-negativity of the mass density $\mu$, and the vacuum Einstein equations are equivalent to Laplace's equation. The focusing of the lightlike geodesics on $(M,g)$ implied by the null convergence condition is the focusing of the trajectories due to the attractive nature of gravity. Observe that if $Q$ is compact a subharmonic function is constant on $Q$, thus on a compact configuration space a natural system satisfies the null convergence condition if and only if $Ric_a\ge 0$ and the potential depends only on time (and hence can be removed with a change of gauge).
\end{example}

\begin{example}[Motion on a moving geometry] \label{rex}
For the Lagrangian $L=\frac{1}{2}a_t(v,v)$ of Sect.\ \ref{gau}, which corresponds to the metric in Rosen form (Sect.\ \ref{ros}),  let $A=a_t^{-1}\dot a_t$. We have
\[
\varrho=-\frac{1}{2}\,\textrm{tr} (a_t^{-1}\ddot a_t)+\frac{1}{4}\,\textrm{tr} A^2=-\p_t \Big(\frac{1}{2}\, \textrm{tr} A\Big)-\textrm{tr}\Big(\frac{A}{2}\Big)^2\le -\p_t \theta-\frac{\theta^2}{d},
\]
where $\theta=\frac{1}{2}\textrm{tr} A=\p_t \ln \sqrt{\det a_t}$ is the rate of expansion of the Riemannian volume. Thus if the null convergence condition holds, which demands $\varrho=\mathcal{R}(t,q,0)\ge 0$, the function $\theta$ satisfies $\p_t \theta\le -\theta^2/d$. If $T=\mathbb{R}$ this inequality implies $\theta=0$, hence $A=0$. We conclude that {\em a metric in Rosen form with $T=\mathbb{R}$ satisfies the null convergence condition if and only if $a_t$ does not depend on time and has non-negative Ricci curvature}. Otherwise the Rosen coordinates necessarily break down in a finite time, consistently with the observations recalled in Sect.\ \ref{ros}.
\end{example}

\begin{example}[Rotating frames] \label{rot}
Let $(Q,a)$ be a Riemannian manifold with a complete Killing vector field $K$ and let
\[
L(q,v)=\frac{1}{2}\,a(v,v)-a(K,v)+\frac{1}{2}\,a(K,K)=\frac{1}{2}\, a(v-K,v-K).
\]
This is the Lagrangian of the geodesic flow of $(Q,a)$ as observed from the reference frame which moves with the flow of $-K$ (Sect.\ \ref{gau}). Using $D_aK_b+D_bK_a=0$ and $D^bD_b K_a=-(Ric_a)_{ab}K^b$ we obtain $F_{ab}=-2D_aK_b$, $\rho=-Ric_a(K,\cdot)$, and by Bochner's formula $\varrho=-\frac{1}{2}\Delta_a\vert K\vert^2+\vert DK\vert^2=Ric_a(K,K)$. Thus
\[
\mathcal{R}(q,v)=Ric_a(v-K,v-K),
\]
and the null convergence condition holds if and only if $Ric_a\ge 0$. For instance, on the round sphere $S^3$ with $K=-\epsilon\xi$, where $\xi$ is the unit Hopf vector field, we obtain the magnetic-type Lagrangian $L=\frac{1}{2}\vert v\vert^2+\epsilon\, \eta(v)+\frac{\epsilon^2}{2}$, where $\eta=a(\xi,\cdot)$ is the standard contact form, for which $\mathcal{R}(q,v)=2\vert v+\epsilon\xi\vert^2$.
\end{example}

\subsection{The Raychaudhuri equation on the classical spacetime}

Let $u$ be a $C^2$ function defined on an open subset of $E$ and let us consider the time dependent vector field on $Q$ and the time dependent endomorphism of $TQ$
\begin{equation} \label{vB}
v_t=a_t^{-1}(\cdot, \p_q u-b_t), \qquad B=\frac{1}{2}\, a_t^{-1}(\p_t a_t+L_{v_t} a_t),
\end{equation}
where $L_{v_t}$ is the Lie derivative on $Q$, that is,
$
B_{ab}=\frac{1}{2}\dot a_{ab}+D_aD_b u-\frac{1}{2}(D_ab_b+D_b b_a)$. The endomorphism $B$ is self-adjoint with respect to $a_t$. Its trace is
\begin{equation} \label{the}
\theta=\textrm{tr} B=\Delta_{a_t}u-\textrm{div}_{a_t} b_t+\p_t \ln \sqrt{\det a_t}\, .
\end{equation}
If $u$ solves the Hamilton-Jacobi equation the integral curves of $\p_t+v_t$ are its base characteristics, the hypersurface $W=\{y=u(t,q)\}$ is lightlike and it is generated by the lightlike geodesics with tangent $n_W=\p_t+v_t+L(t,q,v_t)\p_y$, for which $t$ is an affine parameter (Theor.\ \ref{lkp}). The vectors tangent to $W$ at $(t,q,u(t,q))$ modulo $n_W$ are represented by the vectors $w+\p_q u(w)\p_y$, $w\in T_qQ$, the metric induced on this quotient space is $a_t$. Substituting $\dd y=\p_t u\, \dd t+\p_q u$, the metric induced on $W$ reads in the coordinates $(t,q)$
\begin{align*}
&a_t-\dd t\otimes (\p_q u-b_t)-(\p_q u-b_t)\otimes \dd t-2(V+\p_t u)\,\dd t^2\\
&\qquad =a_t-\dd t\otimes v_t^\flat-v_t^\flat\otimes \dd t+a_t(v_t,v_t)\,\dd t^2,
\end{align*}
where $v_t^\flat=a_t(v_t,\cdot)=\p_q u-b_t$ and we used the Hamilton-Jacobi equation, $\p_t u+V=-\frac{1}{2}a_t(v_t,v_t)$. Its kernel is spanned by $\p_t+v_t$, and half of its Lie derivative along $n_W=\p_t+v_t$ gives on the quotient space the bilinear form $a_t(\cdot, B\cdot)$. Thus {\em $B$ is the null Weingarten map (shape operator) of the lightlike hypersurface $W$, and $\theta$ is its expansion} \cite{hawking73,galloway00,minguzzi14d}. Equivalently, $\theta$ is the logarithmic derivative of the Riemannian volume of $(Q,a_t)$ transported by the flow of the characteristics. The Raychaudhuri equation for $W$, joined with Prop.\ \ref{ric}, gives

\begin{proposition} \label{ray}
Let $u$ be a $C^2$ solution of the Hamilton-Jacobi equation on an open subset of $E$. Along every base characteristic $t\mapsto (t,q(t))$ of $u$ the expansion is $C^1$ and
\begin{equation} \label{ray1}
\frac{\dd \theta}{\dd t}=-\textrm{tr} B^2-\mathcal{R}(t,q(t),\dot q(t))\le -\frac{\theta^2}{d}-\mathcal{R}(t,q(t),\dot q(t)).
\end{equation}
\end{proposition}
A proof of this proposition in terms of quantities living on $E$ is given in Appendix B.

\begin{remark} \label{rayrem}
More generally, a direct calculation, given in Appendix B (we also checked it with a computer algebra system), shows that for every $C^3$ function $u$, defined $v_t$, $B$ and $\theta$ as above, the following identity holds
\begin{equation} \label{ray2}
(\p_t+v_t)\,\theta+\textrm{tr} B^2+\mathcal{R}(t,q,v_t)=\Delta_{a_t} \Big[\frac{\p u}{\p
t}+ H(t,q,\frac{\p u}{\p q})\Big].
\end{equation}
For $a_t$ independent of time, $b_t=0$ and $V=0$ this is Bochner's formula for the function $u(t,\cdot)$.
If $Q$ is compact, integrating Eq.\ (\ref{ray1}) with respect to the Riemannian volume $\mu_t$ of $(Q,a_t)$ we obtain for the solutions of the Hamilton-Jacobi equation
\[
\frac{\dd^2}{\dd t^2} \, \textrm{Vol}(Q,a_t)=\int_Q \big[\theta^2-\textrm{tr} B^2-\mathcal{R}(t,q,v_t)\big]\,\mu_t ,
\]
where we used $\frac{\dd }{\dd t} \textrm{Vol}(Q,a_t)=\int_Q \theta\mu_t$. Since $\theta^2-\textrm{tr} B^2$ does not have a definite sign for $d\ge 2$, this integral identity does not seem sufficient to obtain the results of the next subsection.
\end{remark}

A first consequence of the Raychaudhuri equation is that the classical solutions of the Hamilton-Jacobi equation which exist globally in time are very special. Here $Q$ need not be compact and no periodicity is assumed.

\begin{theorem} \label{rig}
Let $T=\mathbb{R}$, let the null convergence condition (\ref{nec}) hold, and let $u:E\to \mathbb{R}$ be a $C^k$ solution, $2\le k\le r+1$, of the Hamilton-Jacobi equation whose base characteristics are complete (this is the case if the E.-L.\ flow is complete, e.g.\ if $Q$ is compact). Let $v_t$ be given by Eq.\ (\ref{vB}) and let $\nu_t:Q\to Q$, $\nu_0=Id$, be the flow of the base characteristics. Then
\begin{itemize}
\item[(a)] $B=0$, that is,
\begin{equation} \label{kil}
\p_t a_t+L_{v_t} a_t=0,
\end{equation}
namely the characteristic flow is isometric, $\nu_t^* a_t=a_0$. In particular, $u$ satisfies on every slice the linear elliptic equation
\begin{equation} \label{ell}
\Delta_{a_t}u=\textrm{div}_{a_t} b_t-\p_t \ln \sqrt{\det a_t}\, .
\end{equation}
\item[(b)] For every $(t,q,v)$
\begin{equation} \label{rsq}
\mathcal{R}(t,q,v)=Ric_{a_t}(v-v_t,v-v_t),
\end{equation}
that is, $\rho_t=-Ric_{a_t}(v_t,\cdot)$ and $\varrho=Ric_{a_t}(v_t,v_t)$.
\item[(c)] The map $\Phi:\mathbb{R}\times Q\times \mathbb{R}\to M$, $\Phi(t,q',y')=(t,\nu_t(q'),y'+u(t,\nu_t(q')))$ is a $C^{k-1}$ diffeomorphism such that
\[
\Phi^*g=a_0-\dd t\otimes \dd y'-\dd y'\otimes \dd t.
\]
Thus $(M,g)$ is isometric to the product of the Riemannian manifold $(Q,a_0)$, which has non-negative Ricci curvature, and the 2-dimensional Minkowski spacetime. The solutions of the E.-L.\ equations are the curves $q(t)=\nu_t(\gamma(t))$, where $\gamma$ is an affinely parametrized geodesic of $(Q,a_0)$ (possibly constant).
\end{itemize}
\end{theorem}

\begin{proof}
By Prop.\ \ref{ray} along every characteristic the expansion satisfies $\frac{\dd \theta}{\dd t}\le -\frac{\theta^2}{d}$ for every $t\in \mathbb{R}$. If $\theta(t_1)<0$ this inequality implies that $\theta\to -\infty$ before the time $t_1+d/\vert\theta(t_1)\vert$, while if $\theta(t_1)>0$ it implies that $\theta\to +\infty$, in the past direction, after the time $t_1-d/\theta(t_1)$. As $\theta$ is finite we conclude that $\theta=0$ on $E$. By Eq.\ (\ref{ray1}), $\textrm{tr} B^2+\mathcal{R}(t,q,v_t)=0$, and since $B$ is self-adjoint and $\mathcal{R}\ge 0$, we obtain $B=0$ and $\mathcal{R}(t,q,v_t)=0$. Equation (\ref{ell}) is $\theta=0$, see Eq.\ (\ref{the}). From $\frac{\dd}{\dd t}\, \nu_t^* a_t=\nu_t^*(\p_t a_t+L_{v_t} a_t)=0$ we obtain $\nu_t^* a_t=a_0$. The non-negative second order polynomial $v\mapsto \mathcal{R}(t,q,v)$ vanishes at $v_t$, hence it attains there its minimum, thus its differential vanishes at $v_t$, which implies Eq.\ (\ref{rsq}). Point (c) follows from Theorem \ref{ndp}, where the Lagrangian in the new coordinates is $\frac{1}{2}a'_t(v',v')$ with $a'_t=\nu_t^* a_t=a_0$. The Ricci tensor of $a_0$ is non-negative by Prop.\ \ref{ncc}.
\end{proof}

\begin{remark} \label{cls}
The theorem admits a converse: if $(Q,a_0)$ has non-negative Ricci curvature then any spacetime obtained from the product metric $a_0-\dd t\otimes \dd y'-\dd y'\otimes \dd t$ through a change of reference frame $q=\nu_t(q')$ and a change of gauge $y=y'+u(t,q)$ satisfies the null convergence condition, as this condition does not depend on the coordinates. By the formulas of Sect.\ \ref{gau} the mechanical systems obtained in this way are
\begin{equation} \label{cla}
a_t=(\nu_t^{-1})^*a_0, \qquad b_t=\p_q u-a_t(\cdot, \hat v_t), \qquad V=-\p_t u-\frac{1}{2}\, a_t(\hat v_t,\hat v_t),
\end{equation}
where $\hat v_t=(\p_t \nu_t)\circ \nu_t^{-1}$ and $u$ are arbitrary. Therefore, whenever the Hamilton-Jacobi equation admits a global $C^2$ solution with complete characteristics, in particular whenever there are global Rosen coordinates, the null convergence condition holds if and only if the mechanical system is of the form (\ref{cla}), namely if and only if it is the geodesic flow of a (time independent) Riemannian manifold with non-negative Ricci curvature observed from a moving frame (cf.\ Examples \ref{rex} and \ref{rot}). By Theorem \ref{smoa} below, if $Q$ is compact the assumption on the existence of a global $C^2$ solution can be dropped.
\end{remark}

\subsection{Smoothness and rigidity of the global weak KAM solutions}

The previous theorem assumes the existence of a $C^2$ global solution. The main result of this section, Theorem \ref{smoa}, shows that on a compact configuration space the null convergence condition by itself guarantees the existence of such a solution, for arbitrary time dependence: under this condition the global weak KAM solutions, which exist by Theorem \ref{glob}, are classical solutions as differentiable as the Lagrangian. For the time-periodic systems of Sect.\ \ref{s:per} we then obtain further information on the dynamics (Theorem \ref{smo}).

For a locally semi-concave function, such as a weak KAM solution (Prop.\ \ref{glob2}), the second differential exists almost everywhere by Alexandrov's theorem \cite{cannarsa04}, thus the right-hand side of Eq.\ (\ref{the}) is still defined almost everywhere. We start with an inequality which holds with no periodicity assumption. It can be regarded as a Laplacian comparison theorem for the weak KAM solutions, analogous to the superharmonicity of the Busemann functions in Riemannian manifolds of non-negative Ricci curvature \cite{cheeger71}.

\begin{proposition} \label{lap}
Let $T=\mathbb{R}$, let the null convergence condition (\ref{nec}) hold, and let $u$ be a global weak KAM solution such that every point $(t,q)\in E$ is the endpoint of a $u$-calibrated curve defined on $(-\infty,t]$ (by Prop.\ \ref{glob2}(c) this is the case if the E.-L.\ flow is complete, e.g.\ if $Q$ is compact). Then $u$ satisfies almost everywhere on $E$
\begin{equation} \label{lpc}
\Delta_{a_t}u\le \textrm{div}_{a_t} b_t-\p_t \ln \sqrt{\det a_t}\, ,
\end{equation}
where $\Delta_{a_t} u$ is the trace of the second differential of $u(t,\cdot)$ in the sense of Alexandrov.
\end{proposition}

\begin{proof}
By Prop.\ \ref{glob2} the graph $W$ of $u$ is an achronal future horizon of $(M,g)$ whose generators are the light lifts of the maximal $u$-calibrated curves. At an Alexandrov point of $u$ the function $u$ is differentiable, thus by Prop.\ \ref{glob2}(b) just one generator ends at the corresponding point of $W$, and by assumption it is the light lift of a calibrated curve defined on $(-\infty,t]$; since $t$ is an affine parameter, this generator is past complete. The time dual of \cite[Theor.\ 15]{minguzzi14d} (see also \cite{chrusciel01}) states that under the null convergence condition the expansion of the future-directed semitangent field of $W$, calculated at the Alexandrov points of a local graphing function, is non-positive almost everywhere; the argument is local along the generator through the given Alexandrov point, and uses only its past completeness and the null convergence condition along it. The function $u$ is a graphing function for $W$, and by the calculation which precedes Prop.\ \ref{ray}, which at the Alexandrov points of $u$ involves just the second order Taylor expansion of $u$, the expansion relative to the semitangent field normalized with $\dd t(n_W)=1$ is given by Eq.\ (\ref{the}). (At an Alexandrov point of $u$ the function $u(t,\cdot)$ has a second order expansion, since the semi-concave function $u$ has at the point a second order expansion in all the variables.)
\end{proof}

The completeness of the E.-L.\ flow is not needed in Prop.\ \ref{lap}, only the past completeness of the calibrated curves. An application to the Newtonian $N$-body problem, whose configuration space is not compact and whose E.-L.\ flow is incomplete because of the collisions, is given in Prop.\ \ref{nbody} below.

The proof of the main theorem compares the system with its time reversal.

\begin{lemma}[Time reversal] \label{rev}
Let $T=\mathbb{R}$ and let $\tilde L(t,q,v):=L(-t,q,-v)$, namely $\tilde a_t=a_{-t}$, $\tilde b_t=-b_{-t}$ and $\tilde V(t,q)=V(-t,q)$. Let $(M,\tilde g)$ be the spacetime of $\tilde L$ and let $\tilde{\mathcal{R}}$ be the corresponding tensor of Prop.\ \ref{ric}. The map $\Psi(t,q,y)=(-t,q,-y)$ is an isometry of $(M,\tilde g)$ onto $(M,g)$ which reverses the time orientation, and
\[
\tilde{\mathcal{R}}(t,q,v)=\mathcal{R}(-t,q,-v).
\]
In particular, $\tilde L$ satisfies the null convergence condition (\ref{nec}) if and only if $L$ does. Moreover, a function $u$ is a weak KAM solution of positive type for $L$ on $E$ if and only if $\tilde u(t,q):=-u(-t,q)$ is a weak KAM solution of negative type for $\tilde L$ on $E$, and $\Psi$ maps the graph of $\tilde u$ onto the graph of $u$.
\end{lemma}

\begin{proof}
Since $\Psi^*\dd t=-\dd t$ and $\Psi^*\dd y=-\dd y$, Eq.\ (\ref{eis}) gives
\[
\Psi^*g=a_{-t}-\dd t\otimes (\dd y+b_{-t})-(\dd y+b_{-t})\otimes \dd t-2V(-t,q)\, \dd t^2=\tilde g .
\]
The map $\Psi_*$ sends the vector field $\p_t-[\tilde V-\frac{1}{2}]\p_y$, which defines the time orientation of $\tilde g$, to $-(\p_t-[V-\frac{1}{2}]\p_y)$, so it reverses the time orientation. Let $\psi(t,q)=(-t,q)$, so that $\pi\circ \Psi=\psi\circ \pi$. By Prop.\ \ref{ric}, $\pi^*\tilde{\mathcal{R}}=Ric_{\tilde g}=\Psi^*Ric_g=\Psi^*\pi^*\mathcal{R}=\pi^*\psi^*\mathcal{R}$, hence $\tilde{\mathcal{R}}=\psi^*\mathcal{R}$. As $\psi_*(\p_t+v)=-(\p_t-v)$ and $\mathcal{R}$ is quadratic, $\tilde{\mathcal{R}}(t,q,v)=\mathcal{R}(-t,q,-v)$, and the claim on the null convergence condition follows from Prop.\ \ref{ncc}. The reversed curve $s\mapsto \gamma(-s)$ has the same $\tilde L$-action as the $L$-action of $\gamma$, thus the least actions are related by $\tilde S((t_0,q_0),(t_1,q_1))=S((-t_1,q_1),(-t_0,q_0))$. Under the substitution $\tilde u(t,q)=-u(-t,q)$ the domination and calibration conditions which define the weak KAM solutions of positive type for $L$ (Sect.\ \ref{s:subsolutions}) become those which define the weak KAM solutions of negative type for $\tilde L$. Finally, $\Psi(t,q,\tilde u(t,q))=(-t,q,u(-t,q))$.
\end{proof}

\begin{theorem} \label{smoa}
Let $Q$ be compact and connected, let $T=\mathbb{R}$, let the Lagrangian be $C^r$, $r\ge 3$, with arbitrary dependence on time, and let the null convergence condition (\ref{nec}) hold. Then
\begin{itemize}
\item[(a)] The volume of $(Q,a_t)$ does not depend on $t$.
\item[(b)] Every global weak KAM solution $u$ is $C^r$, and it is a classical solution of the Hamilton-Jacobi equation on $E$. It is a weak KAM solution of both negative and positive type, and through every point of $E$ there passes exactly one $u$-calibrated curve defined on $\mathbb{R}$.
\item[(c)] The conclusions (a), (b) and (c) of Theorem \ref{rig} hold. In particular, the Riemannian manifolds $(Q,a_t)$ are all isometric and have non-negative Ricci curvature, the system is the geodesic flow of $(Q,a_0)$ observed from a moving frame, and $(M,g)$ splits as the product of $(Q,a_0)$ and the 2-dimensional Minkowski spacetime.
\item[(d)] The global weak KAM solution is unique up to additive constants, and every weak KAM solution of positive type defined on $E$ coincides with it up to an additive constant.
\end{itemize}
\end{theorem}

\begin{proof}
Since $Q$ is compact the E.-L.\ flow is complete (Theorem \ref{jhy}). Let $u$ be a global weak KAM solution, which exists by Theorem \ref{glob}. By Prop.\ \ref{glob2} its graph $W$ is a future horizon whose generators are the light lifts of the maximal $u$-calibrated curves, and these are defined on $(-\infty,t]$. Let $\theta$ be the density, with respect to the induced measure $\mu_W$, of the absolutely continuous part of the expansion of the future-directed semitangent field of $W$ normalized by $\dd t(n_W)=1$. By Prop.\ \ref{lap} and its proof, $\theta$ is given by Eq.\ (\ref{the}) at the Alexandrov points of $u$, and $\theta\le 0$ almost everywhere.

For $t_1<t_2$ let $D=W\cap\{t_1<t<t_2\}$. Since $W$ is the graph of the continuous function $u$ over $E$ and $Q$ is compact, $D$ is open and relatively compact in $W$. Its boundary is contained in the sections $W_{t_i}=W\cap\{t=t_i\}$, which are compact locally Lipschitz hypersurfaces of $W$, so $D$ has finite perimeter. Since $t$ is an affine parameter on the generators, the generators which cross the boundary of $D$ enter $D$ through $W_{t_1}$ and leave it through $W_{t_2}$. Before regularity is established some generators may have their future endpoint inside $D$; this is allowed by the area theorem. The area theorem \cite[Theor.\ 13]{minguzzi14d} is stated for past horizons. We use its time dual version, obtained by reversing the time orientation, under which the expansion changes sign and the future and past boundaries of $D$ are exchanged. With the weight $\varphi=1$ it gives
\[
\vert A(W_{t_2})\vert-\vert A(W_{t_1})\vert\le \int_D\theta\, \mu_W,
\]
with equality if and only if the local graphing functions of $W$ belong to $W^{2,1}$ over $D$. The coordinate $t$ is not a time function of $(M,g)$, as it is constant along the integral lines of $n$, but by Prop.\ \ref{eik} the function $\tau=y-u+\frac{t}{2}$ is a locally Lipschitz time function, and $\tau\vert_W=\frac{t}{2}$. Thus $W_{t_i}=W\cap\tau^{-1}(t_i/2)$ and, for almost every $t_i$, the area $\vert A(W_{t_i})\vert$ is the area of the metric induced by $g$ on $W_{t_i}$ \cite[Prop.\ 9]{minguzzi14d}. By the calculation which precedes Prop.\ \ref{ray} this metric is $a_{t_i}$, and the projection of $W_{t_i}$ onto $Q$ is a bi-Lipschitz homeomorphism, so this area is $\textrm{Vol}(Q,a_{t_i})$. Thus for almost every $t_1<t_2$
\begin{equation} \label{vdec}
\textrm{Vol}(Q,a_{t_2})-\textrm{Vol}(Q,a_{t_1})\le \int_D\theta\, \mu_W\le 0,
\end{equation}
and since $t\mapsto \textrm{Vol}(Q,a_t)$ is continuous, the volume is non-increasing. This is the time dual of the argument which gives the area theorem for event horizons \cite[Theor.\ 23]{minguzzi14d}, with the time function $\tau$ and with the inequality $\theta\le 0$ provided by Prop.\ \ref{lap}; the restriction to almost every $t_i$ is explained in \cite[Remark 13]{minguzzi14d}.

We now apply the same argument to the Lagrangian $\tilde L$ of Lemma \ref{rev}, which satisfies the null convergence condition and the other assumptions of the theorem, and to one of its global weak KAM solutions, which exists by Theorem \ref{glob}. Indeed, the coefficients $\tilde a_t$, $\tilde b_t$, $\tilde V$ are $C^r$ on the same compact manifold $Q$, and Theorem \ref{glob}, Lemma \ref{sco} and Prop.\ \ref{lap} require nothing else: no condition is imposed on the time dependence. We obtain that $t\mapsto \textrm{Vol}(Q,\tilde a_t)=\textrm{Vol}(Q,a_{-t})$ is non-increasing, namely that $\textrm{Vol}(Q,a_t)$ is non-decreasing. This proves (a). We stress that only weak KAM solutions of negative type, for $L$ and for $\tilde L$, have been used. The time reversal acts on the Lagrangian and not on $u$, and the volumes of the slices are data of the mechanical system, independent of the horizon used to estimate them.

By (a) the left-hand side of (\ref{vdec}) vanishes. Thus $\int_D\theta\,\mu_W=0$, so $\theta=0$ almost everywhere on $D$, and the equality case of the area theorem holds. Hence the local graphing functions of $W$ belong to $W^{2,1}$ over $D$, and by the local regularity result \cite[Theor.\ 17]{minguzzi14d} $W\cap\{t_1<t<t_2\}$ is a $C^r$ hypersurface, as in the equality case of \cite[Theor.\ 23]{minguzzi14d}. As $t_1<t_2$ can be chosen arbitrarily outside a null set, $W$ is a $C^r$ hypersurface. Its tangent hyperplanes are lightlike and contain the tangents to the generators, which are transverse to $n$, hence they are transverse to $n$, and by the implicit function theorem $u$ is $C^r$. By Theorem \ref{lkp} $u$ is a classical solution of the Hamilton-Jacobi equation. Since $u$ is differentiable, by Prop.\ \ref{glob2} each point of $E$ is the endpoint of exactly one maximal $u$-calibrated curve, which is the past half of the base characteristic passing through it. As the characteristics are complete and their light lifts are contained in the graph of $u$, they are $u$-calibrated curves defined on $\mathbb{R}$, and $u$ is also a weak KAM solution of positive type. This proves (b), and (c) follows from Theorem \ref{rig}.

Let $\hat u$ be another global weak KAM solution. By (b) and (c) both $u$ and $\hat u$ are $C^r$ and satisfy Eq.\ (\ref{ell}), thus $\Delta_{a_t}(\hat u-u)=0$. Since $Q$ is compact and connected, by the maximum principle \cite{gilbarg83} $\hat u-u=k(t)$. Then $\p_q \hat u=\p_q u$, and the Hamilton-Jacobi equation gives $\p_t k=0$. Finally, let $u_+$ be a weak KAM solution of positive type on $E$. By Lemma \ref{rev} the function $\tilde u(t,q)=-u_+(-t,q)$ is a global weak KAM solution for $\tilde L$, and by (b) applied to $\tilde L$ it is also of positive type. Again by Lemma \ref{rev}, applied to $\tilde L$, whose time reversal is $L$, $u_+$ is a global weak KAM solution for $L$, hence it coincides with $u$ up to an additive constant. This proves (d).
\end{proof}

\begin{remark}[The volume balance] \label{vrem}
On a compact configuration space, inequality (\ref{vdec}) can also be obtained by integrating Eq.\ (\ref{the}) over the slices, and this identifies the defect in the area theorem. For almost every $t$, almost every $q$ is an Alexandrov point of $u$, and there $\Delta_{a_t}u$ is the Alexandrov Laplacian of the semi-concave function $u(t,\cdot)$. The distributional Laplacian of $u(t,\cdot)$ is a measure whose absolutely continuous part has this density and whose singular part $\Delta^s_{a_t}u(t,\cdot)$ is non-positive \cite{cannarsa04}. Its total mass vanishes, as $Q$ is closed. Since moreover $\int_Q \textrm{div}_{a_t}b_t\,\mu_t=0$ and $\int_Q\p_t\ln\sqrt{\det a_t}\,\mu_t=\frac{\dd}{\dd t}\textrm{Vol}(Q,a_t)$, for almost every $t$
\[
\frac{\dd}{\dd t}\,\textrm{Vol}(Q,a_t)=\int_Q\theta\,\mu_t-\vert\Delta^s_{a_t}u(t,\cdot)\vert(Q)\le \int_Q\theta\,\mu_t\le 0 .
\]
This is the non-smooth version of the identity $\frac{\dd}{\dd t}\textrm{Vol}(Q,a_t)=\int_Q\theta\,\mu_t$ of Remark \ref{rayrem}. The defect is the singular part of the Laplacian of $u$. Its jump part is concentrated on the set where several calibrated curves end. Semi-concavity alone does not exclude a Cantor part, but for uniformly convex Hamiltonians depending on $(t,q)$ the spatial gradient of a viscosity solution is $SBV_{loc}$ \cite{bianchini12}, so that the Cantor part vanishes and, by slicing, for almost every $t$ the singular part is carried by the jump set. The argument above uses only the non-positivity of the whole singular part. The mechanism of the proof of Theorem \ref{smoa} is the following. On a horizon of negative type the volume can only decrease, because a portion of the frame of calibrated curves which expands toward the future would develop a caustic in the past. On a horizon of positive type the volume can only increase. The volumes of the slices do not depend on the horizon, so they are constant, and the volume budget leaves no room for a singular part, in particular for creases.
\end{remark}

\begin{theorem} \label{smo}
Let $Q$ be compact and connected, let the Lagrangian be $C^r$, $r\ge 3$, and periodic in time with period $P$, and let us suppose that the null convergence condition (\ref{nec}) holds, namely that $\mathcal{R}(t,q,v)\ge 0$ for every $(t,q,v)$. Let $u$ be a periodic weak KAM solution at the critical level $c(L)$. Then
\begin{itemize}
\item[(a)] $u$ is $C^r$ and it is a classical solution of the Hamilton-Jacobi equation on $E$. It is a weak KAM solution of both negative and positive type, and through every point of $E$ there passes exactly one $u$-calibrated curve defined on $\mathbb{R}$.
\item[(b)] The conclusions (a), (b) and (c) of Theorem \ref{rig} hold. In particular, the Riemannian manifolds $(Q,a_t)$ are all isometric and have non-negative Ricci curvature, their volume does not depend on time, and the spacetime $(M,g)$ splits as the product of $(Q,a_0)$ and the 2-dimensional Minkowski spacetime.
\item[(c)] The periodic weak KAM solution at the critical level is unique up to additive constants.
\item[(d)] The flow $\nu_t$ of the base characteristics satisfies $\nu_{t+P}=\nu_t\circ \psi$, where $\psi:=\nu_P$ is an isometry of $(Q,a_0)$, and for every $e=(t,q)\in E$ and $k\ge 1$
\begin{equation} \label{sak}
S(e,e+kP)+c(L)kP=\frac{1}{2kP}\, d_0\big(q',\psi^{-k}(q')\big)^2, \qquad q'=\nu_t^{-1}(q),
\end{equation}
where $d_0$ is the Riemannian distance of $(Q,a_0)$. As a consequence, the projected Aubry set is the whole classical spacetime, $\mathcal{A}=E$, the spacetime $M_{c(L)}$ is nowhere strongly causal, and it is causal if and only if $\psi$ has no periodic points.
\end{itemize}
\end{theorem}

\begin{proof}
A periodic weak KAM solution is a global weak KAM solution (Sect.\ \ref{s:per}), thus (a), (b) and (c) follow from Theorem \ref{smoa}. A direct proof, which does not use the time reversal, is given in Remark \ref{dir}.

Since $u+ct$ is periodic, $v_{t+P}=v_t$, hence $\nu_{t+P}=\nu_t\circ \nu_P$, and $\psi=\nu_P$ is an isometry of $(Q,a_0)$ because $\psi^*a_0=\psi^*a_P=a_0$. Under the change of gauge $y'=y-u$ the least action changes to $S'(e_0,e_1)=S(e_0,e_1)-[u(e_1)-u(e_0)]$, and in the co-moving frame of Theorem \ref{rig}, where $L'=\frac{1}{2}a_0(v',v')$, we have  $S'=d_0(q_0',q_1')^2/[2(t_1-t_0)]$ with $q_i'=\nu_{t_i}^{-1}(q_i)$. As $u(e+kP)-u(e)=-ckP$ and $\nu^{-1}_{t+kP}=\psi^{-k}\circ \nu_t^{-1}$, we obtain Eq.\ (\ref{sak}). Its right-hand side is bounded by $(\textrm{diam}\, Q)^2/(2kP)$, thus by Prop.\ \ref{aub} every $e$ belongs to $\mathcal{A}$ and $M_c$ is nowhere strongly causal. By Theorem \ref{cau2} $M_c$ is non-causal if and only if the right-hand side of Eq.\ (\ref{sak}) vanishes for some $q'$ and $k$.
\end{proof}

\begin{remark}[Mechanical interpretation] \label{mec}
Theorem \ref{smo} can be read without any reference to the spacetime $(M,g)$. Let $v_t=a_t^{-1}(\p_q u-b_t)$ be the velocity field determined by the smooth weak KAM solution $u$, and let $\nu_t:Q\to Q$, $\nu_0=Id$, be its flow. By (a) the integral curves of $v_t$ are the $u$-calibrated curves: they are solutions of the E.-L.\ equations, defined for all times, and action minimizing between any two of their points. Let us regard them as the trajectories of the points of a reference frame (Sect.\ \ref{gau}), whose points are labeled by $q'\in Q$ through $q=\nu_t(q')$. This frame is made of actual motions of the system, and it has the following properties.
\begin{itemize}
\item[(i)] {\em It is rigid.} By (b), that is by Theorem \ref{rig}(a), $\nu_t^*a_t=a_0$: the distances between the points of the frame do not change in time.
\item[(ii)] {\em It is inertial.} By Theorem \ref{rig}(c) and Theorem \ref{ndp}, in the variables $(t,q')$ the Lagrangian becomes, up to the total time derivative of $u$, which does not affect the dynamics,
\[
L'(t,q',v')=\frac{1}{2}\, a_0(v',v'),
\]
with no potential, no gyroscopic term and no time dependence. Thus, relative to this frame, every motion of the system is a geodesic of the fixed Riemannian manifold $(Q,a_0)$, which has non-negative Ricci curvature, run at constant speed, the points of the frame being the motions at rest. Newton's first law holds in this frame, and all the forces of the original description, namely those due to the potential $V$, to the gyroscopic 1-form $b_t$ and to the time dependence of $a_t$, are inertial forces due to the motion of the frame, up to gauge terms, see Eq.\ (\ref{cla}).
\item[(iii)] {\em It is periodic only as a rigid body.} By (d), $\nu_{t+P}=\nu_t\circ\psi$: after one period the frame occupies the same configuration, but each of its points has been displaced by the isometry $\psi$ of $(Q,a_0)$. A point of the frame describes a periodic motion if and only if it is a periodic point of $\psi$. For instance, for the Lagrangian $L=\frac{1}{2}a(v-K,v-K)$ of Example \ref{rot}, which is autonomous so that any $P>0$ can be taken as a period, we have $c(L)=0$, $u$ is constant, $v_t=K$, the inertial frame is the one which moves with the flow of $K$, and $\psi$ is the flow of $K$ at time $P$.
\item[(iv)] {\em It is unique up to Galilean boosts.} By (c) the frame is uniquely determined by $L$. The Lagrangians $L-\omega$, $\omega$ closed, have the same motions, and by Remark \ref{coh} their critical weak KAM solutions provide further rigid inertial frames. Relative to the first frame the points of such a frame move along geodesics of $(Q,a_0)$ with a velocity field which is a Killing field dual to a closed 1-form, hence parallel, as in Example \ref{boc}. Thus the frames associated with the different cohomology classes differ by uniform motions along the parallel vector fields of $(Q,a_0)$, e.g.\ by the uniform translations of a flat torus.
\end{itemize}
In the language of weak KAM theory, the graph of $p=\p_q u$ is a smooth Lagrangian graph invariant under the E.-L.\ flow, foliated by the calibrated curves, which are the points at rest of the frame. Moreover, as shown in the proof of (d), in the new variables and after the substitution $w\mapsto w-u$ of the initial data, the least action becomes $d_0(q_0',q_1')^2/[2(t_1-t_0)]$, so the Lax-Oleinik semigroup becomes the Hopf-Lax semigroup of the Riemannian manifold $(Q,a_0)$.

In the spacetime picture, the coordinates $(t,q',y')$, $y'=y-u$, are global Rosen coordinates (Sect.\ \ref{ros}) in which the spatial metric $\nu_t^*a_t$ does not depend on time, and this is the splitting of (b). Without the curvature condition, a global $C^2$ solution of the Hamilton-Jacobi equation with complete characteristics still provides a frame made of actual motions, and global Rosen coordinates (Theorem \ref{ndp}), but in general $\nu_t^*a_t$ depends on time: the frame deforms, and the motions are geodesics of a time-dependent metric. The null convergence condition forbids this deformation through the Raychaudhuri equation (\ref{ray1}): a portion of the frame whose volume is decreasing (resp.\ increasing) at some instant would develop a caustic in finite time in the future (resp.\ past), which is impossible for calibrated curves defined for all times; hence the volume is constant, and then Eq.\ (\ref{ray1}) also excludes the shear.
\end{remark}

Since the change of variables of Remark \ref{mec} is explicit, the known properties of the geodesic flows of the compact Riemannian manifolds with non-negative Ricci curvature translate into properties of the original system. We collect some of them.

\begin{corollary}[Dynamics under the null convergence condition] \label{dyn}
Let $Q$ be compact and connected, let the Lagrangian be $C^r$, $r\ge 3$, with arbitrary dependence on $t\in\mathbb{R}$, and let the null convergence condition (\ref{nec}) hold. Let $u$ be a global weak KAM solution, $v_t=a_t^{-1}(\p_q u-b_t)$, $\nu_t$ the flow of $v_t$, $\nu_0=Id$, and $a_0$ as in Theorem \ref{rig}, which applies by Theorem \ref{smoa}. Then:
\begin{itemize}
\item[(a)] {\em (First integrals)} Along every solution $q(t)$ of the E.-L.\ equations the relative speed $\vert \dot q-v_t(q)\vert_{a_t}$ is constant. Moreover, for every Killing vector field $X$ of $(Q,a_0)$ the quantity $a_t\big((\nu_{t})_*X,\dot q-v_t(q)\big)$ is constant.
\item[(b)] {\em (Conjugate points)} The instants $t_0<t_1$ are conjugate along a solution $q$ if and only if $\gamma(t_0)$ and $\gamma(t_1)$ are conjugate along the $a_0$-geodesic $\gamma(t)=\nu_t^{-1}(q(t))$. If no solution has conjugate instants, then $(Q,a_0)$ is flat.
\item[(c)] {\em (Topology)} The universal covering of $(Q,a_0)$ is isometric to $\mathbb{R}^k\times N$, with $N$ compact and simply connected. The first Betti number satisfies $b_1(Q)\le \dim Q$, and in case of equality $(Q,a_0)$ is a flat torus and the solutions are the curves $q(t)=\nu_t(q_0'+tw)$, $q_0'\in Q$, $w\in\mathbb{R}^{\dim Q}$, in flat coordinates.
\item[(d)] {\em (Periodic solutions)} Let the Lagrangian be periodic in time with period $P$ and let $\psi=\nu_P$ be the isometry of $(Q,a_0)$ of Theorem \ref{smo}(d). For $k\ge 1$ the $kP$-periodic solutions are the curves $q(t)=\nu_t(\gamma(t))$, where $\gamma$ is an $a_0$-geodesic such that $\gamma(t+kP)=\psi^{-k}(\gamma(t))$ for every $t$; namely, either $\gamma$ is constant at a fixed point of $\psi^k$, or $s\mapsto \gamma(kPs)$ is a $\psi^{-k}$-invariant geodesic in the sense of Grove \cite{grove73,grove74}. In particular, there is a $P$-periodic solution with non-vanishing relative speed \cite{grove73}, see also \cite{mazzucchelli15}, and if the fundamental group of $Q$ is infinite and abelian and $Q$ is not a circle, there are infinitely many $P$-periodic solutions \cite{mazzucchelli15,macarini17}.
\item[(e)] {\em (Entropy)} In the periodic case, the period map $\Pi$, which sends $(q(0),\dot q(0))$ to $(q(P),\dot q(P))$, preserves the compact sets $\Sigma_s=\{(q,v): \vert v-v_0(q)\vert_{a_0}=s\}$, $s\ge 0$, and its topological entropy on $\Sigma_s$ is $sP\, h_{top}(a_0)$, where $h_{top}(a_0)$ is the topological entropy of the geodesic flow of $(Q,a_0)$ on the unit tangent bundle.
\end{itemize}
\end{corollary}

\begin{proof}
By Theorem \ref{rig}(c) and Remark \ref{mec} the solutions are the curves $q(t)=\nu_t(\gamma(t))$, with $\gamma$ an affinely parametrized $a_0$-geodesic, and $\dot q=(\nu_{t})_*\dot \gamma+v_t(q)$. Since $\nu_t:(Q,a_0)\to(Q,a_t)$ is an isometry, $\vert \dot q-v_t(q)\vert_{a_t}=\vert\dot\gamma\vert_{a_0}$ and $a_t((\nu_{t})_*X,\dot q-v_t(q))=a_0(X,\dot\gamma)$, which are constant along geodesics, the latter by Noether's theorem. This proves (a). The map $\gamma\mapsto \nu_t\circ\gamma$ sends the one-parameter families of geodesics to the one-parameter families of solutions, and the Jacobi field $J$ of the former to the Jacobi field $(\nu_{t})_*J$ of the latter, which vanishes exactly where $J$ does. If $(Q,a_0)$ has no conjugate points, its total scalar curvature is non-positive and vanishes only if the metric is flat \cite{green58}; as $Ric_{a_0}\ge 0$, the metric is flat. This proves (b). The splitting in (c) is the Cheeger-Gromoll theorem \cite{cheeger71}. By Example \ref{boc} every cohomology class of degree one contains a parallel 1-form, thus $b_1(Q)\le \dim Q$, and in case of equality there is a parallel coframe, so that $a_0$ is flat and $Q$, a compact flat manifold with maximal first Betti number, is a torus; the geodesics are then the affinely parametrized straight lines.

(d). By Theorem \ref{smo}(d), $\nu_{t+kP}=\nu_t\circ\psi^k$, thus $q(t+kP)=q(t)$ if and only if $\nu_t(\psi^k(\gamma(t+kP)))=\nu_t(\gamma(t))$, that is, $\gamma(t+kP)=\psi^{-k}(\gamma(t))$. The isometry $\psi^{-1}$ is homotopic to the identity through the diffeomorphisms $\nu_t^{-1}$, $t\in[0,P]$, thus by a theorem of Grove \cite[Theor.\ 3.7]{grove73}, see also \cite{mazzucchelli15}, it admits an invariant geodesic, on which it acts as a non-trivial translation \cite{grove74}. Its reparametrization $\gamma$ with $\gamma(t+P)=\psi^{-1}(\gamma(t))$ has non-vanishing speed, so that the corresponding solution is $P$-periodic with non-vanishing relative speed. If the fundamental group of $Q$ is infinite and abelian and $Q$ is not a circle, the isometry $\psi^{-1}$, being homotopic to the identity, has infinitely many invariant geodesics by the theorems of Mazzucchelli \cite[Theor.\ 1.1]{mazzucchelli15} (non-cyclic fundamental group) and of Macarini and Mazzucchelli \cite{macarini17} (fundamental group isomorphic to $\mathbb{Z}$), and distinct invariant geodesics give distinct $P$-periodic solutions.

(e). Since $u+c(L)t$ is periodic, $v_{t+P}=v_t$. Let $F(q,v)=(q,v-v_0(q))$ and let $G_t$ be the geodesic flow of $(Q,a_0)$ on $TQ$. As $\nu_0=Id$ and $q(P)=\psi(\gamma(P))$, $\dot q(P)-v_0(q(P))=\psi_*\dot\gamma(P)$, we have $F\circ\Pi\circ F^{-1}=\psi_*\circ G_P$. The differential $\psi_*$ is an isometry of the Sasaki metric of $a_0$ which commutes with $G_t$, so that $(\psi_*\circ G_P)^j=\psi_*^j\circ G_{jP}$ and the Bowen distances $\max_{0\le j<n} d\big((\psi_*G_P)^jx,(\psi_*G_P)^jx'\big)$ coincide with those of $G_P$. Thus $\psi_*\circ G_P$ and $G_P$ have the same topological entropy on every compact set invariant under both maps, such as $F(\Sigma_s)=\{\vert w\vert_{a_0}=s\}$. For $s>0$, $G_P$ on $F(\Sigma_s)$ is conjugate by the rescaling $w\mapsto w/s$ to $G_{sP}$ on the unit tangent bundle, whose entropy is $sP\,h_{top}(a_0)$. For $s=0$, $F(\Sigma_0)$ is the zero section, on which $G_P$ is the identity, so that $F\circ\Pi\circ F^{-1}$ acts there as the isometry $\psi$, whose entropy vanishes.
\end{proof}

\begin{remark}
The null convergence condition does not exclude a complicated dynamics. Knieper and Weiss constructed metrics of positive curvature on $S^2$ whose geodesic flow has positive topological entropy \cite{knieper94}, and Contreras and Paternain proved that positive topological entropy is $C^2$-open and dense among the metrics on $S^2$ \cite{contreras02}, hence also within the open set of the metrics of positive curvature. By Corollary \ref{dyn}(e), for such a positively curved metric $a_0$ with positive entropy every time-periodic system obtained from $L=\frac{1}{2}a_0(v,v)$ by a change of frame which is periodic up to an isometry has a period map with positive entropy, although it satisfies the null convergence condition. The existence of a $P$-periodic solution in (d) follows also from the direct method of the calculus of variations, as for every time-periodic Tonelli Lagrangian; what Corollary \ref{dyn} adds is the identification of all the periodic solutions with the invariant geodesics of the isometries $\psi^{k}$, so that the multiplicity results for isometry-invariant geodesics apply \cite{grove78,mazzucchelli15,macarini17}, as in the last statement of (d).
\end{remark}

\begin{remark}[Direct proof in the periodic case] \label{dir}
In the periodic case points (a)--(c) of Theorem \ref{smo} can be proved without the time reversal, by applying the area theorem to a compact horizon. Let $c=c(L)$. By Theorem \ref{hor} the projection $\bar W$ of the graph of $u$ is a compact achronal future horizon of $(M_c,g)$ with past complete generators. The null convergence condition holds on $M_c$ since $\pi_c:M\to M_c$ is a local isometry, and the metric is $C^3$. By \cite[Theor.\ 15]{minguzzi14d}, in its time dual version, $\theta\le 0$ almost everywhere. By the area theorem \cite[Theor.\ 13]{minguzzi14d} applied to the domain $D=\bar W$, which is compact and has empty boundary,
\[
\int_{\bar W}\theta\, \mu_{\bar W}\ge 0 ,
\]
with equality if and only if the singular part of the expansion vanishes, i.e.\ the local graphing functions belong to $W^{2,1}$. Thus $\theta=0$ almost everywhere, the equality case holds, and by \cite[Theor.\ 17]{minguzzi14d} $\bar W$ is a $C^r$ hypersurface. One then concludes as in the proof of Theorem \ref{smoa}. This is the time dual of the argument of \cite[Theor.\ 18]{minguzzi14d} for compact Cauchy horizons. Alternatively, one can use Theorem \ref{hor}, according to which $\bar W$ is a compact Cauchy horizon for the metric $g'$ given by Eq.\ (\ref{gpr}), and the theorem on the smoothness of compact Cauchy horizons \cite{larsson14,minguzzi14d}. In fact, $g'$ and $g$ coincide on $\bar W$ with their derivatives up to order $r$, so $g'$ satisfies the null convergence condition on $\bar W$, and in the proof of \cite[Theor.\ 18]{minguzzi14d} this condition is used only at the points of the horizon. In this second approach the completeness of the generators could also be inferred from the general result \cite{minguzzi14b}, while here it is a consequence of the completeness of the E.-L.\ flow.

In the periodic case the compactness of the horizon replaces the comparison with the time-reversed system. Point (d) of Theorem \ref{smo}, by contrast, has no analogue for arbitrary time dependence. It describes the motion of the rigid frame over one period, it computes the least action between points separated by an integer number of periods, and it gives the Aubry set and the causal properties of the critical quotient $M_{c(L)}$. Theorem \ref{smoa}(d) also shows that under the null convergence condition a periodic system has no global weak KAM solutions other than the periodic ones at the critical level, up to additive constants. Indeed the latter exist by Theorem \ref{crit} and they are global weak KAM solutions.
\end{remark}

\begin{remark}[Surface gravity] \label{sgr}
Let $\mathcal{H}$ be a $C^2$ null hypersurface transverse to $n$, and let $Z$ be its future-directed lightlike tangent vector field normalized by $g(Z,n)=-1$, that is, $\dd t(Z)=1$. As $\mathcal{H}$ is generated by lightlike geodesics, $\nabla_ZZ=\kappa Z$ for a function $\kappa$, the surface gravity of $Z$, and since $n$ is parallel
\[
-\kappa=g(\nabla_ZZ,n)=Z(g(Z,n))-g(Z,\nabla_Zn)=0 .
\]
Thus the normalization provided by $n$ has vanishing surface gravity, which is another way of stating that $t$ is an affine parameter of the lightlike geodesics transverse to $n$. In particular, the compact horizons $\bar W$ of Theorem \ref{hor}, which have complete generators, are degenerate in the sense of \cite{hounnonkpe25}, and their surface gravity can be normalized to the constant zero wherever they are $C^2$. Without the null convergence condition they are in general only Lipschitz. Under the null convergence condition (\ref{nec}), by Theorems \ref{smo} and \ref{rig}(c), in the coordinates $(t,q',y')$ of the map $\Phi$ we have $g=a_0-\dd t\otimes \dd y'-\dd y'\otimes \dd t$, the horizon is $\{y'=0\}$, $Z=\p_t$ is a parallel, hence Killing, null vector field, and the generator of the deck group of $M_{c(L)}$ acts as $(t,q',y')\mapsto (t+P,\psi^{-1}(q'),y')$. Thus $Z$ descends to $M_{c(L)}$, and $\bar W$ and its translates are compact degenerate Killing horizons, generated by a parallel null vector field. This is the type of conclusion conjectured by Moncrief and Isenberg for the compact Cauchy horizons of vacuum spacetimes \cite{moncrief83}. For general compact Cauchy horizons, the null energy condition is known to imply that the surface gravity can be normalized to a constant, in both the degenerate and the non-degenerate case \cite{hounnonkpe25}; here, thanks to the parallel vector $n$, the normalization is explicit and the horizons are Killing. The generator flow of $\bar W$ is the suspension of the isometry $\psi$, an isometric Riemannian flow in the sense of \cite{hounnonkpe25b}, and its closed orbits correspond to the periodic points of $\psi$, in agreement with Theorem \ref{smo}(d).
\end{remark}

\begin{example}[Degenerate Killing Cauchy horizons under the null convergence condition] \label{kch}
Let $(Q,a)$ be compact with $Ric_a\ge k\, a$, $k>0$, and let $K$ be a Killing vector field of $(Q,a)$ with flow $\zeta_t$. The Lagrangian $L=\frac{1}{2}\, a(v-K,v-K)$ of Example \ref{rot} satisfies the null convergence condition, and being autonomous it is periodic with any period, say $P=1$. The function $u=0$ is a classical solution of the Hamilton-Jacobi equation, with $v_t=K$, hence a periodic weak KAM solution, so that $c(L)=0$ by Theorem \ref{crit}(a), $\nu_t=\zeta_t$, $\psi=\zeta_1$ and $a_0=a$. In the coordinates of Theorem \ref{rig}(c), $M_0=(\mathbb{R}\times Q\times \mathbb{R})/\mathbb{Z}$, where the generator of $\mathbb{Z}$ acts as $(t,q',y')\mapsto (t+1,\psi^{-1}(q'),y')$, and $g=a-\dd t\otimes \dd y'-\dd y'\otimes \dd t$. Let $h(y')=\frac{k}{3}\, e^{1/y'}$ for $y'<0$ and $h(y')=0$ for $y'\ge 0$, a smooth non-negative function with $h''\le k$, and let us choose $\chi=h(y')$ in Theorem \ref{hor}(c), which is possible as $y'=\bar f$. Then
\[
g'=a-\dd t\otimes \dd y'-\dd y'\otimes \dd t+2h(y')\,\dd t^2 .
\]
\begin{itemize}
\item[(i)] $\p_t$ is a Killing vector field of $g'$ which is null on $\bar W=\{y'=0\}$, and the differential of $g'(\p_t,\p_t)=2h(y')$ vanishes on $\bar W$. Thus $\bar W$ is a degenerate Killing horizon. The Killing field is spacelike in the region $\{y'<0\}$.
\item[(ii)] For $s>0$, $\Sigma_s=\{y'=-s\}$ is a smooth compact spacelike hypersurface, as $g'^{-1}(\dd y',\dd y')=-2h(-s)<0$, and $\bar W=H^{+}(\Sigma_s)$ by Theorem \ref{hor}(c).
\item[(iii)] The Ricci tensor of the 2-dimensional factor $g_2=-\dd t\otimes \dd y'-\dd y'\otimes \dd t+2h(y')\dd t^2$ is $h''(y')\, g_2$. For a $g'$-null vector $X=X_Q+X_2$ we have $g_2(X_2,X_2)=-a(X_Q,X_Q)$, thus
\begin{align*}
Ric_{g'}(X,X)&=Ric_a(X_Q,X_Q)-h''(y')\, a(X_Q,X_Q)\\
&\ge (k-h'')\,a(X_Q,X_Q)\ge 0,
\end{align*}
that is, $(M_0,g')$ satisfies the null convergence condition.
\item[(iv)] The closed generators of $\bar W$ correspond to the periodic points of $\psi$. For the round $S^3\subset\mathbb{C}^2$ ($k=2$) and $\zeta_t(z_1,z_2)=(e^{i\alpha t}z_1,e^{i\beta t}z_2)$, with $\alpha/2\pi$ and $\beta/2\pi$ irrational, $\psi$ has no periodic points and no generator is closed. For the round $S^2$ and a rotation by an angle incommensurable with $2\pi$, which gives a 4-dimensional spacetime, the generators over the two poles are closed, while the others are not.
\end{itemize}
Thus, under the null convergence condition, compact Cauchy horizons of compact spacelike hypersurfaces can be degenerate Killing horizons, possibly with no closed generator. The positivity of the Ricci curvature cannot be dropped in this construction: since $h$ vanishes on $[0,+\infty)$ and is positive on $(-\infty,0)$, $h''$ is positive somewhere, and there the null convergence condition requires $Ric_a(w,w)\ge h''\, a(w,w)>0$ for every $w\ne 0$. In particular, it fails if $Ric_a=0$, consistently with the expectation of Moncrief and Isenberg that degenerate compact Cauchy horizons do not exist in vacuum \cite{moncrief18}.
\end{example}

\begin{remark}[Cohomology classes] \label{coh}
Let $\omega$ be a closed 1-form on $Q$. The Lagrangians $L$ and $L-\omega$ have the same E.-L.\ flow and, by Prop.\ \ref{ric}, the same tensor $\mathcal{R}$, since $F$ and $\dot b$ do not change. Thus under the assumptions of Theorem \ref{smoa} the global weak KAM solutions, and in the periodic case those of Theorem \ref{smo}, are $C^r$ for all the Lagrangians $L-\omega$, namely for all the cohomology classes $[\omega]\in H^1(Q,\mathbb{R})$ which label the weak KAM solutions in Mather's theory, where $\alpha([\omega])=c(L-\omega)$ \cite{mather91,contreras99}.
\end{remark}

\begin{example}[Bochner's theorem] \label{boc}
Let $(Q,a)$ be a compact connected Riemannian manifold with $Ric_a\ge 0$, let $\omega$ be a closed 1-form and let $L=\frac{1}{2}a(v,v)-\omega(v)$, so that $\mathcal{R}=Ric_a(v,v)\ge 0$. By the autonomous weak KAM theorem and Theorem \ref{smo} there is a $C^r$ function $u_-:Q\to \mathbb{R}$ such that $u=u_--c(L)t$ solves the Hamilton-Jacobi equation and such that the vector field $v=a^{-1}(\cdot, \dd u_-+\omega)$ satisfies $L_v a=0$, see Eq.\ (\ref{kil}). Thus the 1-form $\dd u_-+\omega$ is closed and dual to a Killing vector field, which implies that it is parallel. We have recovered Bochner's theorem: on a compact Riemannian manifold with non-negative Ricci curvature every cohomology class of degree one contains a parallel form (the harmonic representative). The Hamilton-Jacobi equation shows that $\alpha([\omega])=c(L)=\frac{1}{2}\vert \dd u_-+\omega\vert^2$, namely Mather's $\alpha$ function is half the squared norm of the harmonic representative. The weak KAM solutions for the other periods and for the frames of Remark \ref{cls} provide time-dependent generalizations.
\end{example}

We end this subsection with an application of Prop.\ \ref{lap} in a non-compact setting. It shows that the null convergence condition carries information also when the rigidity of Theorem \ref{smo} is not available.

\begin{proposition}[Newtonian $N$-body problem] \label{nbody}
Let $Q=\Omega\subset (\mathbb{R}^d)^N$, $d\ge 3$, be the non-collision configuration space of the Newtonian $N$-body problem of Subsection \ref{yak}, endowed with the mass metric $a$, and let $b=0$ and $V=-U$, where $U=\sum_{i<j} m_im_j/\Vert \vec{r}_i-\vec{r}_j\Vert$. Let $h\ge 0$ and let $u:\Omega\to \mathbb{R}$ be a function such that $\tilde u(t,q)=u(q)-ht$ is a global weak KAM solution and such that every $q\in \Omega$ is the endpoint of a $\tilde u$-calibrated curve defined on $(-\infty,0]$. Then $u$ is superharmonic with respect to the mass metric: $\Delta_a u\le 0$ almost everywhere, in the sense of Alexandrov, and in the sense of distributions on $\Omega$. In particular, $u$ satisfies the mean value inequality on the balls of the mass metric contained in $\Omega$ and the strong minimum principle.
\end{proposition}

\begin{proof}
The mass metric $a$ is flat and $\Delta_a=\sum_j m_j^{-1}\Delta_{\vec r_j}$. Since in $\mathbb{R}^d$ the Laplacian of $1/r$ is $-(d-3)\,r^{-3}$, each term $1/\Vert \vec{r}_i-\vec{r}_j\Vert$ is superharmonic in $\vec r_i$ and in $\vec r_j$ away from the collisions, and harmonic for $d=3$. Thus $\mathcal{R}=\Delta_a V=-\Delta_a U\ge 0$ on $\Omega$, and the null convergence condition holds (Prop.\ \ref{ncc}); for $d=3$, $(M,g)$ satisfies the vacuum Einstein equations (Subsection \ref{yak}). As the system is autonomous, the $\tilde u$-calibrated curves ending at $(t,q)$ are the time translates of those ending at $(0,q)$, so the assumption of Prop.\ \ref{lap} holds. As $a$ is constant and $b=0$, Eq.\ (\ref{lpc}) reads $\Delta_a u\le 0$ almost everywhere in the sense of Alexandrov. The function $u$ is locally semi-concave (Prop.\ \ref{glob2}(b)), thus its distributional Hessian is a matrix-valued measure whose absolutely continuous part is given by the Alexandrov Hessian and whose singular part is negative semi-definite \cite{cannarsa04}. Taking the trace with respect to $a$ we obtain $\Delta_a u\le 0$ in the sense of distributions. Since $\Delta_a$ is, up to a linear change of coordinates, the Euclidean Laplacian, the continuous function $u$ is superharmonic, and the last statement follows from the classical properties of superharmonic functions.
\end{proof}

Functions $u$ as in Prop.\ \ref{nbody} exist. For $h=0$ they are provided by the weak KAM theorem for the $N$-body problem of Maderna \cite{maderna12}, and for $h>0$ by the horofunctions and Busemann functions constructed by Maderna and Venturelli \cite{maderna20}, who have recently proved that two hyperbolic motions with the same limit shape define the same Busemann function \cite{maderna26}. The singular sets of these and related solutions have been studied by Berti, Polimeni and Terracini \cite{berti26}. In the latter case the calibrating curves are hyperbolic motions with prescribed limit shape defined on $[0,+\infty)$; as $L$ is reversible, their time reversals are calibrated curves defined on $(-\infty,0]$. By Marchal's theorem \cite{marchal02,chenciner02} the calibrated curves, being minimizers between any two of their points, are collision-free at interior times, so the corresponding generators remain in $M$. We do not know whether the superharmonicity of these functions had been observed before. It is not a consequence of semi-concavity, which bounds the Hessian from above only up to a positive constant. Physically, it expresses the focusing of the generators of the horizon $\{y=u(q)-ht\}$ due to the attractive nature of gravity. The assumption $d\ge 3$ is essential for the argument: in the plane the potential $1/r$ is subharmonic and the null convergence condition fails (cf.\ Subsection \ref{yak}).

For instance, for the 2-body problem with total mass $M$ and reduced mass $\mu$ (Subsection \ref{yak}), the function $u=-2\mu\sqrt{2Mr}$ of the relative distance $r$ is a classical solution of the Hamilton-Jacobi equation at $h=0$. Its calibrated curves are the radial parabolic motions which arrive from infinity; they are defined on $(-\infty,0]$ and reach the collision only after the time $0$. Here, for $d\ge 3$, $\Delta_a u=-\frac{1}{2}(2d-3)\sqrt{2M}\, r^{-3/2}<0$, the inequality being strict because these generators focus on the collision set.

The null convergence condition, joined with the homogeneity of $U$, also excludes lightlike lines. The proof is a second variation argument which uses completeness only along the motion considered, and it applies to all the homogeneous pair potentials $r^{-\alpha}$ for which the null convergence condition holds, which include the Newtonian potential $r^{2-d}$ of $\mathbb{R}^d$.

\begin{proposition}[No entire minimizers] \label{nline}
Let $d\ge 3$, $N\ge 2$, $0<\alpha\le d-2$, let $\Omega\subset(\mathbb{R}^d)^N$ be the non-collision configuration space endowed with the mass metric $a$, and let
\[
L=\frac{1}{2}\,a(v,v)+U, \qquad U=\sum_{i<j}\frac{m_im_j}{\Vert \vec r_i-\vec r_j\Vert^{\alpha}} .
\]
For $\alpha=1$ this is the setting of Prop.\ \ref{nbody}. The Lagrangian $L$ satisfies the null convergence condition (\ref{nec}), and for $\alpha=d-2$ its Eisenhart spacetime $(M,g)$ is Ricci flat. Let $q:\mathbb{R}\to \Omega$ be a solution of the E.-L.\ equations. Then there is a smooth variation of $q$ with compact support in time along which the second variation of the fixed-time action is negative. In particular, no motion defined on $\mathbb{R}$ minimizes the fixed-time action between any two of its points, and the only complete achronal lightlike geodesics of $(M,g)$ are the integral lines of $n$.
\end{proposition}

\begin{proof}
For $\xi\in C^\infty_c(\mathbb{R},\mathbb{R}^{dN})$ the second variation of $\int L\,\dd t$, $L=\frac{1}{2}a(v,v)+U$, along $q$ is $\mathcal{I}(\xi,\xi)$, where
\[
\mathcal{I}(\xi,\eta)=\int_{\mathbb{R}} \big[a(\dot\xi,\dot \eta)+D^2U(q(t))[\xi,\eta]\big]\,\dd t .
\]
The variations $q+s\xi$ are admissible for small $\vert s\vert$, since $q$ stays at positive distance from the collision set on the support of $\xi$. Suppose by contradiction that $\mathcal{I}(\xi,\xi)\ge 0$ for every $\xi$. Let $e_1,\dots, e_m$, $m=dN$, be a constant $a$-orthonormal basis. As $\Delta_{\mathbb{R}^d}\, r^{-\alpha}=\alpha(\alpha+2-d)\,r^{-\alpha-2}$, as in the proof of Prop.\ \ref{nbody} we get
\[
\sum_{\beta=1}^m D^2U(q)[e_\beta,e_\beta]=\Delta_a U(q)=-\alpha(d-2-\alpha)\sum_{i<j}\frac{m_i+m_j}{\Vert \vec r_i-\vec r_j\Vert^{\alpha+2}}\le 0 ,
\]
which is the null convergence condition (\ref{nec}), since $a$ is flat, $b=0$ and $\varrho=-\Delta_aU$ (Prop.\ \ref{ric}), with equality for $\alpha=d-2$.
Let $\chi\in C^\infty_c(\mathbb{R})$ be equal to $1$ on $[-1,1]$, let $\chi_R(t)=\chi(t/R)$ and $C=\int \dot\chi^2\,\dd s$. Then
\[
0\le \sum_{\beta=1}^m \mathcal{I}(\chi_R e_\beta,\chi_R e_\beta)=\frac{mC}{R}+\int_{\mathbb{R}} \chi_R^2\, \Delta_a U(q(t))\, \dd t\le \frac{mC}{R},
\]
so that $0\le \mathcal{I}(\chi_R e_\beta,\chi_R e_\beta)\le mC/R$ for every $\beta$. Let $w\in C^\infty_c(\mathbb{R},\mathbb{R}^{dN})$. As $\mathcal{I}$ is non-negative, $\mathcal{I}(\chi_R e_\beta,w)^2\le \mathcal{I}(\chi_R e_\beta,\chi_R e_\beta)\,\mathcal{I}(w,w)\to 0$ as $R\to +\infty$. But for large $R$, $\chi_R=1$ on the support of $w$, so that $\mathcal{I}(\chi_R e_\beta,w)=\int D^2U(q(t))[e_\beta,w(t)]\,\dd t$ does not depend on $R$. Thus this integral vanishes for every $w$ and $\beta$, and by continuity $D^2U(q(t))=0$ for every $t$. This is impossible, since differentiating $U(\lambda q)=\lambda^{-\alpha}U(q)$ twice at $\lambda=1$ gives $D^2U(q)[q,q]=\alpha(\alpha+1)U(q)>0$.

If $q$ minimized the fixed-time action on every compact interval, then $\mathcal{I}$ would be non-negative, which we have excluded. Finally, by Theorem \ref{kdf} a complete lightlike geodesic which is not an integral line of $n$ has $t$ as affine parameter, so it is the light lift of a motion defined on $\mathbb{R}$. If it were achronal, by Eq.\ (\ref{cft}) the increment of $y$ between any two of its points would equal the least action, namely its projection would minimize the fixed-time action between any two of its points.
\end{proof}

\begin{remark}
In spacetime terms, along a lightlike geodesic transverse to $n$ the Jacobi equation reduces to $\ddot \xi=a^{-1}D^2U(q(t))\,\xi$ on the screen space, as follows from the second variation above. The trace of the curvature term is $Ric(k,k)=-\Delta_a U\ge 0$, while by $D^2U(q)[q,q]=\alpha(\alpha+1)U(q)>0$ the curvature term never vanishes, so that the null generic condition holds at every point. The proof above is the index form version of the classical theorem according to which, under the null convergence condition, a complete lightlike geodesic satisfying the generic condition at some point has conjugate points, and hence is not achronal \cite[Sect.\ 4.4]{hawking73}. Only the completeness of that geodesic is used, not that of $(M,g)$.

A free time minimizer, which has energy $0$, minimizes also the fixed-time action, and so does a motion of energy $h>0$ which minimizes the free-time action of $L+h$, namely, up to reparametrization, a line of the Jacobi metric $a_J=2(h+U)\,a$ of Sect.\ \ref{mau}, provided its physical-time parametrization is defined on $\mathbb{R}$. Thus for $d\ge 3$ Prop.\ \ref{nline} contains the theorem of da Luz and Maderna on the absence of free time minimizers defined on $\mathbb{R}$ \cite{daluz14}, which Assouline and Mazzucchelli have recently reproved for $d\ge 3$, and extended to more general pair potentials, through a splitting theorem for the Jacobi-Maupertuis metric \cite{assouline26}. It also shows that the Jacobi metrics of non-negative energy have no lines whose physical-time parametrization is defined on $\mathbb{R}$. This restriction cannot be dropped in general, since completeness in the Jacobi arc length does not imply completeness in physical time. For instance, for $N=2$, $d=4$, $m_1=m_2=1$, $\alpha=2$ and $h=0$, let $R=(\vec r_1+\vec r_2)/2$ and $\vec r_1-\vec r_2=r\omega$, $\omega\in S^3$. Then
\[
a_J=2U a=\frac{4}{r^2}\,\vert \dd R\vert^2+\frac{\dd r^2}{r^2}+\dd\omega^2 ,
\]
so that the curves $R=0$, $\omega=\omega_0$, $r=e^s$, $s\in\mathbb{R}$, are lines, as every curve joining the levels $r=r_1$ and $r=r_2$ has length at least $\vert\ln (r_2/r_1)\vert$. However, the corresponding motions satisfy $r^2=4(t-t_*)$, hence they start from a collision at the finite time $t_*$. For $0<\alpha<2$ the absence of Jacobi lines would follow from the completeness in physical time of the motions which parametrize them, which requires a separate argument. The argument uses only the homogeneity of $U$, its positivity and the inequality $\Delta_aU\le 0$, so it applies to every positive potential, homogeneous of negative degree, which satisfies the null convergence condition. For $\alpha>d-2$, in particular for the Newtonian potential in the plane, the argument fails, since then $\Delta_aU>0$; the theorem of \cite{daluz14} holds also in the plane.
\end{remark}

\subsection{Comparison with the regularity results of weak KAM theory}

It is instructive to compare Theorems \ref{smoa} and \ref{smo} with what is known on the regularity of the solutions of the Hamilton-Jacobi equation in weak KAM theory. We consider Tonelli Lagrangians on compact manifolds, in the autonomous or time-periodic case.

The weak KAM solutions always exist at the critical level. They are Lipschitz and semi-concave, but generically they are not $C^1$. Fathi proved that a $C^1$ solution is automatically $C^{1,1}$ \cite{fathi03}, and the graph of its differential is then a Lipschitz Lagrangian graph invariant under the Hamiltonian flow (see also \cite{arnaud08}). No general mechanism improves this regularity. If one is content with subsolutions the situation is better: Fathi and Siconolfi \cite{fathi04} proved the existence of $C^1$ critical subsolutions, which are strict outside the Aubry set, and Bernard \cite{bernard07} proved the existence of $C^{1,1}$ critical subsolutions, a regularity which is optimal in general; the time-periodic case was treated by Massart \cite{massart07} (see also \cite{bernard07b,contreras13}). These subsolutions correspond, in our language, to $C^1$ or $C^{1,1}$ spacelike or lightlike graphs in the critical spacetime $M_{c(L)}$, namely to regular hypersurfaces which are intersected at most once by the timelike curves. As for the solutions, their existence in the smooth category is equivalent to the existence of smooth invariant exact Lagrangian graphs, which is established by KAM theory for the cohomology classes with Diophantine rotation vector in nearly integrable systems, and which is known to fail under generic perturbations for the other classes. On the opposite side, the existence of smooth solutions for all the cohomology classes is a very strong condition, which for Riemannian metrics on the torus is connected with the absence of conjugate points and hence with flatness by the theorems of Hopf \cite{hopf48} and Burago and Ivanov \cite{burago94}. On other manifolds the two properties are unrelated. For instance, on $S^2\times S^1$ with the product metric and $L=\frac{1}{2}\vert v\vert^2$ every cohomology class has a parallel representative $c\,\dd\vartheta$, so the weak KAM solutions can be chosen constant, while the sphere factor has conjugate points.

Theorem \ref{smoa} is of a different nature. Its assumption is neither a perturbative one nor a condition on the dynamics, but a pointwise inequality on the data of the Lagrangian, the inequality $\mathcal{R}\ge 0$, which is preserved by the addition of closed 1-forms. Its conclusion is that all the global weak KAM solutions, for all the cohomology classes and for arbitrary time dependence, are as regular as the Lagrangian, and that the system is the geodesic flow of a Riemannian manifold with non-negative Ricci curvature observed from a moving frame. The mechanism behind the proof is not that of weak KAM theory. It is the mechanism of the splitting theorems of Riemannian and Lorentzian geometry, in which a convexity property of geometric origin (semi-concavity, i.e.\ the existence of upper support functions given by the action from the points of a calibrated curve) is joined with a differential inequality with the opposite sign due to curvature (Prop.\ \ref{lap}), and with a global argument which forces the equality. In the Cheeger-Gromoll theorem \cite{cheeger71} the global argument is the maximum principle applied to the sum of two Busemann functions, in Galloway's null splitting theorem \cite{galloway00} it is the maximum principle for null hypersurfaces applied to $\p I^{+}(\eta)$ and $\p I^{-}(\eta)$, and here it is the area theorem applied to the horizons of the system and of its time reversal, which share the volumes of the slices (in the periodic case, to a single horizon which is compact thanks to the periodicity and to the choice of the critical value). In this analogy the global weak KAM solutions are the Busemann functions of the classical spacetime (cf.\ Theor.\ \ref{glob}), the conjugate pairs $(u_-,u_+)$ are analogous to the two horizons of a lightlike line, to which they correspond exactly in simple cases (Example \ref{nms}), and the regularity follows from the coincidence $u_-=u_+$. In the autonomous setting of Example \ref{boc}, a Laplacian estimate in the barrier sense for the weak KAM solutions of $L=\frac{1}{2}\vert v\vert^2-\omega(v)$ on closed manifolds with non-negative Ricci curvature, and the characterization of the harmonic forms $\omega$ as those for which the weak KAM solutions are constant, have been obtained recently by Cheng and Wei \cite{chengwei26}. Apart from this, we are not aware of analogous statements in the weak KAM literature; it would be interesting to obtain a proof of Theorem \ref{smoa} which uses only methods of that theory, and to understand whether intermediate regularity results can be obtained from weaker curvature bounds, e.g.\ from a lower bound on $\mathcal{R}$.

\section{Stationary spacetimes and the synchronization of clocks} \label{s:sync}

In this section we apply the previous results to stationary spacetimes, and in particular to a problem of clock synchronization. We shall see that the stationary spacetimes over a compact space are the Maupertuis spacetimes of Section \ref{mau} associated with suitable autonomous Lagrangians, that the threshold between their good and bad causal behavior is Ma\~n\'e's critical value, and that the projections of the horizons constructed in Subsection \ref{s:glob} and Section \ref{s:per} provide a distinguished class of synchronizations of the stationary clocks, the eikonal synchronizations. In the companion paper \cite{minguzzisync} the space of clocks, its optical metric and the Killing time are denoted $(S,h,t)$ instead of $(Q,a,\tau)$, while the drift form $b$ and the Sagnac ratio $\beta$ have the same meaning and sign conventions. No curvature condition is used in Subsections \ref{ss:mauc}-\ref{ss:syn}; the null convergence condition first enters in Subsection \ref{ss:harm}, where it gives the smoothness of the synchronization, and it enters again, for the Lagrangians $L_\chi$, in Subsection \ref{ss:conf}. In Subsection \ref{ss:conf} we consider the representations of a stationary spacetime as a Maupertuis spacetime with a potential, which correspond to its conformal factors, and obtain synchronizations whose slices are orthogonal to observers freely falling in the physical metric.

\subsection{Causality of the Maupertuis spacetimes} \label{ss:mauc}

Throughout this subsection $Q$ is compact and connected and the Lagrangian $L(q,v)=\frac{1}{2}a(v,v)+b(v)-V(q)$ is independent of time. Let $e>\max_Q V$ and let $(\mathcal{M},g^{\mathcal{M}})$ be the Maupertuis spacetime of energy $e$, see Eq.\ (\ref{gmo}), which we shall denote $\mathcal{M}_e$ when we want to stress the dependence on $e$. For a piecewise $C^1$ curve $\gamma$ on $Q$ we consider the Fermat functional of Prop.\ \ref{sqo}
\[
\ell_e(\gamma)=\int \big[\sqrt{a_J(\gamma',\gamma')}+b(\gamma')\big]\, \dd s, \qquad a_J=2(e-V)\,a,
\]
which does not depend on the parametrization, and we define the {\em Fermat pre-distance} $d_e(q_0,q_1)\in [-\infty,+\infty)$ as the infimum of $\ell_e$ over the curves which connect $q_0$ to $q_1$, the constant curve being included if $q_0=q_1$. As $b$ is arbitrary, the Fermat functional need not be positive.

\begin{lemma} \label{fer}
\begin{itemize}
\item[(a)] The point $(\tilde y_1,q_1)$ belongs to $I^{+}((\tilde y_0,q_0))$ if and only if $\tilde y_1-\tilde y_0>d_e(q_0,q_1)$. For every future-directed causal curve $s\mapsto (\tilde y(s),\gamma(s))$ we have $\tilde y_1-\tilde y_0\ge \ell_e(\gamma)$.
\item[(b)] For every $q_0,q_1\in Q$, $d_e(q_0,q_1)=\inf_{T>0}\, [S((0,q_0),(T,q_1))+eT]$, namely $d_e$ is Ma\~n\'e's action potential of the autonomous theory at the energy level $e$.
\item[(c)] $e\ge c(L)$ if and only if $\ell_e(\gamma)\ge 0$ for every closed curve $\gamma$.
\end{itemize}
\end{lemma}

\begin{proof}
(a). Let $x(s)=(\tilde y(s),\gamma(s))$ be a future-directed causal curve. Since $g^{\mathcal{M}}(\p_{\tilde y},x')=-(\tilde y'-b(\gamma'))<0$ and $g^{\mathcal{M}}(x',x')=a_J(\gamma',\gamma')-(\tilde y'-b(\gamma'))^2\le 0$, we have $\tilde y'\ge \sqrt{a_J(\gamma',\gamma')}+b(\gamma')$, with strict inequality if the curve is timelike. Integrating we get the last statement of (a) and the `only if' part of the first. Conversely, if $\tilde y_1-\tilde y_0>d_e(q_0,q_1)$ there is a curve $\gamma:[0,1]\to Q$ connecting $q_0$ to $q_1$ such that $\kappa:=\tilde y_1-\tilde y_0-\ell_e(\gamma)>0$, and the curve $s\mapsto (\tilde y_0+\int_0^s [\sqrt{a_J(\gamma',\gamma')}+b(\gamma')+\kappa]\,\dd s', \gamma(s))$ is timelike, future-directed, and connects the two points.

(b). For every curve $q:[0,T]\to Q$, by the inequality between arithmetic and geometric means,
\begin{equation} \label{amgm}
\int_0^T [L(q,\dot q)+e]\,\dd t=\int_0^T \Big[\frac{1}{2}\,a(\dot q,\dot q)+(e-V)+b(\dot q)\Big]\dd t\ge \ell_e(q),
\end{equation}
with equality if and only if $q$ has energy $e$. Every non-constant curve can be reparametrized so as to have energy $e$ on a finite interval, since $e-V$ is bounded away from zero, and for the constant curve at $q_0$ the left-hand side equals $(e-V(q_0))T$, whose infimum over $T>0$ vanishes. Thus the two infima coincide.

(c). By Theorem \ref{crit}(b), applied with an arbitrary period $P$ (for an autonomous Lagrangian $c(L)$ does not depend on $P$, see Sect.\ \ref{s:aut}), $c(L)=\sup_{q\in Q,\, T>0} -S((0,q),(T,q))/T$. Thus $e\ge c(L)$ if and only if the action of $L+e$ is non-negative on every closed curve (cf.\ Prop.\ \ref{critchar}(c)), which by Eq.\ (\ref{amgm}) and the reparametrization argument of (b) holds if and only if $\ell_e$ is non-negative on every closed curve.
\end{proof}

The next result, the analog for the Maupertuis reduction of Theorem \ref{cau2}, shows that Ma\~n\'e's critical value is also the threshold for the causality of the Maupertuis spacetimes. It sharpens Theorem \ref{stc}: since $c(L)\le \max_Q H(q,0)=\max_Q [V+\frac{1}{2}a^{-1}(b,b)]$ (Sect.\ \ref{s:aut}), the assumption of that theorem implies $e>c(L)$, and its time function $\tilde y$ is obtained by running the argument of point (iii) below with $u=0$, which is dominated by $L+\max_Q H(q,0)$.

\begin{theorem} \label{mcau}
Let $Q$ be compact and connected, let $L$ be time independent and let $e>\max_Q V$.
\begin{itemize}
\item[(i)] $\mathcal{M}_e$ is totally vicious if and only if $e<c(L)$.
\item[(ii)] $\mathcal{M}_e$ is chronological but not distinguishing if and only if $e=c(L)$.
\item[(iii)] $\mathcal{M}_e$ is globally hyperbolic if and only if it is distinguishing, if and only if $e>c(L)$. In this case, for every function $u:Q\to \mathbb{R}$ dominated by $L+c(L)$, for instance for every weak KAM solution, the function $\tilde y-u$ is a time function whose level sets are compact Lipschitz Cauchy hypersurfaces.
\end{itemize}
\end{theorem}

\begin{proof}
Let $c=c(L)$. (iii). Suppose $e>c$, let $u$ be dominated by $L+c$ and let $x(s)=(\tilde y(s),\gamma(s))$ be a future-directed causal curve from $(\tilde y_0,q_0)$ to $(\tilde y_1,q_1)$. Domination gives $u(q_1)-u(q_0)\le \int_0^T [L(q,\dot q)+c]\,\dd t$ for every parametrization $q:[0,T]\to Q$ of $\gamma$, and taking the infimum over the parametrizations (approximating, where $c-V$ vanishes, by curves of energy slightly larger than $c$) we get $u(q_1)-u(q_0)\le \int [\sqrt{2(c-V)}\,\vert \gamma'\vert_a+b(\gamma')]\,\dd s$. By Lemma \ref{fer}(a),
\[
[\tilde y_1-u(q_1)]-[\tilde y_0-u(q_0)]\ge \int \big[\sqrt{2(e-V)}-\sqrt{2(c-V)}\big]\, \vert \gamma'\vert_a \,\dd s\ge m\, \ell_a(\gamma),
\]
where $m=\min_Q [\sqrt{2(e-V)}-\sqrt{2(c-V)}]>0$ and $\ell_a$ is the $a$-length. If $\gamma$ is not constant the right-hand side is positive, while if $\gamma$ is constant then $\tilde y_1>\tilde y_0$. Thus $\tilde y-u$, which is continuous, is a time function. Its level sets are graphs over $Q$, hence compact. Let $x(s)$ be a future inextendible causal curve. If the $a$-length of its projection is infinite, then $\tilde y-u\to +\infty$ by the previous inequality. Otherwise the projection converges to a point of $Q$. The function $f=\tilde y-u$ is increasing along the curve; if it were bounded from above it would converge, and then, by the continuity of $u$, also $\tilde y=f+u(q)$ would converge, so that the curve would converge, contradicting inextendibility. (The coordinate $\tilde y$ itself need not be monotone along causal curves if $b\ne 0$.) In both cases $\tilde y-u\to+\infty$ and, analogously, $\tilde y-u\to-\infty$ in the past direction. Thus each level set is intersected exactly once by every inextendible causal curve, namely it is a Cauchy hypersurface, and $\mathcal{M}_e$ is globally hyperbolic. A function $u$ dominated by $L+c$ exists by the weak KAM theorem \cite{fathi08}.

(i). If $e<c$, by Lemma \ref{fer}(c) there is a closed curve $\gamma_0$ with $\ell_e(\gamma_0)<0$. For every $q_0,q_1\in Q$ we can join $q_0$ to a point of $\gamma_0$, go around $\gamma_0$ $N$ times, and then join $q_1$, thus $d_e(q_0,q_1)=-\infty$ and by Lemma \ref{fer}(a) $I^{+}(x)=\mathcal{M}$ for every $x$.

(ii). If $e=c$, by Lemma \ref{fer}(c) $d_e(q,q)=0$ for every $q$, thus by Lemma \ref{fer}(a) no point belongs to its own chronological future. Suppose that $\mathcal{M}_e$ is distinguishing. As $\p_{\tilde y}$ is a complete timelike Killing vector field, by the theorem of Javaloyes and S\'anchez \cite{javaloyes08} there is a smooth spacelike hypersurface which intersects every integral curve of $\p_{\tilde y}$ exactly once, namely the graph $\tilde y=\alpha(q)$ of a smooth function $\alpha$. The calculation in the proof of Theorem \ref{stc}, with $\tilde y$ replaced by $\tilde y-\alpha$, gives $(g^{\mathcal{M}})^{-1}(\dd \tilde y-\dd\alpha,\dd\tilde y-\dd\alpha)=a_J^{-1}(b-\dd\alpha,b-\dd \alpha)-1$, thus the graph is spacelike if and only if $a_J^{-1}(b-\dd\alpha,b-\dd \alpha)<1$, namely if and only if $H(q,\dd \alpha)<e$ on $Q$. Thus $\max_Q H(q,\dd\alpha)<e$, and Prop.\ \ref{critchar}(d) gives $c\le \max_Q H(q,\dd\alpha)<e$, a contradiction.

The three cases are mutually exclusive and exhaustive, and globally hyperbolic spacetimes are distinguishing, while totally vicious spacetimes are not chronological. This proves all the equivalences.
\end{proof}

\begin{remark}
In the passage from the Eisenhart spacetime to the Maupertuis spacetime, Ma\~n\'e's potential $d_e$ plays the role played by the least action $S$ in $(M,g)$ (Lemma \ref{pcf}) and by $\Phi_c$ in the compactified spacetimes of Sect.\ \ref{s:per} (Remark \ref{mpo}). At the critical energy one can show, arguing as in Prop.\ \ref{aub}, that strong causality fails at the points which project on the projected Aubry set of the autonomous theory. Lemma \ref{fer}(a) also shows that $\mathcal{M}_e$ is causal if and only if it is chronological and there is no non-constant closed curve with $\ell_e=0$, since the light lift of such a curve closes up. For stationary spacetimes the description of the causal relations through the Fermat functional, and the translation of the causal ladder into properties of the associated Randers metric, are well known when the Fermat functional is positive \cite{levicivita18,perlick90,gibbons09,caponio11}. The case in which it is not positive was studied in \cite{herrera21}, where the levels of the causal ladder are characterized in terms of the Fermat pre-distance. The new element here is the identification of the threshold with Ma\~n\'e's critical value, which allows us to bring in the weak KAM theory.
\end{remark}

\subsection{Stationary spacetimes and the Sagnac ratio} \label{ss:sag}

Let $(Q,a)$ be a compact connected Riemannian manifold and let $b$ be a 1-form on $Q$. We consider the stationary spacetime
\begin{equation} \label{stat}
\mathcal{M}=\mathbb{R}\times Q, \qquad \hat g=-(\dd \tau-b)\otimes (\dd \tau-b)+a,
\end{equation}
time oriented by the complete Killing vector field $\p_\tau$, which satisfies $\hat g(\p_\tau,\p_\tau)=-1$. Every spacetime $(N,g)$ which admits a complete timelike Killing vector field $k$ whose flow makes $N$ a principal $(\mathbb{R},+)$-bundle over a compact manifold $Q$ is conformal to a spacetime of this form, $k$ corresponding to $\p_\tau$. Indeed, the bundle is trivial, so that a global section provides a coordinate $\tau$ with $k=\p_\tau$, and for the conformal metric $\hat g=g/\vert g(k,k)\vert$ the 1-form $-\hat g(k,\cdot)$ takes the value $1$ on $k$ and is $k$-invariant, thus it reads $\dd \tau-b$ for a 1-form $b$ on $Q$, while $\hat g+(\dd\tau-b)^2$ is the pullback of a Riemannian metric $a$ on $Q$ (cf.\ the reduction of Sect.\ \ref{mau}, where the Killing vector was spacelike). By a theorem of Harris \cite{harris92}, the flow of a complete timelike conformal Killing vector field of a chronological spacetime is free and proper, thus this is the case whenever $N$ is chronological and its space of orbits is compact; see also \cite{javaloyes08,costaesilva17,minguzzisync}. The conformal factor plays no role in the causal statements that follow, since the causal structure is conformally invariant; it will enter in Subsection \ref{ss:conf}. The integral curves of $\p_\tau$ are the worldlines of the {\em stationary observers}, and the function $\tau$ is a {\em synchronization} of their clocks. The synchronization is not unique: for every function $\alpha$ on $Q$, $\tau'=\tau-\alpha$ is another synchronization, and in the new coordinates the metric has the same form with $b$ replaced by
\[
b'=b-\dd \alpha .
\]
For a future-directed lightlike vector $\dd\tau'(x')=b'(\gamma')+\vert \gamma'\vert_a$, thus with respect to the synchronization $\tau'$ the one-way coordinate speed of light in the $a$-unit direction $\epsilon$ is $1/(1+b'(\epsilon))$, and the anisotropy of the one-way speed of light is measured by $\vert b'\vert_a$. The synchronization $\tau'$ is a time function with timelike gradient if and only if $\vert b'\vert_a<1$ on $Q$, as the slices $\tau'=const.$ have induced metric $a-b'\otimes b'$, equivalently, by Eq.\ (\ref{oms}) below, $\hat g^{-1}(\dd \tau',\dd\tau')=-1+\vert b'\vert_a^2$.

In an unpublished 2011 manuscript, now revised and posted as the companion paper \cite{minguzzisync}, it was proposed to select the synchronization which minimizes the anisotropy functional $\int_Q \vert b'\vert_a^2 \,\dd\mu_a$. By the Hodge decomposition the minimizer exists, is unique up to an additive constant, and is characterized by $\delta b'=0$, which is equivalent to the wave equation $\Box_{\hat g}\tau'=0$ \cite[Sect.\ 5]{minguzzisync}. Indeed, since $\det \hat g=-\det a$ and
\begin{equation} \label{oms}
\hat g^{-1}=-\p_{\tau'}\otimes \p_{\tau'}+a^{cd}(\p_c+b'_c\p_{\tau'})\otimes(\p_d+b'_d\p_{\tau'})
\end{equation}
in the coordinates $(\tau',q)$, we have $\Box_{\hat g}\tau'=\textrm{div}_a b'$. We call it the {\em harmonic synchronization}. It is smooth by elliptic regularity, and the question left open in the unpublished 2011 version of \cite{minguzzisync} is whether it is a time function whenever some synchronization is. We shall see that the answer is negative in general and positive under a curvature condition.

Comparing with Eq.\ (\ref{gmo}) we see that $(\mathcal{M},\hat g)$ is the Maupertuis spacetime, with $\tilde y=\tau$, of the Lagrangian
\begin{equation} \label{lb}
L_b(q,v)=\frac{1}{2}\,a(v,v)+b(v)
\end{equation}
at the energy $e=1/2$. More generally, for $\lambda>0$ the stationary spacetime $(\mathcal{M},\hat g_\lambda)$, $\hat g_\lambda=-(\dd \tau-b)^2+\lambda^2 a$, is the Maupertuis spacetime of $L_b$ at the energy $e=\lambda^2/2$, so the parameter $\lambda$ controls the opening of the light cones, which get narrower as $\lambda$ increases. In particular, $(\mathcal{M},\hat g)$ is the quotient of the Eisenhart spacetime of $L_b$ under the flow of the spacelike Killing vector field $\p_t-\frac{1}{2}\p_y$, see Sect.\ \ref{mau}.

\begin{definition}
The {\em Sagnac ratio} of $b$ is
\[
\beta(b)=\sup_\gamma \frac{\oint_\gamma b}{\ell_a(\gamma)}\in [0,+\infty),
\]
where the supremum is over the non-constant closed piecewise $C^1$ curves and $\ell_a$ denotes the $a$-length.
\end{definition}

The supremum is non-negative because the circulation changes sign under reversal of the orientation, and $\beta(b)$ depends only on the class of $b$ modulo exact forms, namely it does not depend on the synchronization. The quantity $2\oint_\gamma b$ is the Sagnac effect, namely the difference between the coordinate times needed by light to go around $\gamma$ in the two directions, thus $\beta(b)$ is the maximal ratio between half the Sagnac delay and the length of the circuit.

\begin{proposition} \label{sagn}
Let $\beta=\beta(b)$. Then
\begin{itemize}
\item[(a)] $c(L_b)=\frac{1}{2}\beta^2$;
\item[(b)] $\beta=\min_\alpha\, \textrm{ess\,sup}_Q\,\vert b-\dd \alpha\vert_a$, where the minimum is over the Lipschitz functions on $Q$. It is attained by every weak KAM solution $u$ of $L_b$, which satisfies $\vert \dd u-b\vert_a=\beta$ at every point of differentiability;
\item[(c)] $\beta=\inf_\alpha \max_Q \vert b-\dd \alpha\vert_a$, where the infimum is over the smooth functions on $Q$.
\end{itemize}
\end{proposition}

\begin{proof}
(a). Let $\gamma$ be a closed curve of $a$-length $\ell$ and circulation $\oint_\gamma b=-B$. Parametrized with constant speed on $[0,T]$ its action is $\ell^2/(2T)-B$, and among the parametrizations on $[0,T]$ this is the minimum. Thus $\sup_{T>0}-S_T/T$ over the closed curves with the shape of $\gamma$ is $\sup_{T>0}(B-\ell^2/(2T))/T$, which equals $B^2/(2\ell^2)$ if $B>0$ and $0$ otherwise. By the formula for $c(L)$ used in the proof of Lemma \ref{fer}(c), and since constant curves give $c(L)\ge 0$, we obtain $c(L_b)=\frac{1}{2}\beta^2$.

(b). If $\alpha$ is Lipschitz with $\vert \dd \alpha-b\vert_a\le m$ almost everywhere, then $\alpha-\frac{1}{2}m^2t$ is dominated by $L_b$ \cite{fathi08}, thus, as in the proof of Theorem \ref{mcau}, $\alpha(q_1)-\alpha(q_0)\le \int_\gamma [m\vert \gamma'\vert_a+b(\gamma')]$ for every curve $\gamma$. For closed curves this gives $-\oint_\gamma b\le m\,\ell_a(\gamma)$, hence $\beta\le m$. A weak KAM solution $u$ is Lipschitz and satisfies the Hamilton-Jacobi equation $\frac{1}{2}\vert \dd u-b\vert_a^2=c(L_b)$ at its differentiability points \cite{fathi08}, hence $\vert \dd u-b\vert_a=\beta$ almost everywhere, by (a).

(c). By Prop.\ \ref{critchar}(d), $c(L)=\inf_\alpha \max_Q H(q,\dd \alpha)$, the infimum being over the smooth functions, and for $L_b$ we have $H(q,p)=\frac{1}{2}\vert p-b\vert_a^2$.
\end{proof}

Theorem \ref{mcau} now gives the causal classification of the stationary spacetimes over a compact space, which can be compared with the classification of Harris in terms of the weight of the fundamental cocycle \cite{harris15}, see the paragraph after the corollary.

\begin{corollary} \label{scau}
The stationary spacetime (\ref{stat}) is totally vicious if $\beta(b)>1$, chronological but not distinguishing if $\beta(b)=1$, and globally hyperbolic if $\beta(b)<1$. More generally, the same holds for $(\mathcal{M},\hat g_\lambda)$ with $\beta(b)$ replaced by $\beta(b)/\lambda$.
\end{corollary}

\paragraph{Relation with the work of Garfinkle and Harris} \label{harr}
The Sagnac ratio coincides with the weight of the fundamental cocycle introduced by Garfinkle and Harris \cite{garfinkle97} and developed by Harris \cite{harris15}. Indeed, for the spacetime (\ref{stat}) the space of Killing orbits is $Q$, the Killing vector field has unit length, so that the conformal metric of the orbit space used in \cite{garfinkle97,harris15} is $a$, and the fundamental cocycle is the functional $\gamma\mapsto\oint_\gamma b$ on loops. For stationary-complete spacetimes Harris proved that the spacetime is chronologically vicious if the weight is larger than one, and causally continuous if it is smaller than one, in which case it is globally hyperbolic provided the orbit space is complete; if the weight is equal to one, the spacetime is chronological and its position in the causal ladder depends on the behavior of the sequences of loops which realize the weight \cite[Thm.\ 2.18, Cor.\ 2.20]{harris15}. Moreover, Harris proved that future distinction is precisely the condition for the existence of a standard stationary presentation \cite[Thm.\ 2.18, Cor.\ 2.20]{harris15}, namely of a smooth synchronization $\alpha$ with $\vert b-\dd\alpha\vert_a<1$. Over a compact $Q$ this inequality is uniform, so that $\beta(b)\le \max_Q\vert b-\dd\alpha\vert_a<1$. Thus the whole of Corollary \ref{scau}, including the non-distinction for $\beta(b)=1$, follows from Harris's classification and the compactness of $Q$. What the present approach adds is the identification $\beta(b)^2=2c(L_b)$ with Ma\~n\'e's critical value and the characterization of $\beta(b)$ in Prop.\ \ref{sagn}(b)-(c) as the minimal sup norm of the drift form $b-\dd\alpha$ over all the synchronizations. The latter clarifies a point raised by Garfinkle and Harris \cite{garfinkle23}: the weight is invariant under change of synchronization, while the curvature estimates are more conveniently expressed through a synchronization in which the norm of the drift form is bounded by a constant $w<1$; by Prop.\ \ref{sagn}(c), for compact $Q$ such a synchronization can be chosen smooth for every $w\in(\beta(b),1)$, and by Prop.\ \ref{sagn}(b) the value $w=\beta(b)$ is attained by the Lipschitz synchronizations $\tau-u$, $u$ a weak KAM solution.

In the stationary context the Hamilton-Jacobi equation of $L_b$ at the critical level reads
\begin{equation} \label{seik}
\vert \dd u-b\vert_a=\beta(b),
\end{equation}
which we call the {\em synchronization equation}. For a Lipschitz solution, understood in the viscosity sense, it states that in the synchronization $\tau-u$ the anisotropy $\vert b-\dd u\vert_a$ is the same at every point of $Q$ where $u$ is differentiable.

\subsection{The synchronization theorem} \label{ss:syn}

\begin{theorem}[Synchronization theorem] \label{sync}
Let $(Q,a)$ be a compact connected Riemannian manifold, let $b$ be a 1-form on $Q$, let $\beta=\beta(b)$, and let $u$ be a weak KAM solution of the Lagrangian $L_b$ of Eq.\ (\ref{lb}), namely a Lipschitz and semi-concave viscosity solution of the synchronization equation (\ref{seik}). Such a solution always exists, but in general it is not unique up to additive constants (Remark \ref{exist}). Let
\[
\tau_*=\tau-u .
\]
\begin{itemize}
\item[(a)] (Minimal anisotropy) The synchronization $\tau_*$ minimizes the maximal anisotropy of the one-way speed of light: for every synchronization $\tau-\alpha$ with $\alpha$ Lipschitz
\[
\textrm{ess\,sup}_Q\vert b-\dd u\vert_a=\beta\le \textrm{ess\,sup}_Q\vert b-\dd \alpha\vert_a .
\] Moreover, the anisotropy $\vert b-\dd u\vert_a$ equals $\beta$ at every differentiability point of $u$.
\item[(b)] (Causality) If the stationary spacetime (\ref{stat}) is distinguishing, then $\tau_*$ is a time function whose level sets are compact Lipschitz Cauchy hypersurfaces, and the spacetime is globally hyperbolic. More precisely, $\beta<1$ and for every future-directed causal curve from $x_0$ to $x_1$, whose projection on $Q$ is $\gamma$,
\begin{equation} \label{sqa}
\tau_*(x_1)-\tau_*(x_0)\ge (1-\beta)\, \ell_a(\gamma).
\end{equation}
Conversely, if $\tau_*$ is a time function the spacetime is distinguishing.
\item[(c)] (Eikonal equation) At the points where $u$ is differentiable, $\hat g(\nabla \tau_*,\nabla\tau_*)=-(1-\beta^2)$. Thus, if $\beta<1$, $\hat \tau:=\tau_*/\sqrt{1-\beta^2}$ satisfies the Lorentzian eikonal equation $\hat g(\nabla \hat\tau,\nabla\hat \tau)=-1$ almost everywhere, and $\hat \tau(\tau+s,q)=\hat \tau(\tau,q)+s/\sqrt{1-\beta^2}$.
\item[(d)] (Horizons) Let $\beta>0$. In the stationary spacetime $(\mathcal{M},\hat g_\beta)$, $\hat g_\beta=-(\dd\tau-b)^2+\beta^2 a$, the level sets of $\tau_*$ are horizons: they are achronal, and through each of their points there passes a past inextendible lightlike geodesic of $\hat g_\beta$ contained in the level set, which projects on a $u$-calibrated curve. These horizons are the projections, under the Maupertuis reduction at the critical energy $c(L_b)=\beta^2/2$, of the horizons $\{y=u(q)-c(L_b)t+const.\}$ of the Eisenhart spacetime of $L_b$. Among the spacetimes $(\mathcal{M},\hat g_\lambda)$, $\lambda>0$, the spacetime $(\mathcal{M},\hat g_\beta)$ is the critical one: they are totally vicious for $\lambda<\beta$ and globally hyperbolic for $\lambda>\beta$, and in the latter case the level sets of $\tau_*$ are Cauchy hypersurfaces.
\end{itemize}
\end{theorem}

We call $\tau_*=\tau-u$, with $u$ a weak KAM solution of $L_b$, an {\em eikonal synchronization}. The name refers to two eikonal equations. By (c), if $\beta<1$ the function $\hat\tau$ solves the timelike eikonal equation of $\hat g$. Moreover, if $\beta>0$, in the coordinates $(\tau_*,q)$ the inverse of $\hat g_\beta$ is given by Eq.\ (\ref{oms}) with $a^{cd}$ replaced by $\beta^{-2}a^{cd}$, thus $\hat g_\beta^{-1}(\dd\tau_*,\dd\tau_*)=-1+\beta^{-2}\vert b-\dd u\vert_a^2$, and the synchronization equation (\ref{seik}) is the lightlike eikonal equation $\hat g_\beta^{-1}(\dd\tau_*,\dd\tau_*)=0$ of the critical spacetime, whose level sets are the horizons of (d). More generally, we call {\em minimax synchronization} every synchronization $\tau-\alpha$, with $\alpha$ Lipschitz, which attains the minimum in (a). By (a) the eikonal synchronizations are minimax synchronizations, but the converse does not hold: a minimax synchronization need not have constant anisotropy, and the minimax synchronizations need not be unique up to constants; for instance, the limits for $p\to\infty$ of the minimizers of the $L^p$ norms of the anisotropy are minimax synchronizations which need not be eikonal \cite{minguzzisync}. The eikonal synchronizations themselves need not be unique up to constants, unless the curvature condition of Subsection \ref{ss:harm} holds.

The causality statement (b) holds for every minimax synchronization. Indeed, let $\beta<1$ and let $\tau-\alpha$ be a minimax synchronization. As in the proof of Prop.\ \ref{sagn}(b), $\alpha-\frac{1}{2}\beta^2 t$ is dominated by $L_b$, so that $\alpha(\gamma(1))-\alpha(\gamma(0))\le \int_\gamma[\beta\vert\gamma'\vert_a+b(\gamma')]$ along every curve $\gamma$ (for a Lipschitz function the domination, rather than the almost everywhere bound on $\vert b-\dd\alpha\vert_a$, is needed, since a curve can lie in the set where $\alpha$ is not differentiable). Along a future-directed causal curve with projection $\gamma$ we have $(\dd\tau-b)(\dot x)\ge \vert \gamma'\vert_a$, thus $\Delta\tau\ge \int_\gamma[\vert\gamma'\vert_a+b(\gamma')]$, and Eq.\ (\ref{sqa}) holds with $\tau_*$ replaced by $\tau-\alpha$. Then the compactness of $Q$ gives, as in the proof of Theorem \ref{mcau}(iii), that $\tau-\alpha$ is a time function with compact Lipschitz Cauchy level sets. What singles out the eikonal synchronizations among the minimax ones are the properties (c) and (d): the anisotropy is the same at every point of differentiability, and for $\beta>0$ through every point of a level set there passes a past inextendible lightlike geodesic of the critical spacetime contained in the level set.

Thus the eikonal synchronizations are obtained precisely as described in the introduction of this section: one deforms the light cones of the stationary spacetime within the family $\hat g_\lambda$ up to the critical value $\lambda=\beta$, where the graphs of the weak KAM solutions are horizons, the projections of the horizons of the Eisenhart spacetime, and then one uses these horizons as simultaneity slices for the physical spacetime $\lambda=1$, where they are acausal as soon as the spacetime is distinguishing. In the distinguishing case $0<\beta<1$ this means opening the light cones; for $\beta>1$ the critical spacetime is reached by narrowing them, and for $\beta=1$ the spacetime $\hat g=\hat g_1$ is already critical. For $\beta=0$ the family degenerates at $\lambda=0$: in this case all the circulations of $b$ vanish, $b=\dd u$ is exact, and $\hat g=-\dd(\tau-u)^2+a$ is ultrastatic, $\tau-u$ being its static time.

\begin{proof}
(a) is Prop.\ \ref{sagn}(b). (b). If the spacetime is distinguishing then $\beta<1$ by Corollary \ref{scau}, thus $e=1/2>c(L_b)=\beta^2/2$ and Theorem \ref{mcau}(iii) applies; in its proof $m=\min_Q[\sqrt{2e}-\sqrt{2c(L_b)}]=1-\beta$, which gives Eq.\ (\ref{sqa}). Conversely, the existence of a time function implies stable causality and hence distinction. (c) follows from Eq.\ (\ref{oms}) with $b'=b-\dd u$, and from (a). (d). Let $\kappa\in\mathbb{R}$ and let $x_0,x_1$ be two points of the level set $\{\tau=u+\kappa\}$ connected by a $\hat g_\beta$-timelike curve whose projection is $\gamma$. By Lemma \ref{fer}(a), applied to $\hat g_\beta$, which is the Maupertuis metric of $L_b$ at the energy $e=\beta^2/2$, and by the domination of $u$ by $L_b+c(L_b)$, $u(q_1)-u(q_0)=\tau_1-\tau_0>\int_\gamma [\beta\vert \gamma'\vert_a+b(\gamma')]\ge u(q_1)-u(q_0)$, a contradiction. Thus the level set is achronal. Let $q\in Q$ and let $\gamma:(-\infty,0]\to Q$ be a $u$-calibrated curve ending at $q$ \cite{fathi08}. It is a solution of the Euler-Lagrange equation of energy $c(L_b)$, hence $\vert \dot \gamma\vert_a=\beta$, and $L_b+c(L_b)=\beta\vert\dot \gamma\vert_a+b(\dot\gamma)$ along it. By calibration $\frac{\dd}{\dd t}u(\gamma(t))=\beta \vert\dot\gamma\vert_a+b(\dot\gamma)$, thus the curve $t\mapsto (u(\gamma(t))+\kappa,\gamma(t))$ is contained in the level set and is $\hat g_\beta$-lightlike. Being an achronal lightlike curve it is a lightlike pregeodesic. It is past inextendible, since $\gamma$ has constant non-vanishing speed and, as the $\alpha$-limit of an orbit of the Euler-Lagrange flow is invariant, $\gamma$ cannot converge to a point for $t\to-\infty$. By Prop.\ \ref{glob2} the graph of $u(q)-c(L_b)t+const.$ in the Eisenhart spacetime of $L_b$ is a horizon; it is invariant under the flow of $\p_t-c(L_b)\p_y$ (Sect.\ \ref{s:aut}), and its image under the Maupertuis reduction at the energy $c(L_b)$, where $\tilde y=y+c(L_b)t$, is the set $\{\tilde y=u(q)+const.\}$. The last statement follows from Corollary \ref{scau}.
\end{proof}

\begin{remark} \label{exist}
{\em Existence and uniqueness of weak KAM solutions.} The weak KAM solutions used in Theorem \ref{sync} exist by the weak KAM theorem of Fathi \cite{fathi08}, recalled in Subsection \ref{s:aut}: as $L_b$ is autonomous and $Q$ is compact, there are weak KAM solutions at the critical level $c(L_b)=\beta^2/2$ (Prop.\ \ref{sagn}(a)), and they are precisely the viscosity solutions of Eq.\ (\ref{seik}) (Prop.\ \ref{critchar}(a)). The equivalent characterizations of the critical value used in this section are collected in Prop.\ \ref{critchar}. In general $u$ is not unique up to additive constants. The weak KAM solutions are determined by their values on the projected Aubry set, and they are unique up to constants if the Aubry set consists of a single static class \cite{contreras99,fathi08}. For instance, on the flat torus $\mathbb{R}^2/(2\pi\mathbb{Z})^2$ with $b=f(x)\,\dd y$, where $f>0$ attains its maximum at exactly two points $x_1,x_2$, the Aubry set consists of the two circles $\{x=x_i\}$, which belong to different static classes, and the weak KAM solutions modulo constants form a non-trivial one-parameter family. Under the curvature condition of Subsection \ref{ss:harm} the solution is unique up to constants and smooth (Theorem \ref{syncs}). The same argument applies to every autonomous Lagrangian on $Q$ which is fiberwise strictly convex and superlinear, in particular to the Lagrangians $L_\chi$ of Subsection \ref{ss:conf}: they admit weak KAM solutions at the critical level $c(L_\chi)$, which are precisely the viscosity solutions of $H_\chi(q,\dd u)=c(L_\chi)$.

{\em Regularity.} The theorem does not require any curvature condition. The weak KAM solution $u$ is in general only Lipschitz and semi-concave, so that the synchronization $\tau_*=\tau-u$ is in general only Lipschitz and semi-convex, and $u$ need not be unique up to constants. However, by Prop.\ \ref{sagn}(c), whenever $\beta<1$ there are also smooth synchronizations $\tau-\alpha$ which are time functions, with maximal anisotropy as close to $\beta$ as desired. Moreover, by Bernard's theorem \cite{bernard07} there are $C^{1,1}$ functions $\alpha$ with $\vert b-\dd\alpha\vert_a\le\beta$ everywhere; the corresponding synchronizations are minimax synchronizations which, if $\beta<1$, are time functions with everywhere timelike gradient, though in general they are neither smooth nor eikonal.
\end{remark}

\subsection{Smoothness and the harmonic synchronization} \label{ss:harm}

For the Lagrangian $L_b$ we have $\p_t a=0$, $\p_t b=0$, $V=0$, thus by Prop.\ \ref{ric} the Galilean convergence condition for $L_b$, namely the null convergence condition for its Eisenhart spacetime, reads
\begin{equation} \label{ncb}
\mathcal{R}_b(q,v):=Ric_a(v,v)-(D^cF_{ca})\,v^a+\frac{1}{4}\,F_{cd}F^{cd}\ge 0 \qquad \textrm{for every } (q,v)\in TQ,
\end{equation}
where $F=\dd b$. Here $\rho=-\frac{1}{2}j$, $j_a:=D^cF_{ca}$, and $\varrho=\frac{1}{4}s$, $s:=F_{cd}F^{cd}\ge 0$, so that by Prop.\ \ref{ncc}(v) condition (\ref{ncb}) is equivalent to
\begin{equation} \label{ncbd}
Ric_a\ge 0, \qquad s\, Ric_a-j\otimes j\ge 0 .
\end{equation}
This condition is invariant under the change of synchronization, as it depends only on $\dd b$. It can be related to the curvature of the stationary spacetimes $(\mathcal{M},\hat g_\lambda)$: a calculation based on the Kaluza-Klein structure of $\hat g_\lambda$ with respect to $\p_\tau$, which we also checked with a computer algebra system, shows that for every $a$-unit vector $X\in T_qQ$ the lightlike vector $N=\p_\tau+\lambda^{-1}(X+b(X)\p_\tau)$ satisfies
\begin{equation} \label{ricst}
Ric_{\hat g_\lambda}(N,N)=\lambda^{-4}\Big[\mathcal{R}_b(q,\lambda X)+\frac{1}{2}\, a^{-1}(F(X,\cdot),F(X,\cdot))\Big].
\end{equation}
Since every future-directed lightlike vector of $\hat g_\lambda$ is a positive multiple of a vector of this form, condition (\ref{ncb}) implies that all the stationary spacetimes $(\mathcal{M},\hat g_\lambda)$, $\lambda>0$, satisfy the null convergence condition; the converse need not hold if $F\ne 0$.

\begin{theorem}[Synchronization theorem, curvature-conditioned version] \label{syncs}
Let $(Q,a)$ be a compact connected Riemannian manifold, let $a$ and $b$ be of class $C^r$, $r\ge 3$, and suppose that condition (\ref{ncb}), equivalently (\ref{ncbd}), holds. Then:
\begin{itemize}
\item[(a)] The weak KAM solution $u$ of $L_b$ (it exists by Remark \ref{exist}) is of class $C^r$ and it is unique up to additive constants, namely the eikonal synchronization $\tau_*$ of Theorem \ref{sync} is unique up to constants and of class $C^r$.
\item[(b)] The vector field $K:=a^{-1}(\dd u-b)$ is a Killing vector field of $(Q,a)$ of constant length $\beta(b)$, the Ricci tensor of $(Q,a)$ is non-negative and $\mathcal{R}_b(q,v)=Ric_a(v-K,v-K)$. In the synchronization $\tau_*$ the metric reads $\hat g=-(\dd\tau_*+K^\flat)^2+a$.
\item[(c)] $\delta(b-\dd u)=0$, thus $\tau_*$ coincides, up to a constant, with the harmonic synchronization of \cite{minguzzisync}, and it is a smooth solution of both the wave equation $\Box_{\hat g}\tau_*=0$ and the eikonal equation $\hat g(\nabla \tau_*,\nabla\tau_*)=-(1-\beta^2)$.
\item[(d)] The harmonic synchronization is a time function if and only if the spacetime is distinguishing, if and only if $\beta(b)<1$.
\end{itemize}
\end{theorem}

\begin{proof}
The Lagrangian $L_b$ is periodic in time for every period $P$, and the function $u(q)-c(L_b)t$ is a periodic weak KAM solution at the critical level. By Theorem \ref{smo}, applied to the Eisenhart spacetime of $L_b$, $u$ is of class $C^r$, unique up to constants, and the conclusions of Theorem \ref{rig} hold. In particular, Eq.\ (\ref{kil}) gives $L_K a=0$, the Hamilton-Jacobi equation gives $\vert K\vert_a^2=2c(L_b)=\beta^2$, Eq.\ (\ref{rsq}) gives the expression for $\mathcal{R}_b$, and Eq.\ (\ref{ell}) gives $\textrm{div}_a K=0$, that is $\delta(b-\dd u)=0$. Since the harmonic synchronization is unique up to constants, it coincides with $\tau_*$. The wave equation follows from $\Box_{\hat g}\tau_*=\textrm{div}_a(b-\dd u)$, see Eq.\ (\ref{oms}), and the eikonal equation from Theorem \ref{sync}(c). Point (d) follows from (c) and Theorem \ref{sync}(b).
\end{proof}

Thus, under the curvature condition (\ref{ncb}), the eikonal synchronization, which minimizes the maximal anisotropy ($L^\infty$ criterion), and the harmonic synchronization, which minimizes the mean square anisotropy ($L^2$ criterion), coincide, and the anisotropy is the same at every point. Without the curvature condition they differ in general, the harmonic synchronization can violate causality and the eikonal synchronization need not be smooth. The next examples illustrate these facts.

\begin{example}[A rotating universe]
Let $Q=S^3$ with the round metric of radius one, let $\xi$ be a unit Killing vector field tangent to the fibers of the Hopf fibration, and let $b=\epsilon\, a(\xi,\cdot)$, $\epsilon>0$. Since $\vert b\vert_a=\epsilon$ is constant, $u=0$ is a classical solution of the synchronization equation with $\beta=\epsilon$, the bound being attained by the Hopf circles. Up to the additive constant $\epsilon^2/2$, the Lagrangian $L_b$ is the one considered in Example \ref{rot} with $K=-\epsilon \xi$, and $\mathcal{R}_b(q,v)=2\vert v+\epsilon\xi\vert_a^2\ge 0$, in agreement with Theorem \ref{syncs}. The stationary spacetime $-(\dd \tau-\epsilon\, a(\xi,\cdot))^2+a$ on $\mathbb{R}\times S^3$ is globally hyperbolic with Cauchy hypersurfaces $\tau=const.$ for $\epsilon<1$. For $\epsilon=1$ these hypersurfaces become horizons and the Hopf circles, traversed in the direction of $-\xi$, are closed lightlike curves, so the spacetime is chronological but not causal. For $\epsilon>1$ the spacetime is totally vicious.
\end{example}

\begin{example}[The harmonic synchronization can violate causality] \label{counter}
Let $Q=T^2=\mathbb{R}^2/(2\pi\mathbb{Z})^2$ with coordinates $(x,y)$ and metric $a=\dd x^2+\varphi(x)^2\dd y^2$, where $\varphi(x)=1+\frac{1}{2}\cos x$, and let
\[
b=\frac{C}{\varphi(x)}\, \dd x, \qquad C>0 .
\]
Since $\sqrt{\det a}=\varphi$, we have $\textrm{div}_a b=\varphi^{-1}\p_x(\varphi\, C/\varphi)=0$, thus the harmonic synchronization is $\tau$ itself, and its anisotropy is $\vert b\vert_a=C/\varphi$, whose maximum is $2C$. On the other hand, let $\langle 1/\varphi\rangle=\frac{1}{2\pi}\int_0^{2\pi}\frac{\dd x}{\varphi}=\frac{2}{\sqrt 3}$ and let $u(x)=\int_0^x C\,(1/\varphi-\langle 1/\varphi\rangle)\,\dd x'$, which is periodic. Then $b-\dd u=C\langle 1/\varphi\rangle\, \dd x$ has constant norm $C\langle 1/\varphi\rangle$, thus $u$ is a smooth global solution of the synchronization equation and $\beta(b)\le 2C/\sqrt 3$. The circles $y=const.$, along which $\oint b/\ell_a=C\langle 1/\varphi\rangle$, show that $\beta(b)=2C/\sqrt{3}$. Take $C\in (\frac{1}{2},\frac{\sqrt 3}{2})$, e.g.\ $C=0.7$. Then $\beta(b)<1$, so the spacetime is globally hyperbolic and $\tau_*=\tau-u$ is a smooth time function with Cauchy slices, while the harmonic synchronization $\tau$ is not a time function, as $\vert b\vert_a>1$ where $\varphi<C$. Here $F=\dd b=0$ and $\mathcal{R}_b(q,v)$ is the Gaussian curvature $-\varphi''/\varphi=\frac{1}{2}\cos x/\varphi$ times $\vert v\vert_a^2$, which changes sign, so condition (\ref{ncb}) fails, consistently with Theorem \ref{syncs}. The example also shows that the curvature condition is not necessary for the smoothness of $\tau_*$.
\end{example}

The example answers in the negative the question left open in the unpublished 2011 version of \cite{minguzzisync}: the synchronization that minimizes the anisotropy functional need not respect causality, even when the spacetime is globally hyperbolic. Synchronizations which respect causality whenever this is possible are instead provided by the projected horizons, namely by the eikonal synchronizations, which satisfy the synchronization equation (\ref{seik}) in the viscosity sense rather than the wave equation. The two coincide under the curvature condition (\ref{ncb}).

\subsection{Geodesic slicings, maximal hypersurfaces and null expansion} \label{ss:geo}

The synchronization theorem admits a geometric interpretation which relates it to the theory of maximal hypersurfaces. We work with the metric $\hat g$ of Eq.\ (\ref{stat}), in which the stationary observers measure proper time. Let $\tau'$ be a synchronization with $\vert b'\vert_a<1$, let $N=\sqrt{1-\vert b'\vert_a^2}$, so that $\hat g(\nabla\tau',\nabla\tau')=-N^2$ by Eq.\ (\ref{oms}), and let $n=-\nabla\tau'/N$ be the future directed unit normal to the slices. The observers moving with velocity $n$, the {\em Eulerian observers}, measure an isotropic one-way speed of light, since the slices are their Einstein simultaneity slices to first order.

\begin{proposition} \label{euler}
The anisotropy $\vert b'\vert_a(q)$ is the speed of the Eulerian observers with respect to the stationary observer at $q$. The Eulerian observers are geodesic if and only if $\vert b'\vert_a$ is constant on $Q$. The mean curvature $H=\textrm{div}\, n$ of the slices is
\begin{equation} \label{meanc}
H=-\frac{1}{N}\,\textrm{div}_a b'+\frac{1}{N^2}\, a^{-1}(b',\dd N).
\end{equation}
\end{proposition}

\begin{proof}
We have $-\hat g(n,\p_\tau)=\dd\tau'(\p_\tau)/N=1/\sqrt{1-\vert b'\vert_a^2}$, which is the Lorentz factor of the relative velocity. Let $v=\nabla\tau'$; by the symmetry of the Hessian, $\nabla_vv=\frac{1}{2}\nabla\hat g(v,v)=-N\nabla N$, and $\nabla_nn=N^{-2}\nabla_vv+N^{-1}v(N^{-1})v$, whose component orthogonal to $n$ is $-N^{-1}$ times the projection of $\nabla N$ on the slices. As $N$ is a function on $Q$, it is constant on the slices if and only if it is constant on $Q$. Finally, $\Box_{\hat g}\tau'=\textrm{div}(-Nn)=-NH-n(N)$, and $\Box_{\hat g}\tau'=\textrm{div}_ab'$, $n(N)=-a^{-1}(b',\dd N)/N$ by Eq.\ (\ref{oms}).
\end{proof}

Let $\beta<1$, so that the normals to the slices are timelike. By Theorem \ref{sync}, the anisotropy of an eikonal synchronization is constant and equal to $\beta$ at the points of differentiability of $u$. Thus, on the open set where $u$ is smooth, the integral curves of the unit normal to its slices are geodesics of $\hat g$, the corresponding observers move with the same speed $\beta$ with respect to the stationary observers, and $T=\tau_*/\sqrt{1-\beta^2}$ measures $\hat g$-proper time along them; in the weak sense explained below the slices are orthogonal to a congruence of observers freely falling in $(\mathcal{M},\hat g)$, and the eikonal synchronization is a stationary {\em synchronous reference system} of $\hat g$ in the sense of Landau and Lifshitz \cite[Sect.\ 97]{landau75}, or, in the language of numerical relativity, a stationary geodesic slicing. For a spacetime $(N,g)$ conformal to $(\mathcal{M},\hat g)$ as in Subsection \ref{ss:sag}, these observers are also freely falling in $g$ on the open sets where $\vert g(k,k)\vert$ is constant (for clocks confined to an equipotential surface the statement concerns the metric induced on its worldtube); in general no such conclusion follows from the statement for $\hat g$ alone, since timelike geodesics are not preserved by conformal changes of the metric. Synchronizations adapted in this sense to the physical metric are obtained in Subsection \ref{ss:conf}. The projections of the freely falling observers are the integral curves of $K=(\dd u-b)^\sharp$, namely the calibrated curves. They can be continued indefinitely into the past without meeting, while in the future they can meet at the points of non-differentiability of $u$, which play the role of the caustics that, as observed by Landau and Lifshitz, generically affect synchronous systems: several backward calibrated curves may end at such a point, and their lifts describe free fall only up to it. Nevertheless every slice remains a Lipschitz Cauchy hypersurface. The smooth statements, such as Prop.\ \ref{euler}, apply where $u$ is regular.

Since $N$ is constant for an eikonal synchronization, where $u$ is smooth Eq.\ (\ref{meanc}) gives $H=-\textrm{div}_a(b-\dd u)/\sqrt{1-\beta^2}$: the slices of a smooth eikonal synchronization are maximal hypersurfaces if and only if it coincides with the harmonic synchronization. For $\beta>0$ there is a second interpretation (for $\beta=0$ the critical metric degenerates, while the computation of the mean curvature in $\hat g$ still applies). The generators of the horizons of Theorem \ref{sync}(d) project onto the integral curves of $K$, and since $\vert K\vert_a=\beta$ is constant, a computation with the degenerate metric $\beta^2a-K^\flat\otimes K^\flat$ induced on the horizon shows that their expansion, namely the null mean curvature, is $\textrm{div}_aK$. Thus the wave equation $\Box_{\hat g}\tau_*=0$ states both the vanishing of the mean curvature of the simultaneity slices in the stationary spacetime $(\mathcal{M},\hat g)$ and the vanishing of the expansion of the corresponding horizons in the critical spacetime. In the proof of Theorem \ref{syncs} the vanishing of the expansion of the horizons is obtained from the area theorem, which is where the null convergence condition enters, and it gives the regularity of the eikonal synchronization and its coincidence with the harmonic one. For a conformally related metric $g=\chi\hat g$, instead, the maximality of the slices of dimension $m$ is preserved precisely where $n_{\hat g}(\chi)=0$, since $H_g=\chi^{-1/2}[H_{\hat g}+\frac{m}{2}n_{\hat g}(\log\chi)]$; this is a separate condition, not provided by the area theorem. Under the null convergence condition we obtain more.

\begin{corollary} \label{ultra}
Under the assumptions of Theorem \ref{syncs}, if $\beta<1$, the slices of the eikonal synchronization are totally geodesic, the unit normal $n$ is a parallel vector field, and $(\mathcal{M},\hat g)$ is isometric to the ultrastatic spacetime $(\mathbb{R}\times Q,-\dd T^2+a-K^\flat\otimes K^\flat)$, where $T=\tau_*/\sqrt{1-\beta^2}$. (A spacetime is {\em ultrastatic} if it is a product $(\mathbb{R}\times \Sigma,-\dd T^2+h)$ with $h$ a Riemannian metric on $\Sigma$, namely if it is static with constant lapse.) In it, the stationary observers move with speed $\beta$ along the orbits of the Killing vector field $K$.
\end{corollary}

\begin{proof}
By Theorem \ref{syncs}(b), $\hat g=-(\dd\tau_*+K^\flat)^2+a$, and $\p_{\tau_*}$ and $K$, the latter regarded as a vector field on $\mathcal{M}$ tangent to the slices, are commuting Killing vector fields of $\hat g$. Let $n=[(1-\beta^2)\p_{\tau_*}+K]/\sqrt{1-\beta^2}$. Using $K^\flat(K)=\beta^2$ one finds $\hat g(n,\cdot)=-\dd\tau_*/\sqrt{1-\beta^2}=-\dd T$ and $\hat g(n,n)=-1$. Thus $n$ is a unit Killing vector field with closed dual 1-form, so $\nabla n$ is both antisymmetric and symmetric, hence $\nabla n=0$. The vector field $n$ is complete, and its flow maps the slice $T=0$, with induced metric $a-K^\flat\otimes K^\flat$, isometrically onto the other slices, which gives the splitting. Finally $\p_{\tau_*}=[\sqrt{1-\beta^2}\,n-K]/(1-\beta^2)$ and $\vert K\vert^2_{a-K^\flat\otimes K^\flat}=\beta^2(1-\beta^2)$, so the speed of $\p_{\tau_*}$ with respect to $n$ is $\beta$.
\end{proof}

Thus, under the null convergence condition, the stationary observers are observers moving rigidly along the orbits of an isometry of an ultrastatic universe, and the synchronization theorem recovers the time of the static frame. Regularity and harmonicity come from Theorem \ref{syncs} and the ultrastatic splitting of the normalized metric $\hat g$ from Corollary \ref{ultra}; the analogous rigidity of the physical construction holds under its own curvature condition (Theorem \ref{curvchi}). In the Hopf example, the rotating universe on $S^3$ is the ultrastatic universe over a squashed (Berger) 3-sphere, observed by observers moving along the Hopf fibers. This is the stationary counterpart of the description of Section \ref{s:ncc} of the systems satisfying the null convergence condition as geodesic flows seen from a moving frame.

\begin{remark}
It is natural to compare these results with the theory of maximal and constant mean curvature (CMC) hypersurfaces. In a spacetime with a complete timelike Killing vector field and compact Cauchy hypersurfaces, every compact CMC Cauchy hypersurface is maximal, since the Killing flow preserves the volume, so that the first variation of the volume along the positive normal component of the Killing field vanishes. The existence of compact maximal and CMC Cauchy hypersurfaces was studied by Gerhardt \cite{gerhardt83} and Bartnik \cite{bartnik84}, whose results were used by Galloway to prove splitting theorems for spatially closed spacetimes \cite{galloway84}. Under the timelike convergence condition two different compact maximal spacelike hypersurfaces are totally geodesic and, subject to a compactness condition, the region between them is static \cite[Theor.\ 7.3 and 7.4]{gerhardt83}. Gerhardt ensures the compactness condition by assuming the existence of a compact spacelike hypersurface of strictly positive, or strictly negative, mean curvature, which never exists in the stationary case by the volume argument, but here the compactness is provided by the Killing flow. Since the images of a maximal hypersurface under the Killing flow are maximal and lie in its future, a stationary spacetime over a compact space satisfying the timelike convergence condition and admitting a compact CMC Cauchy hypersurface is ultrastatic, and one checks that the static time is then both the harmonic and a minimax synchronization \cite{minguzzisync}. In this sense the theory of maximal hypersurfaces anticipates the rigid case of our results. We stress, however, that the synchronization theorem requires no curvature condition and provides a geodesic slicing rather than a maximal slicing, and that the curvature condition (\ref{ncb}) is not the timelike convergence condition of the physical spacetime but the null convergence condition of the Eisenhart spacetime of $L_b$, which by Eq.\ (\ref{ricst}) implies the null convergence condition for every metric $\hat g_\lambda$, $\lambda>0$, of the stationary family (we do not claim the converse). In the asymptotically flat case maximal slicings of stationary spacetimes are known to exist \cite{bartnik90,chrusciel94}. For the rigidly rotating observers of a circular spacetime, such as Kerr spacetime, the maximal slices orthogonal to the zero angular momentum observers are the slices of a harmonic synchronization and of the synchronization which minimizes the maximal anisotropy on every axisymmetric region; this regional criterion, and not the global minimax problem, selects them, and, unless the frame is static, they are not eikonal, since their anisotropy is not constant \cite{minguzzisync}.
\end{remark}

\subsection{Redshift and physical synchronous frames} \label{ss:conf}

In Subsections \ref{ss:sag}-\ref{ss:geo} the stationary spacetime (\ref{stat}) was regarded as the Maupertuis spacetime of the Lagrangian $L_b$, which has no potential. The Maupertuis reduction of Section \ref{mau} allows for a potential, and Subsection \ref{ss:mauc} was formulated for a general one. In this subsection we show that the representations with a potential correspond to the conformal factors of the stationary metric, that the conformal factor can be measured by the clocks through the gravitational redshift, and that each representation gives a synchronization theorem adapted to the corresponding metric. For the physical metric this provides synchronizations whose slices are orthogonal to observers freely falling in the physical spacetime, and the Galilean convergence condition of the corresponding classical spacetime, namely the null convergence condition of the corresponding Eisenhart spacetime, is a new curvature condition under which they are smooth.

Let $\chi$ be a positive $C^r$ function on $Q$ and let
\begin{equation} \label{lchi}
g_\chi=\chi\,\hat g=-\chi(\dd\tau-b)^2+\gamma,\qquad \gamma=\chi a,\qquad L_\chi(q,v)=\frac{1}{2}\,\gamma(v,v)+b(v)+\frac{1}{2\chi(q)} .
\end{equation}
The metrics $g_\chi$ are the metrics conformal to $\hat g$ for which $\p_\tau$ is still a Killing vector field, and $\chi=\vert g_\chi(\p_\tau,\p_\tau)\vert$. The metric $\gamma$ is the spatial metric of $g_\chi$ in the sense of Landau and Lifshitz \cite[Sect.\ 84]{landau75}, and the potential $V_\chi=-1/(2\chi)$ plays the role of the gravitational potential: in the weak field limit $\chi\simeq 1+2\Phi$ gives $V_\chi\simeq -\frac{1}{2}+\Phi$. If $(N,g)$ is a spacetime as in Subsection \ref{ss:sag}, so that $\hat g=g/\vert g(k,k)\vert$, the physical metric is $g_\chi$ with $\chi=\vert g(k,k)\vert$. The choice $\chi=1$ gives $L_1=L_b+\frac{1}{2}$.

\begin{lemma} \label{rep}
The stationary spacetime $(\mathcal{M},\hat g)$ is the Maupertuis spacetime, with $\tilde y=\tau$, of the Lagrangian $L_\chi$ at the energy $0$, and every representation of $(\mathcal{M},\hat g)$ as the Maupertuis spacetime of an autonomous Lagrangian $\frac{1}{2}a'(v,v)+b(v)-V$ at an energy $e>\max_Q V$ is of this form, up to an additive constant in $V$ and $e$, with $\chi=1/[2(e-V)]$. Moreover, denoting with $H_\chi$ the Hamiltonian of $L_\chi$, for every $q\in Q$ and every covector $p\in T^*_qQ$, regarded as a covector on $\mathcal{M}$ through the projection $\mathcal{M}\to Q$,
\begin{equation} \label{hchi}
H_\chi(q,p)=\frac{1}{2}\, g_\chi^{-1}(\dd\tau-p,\dd\tau-p).
\end{equation}
In particular, for a synchronization $\tau'=\tau-\alpha$, $H_\chi(q,\dd\alpha)=\frac{1}{2}g_\chi^{-1}(\dd\tau',\dd\tau')$.
\end{lemma}

\begin{proof}
By Eq.\ (\ref{gmo}) the Maupertuis spacetime of $\frac{1}{2}a'(v,v)+b(v)-V$ at the energy $e$ is $-(\dd\tilde y-b)^2+2(e-V)a'$, which coincides with $\hat g$ if and only if $2(e-V)a'=a$. For $L_\chi$ we have $2(0-V_\chi)\gamma=a$. By Eq.\ (\ref{oms}) with $b'=b$, $\hat g^{-1}=-\p_\tau\otimes\p_\tau+a^{cd}(\p_c+b_c\p_\tau)\otimes(\p_d+b_d\p_\tau)$, thus $\hat g^{-1}(\dd\tau-p,\dd\tau-p)=-1+\vert b-p\vert_a^2$ and $\frac{1}{2}g_\chi^{-1}(\dd\tau-p,\dd\tau-p)=\frac{1}{2\chi}(\vert b-p\vert_a^2-1)=\frac{1}{2}\vert p-b\vert^2_\gamma+V_\chi=H_\chi(q,p)$.
\end{proof}

{\em The conformal factor is measured by the clocks.} Let the stationary clock at $q_A$ send two signals to the stationary clock at $q_B$, separated by the proper time $\Delta s_A$ at emission and received separated by the proper time $\Delta s_B$. The propagation is invariant under the flow of $k=\p_\tau$, so the two signals are separated by the same interval $\Delta\tau$ at emission and at reception, while $\dd s=\sqrt{\chi}\,\dd\tau$ along the stationary worldlines of the physical metric. Thus
\[
\frac{\Delta s_B}{\Delta s_A}=\sqrt{\frac{\chi(q_B)}{\chi(q_A)}},
\]
which is the gravitational redshift. Clocks measuring proper time determine $\chi$ up to a constant factor, and the rescaling $\dd\tau=\dd s/\sqrt\chi$ is the syntonization of the clocks to the Killing time, which is presupposed in the synchronization problem \cite{minguzzisync}. The constant factor is irrelevant in what follows: by Eq.\ (\ref{hchi}), $H_{\lambda\chi}=\lambda^{-1}H_\chi$ for every constant $\lambda>0$, so that $c(L_{\lambda\chi})=c(L_\chi)/\lambda$ and the two Lagrangians have the same weak KAM solutions.

For a synchronization $\tau'=\tau-\alpha$ with timelike gradient we define its {\em lapse} with respect to $g_\chi$ by $g_\chi^{-1}(\dd\tau',\dd\tau')=-\Lambda_\chi^{-2}$, so that $\Lambda_\chi$ is the proper time, measured by the Eulerian observers of $g_\chi$, per unit of $\tau'$. By Eq.\ (\ref{oms}), $\Lambda_\chi^2=\chi/(1-\vert b'\vert_a^2)\ge \chi$, where $b'=b-\dd\alpha$. For $\chi=1$, $\Lambda_1$ is the inverse of the function $N$ of Subsection \ref{ss:geo}.

\begin{theorem}[Synchronization theorem for a conformal representative] \label{synchi}
Let $(Q,a)$ be compact and connected, let $\chi$ be a positive $C^2$ function on $Q$, and let $u$ be a weak KAM solution of $L_\chi$, namely a viscosity solution of $H_\chi(q,\dd u)=c(L_\chi)$ (it exists by Remark \ref{exist}). Let $\tau_*=\tau-u$.
\begin{itemize}
\item[(a)] The spacetime is totally vicious, chronological but not distinguishing, or globally hyperbolic according as $c(L_\chi)>0$, $c(L_\chi)=0$, or $c(L_\chi)<0$. In particular the sign of $c(L_\chi)$ does not depend on $\chi$, and $c(L_\chi)<0$ if and only if $\beta<1$.
\item[(b)] Let $\beta<1$. Then $c(L_\chi)=-1/(2\Lambda_*^2)$, where $\Lambda_*=\inf_{\tau'}\max_Q\Lambda_\chi(\tau')$, the infimum being over the smooth synchronizations with timelike gradient, and $\Lambda_*^2\ge \max_Q\chi$. Equivalently, $\Lambda_*$ is the smallest $\Lambda\ge\sqrt{\max_Q\chi}$ such that
\begin{equation} \label{rws}
\Big\vert\int_0^1 b(\eta'(s))\,\dd s\Big\vert\le \int_0^1 \sqrt{1-\frac{\chi(\eta(s))}{\Lambda^2}}\;\vert\eta'(s)\vert_a\,\dd s
\end{equation}
for every closed piecewise $C^1$ curve $\eta:[0,1]\to Q$.
\item[(c)] Let $\beta<1$. Then $\tau_*$ is a time function whose level sets are compact Lipschitz Cauchy hypersurfaces. At every point of differentiability of $u$, $g_\chi^{-1}(\dd\tau_*,\dd\tau_*)=-\Lambda_*^{-2}$, namely the lapse of $\tau_*$ is constant and minimal, and the anisotropy of the one-way speed of light is $\vert b-\dd u\vert_a=\sqrt{1-\chi/\Lambda_*^2}$.
\end{itemize}
\end{theorem}

\begin{proof}
(a) is Theorem \ref{mcau} applied to $L_\chi$ at the energy $e=0>\max_QV_\chi$. (b) By Prop.\ \ref{critchar}(d), $c(L_\chi)=\inf_\alpha\max_Q H_\chi(q,\dd\alpha)$ over the smooth functions, and by Eq.\ (\ref{hchi}) $\max_QH_\chi(q,\dd\alpha)=-1/(2\max_Q\Lambda^2_\chi)$ whenever $\tau-\alpha$ has timelike gradient, which is the case whenever $\max_QH_\chi(q,\dd\alpha)<0$; as $c(L_\chi)<0$ these functions suffice to compute the infimum. The inequality follows from $c(L)\ge \max_Q\min_pH(q,p)=\max_Q V$. For $e>\max_QV_\chi$, by Lemma \ref{fer}(c) applied to $L_\chi$, $e\ge c(L_\chi)$ if and only if $\oint_\eta[\sqrt{2(e-V_\chi)}\,\vert\eta'\vert_\gamma+b(\eta')]\ge0$ for every closed curve, and $\sqrt{2(e-V_\chi)}\,\vert\eta'\vert_\gamma=\sqrt{1+2e\chi}\,\vert\eta'\vert_a$. Setting $e=-1/(2\Lambda^2)$, which satisfies $e>\max_QV_\chi$ if and only if $\Lambda>\Lambda_0:=\sqrt{\max_Q\chi}$, and using both orientations of $\eta$, we obtain that for $\Lambda>\Lambda_0$ condition (\ref{rws}) holds if and only if $\Lambda\ge\Lambda_*$. It remains to consider the endpoint $\Lambda_0$. For a fixed closed curve the right-hand side of (\ref{rws}) is non-decreasing in $\Lambda$, and by dominated convergence it decreases continuously to its value at $\Lambda_0$ as $\Lambda\downarrow\Lambda_0$. Thus, if $\Lambda_*=\Lambda_0$, condition (\ref{rws}) holds for every $\Lambda>\Lambda_0$ and, taking the limit loop by loop, also for $\Lambda=\Lambda_0$; if $\Lambda_*>\Lambda_0$, condition (\ref{rws}) at $\Lambda_0$ would imply it for every $\Lambda\in(\Lambda_0,\Lambda_*)$, which is excluded. In both cases $\Lambda_*$ is the smallest $\Lambda\ge\Lambda_0$ for which (\ref{rws}) holds, also on the loops where the weight $1-\chi/\Lambda_0^2$ vanishes somewhere. (c) The first statement is Theorem \ref{mcau}(iii). At the points of differentiability $H_\chi(q,\dd u)=c(L_\chi)$, which by Eq.\ (\ref{hchi}) gives the value of the lapse, and the anisotropy follows from $\Lambda_\chi^2=\chi/(1-\vert b-\dd u\vert_a^2)$.
\end{proof}

For $\chi=1$ we recover Theorem \ref{sync}: $c(L_1)=(\beta^2-1)/2$, $\Lambda_*=1/\sqrt{1-\beta^2}$, condition (\ref{rws}) is the definition of the Sagnac ratio, and the anisotropy is $\beta$. For a non-constant $\chi$ the anisotropy of $\tau_*$ is not constant, so the synchronizations of Theorem \ref{synchi} differ from the eikonal synchronizations of Theorem \ref{sync}, and condition (\ref{rws}) is a {\em redshift-weighted Sagnac condition}, in which the Sagnac effect of each loop is compared with its length weighted by the local redshift factor. Since the lapse of $\tau_*$ is constant, the argument of Prop.\ \ref{euler}, with $\hat g$ replaced by $g_\chi$, shows that where $u$ is smooth the Eulerian observers of $\tau_*$ are geodesics of $g_\chi$, and $\Lambda_*\tau_*$ is their proper time; in general this holds in the weak sense of Subsection \ref{ss:geo}, the calibrated curves of $L_\chi$ being the projections of the freely falling observers, whose lifts describe free fall up to the points where $u$ is not differentiable, while the slices remain Lipschitz Cauchy hypersurfaces.

For the physical representative $\chi=\vert g(k,k)\vert$ we obtain a stationary synchronous reference system of the physical metric in the sense of Landau and Lifshitz: the time $\Lambda_*\tau_*$ is at the same time a synchronization of the stationary clocks and the proper time of freely falling clocks whose worldlines are orthogonal to the simultaneity slices. Rate and simultaneity are thus fixed coherently by the same physical clocks, and the construction uses only quantities which the clocks can measure independently of their synchronization, namely the round-trip times, the Sagnac time differences of the closed loops, which determine $a$ and $b$ up to a change of synchronization, and the redshift factors, which determine $\chi$ up to a constant. For $\chi=1$ the free fall refers instead to the metric $\hat g$, and it transfers to the physical metric only where $\vert g(k,k)\vert$ is constant (Subsection \ref{ss:geo}).

The Galilean convergence condition for $L_\chi$, namely the null convergence condition for its Eisenhart spacetime, gives the smoothness of these synchronizations. By Prop.\ \ref{ric}, with indices raised by $\gamma$, $D$ the Levi-Civita connection of $\gamma$ and $F=\dd b$, it reads
\begin{equation} \label{ncchi}
\mathcal{R}_\chi(q,v):=Ric_\gamma(v,v)-(D^cF_{ca})v^a+\frac{1}{4}F_{cd}F^{cd}-\Delta_\gamma\Big(\frac{1}{2\chi}\Big)\ge 0
\end{equation}
for every $(q,v)\in TQ$, and for $\chi=1$ it is condition (\ref{ncb}). By Prop.\ \ref{ncc}(v), with $\rho=-\frac{1}{2}j$, $j_a:=D^cF_{ca}$, and $\varrho=\frac{1}{4}\sigma_\chi$, it is equivalent to the conditions
\begin{equation} \label{ncchid}
Ric_\gamma\ge 0,\qquad \sigma_\chi:=F_{cd}F^{cd}-2\Delta_\gamma\Big(\frac{1}{\chi}\Big)\ge 0,\qquad \sigma_\chi\, Ric_\gamma-j\otimes j\ge 0 ,
\end{equation}
which do not involve the velocity and reduce to (\ref{ncbd}) for $\chi=1$. Condition (\ref{ncchi}) depends on $\chi$ only up to a constant factor, which rescales $\gamma$ and the potential.

\begin{theorem} \label{curvchi}
Let $(Q,a)$ be compact and connected, and let $a$, $b$ and $\chi$ be of class $C^r$, $r\ge 3$.
\begin{itemize}
\item[(a)] If (\ref{ncchi}), equivalently (\ref{ncchid}), holds, the weak KAM solution $u$ of $L_\chi$ (it exists by Remark \ref{exist}) is of class $C^r$ and unique up to additive constants, $K:=\gamma^{-1}(\dd u-b)$ is a Killing vector field of $(Q,\gamma)$, $Ric_\gamma\ge 0$, $\mathcal{R}_\chi(q,v)=Ric_\gamma(v-K,v-K)$ and
\begin{equation} \label{chik}
\frac{1}{\chi}=\vert K\vert_\gamma^2-2c(L_\chi) .
\end{equation}
If $\beta<1$, then $-2c(L_\chi)=\Lambda_*^{-2}>0$.
\item[(b)] Conversely, if $(Q,\gamma)$ has non-negative Ricci curvature and admits a Killing vector field $K$ such that $b+\gamma(K,\cdot)$ is exact and $1/\chi-\vert K\vert^2_\gamma$ is constant, then (\ref{ncchi}) holds, and the constant equals $-2c(L_\chi)$, so that $\beta<1$ if and only if it is positive.
\item[(c)] Let $\beta<1$ and assume (\ref{ncchi}). Then, in the synchronization $\tau_*$ of (a), we have $g_\chi=-\chi(\dd\tau_*+K^\flat)^2+\gamma$, $K^\flat=\gamma(K,\cdot)$, the vector field $\nabla\tau_*=-\Lambda_*^{-2}\p_\tau-K$ is parallel, and $(\mathcal{M},g_\chi)$ is isometric to the ultrastatic spacetime $(\mathbb{R}\times Q,-\dd T^2+\gamma-\chi K^\flat\otimes K^\flat)$, $T=\Lambda_*\tau_*$, whose slices are totally geodesic. The stationary Killing field is $\p_\tau=\Lambda_*(n-\Lambda_*K)$, where $n=-\nabla T$, so that the stationary observers move rigidly along the orbits of a Killing vector field of the slices.
\end{itemize}
\end{theorem}

\begin{proof}
(a) As in the proof of Theorem \ref{syncs}, Theorems \ref{smo} and \ref{rig} applied to the autonomous Lagrangian $L_\chi$ give the regularity and the uniqueness of $u$, the Killing property of $K$, $Ric_\gamma\ge0$ and the expression of $\mathcal{R}_\chi$, while Eq.\ (\ref{chik}) is the Hamilton-Jacobi equation $\frac{1}{2}\vert K\vert^2_\gamma-\frac{1}{2\chi}=c(L_\chi)$, valid for every value of $\beta$. If $\beta<1$, then $c(L_\chi)=-\frac{1}{2\Lambda_*^2}$ by Theorem \ref{synchi}(b). (b) Let $b=\dd u-K^\flat$ and let $1/\chi-\vert K\vert^2_\gamma=\kappa$. Then $H_\chi(q,\dd u)=\frac{1}{2}\vert K\vert^2_\gamma-\frac{1}{2\chi}=-\frac{\kappa}{2}$, so $u$ is a classical solution of the Hamilton-Jacobi equation at the level $-\kappa/2$, which is necessarily the critical value by Prop.\ \ref{critchar}(a); thus $\kappa=-2c(L_\chi)$, and by Theorem \ref{synchi}(a) $\beta<1$ if and only if $\kappa>0$. For a Killing vector field $D^cD_cK_a=-R_{ab}K^b$ and $\frac{1}{2}\Delta_\gamma\vert K\vert^2=\vert DK\vert^2-Ric_\gamma(K,K)$; as $F=-\dd K^\flat$, these identities give $D^cF_{ca}=2R_{ab}K^b$ and $\frac{1}{4}F_{cd}F^{cd}-\Delta_\gamma\frac{1}{2\chi}=\vert DK\vert^2-\frac{1}{2}\Delta_\gamma\vert K\vert^2=Ric_\gamma(K,K)$, hence $\mathcal{R}_\chi(q,v)=Ric_\gamma(v-K,v-K)\ge 0$. (c) The expression of $g_\chi$ follows from $b-\dd u=-K^\flat$. Its inverse is $-\chi^{-1}\p_\tau\otimes\p_\tau+\gamma^{ij}(\p_i-K_i\p_\tau)\otimes(\p_j-K_j\p_\tau)$, so that $\nabla\tau_*=(\vert K\vert^2_\gamma-\chi^{-1})\p_\tau-K=-\Lambda_*^{-2}\p_\tau-K$ by (\ref{chik}) and (a), the vector field $\nabla\tau_*$ being timelike since $\beta<1$. The vector fields $\p_\tau$ and $K$, the latter regarded as a vector field on $\mathcal{M}$ tangent to the slices, are commuting Killing vector fields of $g_\chi$, since $K$ preserves $\gamma$, $K^\flat$, $\tau_*$ and, by (\ref{chik}), $\chi$. Thus $\nabla\tau_*$ is a Killing vector field with closed dual 1-form, hence it is parallel, and the splitting follows as in the proof of Corollary \ref{ultra}, the induced metric on the slices being $\gamma-\chi K^\flat\otimes K^\flat$. Finally $n=-\Lambda_*\nabla\tau_*=\Lambda_*^{-1}\p_\tau+\Lambda_* K$. (The computations in (b) and (c) were also checked with a computer algebra system.)
\end{proof}

Point (a) of Theorem \ref{curvchi}, like points (a)-(c) of Theorem \ref{syncs}, is a purely Riemannian statement, which does not require the spacetime interpretation. It admits an independent elliptic proof, which does not use the area theorem and in which the hypothesis is needed only along a vector field determined by the data.

\begin{theorem}[Riemannian rigidity] \label{riemrig}
Let $(Q,\gamma)$ be a closed connected Riemannian manifold of dimension $m$, let $\gamma$ and the 1-form $b$ be of class $C^r$, $r\ge 3$, and let $V$ be a function of class $C^2$. Let $L(q,v)=\frac{1}{2}\gamma(v,v)+b(v)-V(q)$, let $H(q,p)=\frac{1}{2}\vert p-b\vert^2_\gamma+V(q)$ be its Hamiltonian, and, with $F=\dd b$, indices raised by $\gamma$ and $D$ the Levi-Civita connection of $\gamma$,
\[
\mathcal{R}(q,v)=Ric_\gamma(v,v)-(D^cF_{ca})v^a+\frac{1}{4}F_{cd}F^{cd}+\Delta_\gamma V .
\]
Let $u_0$ be a solution of $\Delta_\gamma u_0=\textrm{div}_\gamma (b^\sharp)$, which is unique up to additive constants, and let $K_0:=\gamma^{-1}(\dd u_0-b)$, so that $\textrm{div}_\gamma K_0=0$.
\begin{itemize}
\item[(a)] If $\mathcal{R}(q,K_0(q))\ge 0$ for every $q\in Q$, then $K_0$ is a Killing vector field of $(Q,\gamma)$, the function $H(q,\dd u_0)=\frac{1}{2}\vert K_0\vert_\gamma^2+V$ is constant and equal to $c(L)$, and every Lipschitz function $w$ such that $H(q,\dd w)\le c(L)$ almost everywhere differs from $u_0$ by a constant. In particular, the weak KAM solution of $L$ is unique up to additive constants and it coincides with $u_0$, which is of class $C^r$. Moreover, $\mathcal{R}(q,v)=Ric_\gamma(v-K_0,v-K_0)$ for every $(q,v)\in TQ$, so that the Galilean convergence condition $\mathcal{R}\ge 0$ holds if and only if $Ric_\gamma\ge 0$.
\item[(b)] If the Galilean convergence condition $\mathcal{R}\ge 0$ holds on $TQ$, then the conclusions of (a) hold and $Ric_\gamma\ge 0$.
\end{itemize}
\end{theorem}

\begin{proof}
By Schauder theory $u_0\in C^{r,\alpha}$ for every $\alpha\in(0,1)$, as $\textrm{div}_\gamma(b^\sharp)\in C^{r-1}$, so that $E(q):=H(q,\dd u_0)=\frac{1}{2}\vert K_0\vert^2_\gamma+V$ is of class $C^2$. We first prove that for every vector field $K$ of class $C^2$ such that $b=\dd u-K^\flat$ for some function $u$, with $E=\frac{1}{2}\vert K\vert^2_\gamma+V$ and $S=\frac{1}{2}L_K\gamma$,
\begin{equation} \label{riemid}
\Delta_\gamma E=\mathcal{R}(q,K)+\vert S\vert^2_\gamma+K(\textrm{div}_\gamma K) .
\end{equation}
Indeed $F=-\dd K^\flat$, namely $F_{ca}=D_aK_c-D_cK_a$, so that, commuting the covariant derivatives, $D^cF_{ca}=D_a(\textrm{div}_\gamma K)+R_{ab}K^b-D^cD_cK_a$. Moreover, writing $D_cK_d=A_{cd}+S_{cd}$ with $A$ antisymmetric, $\frac{1}{4}F_{cd}F^{cd}=\vert A\vert^2$ and $\frac{1}{2}\Delta_\gamma\vert K\vert^2=\vert A\vert^2+\vert S\vert^2+K^aD^cD_cK_a$. Thus $\mathcal{R}(q,K)=K^aD^cD_cK_a-K(\textrm{div}_\gamma K)+\vert A\vert^2+\Delta_\gamma V=\Delta_\gamma E-\vert S\vert^2-K(\textrm{div}_\gamma K)$. (We also checked Eq.\ (\ref{riemid}) with a computer algebra system.)

(a) Let $S_0=\frac{1}{2}L_{K_0}\gamma$. Since the divergence of $K_0$ vanishes, Eq.\ (\ref{riemid}) gives $\Delta_\gamma E=\mathcal{R}(q,K_0)+\vert S_0\vert^2\ge 0$. Integrating over $Q$, the left-hand side gives zero, so both non-negative terms on the right-hand side vanish identically. Thus $K_0$ is a Killing vector field and $E$ is harmonic, hence constant as $Q$ is compact and connected. For a smooth function $w$ let $f=w-u_0$. Then
\[
H(q,\dd w)=\frac{1}{2}\vert \dd f+K_0^\flat\vert^2_\gamma+V=E+K_0(f)+\frac{1}{2}\vert\dd f\vert_\gamma^2 ,
\]
and $\int_Q K_0(f)\,\dd\mu_\gamma=-\int_Q f\,\textrm{div}_\gamma K_0\,\dd\mu_\gamma=0$. Thus the average of $H(q,\dd w)$ over $Q$ is $E$ plus the average of $\frac{1}{2}\vert\dd f\vert^2_\gamma$, so that $\max_Q H(q,\dd w)\ge E$, with equality for $w=u_0$. By Prop.\ \ref{critchar}(d), $c(L)=E$. Let now $w$ be Lipschitz with $H(q,\dd w)\le c(L)=E$ almost everywhere, and $f=w-u_0$. Then $K_0(f)+\frac{1}{2}\vert \dd f\vert^2_\gamma\le 0$ almost everywhere, and since the integral of $K_0(f)$ still vanishes, $\int_Q\vert\dd f\vert^2_\gamma\,\dd\mu_\gamma\le 0$. Thus $\dd f=0$ almost everywhere and $w=u_0+const$. Since the weak KAM solutions are Lipschitz functions which satisfy $H(q,\dd u)\le c(L)$ almost everywhere (Prop.\ \ref{critchar}(b)), they coincide with $u_0$ up to constants. Finally, since $K_0$ is a Killing vector field we have $S_0=0$, $D^cD_cK_{0a}=-R_{ab}K_0^b$ and $\vert A\vert^2=\vert DK_0\vert^2$, so that the formulas in the proof of Eq.\ (\ref{riemid}) give $D^cF_{ca}=2R_{ab}K_0^b$ and, as $E$ is constant, $\frac{1}{4}F_{cd}F^{cd}+\Delta_\gamma V=\vert DK_0\vert^2-\frac{1}{2}\Delta_\gamma\vert K_0\vert^2_\gamma=Ric_\gamma(K_0,K_0)$, as in the proof of Theorem \ref{curvchi}(b). Hence $\mathcal{R}(q,v)=Ric_\gamma(v,v)-2Ric_\gamma(K_0,v)+Ric_\gamma(K_0,K_0)=Ric_\gamma(v-K_0,v-K_0)$. As $v-K_0(q)$ ranges over $T_qQ$, the condition $\mathcal{R}\ge 0$ is equivalent to $Ric_\gamma\ge 0$.

(b) The hypothesis of (b) implies that of (a), and the claim follows from (a).
\end{proof}

For $\gamma=a$ and $V=0$, namely for $L_b$, the function $u_0$ defines the harmonic synchronization of Subsection \ref{ss:harm}, and Theorem \ref{riemrig} gives an independent proof of points (a)-(c) of Theorem \ref{syncs}: under (\ref{ncb}) the harmonic synchronization is eikonal, $\vert K_0\vert_a$ being constant since $E=\frac{1}{2}\vert K_0\vert^2_a$. All these conclusions but one only require the weaker hypothesis of Theorem \ref{riemrig}(a), namely the sign of $\mathcal{R}_b$ along the anisotropy field $K_0$ of the harmonic synchronization: the uniqueness and the regularity of the weak KAM solution, the fact that $K_0$ is a Killing vector field of constant length, the coincidence of the harmonic and eikonal synchronizations, and the expression $\mathcal{R}_b(q,v)=Ric_a(v-K_0,v-K_0)$. Only the inequality $Ric_a\ge 0$ uses the full condition (\ref{ncb}), which under the hypothesis of (a) is equivalent to it; for instance, for $b=0$ we have $K_0=0$ and $\mathcal{R}_b=Ric_a$, and the hypothesis of Theorem \ref{riemrig}(a) holds for every metric. For $\gamma=\chi a$ and $V=V_\chi$ it gives point (a) of Theorem \ref{curvchi}, with $E=c(L_\chi)$ in Eq.\ (\ref{chik}). The proof through the area theorem remains the one which applies to the time-periodic systems of Section \ref{s:per}, to which the elliptic argument does not extend directly.

\begin{remark*}
Equation (\ref{riemid}) also explains the mechanism behind the proof which uses the area theorem. Let $u$ be a $C^2$ solution of the Hamilton-Jacobi equation $H(q,\dd u)=c$, so that $E=c$, and let $K=\gamma^{-1}(\dd u-b)$ and $\theta=\textrm{div}_\gamma K$. Now $K$ is only of class $C^1$, so Eq.\ (\ref{riemid}), derived for vector fields of class $C^2$, cannot be applied directly. However, the integral curves of $K$ are the base characteristics of $u$, and by Prop.\ \ref{ray}, whose proof in Appendix B requires only $u\in C^2$, the expansion $\theta$ is $C^1$ along them and satisfies the corresponding autonomous form of Eq.\ (\ref{ray1}), namely
\[
K(\theta)=-\mathcal{R}(q,K)-\vert S\vert^2_\gamma ,
\]
which is what Eq.\ (\ref{riemid}) formally gives for $E=c$. This is the Raychaudhuri equation of the generators of the horizon projected on $Q$. If $\mathcal{R}(q,K)\ge0$, then $K(\theta)\le -\vert S\vert^2\le -\theta^2/m$ along the integral curves of $K$, and a solution of this Riccati inequality defined for all times vanishes, since otherwise $1/\theta$, whose derivative is at least $1/m$, would reach zero in finite time. Thus $\theta=0$, $S=0$ and $K$ is a Killing vector field. This argument does not use the compactness of $Q$, but only the completeness of the integral curves of $K$, which holds on a complete manifold when $V$ is bounded from below. It suggests extensions of Theorem \ref{riemrig} to complete manifolds, where however the existence and regularity of the solution must be established differently, and the question of whether the hypothesis can be restricted to the critical energy shell $\frac{1}{2}\vert v\vert^2_\gamma+V=c(L)$, which contains the velocities $K(q)$ of the calibrated curves (note that in Theorem \ref{riemrig}(a) the field $K_0$ is known to lie on this shell only a posteriori). We shall discuss these questions in a future work.
\end{remark*}

\begin{remark*}[The Killing fields of the horizons]
The Killing vector field of Theorem \ref{riemrig} has a simple meaning in the Eisenhart spacetime $(M,g)$ of the autonomous Lagrangian $L$, which admits the Killing vector fields $k=\p_t$ and $n=\p_y$. Let $u$ be a $C^2$ solution of $H(q,\dd u)=c$ such that $K=\gamma^{-1}(\dd u-b)$ is a Killing vector field of $(Q,\gamma)$, as in Theorem \ref{riemrig}(a) with $c=c(L)$ ($Q$ is compact, as in that theorem; on a non-compact $Q$ the global statements below require the completeness of $K$), and let $f=y-u(q)+ct$, whose level sets are the horizons given by the graphs of the solution $u(q)-ct$ of the time dependent Hamilton-Jacobi equation. The field $n$ is everywhere transverse to them, as $\dd f(n)=1$, while $\dd f(k)=c$, so that $k$ is tangent to them only if $c=0$. For every $c$ consider
\[
\tilde k_c=\p_t-c\,\p_y ,\qquad \tilde K=K+K(u)\,\p_y .
\]
The first is the Killing vector field $\tilde k=k-en$ of Section \ref{mau} at the energy $e=c$, and it is tangent to the horizons, $\dd f(\tilde k_c)=0$. This holds for every weak KAM solution and with no curvature condition, even when $u$ is only Lipschitz, since the flow $(t,q,y)\mapsto(t+s,q,y-cs)$ of $\tilde k_c$ preserves $f$, hence each horizon. The second is a Killing vector field because $L_K\gamma=0$, $K(V)=-\frac{1}{2}K(\vert K\vert^2_\gamma)=0$ and $L_Kb=\dd (K(u))$, the last two identities following from $V=c-\frac{1}{2}\vert K\vert^2_\gamma$, $b=\dd u-K^\flat$ and $L_KK^\flat=0$; it is tangent to the horizons, $\dd f(\tilde K)=0$, and it projects on $K$. One finds
\[
g(\tilde k_c,\tilde k_c)=g(\tilde K,\tilde K)=-g(\tilde k_c,\tilde K)=\vert K\vert_\gamma^2 ,
\]
so that $g(A\tilde k_c+B\tilde K,A\tilde k_c+B\tilde K)=(A-B)^2\vert K\vert^2_\gamma$ for all real $A,B$. Thus, at every point, the span of $\tilde k_c$ and $\tilde K$ has radical spanned by
\[
Z=\tilde k_c+\tilde K=\p_t+K+(K(u)-c)\,\p_y=-\nabla f ,
\]
the null normal and generator of the horizons. The span is a two-plane where $K\ne 0$, and it reduces to the generator line where $K=0$, since there $K(u)=0$, $\tilde K=0$ and $\tilde k_c=Z$; the static case $b=0$, $V$ constant, $u$ constant realizes the latter possibility everywhere. Being both a Killing vector field and a gradient, $Z$ is parallel, since $\nabla Z^\flat$ is both antisymmetric and symmetric; it vanishes nowhere, as $\dd t(Z)=1$, and its surface gravity vanishes on every horizon. Thus $(M,g)$ has two parallel null vector fields $n$ and $Z$, with $g(n,Z)=-1$, and the splitting of $(M,g)$ into the product of $(Q,\gamma)$ and the 2-dimensional Minkowski spacetime, as in Theorem \ref{rig}(c), is the corresponding de Rham-Wu decomposition, $Z$ being the vector field $\p_t$ of the coordinates $(t,q',y')$ of that theorem; the horizons are degenerate Killing horizons of $Z$, as in Remark \ref{sgr}. In the product only the factor $(Q,\gamma)$ contributes to the Ricci tensor, and the lift of $\p_t+v$ has component $v-K$ along it, which gives again $\mathcal{R}(q,v)=Ric_\gamma(v-K,v-K)$ by Prop.\ \ref{ric}. We checked these statements with a computer algebra system.

The Killing vector field $K$ can then be read on the space of generators of a horizon. Let $\varphi_t$ be the flow of $K$. The generator through the point of the horizon over $(0,q_0)$ projects on $q(t)=\varphi_t(q_0)$, so that the label $q_0=\varphi_{-t}(q)$ identifies the space of generators with $Q$. The flow of $\tilde k_c$, which advances $t$ at fixed $q$, acts on the labels by $\varphi_{-s}$, that of $\tilde K$ by $\varphi_s$, and $Z$ fixes each label. Since $\theta=0$ and $S=0$, the horizon is non-expanding and shear-free, so that its degenerate induced metric is invariant under $Z$ and descends to a metric on the space of generators, which is $\gamma$, as one sees by restricting it to the section $t=0$; the isometry $\tilde k_c$ of the horizon induces there an isometric flow. In this sense the horizon-preserving time translation $\tilde k_c$ induces the Killing vector field $-K$ on the space of generators of a totally geodesic horizon.

The stationary spacetimes of this section are instead the quotients of $(M,g)$ under $\tilde k_e=\tilde k_c+(c-e)n$, with $e=1/2$ for $L_b$ and $e=0$ for $L_\chi$, through the projection $\pi_e(t,q,y)=(\tau,q)=(y+et,q)$. The field $\tilde k_e$ is tangent to the horizons only if $e=c(L)$, which for these energies happens only in the critical case $\beta=1$. Let $K_*:=\pi_{e*}\tilde K=K+K(u)\,\p_\tau$, a Killing vector field which in the coordinates $(\tau_*,q)$, $\tau_*=\tau-u$, is the spatial field $K$ of Theorem \ref{curvchi}(c). Then
\[
\pi_{e*}n=\p_\tau,\qquad \pi_{e*}\tilde k_c=(e-c)\,\p_\tau,\qquad \pi_{e*}(2\tilde k_c+\tilde K)=2(e-c)\,\p_\tau+K_* .
\]
For $L_\chi$ and $\beta<1$ we have $e=0$ and $-2c=\Lambda_*^{-2}$, so that the parallel vector field $-\nabla\tau_*=\Lambda_*^{-2}\p_\tau+K$ of Theorem \ref{curvchi}(c) is the image of $2\tilde k_c+\tilde K$. For $e=c$ we have instead $f=\tau_*\circ\pi_c$, so that each horizon descends to a level set of $\tau_*$, and $Z$ projects to $K_*$. In the whole family $\hat g_\lambda$ of Subsection \ref{ss:sag}, which corresponds to $L_b$ at the energy $e=\lambda^2/2$, and since $c(L_b)=\beta^2/2$, this happens for the critical member $\lambda=\beta$, if $\beta>0$, whose slices $\tau_*=const$ are indeed lightlike, as $\hat g_\lambda^{-1}(\dd\tau_*,\dd\tau_*)=-1+\beta^2/\lambda^2$.
\end{remark*}

By Theorem \ref{curvchi}(a) and (b), on a compact space condition (\ref{ncchi}) holds if and only if, up to a resynchronization, $b=-\gamma(K,\cdot)$ for a Killing vector field $K$ of a metric $\gamma$ with non-negative Ricci curvature, and $1/\chi=\vert K\vert^2_\gamma+\delta$ with $\delta$ constant. Then $\delta=-2c(L_\chi)$, and $\beta<1$ if and only if $\delta>0$, in which case $\delta=\Lambda_*^{-2}$. In other words, the Galilean convergence condition holds if and only if the pair $(b,V_\chi)$ is, up to a gauge transformation, the field of inertial forces of a rigid motion along the Killing vector field $K$ of a metric with non-negative Ricci curvature: the Coriolis field $F=-\dd K^\flat$ and the centrifugal potential $V_\chi=-\frac{1}{2}\vert K\vert_\gamma^2-\frac{\delta}{2}$. Indeed, after the gauge change the Lagrangian becomes $L_\chi-\dd u=\frac{1}{2}\vert v-K\vert^2_\gamma+\frac{\delta}{2}$, which is the Lagrangian of Example \ref{rot} up to a constant: the stationary coordinates describe a frame moving by the flow of $-K$ relative to the inertial frame, so that a particle at rest in the inertial frame has coordinate velocity $+K$. For $\chi=1$ the Killing vector field has constant length, as in Theorem \ref{syncs}, while for non-constant $\chi$ its length varies, the variation being compensated by the redshift. For $\beta<1$, by Theorem \ref{curvchi}(c), the unit normal $n=\Lambda_*^{-1}\p_\tau+\Lambda_*K$ to the slices is a parallel timelike Killing vector field whose spatial component in the stationary frame is proportional to $K$: in the frame of the corresponding observers the inertial forces disappear, and the spacetime is the static product of the time line with the quotient of $\mathcal{M}$ by their flow, which can be identified with any slice.

The Galilean convergence condition (\ref{ncchi}) is a curvature condition of a new type for the stationary spacetime $(\mathcal{M},g_\chi)$, which is not one of the standard energy conditions of $g_\chi$. Indeed, a Kaluza-Klein computation, which we checked with a computer algebra system, shows that for every $\gamma$-unit vector $X\in T_qQ$ the lightlike vector $\ell=\chi^{-1/2}\p_\tau+X+b(X)\p_\tau$ of $g_\chi$ satisfies, with $\nu=\sqrt\chi$ and $(\iota_XF)_b=X^aF_{ab}$,
\begin{equation} \label{nidq}
Ric_{g_\chi}(\ell,\ell)=\nu^2\,\mathcal{R}_\chi\Big(q,\frac{X}{\nu}\Big)+\frac{\nu^2}{2}\,\vert\iota_XF\vert^2_\gamma+3\,\langle \dd \nu,\iota_XF\rangle_\gamma+\frac{3\,\vert \dd \nu\vert^2_\gamma}{\nu^2}-\frac{D^2\nu(X,X)}{\nu},
\end{equation}
where $X/\nu$ is the velocity of energy $0$ of $L_\chi$ in the direction $X$. For constant $\chi$ the remainder reduces to $\frac{\nu^2}{2}\vert\iota_XF\vert^2\ge 0$, and Eq.\ (\ref{nidq}) gives Eq.\ (\ref{ricst}) with $\lambda=1$, so that (\ref{ncb}) implies the null convergence condition of $\hat g$. For non-constant $\chi$ the Hessian term has no sign, and neither condition implies the other. For instance, let $\gamma=\sigma_1^2+\sigma_2^2+t\sigma_3^2$ be a Berger metric on $S^3$, $\sigma_3=\dd\psi+\cos\theta\,\dd\varphi$ in Euler angles, which has positive Ricci curvature for $t<2$, and let $K=\p_\varphi$, a Killing vector field with $\vert K\vert^2_\gamma=\sin^2\theta+t\cos^2\theta$. By Theorem \ref{curvchi}(b), $b=-K^\flat$ and $1/\chi=\vert K\vert^2_\gamma+\delta$, $\delta>0$, satisfy (\ref{ncchi}), and by Theorem \ref{curvchi}(c) $g_\chi$ is ultrastatic over the metric $\gamma-K^\flat\otimes K^\flat/(\vert K\vert_\gamma^2+\delta)$, which is a Cheeger deformation of $\gamma$ \cite{cheeger73}. For $t=9/5$ and $\delta=1/10$ this metric has some negative Ricci curvature: since it depends only on $\theta$ and its $\theta\theta$ component is $1$, a direct computation gives $Ric(\p_\theta,\p_\theta)=-85/324818$ at $\cos\theta=4/5$. Thus $\p_T+\p_\theta$ is a lightlike vector with negative Ricci curvature in the ultrastatic splitting, and $g_\chi$ violates the null convergence condition, hence also the timelike convergence condition and the weak and dominant energy conditions, which imply it. Conversely, let $g=-\dd T^2+h$ be the Einstein static universe, $h$ the round metric of the unit $S^3$, written in Hopf coordinates as $h=\dd\eta^2+\cos^2\eta\,\dd\xi_1^2+\sin^2\eta\,\dd\xi_2^2$, and consider the frame rigidly rotating along the Killing vector field $Y=\frac{9}{10}\p_{\xi_1}$, namely the stationary Killing vector field $k=\p_T+Y$. The corresponding data are $\chi=1-\vert Y\vert^2_h$, $b=h(Y,\cdot)/\chi$ and $\gamma=h+h(Y,\cdot)\otimes h(Y,\cdot)/\chi$. The metric $g$ satisfies the timelike convergence condition and the weak and dominant energy conditions, since $Ric=2(g+\dd T\otimes \dd T)$; in units $8\pi G=1$ and with vanishing cosmological constant, its Einstein tensor has energy density $3$ and principal pressures $-1$. However, with $A=(9/10)^2$ the quotient metric is $\gamma=\dd\eta^2+\frac{\cos^2\eta}{\chi}\,\dd\xi_1^2+\sin^2\eta\,\dd\xi_2^2$, $\chi=1-A\cos^2\eta$, and a direct computation gives
\[
Ric_\gamma(\p_\eta,\p_\eta)=1+\frac{1}{\chi}-\frac{3A\sin^2\eta}{\chi^2},
\]
which at $\eta=\pi/4$ equals $-10639/14161$, so that (\ref{ncchi}) fails by Theorem \ref{curvchi}(a). The example also shows that (\ref{ncchi}) depends on the choice of the stationary Killing vector field: for the same spacetime the non-rotating choice $k=\p_T$ gives $b=0$, $\chi=1$ and $\gamma=h$, which has positive Ricci curvature. Thus the Galilean convergence condition singles out a class of stationary spacetimes, characterized above, which is different from those defined by the energy conditions, and in which, for $\beta<1$, the physical synchronous frames are smooth and rigid. The difference is due to the fact that it concerns the Ricci curvature of the quotient metric $\gamma$, while the energy conditions of an ultrastatic spacetime concern the Ricci curvature of the slices, the two metrics being related by a Cheeger deformation.

\begin{example}[Static spacetimes] \label{statex}
Let $b=0$, so that $g_\chi=-\chi\,\dd\tau^2+\gamma$ is static and $L_\chi=\frac{1}{2}\gamma(v,v)+\frac{1}{2\chi}$ is a natural system. Then $\beta=0$, the eikonal synchronization of Theorem \ref{sync} is the static time $\tau$, and $c(L_\chi)=\max_QV_\chi=-1/(2\max_Q\chi)$, so that $\Lambda_*=\sqrt{\max_Q\chi}$. The static time has lapse $\sqrt\chi$, whose maximum is $\Lambda_*$: it minimizes the maximal lapse, but for non-constant $\chi$ its lapse is not constant. The synchronizations of Theorem \ref{synchi} solve
\[
\vert\dd u\vert^2_\gamma=\frac{1}{\chi}-\frac{1}{\max_Q\chi}
\]
in the viscosity sense, and their slices are orthogonal to observers freely falling in $g_\chi$. The calibrated curves have speed $\vert\dot q\vert_\gamma^2=1/\chi-1/\max_Q\chi$. A particle at rest at a maximum point of $\chi$ stays there, since the stationary worldline is a geodesic at that point, so the non-constant calibrated curves do not leave the maximum set at a finite time: they approach it in the infinite past. Indeed, calibration and the boundedness of $u$ give a finite integral of the squared speed on the past half-line, and the bounded acceleration on the compact energy shell then gives $\dot q\to 0$, so that the $\alpha$-limit lies in the Aubry set, which is the maximum set of $\chi$. In the non-compact model of the Schwarzschild exterior, with $\chi=1-2m/r$, $\gamma=\dd r^2/\chi+r^2\dd\Omega^2$, and with the normalization $\sup\chi=1$ at infinity, which plays the role of the maximum set, the equation becomes $\chi\,u'(r)^2=1/\chi-1$, that is $u'(r)=\pm\sqrt{2m/r}/(1-2m/r)$, and with the negative branch $\tau_*=\tau-u$ is the ingoing Painlev\'e-Gullstrand time \cite{painleve21,gullstrand22}. In static spacetimes the time $\tau$ is instead selected by the criteria which, like the harmonic synchronization, are symmetric under time reversal. Since $\Lambda_\chi^2=\chi/(1-\vert\dd\alpha\vert_a^2)\ge\chi$ with equality if and only if $\dd\alpha=0$, the static time minimizes the maximal lapse, although other synchronizations can attain the same maximum, and it is the unique minimizer, up to constants, of every functional of the form $\int_Q f(\Lambda_\chi)\,\rho\,\dd\mu_a$ with $f$ strictly increasing and $\rho$ a positive smooth density. Similarly, the synchronizations which minimize $\int g_\chi^{-1}(\dd\tau',\dd\tau')\,\dd\textrm{vol}_{g_\chi}$ are the solutions of the wave equation $\Box_{g_\chi}\tau'=0$ of the physical metric, and for $b=0$ the static time is one of them. For the circular stationary axisymmetric spacetimes, e.g.\ Kerr spacetime, the Boyer-Lindquist time satisfies both $\Box_g t=0$ and $\Box_{\hat g}t=0$ \cite{minguzzisync}, since $g^{\mu t}$ vanishes for $\mu\ne t,\varphi$ and the metric does not depend on $t$ and $\varphi$. Finally, for $b=0$ condition (\ref{ncchi}) reads $Ric_\gamma\ge 0$ and $\Delta_\gamma(1/\chi)\le 0$, which on a compact $Q$ forces $\chi$ to be constant, in agreement with Theorem \ref{curvchi}.
\end{example}

\section{Conclusions}
In this work we advocated, following \cite{minguzzi06d}, that working
on Lagrangian classical mechanics or on the Lorentzian geometry of
spacetimes admitting a covariantly constant null vector is one and
the same thing. Thus it is possible to import methods and ideas from
classical mechanics to Lorentzian geometry and the other way around.
For instance, the direct methods of the calculus of variations on the
base $E$ for the functional $\mathcal{S}$ are related to the limit
curve theorems of Lorentzian geometry on $(M,g)$.

Equation (\ref{cgf1}) shows that the chronological relation can
be characterized through the least action function $S$ and that all
the conformally invariant properties of $(M,g)$ which are stronger
than distinction (with distinction included, Theorem \ref{psf}) must
be expressible in terms of $S$. Actually, some properties such as
strong causality are known to admit a simple characterization in
terms of the distance function \cite{beem96}, and since $d$ and $S$
have been proved to be related by the main result Eq. (\ref{cft}),
it is clear that analogous characterizations of those properties in
terms of $S$ hold.

This work suggests considering two distinct problems: (i) the
determination of the causal properties of the spacetime from the
behavior of $S$ (most notably continuity and finiteness properties);
(ii) the determination of $S$ given ingredients such as the time
dependent space metric, potential or Coriolis 1-form. Clearly, (ii)
is a pure variational problem from classical mechanics.

There is another advantage in using only $S$ to characterize the
causal properties of the spacetime: $S$, being related to the
Lorentzian distance, is independent of the splitting of $E$ into
space and time. In other words, different Newtonian flows on $E$ (see
\cite{minguzzi06d}) may give different splittings and hence
different fields $a_t, b_t, V$, while keeping the same $S$. The
results obtained in this way are thus naturally independent of the
splitting of $E$.

The correspondence extends to the non-smooth solutions of the Hamilton-Jacobi equation. The objects of weak KAM theory are objects of causality theory: the dominated functions are the closed future sets, the Lax-Oleinik evolution is the causal propagation of an achronal graph, the weak KAM solutions are horizons, and the calibrated curves are the projections of their generators. For time-periodic systems the compactification of the time direction turns Ma\~n\'e's critical value into the threshold of chronology, Ma\~n\'e's potential into the chronological relation, the Aubry set into the locus where strong causality fails, and the periodic weak KAM solutions into compact (Cauchy) horizons. For autonomous systems the Maupertuis reduction turns the spacetimes of the Lagrangians $\frac{1}{2}a(v,v)+b(v)$ into the stationary spacetimes over $Q$, Ma\~n\'e's critical value into the threshold of distinction, and the weak KAM solutions into synchronizations of the stationary clocks. The main entries of the dictionary are collected in Table \ref{dic}.

\begin{table}[p]
\vspace*{-1.5cm}
\begin{center}
{\footnotesize \renewcommand{\arraystretch}{1.15}
\begin{tabular}{p{0.45\textwidth}p{0.47\textwidth}}
\hline
{\bf Classical mechanics on $E$} & {\bf Lorentzian geometry of $(M,g)$} \\
\hline
solution of the E.-L.\ equations & lightlike geodesic transverse to $n$ (light lift) \\
action functional $\mathcal{S}$ & extra coordinate $y$ of the light lift \\
least action $S$ & chronological relation; Lorentzian distance, Eq.\ (\ref{cft}) \\
action minimizing motion & achronal lightlike geodesic \\
Legendre transform & musical isomorphism on the light cone \\
Hamilton-Jacobi equation & lightlike condition for $\dd(y-u)$; timelike eikonal equation for $y-u+t/2$ \\
$C^1$ solution of the H.-J.\ equation & $C^1$ lightlike hypersurface transverse to $n$ \\
global solution with complete characteristics & global Rosen coordinates \\
change of Lagrangian by a total derivative & change of section of $\pi: M\to E$ (gauge) \\
lower semi-continuity of $S$ on the diagonal & strong causality (stable causality) \\
lower semi-continuity of $S(e,\cdot)$, $S(\cdot,e)$ at $e$ & distinction \\
lower semi-continuity of $S$ and existence of minimizers (Tonelli's theorem) & causal simplicity \\
coercivity of $\mathcal{S}$ & global hyperbolicity \\
Maupertuis' principle, Jacobi metric & spacelike dimensional reduction, Fermat's principle \\
dominated function & closed future set, achronal boundary \\
Lax-Oleinik evolution of $u_{t_0}$ & boundary of the causal future of the graph of $u_{t_0}$ \\
calibrated curve & generator of the achronal boundary \\
semi-concavity & semi-convexity of horizons \\
weak KAM solution of negative (positive) type & future (past) horizon transverse to $n$ with inextendible generators \\
globally minimizing orbit & lightlike line \\
conjugate pair $(u_-,u_+)$ & $\p I^{+}(\eta)$, $\p I^{-}(\eta)$, with $\eta$ lightlike line (Example \ref{nms}; in general an analogy) \\
\hline
{\em Time-periodic systems, $Q$ compact} & {\em Quotient spacetime $M_c$} \\
\hline
addition of a constant $c$ to $L$ & narrowing of the light cones \\
Ma\~n\'e's critical value $c(L)$ & threshold of chronology; $c>c(L)$ iff globally hyperbolic \\
Ma\~n\'e's potential & chronological relation of $M_c$ \\
projected Aubry set & locus where strong causality fails in $M_{c(L)}$ \\
periodic weak KAM solution & compact horizon; compact Cauchy horizon for $g'$ \\
critical subsolution & compact achronal hypersurface of $M_{c(L)}$ \\
\hline
\end{tabular}
}
\end{center}
\caption{A dictionary between Lagrangian mechanics and the geometry of spacetimes admitting a parallel null vector (continued in Table \ref{dic2}).} \label{dic}
\end{table}

\begin{table}[!tp]
\begin{center}
{\footnotesize \renewcommand{\arraystretch}{1.15}
\begin{tabular}{p{0.45\textwidth}p{0.47\textwidth}}
\hline
{\bf Classical mechanics on $E$} & {\bf Lorentzian geometry of $(M,g)$} \\
\hline
{\em Curvature} &  \\
\hline
Galilei (Newton-Cartan) connection on $E$ & projection of the Levi-Civita connection \\
its Ricci tensor $\mathcal{R}$ & Ricci tensor of $(M,g)$, $Ric=\pi^*\mathcal{R}$ \\
$\mathcal{R}(t,q,v)\ge 0$, Galilean convergence condition & null (equivalently, timelike) convergence condition; dominant energy condition \\
Poisson's equation $\Delta V=4\pi G\mu$ & {\raggedright Einstein's equations with null dust source \mbox{$T=\frac{\mu}{2}\,\dd t\otimes \dd t$}\par} \\
mass density $\mu\ge 0$ & null energy condition for $T$ \\
equation for $\Delta_{a_t} u$ along the characteristics & Raychaudhuri equation \\
Laplacian comparison for weak KAM solutions & non-positivity of the expansion of future horizons \\
smoothness of the periodic weak KAM solutions & smoothness of compact Cauchy horizons \\
frame of calibrated curves, $t$-parametrized & degenerate Killing horizon, parallel generator \\
\hline
{\em Autonomous systems, $L_b=\frac{1}{2}a(v,v)+b(v)$} & {\em Stationary spacetimes $-(\dd\tau-b)^2+\lambda^2a$} \\
\hline
energy level $e=\lambda^2/2$ & Maupertuis reduction; opening of the light cones \\
Ma\~n\'e's potential at the level $e$ & Fermat pre-distance; chronological relation \\
$\sqrt{2c(L_b)}$ & Sagnac ratio $\beta$; $\beta<\lambda$ iff globally hyperbolic \\
weak KAM solution $u$ & eikonal synchronization $\tau-u$ (minimal, constant anisotropy); projected horizon \\
$\mathcal{R}_b\ge 0$ & {\raggedright coincidence of the eikonal and harmonic ($\Box\tau=0$) synchronizations\par} \\
$\Delta_a u_0=\textrm{div}_a b^\sharp$, $\mathcal{R}_b(q,K_0)\ge0$ & {\raggedright harmonic synchronization $\tau-u_0$ is eikonal (Theorem \ref{riemrig})\par} \\
\hline
{\em Autonomous systems with potential, $L_\chi=\frac{1}{2}\chi a(v,v)+b(v)+\frac{1}{2\chi}$} & {\em Stationary spacetimes $g_\chi=\chi[-(\dd\tau-b)^2+a]$} \\
\hline
potential $V_\chi=-1/(2\chi)$ & conformal factor, measured by the redshift \\
$c(L_\chi)=-1/(2\Lambda_*^2)$, $\beta<1$ & minimal maximal lapse; redshift-weighted Sagnac condition \\
weak KAM solution $u$, $\beta<1$ & {\raggedright synchronization $\tau-u$ with constant lapse; synchronous reference system of $g_\chi$\par} \\
$\mathcal{R}_\chi\ge 0$, $\beta<1$ & smooth synchronous frame; $g_\chi$ ultrastatic \\
$\mathcal{R}_\chi\ge 0$ & {\raggedright $(b,V_\chi)$ inertial field of a rigid motion along a Killing field (Theorems \ref{curvchi}, \ref{riemrig})\par} \\
\hline
\end{tabular}
}
\end{center}
\caption{The dictionary, continued from Table \ref{dic}.} \label{dic2}
\end{table}

The dictionary is useful in both directions. On the one hand, results of the calculus of variations and of weak KAM theory give information on the causal structure: for instance, the weak KAM theorem provides the critical spacetimes $M_{c(L)}$, a large family of chronological spacetimes which are not strongly causal and which are foliated by compact horizons. On the other hand, results of mathematical relativity can be read as results on mechanical systems. The most significant example in this work is the theorem on the smoothness of the weak KAM solutions under the null convergence condition, which for these spacetimes coincides with the dominant energy condition (Theorems \ref{smoa} and \ref{smo}), whose periodic case is the classical counterpart of the smoothness of compact Cauchy horizons, and whose proof passes through the area theorem. Under the same condition the motion is a geodesic flow of non-negative Ricci curvature observed from a moving frame, so that, for instance, Grove's theorem on isometry-invariant geodesics gives periodic motions (Corollary \ref{dyn}), while the compact horizons of the periodic case are degenerate Killing horizons, from which we obtain compact degenerate Killing Cauchy horizons in spacetimes satisfying the null convergence condition (Example \ref{kch}). For the Newtonian $N$-body problem in dimension $d\ge 3$, the null convergence condition and the homogeneity of the potential exclude the motions defined on $\mathbb{R}$ which minimize the fixed-time action (Prop.\ \ref{nline}). Another example is the synchronization theorem (Theorem \ref{sync}): the projections of the horizons give the eikonal synchronizations of the clocks of a stationary spacetime over a compact space, which minimize the maximal anisotropy of the one-way speed of light, are time functions whenever the spacetime is distinguishing, and under the null convergence condition are unique up to constants and coincide with the harmonic synchronization. If $\beta<1$, their slices are orthogonal, in a weak sense which allows for caustics, to observers freely falling in the normalized metric $\hat g$, and their relevance for the synchronization of clocks and of computer networks is discussed in the companion paper \cite{minguzzisync}. By representing the stationary spacetime as a Maupertuis spacetime with a potential, which amounts to the choice of a conformal factor measurable by the clocks through the redshift, one obtains, for $\beta<1$, synchronizations orthogonal to observers freely falling in the physical metric, governed by a redshift-weighted Sagnac condition, and the Galilean convergence condition of the corresponding classical spacetime, namely the null convergence condition of the corresponding Eisenhart spacetime, emerges as a new curvature condition on stationary spacetimes, distinct from the standard energy conditions, under which these synchronizations are smooth (Theorems \ref{synchi} and \ref{curvchi}). On a compact space it characterizes the fields of inertial forces of rigid motions along Killing vector fields, a rigidity phenomenon of Riemannian nature for which we have also given an independent elliptic proof (Theorem \ref{riemrig}). In this proof the regularity is obtained by starting from the solution of a Poisson equation, which for $L_b$ is the harmonic synchronization, and the sign of the curvature is needed only along its anisotropy field: the harmonic synchronization is eikonal as soon as this sign condition holds. In the Eisenhart spacetime the Killing vector field is, up to sign, the action of the horizon-preserving time translation on the space of generators of the horizons, which are then totally geodesic, and their generator is a second parallel null vector field. The extension of this rigidity to complete manifolds, and the possible restriction of the hypothesis to the critical energy shell, deserve further study.

Several questions remain open. It would be interesting to extend the results of Section \ref{s:per} to non-compact configuration spaces, where weak KAM theory is available under suitable assumptions \cite{fathi07}, and to understand the causal counterpart of the Mather measures and of the asymptotic behavior of the Lax-Oleinik evolution, which in the time-periodic case need not converge \cite{fathi00}. One could also investigate whether weaker curvature conditions, such as a lower bound on the function $\mathcal{R}$, lead to quantitative regularity results for the weak KAM solutions, and identify the curvature conditions, weaker than (\ref{ncb}), under which the harmonic synchronization of a distinguishing stationary spacetime is a time function. Finally, for the $N$-body problem it remains to be understood whether the absence of lightlike lines extends to the planar case, and whether, for the weak-force pair potentials, the lines of the Jacobi metrics of non-negative energy correspond to motions defined on $\mathbb{R}$.

\section*{Acknowledgments}
As explained in the introduction, the whole program was laid out about fifteen years ago and the first half of it was developed at the time. For the writing of the second part we had the assistance of AI tools from Anthropic (Claude Opus~5.5) and OpenAI (GPT-6~Astra). This work was supported in part by a grant of access to OpenAI models through the ChatGPT for Academic Researchers program. This work has also been partially supported by GNFM of INDAM and by FQXi.
%
%
%

\phantomsection
\addcontentsline{toc}{section}{Appendix A. Triangle and reverse triangle inequalities: an abstract framework}
\section*{Appendix A. Triangle and reverse triangle inequalities: an abstract framework}


With Eq. (\ref{cft}) we have established a relation between a
two-point function satisfying the triangle inequality and a
two-point function satisfying the reverse triangle inequality. In
this section we wish to study this type of dependence in an abstract
way showing how the triangle inequality property gets related with
the reverse triangle inequality property.

A preordered space $(X,\preceq)$ is a set $X$ endowed with a
reflexive and transitive relation (that is $x\preceq x$; $x\preceq
y$ and $y\preceq z$ implies $x\preceq z$). We define $x \prec y$ if
$x \preceq y$ and $y\npreceq x$, and we define $x\sim y$ if $x
\preceq y$ and $y\preceq x$. The relation $\sim$ is the equivalence
relation part of $\preceq$ while $\prec$ is the irreflexive part of
$\preceq$ (i.e. $x\nprec x$). The reflexive closure $\prec_{rc}$ of
$\prec$ is given by `` $x\prec y$ or $x=y$'' and is the smallest
reflexive relation containing $\prec $.

 The preorder is an order (or partial order) if
$\preceq$ is antisymmetric, namely $x\sim y\Rightarrow x=y$. An
isotone function is a function $f: X \to \mathbb{R}$ such that
$x\preceq y \Rightarrow f(x)\le f(y)$. A utility function is an
isotone function such that $x\prec y \Rightarrow f(x)< f(y)$. If
$\preceq$ is total ($x\preceq y$ or $y \preceq x$) then $f$ is an
utility if and only if ``$x\preceq y \Leftrightarrow f(x)\le f(y)$''. (The
existence of a continuous utility for total preorders follows under
mild conditions on the topology of $X$ from the closure of $\preceq$
regarded as a subset of $X\times X$, this is the Eilenberg-Debreu
theorem \cite{eilenberg41,debreu54,bridges95}.) The utility is not
unique because any function $f(t(\cdot))$ with $f:\mathbb{R} \to
\mathbb{R}$ increasing would give another utility.

The next proposition is an improvement over \cite[Prop.
1,2]{kunzi08b}.

\begin{proposition} \label{vxc}
On a set $X$ let $s:X\times X \to [-\infty,+\infty]$ be a function
with the properties
\begin{itemize}
\item[($\ast$)] $s(x,x)<+\infty$
\item[($\star$)] $s(w,x)<+\infty$ and $s(x,y)=-\infty$ and $s(y,z)<+\infty$
$\Rightarrow s(w,z)=-\infty$.
\end{itemize}
On the cartesian product $X\times \mathbb{R}$ define the relation
\begin{equation}
(x,a) \le (y,b) \quad  \textrm{ if } \quad b-a\ge s(x,y)
\end{equation}
and consider the properties of $\, \, s$
\begin{itemize}
\item[(a)] for all  $x,y,z$, $s(x,z)\le s(x,y)+s(y,z)$,
\item[(b)] for all  $x$, $s(x,x)=0$ or $s(x,x)=-\infty$,
\item[(c)] for all $x$, $s(x,x)\ne -\infty$ and for all $x,y$, $s(x,y)+s(y,x)=0 \Rightarrow x=y$,
\end{itemize}
where in (a) and (c) it is understood that the  convention
$(+\infty)+(-\infty)=+\infty$ applies.

The relation $\le$ is transitive if and only if (a); it is a preorder if and only if (a)
and (b); it is an order if and only if (a), (b) and (c).
\end{proposition}

\begin{remark}
Under (a) condition (b) is equivalent to $s(x,x)\le 0$ because if
$s(x,x)\le 0$ then from (a), $s(x,x)\le s(x,x)+s(x,x)$, thus if
$s(x,x)$ is  finite, $0\le s(x,x)$ which implies $s(x,x)=0$.
\end{remark}

\begin{proof}
Assume that $\le$ is  transitive.  Let $x,y,z\in X$, if
$s(x,y)=+\infty$ or $s(y,z)=+\infty$ then the triangle inequality is
satisfied. Analogously, by ($\ast$) and ($\star$) it is easy to
check that if $s(x,y)=-\infty$ or $s(y,z)=-\infty$ the triangle
inequality is satisfied. In the remaining case, $\vert s(x,y)\vert
<+\infty$ and $\vert s(y,z)\vert<+\infty$ and we can find constants
$a,b,c$ such that $b-a= s(x,y)$ and $c-b= s(y,z)$ thus $(x,a)\le
(y,b)\le (c,z)$ and by transitivity $(x,a)\le (z,c)$ which reads
$c-a\ge s(y,z)$, thus $s(x,y)+s(y,z)=c-a\ge s(y,z)$ and (a) is
proved.

For the converse assume (a). If $(x,a)\le (y,b)\le (z,c)$ then
$b-a\ge s(x,y)$ and $c-b \ge s(y,z)$, in particular $s(x,z)$ and
$s(y,z)$ are different from $+\infty$. If one of them is $-\infty$
then by ($\ast$) and ($\star$), $s(x,z)=-\infty$, which implies
$c-a\ge s(x,z)$ and $(x,a)\le (z,c)$. Thus let us assume that
$s(x,z)$ and $s(y,z)$ are finite. Summing the two equations above
and using (a), $c-a\ge s(x,y)+s(y,z)\ge s(x,z)$, thus $(x,a)\le
(z,c)$ and hence $\le$ is transitive.

Assume that $\le$ is a preorder, since we have already proved (a),
we have only to prove (b). From reflexivity $(x,a)\le (x,a)$ and
thus $s(x,x)\le 0$. In particular $s(x,x)<+\infty$ and using the
triangle inequality of (a), $s(x,x)\le s(x,x)+s(x,x)$ thus either
$s(x,x)=-\infty$ or $s(x,x)$ is finite in which case the inequality
reads $s(x,x)\ge 0$ and hence $s(x,x)=0$.

For the converse, assume (a) and (b). We already know that $\le$ is
transitive and we have only to prove reflexivity. We have $(x,a)\le
(x,a)$ because $s(x,x)\le 0$ by (b).

Assume $\le$ is an order, we already know that (a) and (b) hold. If,
by contradiction, there is some $x$ such that $s(x,x)=-\infty$ then
taken $b\ne a$ we have $b-a\ge s(x,x)$ and $a-b\ge s(x,x)$ which
implies $(x,a)\le (x,b)\le (x,a)$ thus $(x,a)=(x,b)$ that is, $a=b$.
The contradiction proves that for every $x$, $s(x,x)=0$.

Let $x,y$ be such that $s(x,y)+s(y,x)=0$, then none of the two terms
is $+\infty$ and we can assume without loss of generality that
$s(x,y)\ge 0$, $s(y,x)\le 0$ and that both are finite (otherwise
exchange $x$ with $y$). We can find constants $a,b$ such that
$b-a=s(x,y)$ which implies $(x,a)\le (y,b)$, and
$a-b=-s(x,y)=s(y,x)$ which gives $ (y,b)\le (x,a)$ thus from the
antisymmetry of $\le$, $x=y$.

Conversely, assume (a), (b) and (c) and let $(x,a)\le (y,b)\le
(x,a)$ then $b-a\ge s(x,y)$ and $a-b\ge s(y,x)$ from which we get
that both $s(x,y)$ and $s(y,x)$ are different  from $+\infty$ and
that $s(x,y)+s(y,x)\le 0$.


From the triangle inequality $s(x,x)\le s(x,y)+s(y,x)$ and by (c)
$s(x,x)=0$ thus  $s(x,y)+s(y,x)=0$ which gives by (c) $x=y$.
Moreover as $b-a\ge s(x,y)=s(x,x)=0$ and analogously $a-b\ge
s(y,x)=s(x,x)=0$ we get $a=b$ and hence the antisymmetry of $\le$.

\end{proof}

\begin{remark}
Non-negative functions $s$ which satisfy (a) and (b) are  called
quasi-pseudometrics, while non-negative functions with satisfy (a),
(b) and (c') where (c') is a stronger version of (c) in which
$s(x,y)=0$ implies $x=y$ are called quasimetrics. The conditions
(a), (b) and (c) together with $s\ge 0$ imply that on the quotient
space given by $x\approx y$ if and only if $s(x,y)=s(y,x)=0$, $s$ is a
quasimetric. The question whether $s$ is non-negative or could be
replaced by a parent non-negative function is therefore important
and will be considered in what follows.
\end{remark}

The fact that  given two functions $d$ and $S$ related as in
equation (\ref{cft}), $d$ satisfies the reverse triangle inequality
if and only if $S$ satisfies the triangle inequality, can  be proved in
generality, with no mention to their actual physical meaning, and
represents an interesting clue given by Eq. (\ref{cft}). We need,
however to enrich our abstract framework with a total preorder and
an utility on $X$, and to demand the function $s$ to be related in a
suitable way with that preorder.

\begin{proposition}
Let $(X,\preceq)$ be a total preordered space, let $t: X\to
\mathbb{R}$ be an utility function, and let $s:X\times X \to
[-\infty,+\infty]$ be a function as in Prop. \ref{vxc} which
satisfies (b) and let $\le$ be the relation on $X\times \mathbb{R}$
there defined. Assume additionally the property of $s$
\begin{itemize}
\item[($\,\diamond$)] if $x\npreceq y$ then $s(x,y)=+\infty$, and if $x\nprec_{rc} y$ then $s(x,y)=\pm \infty$.
\end{itemize}
Define $d: (X\times \mathbb{R})^2 \to [0,+\infty]$ by
\begin{equation}
d((x,a),(y,b))=\sqrt{2[b- a-s(x,y)](t(y)-t(x))}
\end{equation}
if $(x,a)\le (y,b)$ and $0$ otherwise (here the convention
$(+\infty)0=+\infty$ applies). Given $x_1,x_2,x_3\in X$, the reverse
triangle inequality (shortened notation)
\begin{equation} \label{pdk}
d((x_1,a_1),(x_2,a_2))+d((x_2,a_2),(x_3,a_3)) \le
d((x_1,a_1),(x_3,a_3))
\end{equation}
holds for every triple $(x_1,a_1)\le (x_2,a_2)\le (x_3,a_3)$ if and
only if the triangle inequality (a) of Prop. \ref{vxc} holds for
$x_1,x_2,x_3$. Assuming that this is the case then, for a given
triple, the equality in Eq. (\ref{pdk}) holds if and only if the equality holds
in (a) of Prop. \ref{vxc} and either (i): there is a constant $k\ge
0$ such that
\begin{equation}
a_i-a_j-s(x_i,x_j)=k(t(x_j)-t(x_i)), \qquad i,j=1,2,3 \quad i>j.
\end{equation}
or (ii): $s(x_1,x_2)$ or $s(x_2,x_3)$ (and hence $s(x_1,x_3)$) is
$-\infty$.
\end{proposition}

\begin{remark}
By property ($\diamond$) if $(x,a)\le (y,b)$ then $s(x,y)<+\infty$
and hence $x\preceq y$ and $t(x)\le t(y)$.
\end{remark}

\begin{proof}
Let $(x_1,a_1)\le (x_2,a_2)\le (x_3,a_3)$ so that necessarily
$s(x_1,x_2)$ and $s(x_2,x_3)$ are  different  from $+\infty$.
Furthermore, assume that Eq. (\ref{pdk}) holds. If
$s(x_1,x_2)=-\infty$ or $s(x_2,x_3)=-\infty$ then by properties
($\ast$) and ($\star$) of $s$, $s(x_1,x_3)=-\infty$ and  the
triangle inequality (a) of Prop. \ref{vxc} is satisfied (both
inequalities become equalities and (ii) holds). Therefore, let us
assume that $s(x_1,x_2)$ and $s(x_2,x_3)$ are finite.  By
($\diamond$) $x_1\prec_{rc}x_2\prec_{rc} x_3$ and hence, as $t$ is
an utility, $t(x_1)<t(x_3)$ unless  $x_1=x_2=x_3$ in which case  the
triangle inequality is satisfied because by (b) and the finiteness
of $s(x_1,x_2)$,  $s(x_1,x_1)= 0$.

We can find constants $a_1$,$a_2$, $a_3$, such that
$a_2-a_1=s(x_1,x_2)$ and $a_3-a_2=s(x_2,x_3)$, so that
$a_3-a_1=s(x_1,x_2)+s(x_2,x_3)$, thus replacing in Eq. (\ref{pdk}),
$[s(x_1,x_2)+s(x_2,x_3)-s(x_1,x_3)](t(x_3)-t(x_1))\ge 0$ and
recalling $t(x_1)<t(x_3)$ we find that the triangle inequality is
satisfied.

For the converse, assume that the triangle inequality $s(x_1,x_3)\le
s(x_1,x_2)+s(x_2,x_3)$ holds and let  $(x_1,a_1)\le (x_2,a_2)\le
(x_3,a_3)$ so that $s(x_1,x_2)$ and $s(x_2,x_3)$ are both different
from $+\infty$. If one of them is $-\infty$ then by ($\ast$) and
($\star$) $s(x_1,x_3)=-\infty$ from which it follows
$d_{13}=+\infty$ (obvious shortened notation) and hence Eq.
(\ref{pdk}) holds. We are left with the case in which $s(x_1,x_2)$,
$s(x_2,x_3)$ and $s(x_1,x_3)$ are finite. By  ($\diamond$)
$x_1\prec_{rc}x_2\prec_{rc} x_3$ and hence, as $t$ is an utility,
$t(x_1)<t(x_2)<t(x_3)$ unless  $x_1=x_2$ or $x_2=x_3$. However, in
the case $x_1=x_2$ as $t(x_1)=t(x_2)$, we have $d_{12}=0$ so that
Eq. (\ref{pdk}) reads $d_{23}\le d_{13}$, that is
\[
a_3-a_2-s(x_1,x_3)\le a_3-a_1-s(x_1,x_3),
\]
 or equivalently, $a_1\le a_2$. This inequality holds because $a_2-a_1\ge
 s(x_1,x_2)$ and $s(x_1,x_2)=s(x_1,x_1)=0$.
 Analogously,  if $x_2=x_3$ Eq. (\ref{pdk})
holds. Let us consider the remaining case $t(x_1)<t(x_2)<t(x_3)$. It
is not hard to check the equality (obvious shortened notation)
\[
d_{13}^2-(d_{12}+d_{23})^2=\left( \sqrt{\frac{t_2-t_1}{t_3-t_2}}\,
d_{23}-\sqrt{\frac{t_3-t_2}{t_2-t_1}}\,
d_{12}\right)^2+2(t_3-t_1)(s_{12}+s_{23}-s_{13}),
\]
from which we get $d_{13}\ge d_{12}+d_{23}$, and that the equality
holds if and only if it holds in the triangle inequality and
additionally the square on the right-hand side vanishes, namely
\[
\frac{d_{12}}{t_2-t_1}=\frac{d_{23}}{t_3-t_2}.
\]
Denoting this ratio with $k$, from the expression for $d_{12}$ we
get thesis.
\end{proof}

\begin{proposition}
Let $(X,\preceq)$ be a total preordered space, let $t: X\to
\mathbb{R}$ be an utility function, and let $s:X\times X \to
[-\infty,+\infty]$ be a function as in Prop. \ref{vxc} which
satisfies (a),(b), (c) and ($\diamond$). For every $\alpha\in
\mathbb{R}$ the function $s_{\alpha}(x,y)=s(x,y)-\alpha (t(y)-t(x))$
satisfies the same properties and the relations on $X$,
$\preceq_{\alpha}:=\{(x,y): s_{\alpha}(x,y)\le 0\}$, for $\alpha \in
\mathbb{R}$ and $\preceq_{+\infty}:=\preceq$ are orders such that
$\alpha <\beta$ implies that $\preceq_{\beta}$ extends
$\preceq_{\alpha}$.
\end{proposition}

\begin{proof}
The first statement follows easily from the fact that $s(x,y)$
equals $-\infty$ or  $+\infty$ if and only if the same value is taken by
$s_{\alpha}(x,y)$.

The relations $\preceq_{\alpha}$, $\alpha \in \mathbb{R}$, and
$\preceq_{+\infty}$ are orders by Prop. \ref{vxc}.

%

Assume that $s_{\alpha}(x,y)\le 0$ then $s(x,y)\le -\alpha
(t(x)-t(y))<+\infty$ and from ($\diamond$) we get $x \preceq y$,
thus $\preceq$ extends $\preceq_{\alpha}$.

Let $\alpha <\beta<+\infty$, and $s_{\alpha}(x,y)\le 0$, since
$s_{\alpha}(x,y)=s_{\beta}(x,y)+(\beta-\alpha) (t(y)-t(x))$ we get
$s_{\beta}(x,y)\le 0$ and hence $\preceq_{\beta}$ extends
$\preceq_{\alpha}$.

\end{proof}

\begin{definition}
Let $B$ be the set of $\beta\in (-\infty,+\infty]$ such  that
$\preceq_{\beta}=\preceq$, then
\begin{align*}
u&=\inf_{\beta\in B} \beta,
\end{align*}
is the {\em critical upper slope}.

Let $C$ be the set of $\gamma\in (-\infty,+\infty]$ such  that $x
\preceq_{\gamma} y$ implies $x \sim y$, then
\begin{align*}
l&=\sup_{\gamma\in C} \gamma, \qquad \textrm{ if } C\ne \emptyset
\end{align*}
and $l=-\infty$ if $C=\emptyset$, is the {\em critical lower slope}.
\end{definition}
Note that $B\ne \emptyset$ because $\preceq_{+\infty}=\preceq$ and
$B$ is increasing,  that is, if $\beta_1<\beta_2$ with $\beta_1\in
B$ then $\beta_2\in B$.

\begin{theorem}
We have $l \le u$, and if $\preceq=\,\sim$ then $l=u$.
\end{theorem}

\begin{theorem}
If there is at least a pair $x\prec y$ such that $s(x,y)>-\infty$
then  $c>-\infty$ and $\preceq_c=\preceq$ that is, $c \in B$.
\end{theorem}

\begin{proof}
If there are  $x, y \in X$,  $x \prec y$ such that $s(x,y)\ne
-\infty$ since $\preceq_{\beta}=\preceq$ for $\beta \in
B\backslash\{+\infty\}$ it must be $x \preceq_{\beta} y$ and hence
$s(x,y)-\beta(t(y)-t(x))\le 0$, and since $t$ is an utility
$\beta\ge s(x,y)/(t(y)-t(x))>-\infty$ which implies $c>-\infty$.

 We already know that $x \preceq_c y$ implies $x
\preceq y$. Assume $x \preceq y$ then for every $\beta>c$
\[
s(x,y)-\beta(t(y)-t(x))\le 0.
\]
Taking the limit $\beta \to c$ we get $s(x,y)-c(t(y)-t(x))\le 0$ and
hence $x\preceq_c y$ (note that the continuity of $s$ is not used).

\end{proof}

\phantomsection
\addcontentsline{toc}{section}{Appendix B. The Ricci tensor and the Raychaudhuri identity}
\section*{Appendix B. The Ricci tensor and the Raychaudhuri identity}

In this appendix we prove Proposition \ref{ric} and the identity (\ref{ray2}) of Remark \ref{rayrem}, and we complete the proof of Proposition \ref{ray}. We use the notation of Section \ref{s:ncc}: $D$ is the Levi-Civita connection of $(Q,a_t)$ at fixed $t$, the indices $a,b,c,d$ refer to local coordinates $q^a$ on $Q$ and are raised and lowered with $a_t$, a dot denotes the partial derivative with respect to $t$ at fixed $q$, $F_{ab}=\p_a b_b-\p_b b_a$, and
\[
\textrm{tr}\,\dot a:=a^{ab}\dot a_{ab}=2\,\p_t \ln\sqrt{\det a_t}\, .
\]

\subsection*{B.1 The connection induced on the classical spacetime}

In the coordinates $(t,q^a,y)$ the non-vanishing components of the metric (\ref{eis}) and of its inverse are
\begin{gather*}
g_{ab}=a_{ab},\quad g_{ta}=b_a,\quad g_{tt}=-2V, \quad g_{ty}=-1, \\
g^{ab}=a^{ab}, \quad g^{ay}=b^a, \quad g^{yy}=2V+b_ab^a,\quad g^{ty}=-1 .
\end{gather*}
Let $\Gamma_{\alpha\beta\gamma}=\frac{1}{2}(\p_\beta g_{\alpha\gamma}+\p_\gamma g_{\alpha\beta}-\p_\alpha g_{\beta\gamma})$ be the Christoffel symbols of the first kind, so that $\Gamma^\mu_{\beta\gamma}=g^{\mu\alpha}\Gamma_{\alpha\beta\gamma}$. Since the components $g_{y\alpha}$ are constant and no component depends on $y$, we have $\Gamma_{y\beta\gamma}=0$ and $\Gamma_{\alpha y \gamma}=0$. As $g^{t\alpha}=-\delta^\alpha_y$ and $g^{a\alpha}\Gamma_{\alpha\beta\gamma}=a^{ac}\Gamma_{c\beta\gamma}+b^a\Gamma_{y\beta\gamma}$, we obtain
\begin{equation} \label{chr1}
\Gamma^t_{\beta\gamma}=0, \qquad \Gamma^\mu_{y\gamma}=0, \qquad \Gamma^a_{\beta\gamma}=a^{ac}\,\Gamma_{c\beta\gamma}\, ,
\end{equation}
and the computation of the last expression gives
\begin{equation} \label{chr2}
\Gamma^a_{bc}=\Gamma(a_t)^a_{bc}, \qquad \Gamma^a_{bt}=\Lambda^a{}_b:=\frac{1}{2}\,a^{ac}(\dot a_{cb}-F_{cb}), \qquad \Gamma^a_{tt}=\mathcal{G}^a:=a^{ab}(\dot b_b+\p_b V),
\end{equation}
where $\Gamma(a_t)^a_{bc}$ are the Christoffel symbols of $(Q,a_t)$. The first two identities in (\ref{chr1}) express the fact that $\dd t$ and $n=\p_y$ are parallel. Observe that, for each $t$, $\Lambda^a{}_b$ and $\mathcal{G}^a$ are tensor fields on $Q$, and that they do not depend on $y$.

Let us denote with capital indices $A,B,C,\ldots$ the coordinates $t,q^1,\ldots,q^d$ of $E$. The symbols $\Gamma^A_{BC}$, given by (\ref{chr1})-(\ref{chr2}), are the Christoffel symbols of a torsion-free connection $\nabla^E$ on $E$, which is the projection of the Levi-Civita connection of $g$: if $X$ and $Y$ are vector fields on $M$ invariant under the flow of $n$, then $\pi_*(\nabla_X Y)=\nabla^E_{\pi_* X} \pi_* Y$. The connection $\nabla^E$ preserves $\dd t$ and the degenerate contravariant metric $a^{ab}\p_a\otimes \p_b$, namely it is a Galilei connection, and the construction is that of the Bargmann structures of Duval, Burdet, K\"unzle and Perrin \cite{duval85} (see also \cite[Appendix]{minguzzi06d}). The geodesics of $\nabla^E$, parametrized with $t$, satisfy
\[
\ddot q^a+\Gamma(a_t)^a_{bc}\,\dot q^b\dot q^c+2\Lambda^a{}_b\,\dot q^b+\mathcal{G}^a=0 ,
\]
which is the Euler-Lagrange equation (\ref{ele}). In particular, for a particle in a Newtonian gravitational field, $a_t=\delta$ and $b_t=0$, the only non-vanishing symbols are $\Gamma^a_{tt}=\p_a V$ and $\nabla^E$ is the connection of the geometrized formulation of Newtonian gravity, also known as Newton-Cartan theory \cite{cartan23,trautman65,kunzle72,malament12}.

\subsection*{B.2 The Ricci tensor}

Let $R_{\beta\gamma}=\p_\mu\Gamma^\mu_{\beta\gamma}-\p_\gamma\Gamma^\mu_{\mu\beta}+\Gamma^\mu_{\mu\delta}\Gamma^\delta_{\beta\gamma}-\Gamma^\mu_{\gamma\delta}\Gamma^\delta_{\mu\beta}$ be the Ricci tensor of $g$. By (\ref{chr1}) and since nothing depends on $y$, all the terms in which $\mu$ or $\delta$ take the value $t$ or $y$ vanish, and so do all the components with $\beta=y$ or $\gamma=y$. Thus $R_{y\gamma}=0$ and
\begin{equation} \label{rred}
R_{BC}=\p_a\Gamma^a_{BC}-\p_C\Gamma^a_{aB}+\Gamma^a_{ad}\Gamma^d_{BC}-\Gamma^a_{Cd}\Gamma^d_{aB},
\end{equation}
namely $Ric=\pi^*Ric^E$ where $Ric^E$ is the Ricci tensor of the Galilei connection $\nabla^E$. Therefore, the tensor $\mathcal{R}$ of Proposition \ref{ric} is the Ricci tensor of $\nabla^E$. Let us compute its components.

For $B=b$, $C=c$, all the symbols in (\ref{rred}) are those of $(Q,a_t)$, thus $R_{bc}=(Ric_{a_t})_{bc}$, which is Eq.\ (\ref{ri1}).

For $B=b$, $C=t$, using $\Gamma^a_{ab}=\p_b \ln\sqrt{\det a_t}$ we obtain
\begin{align*}
R_{bt}&=\p_a\Lambda^a{}_b+\Gamma^a_{ad}\Lambda^d{}_b-\Gamma^d_{ab}\Lambda^a{}_d-\p_t\Gamma^a_{ab}=D_a\Lambda^a{}_b-\frac{1}{2}\,\p_b\, \textrm{tr}\, \dot a\\
&=\frac{1}{2}\,\big[D^a\dot a_{ab}-\p_b\,\textrm{tr}\, \dot a-D^aF_{ab}\big],
\end{align*}
which is Eq.\ (\ref{ri2}).

For $B=C=t$,
\[
R_{tt}=\p_a\mathcal{G}^a+\Gamma^a_{ad}\mathcal{G}^d-\p_t\Lambda^a{}_a-\Lambda^a{}_d\Lambda^d{}_a=D_a\mathcal{G}^a-\p_t \Lambda^a{}_a-\Lambda^a{}_d\Lambda^d{}_a .
\]
Here $D_a\mathcal{G}^a=\Delta_{a_t}V+\textrm{div}_{a_t}\dot b$. Since $F_{ab}$ is antisymmetric, $\Lambda^a{}_a=\frac{1}{2}\textrm{tr}\,\dot a$, and since $\p_t a^{ab}=-\dot a^{ab}$, where $\dot a^{ab}:=a^{ac}a^{bd}\dot a_{cd}$, we have $\p_t\,\textrm{tr}\,\dot a= a^{ab}\ddot a_{ab}-\dot a_{ab}\dot a^{ab}$. Moreover, as the products between the symmetric tensor $\dot a$ and the antisymmetric tensor $F$ have vanishing trace, $\Lambda^a{}_d\Lambda^d{}_a=\frac{1}{4}(\dot a_{ab}\dot a^{ab}-F_{ab}F^{ab})$. Thus
\[
R_{tt}=\Delta_{a_t}V+\textrm{div}_{a_t}\dot b-\frac{1}{2}\,a^{ab}\ddot a_{ab}+\frac{1}{2}\,\dot a_{ab}\dot a^{ab}-\frac{1}{4}\,\dot a_{ab}\dot a^{ab}+\frac{1}{4}\,F_{ab}F^{ab},
\]
which is Eq.\ (\ref{ri3}). This concludes the proof of Proposition \ref{ric}.

\subsection*{B.3 The Raychaudhuri identity}

Let $u$ be a $C^3$ function defined on an open subset of $E$ and let
\begin{gather*}
w_a:=\p_a u-b_a, \qquad S_{ab}:=\frac{1}{2}\,(D_aw_b+D_bw_a), \\
h:=\p_t u+H(t,q,\p_q u)=\p_t u+\frac{1}{2}\,w_aw^a+V,
\end{gather*}
so that $v_t^a=w^a$ in Eq.\ (\ref{vB}). Since $\p_a\p_b u$ is symmetric, $D_aw_b-D_bw_a=-F_{ab}$, thus $D_aw_b=S_{ab}-\frac{1}{2}F_{ab}$, and Eqs.\ (\ref{vB}) and (\ref{the}) read
\[
B_{ab}=\frac{1}{2}\,\dot a_{ab}+S_{ab}, \qquad \theta=\frac{1}{2}\,\textrm{tr}\,\dot a+D_aw^a .
\]
Using $\p_a\p_t u=\dot w_a+\dot b_a$ and $w^bD_aw_b=w^bD_bw_a+F_{ba}w^b$ we obtain
\begin{equation} \label{dh}
D_a h=\dot w_a+w^bD_bw_a+F_{ba}w^b+\dot b_a+\p_a V.
\end{equation}
(For a solution of the Hamilton-Jacobi equation $h=0$, and Eq.\ (\ref{dh}) states that the characteristics solve the Euler-Lagrange equation.) Let us calculate the divergence of each term on the right-hand side of (\ref{dh}).
\begin{itemize}
\item[(i)] From $\p_t a^{ab}=-\dot a^{ab}$ and the variation formula $\p_t\Gamma(a_t)^c_{ab}=\frac{1}{2}a^{cd}(D_a\dot a_{bd}+D_b\dot a_{ad}-D_d\dot a_{ab})$ we get
\[
D^a\dot w_a=\p_t(D^aw_a)+\dot a^{ab}S_{ab}+w^c\big(D^a\dot a_{ac}-\frac{1}{2}\,\p_c\,\textrm{tr}\,\dot a\big).
\]
\item[(ii)] By the Ricci identity $D^aD_bw_a=D_bD^aw_a+(Ric_{a_t})_{bc}w^c$, thus
\begin{align*}
D^a(w^bD_bw_a)&=(D^aw^b)(D_bw_a)+w^bD_b(D^aw_a)+Ric_{a_t}(w,w)\\
&=S_{ab}S^{ab}-\frac{1}{4}\,F_{ab}F^{ab}+w^bD_b(D^aw_a)+Ric_{a_t}(w,w).
\end{align*}
\item[(iii)] $D^a(F_{ba}w^b)=-w^cD^aF_{ac}+F_{ba}(S^{ab}-\frac{1}{2}F^{ab})=-w^cD^aF_{ac}+\frac{1}{2}\,F_{ab}F^{ab}$.
\item[(iv)] $D^a(\dot b_a+\p_aV)=\textrm{div}_{a_t}\dot b+\Delta_{a_t}V$.
\end{itemize}
Summing these expressions and using Eq.\ (\ref{ri2}) in the form $w^c(D^a\dot a_{ac}-\frac{1}{2}\p_c\,\textrm{tr}\,\dot a-D^aF_{ac})=2\rho_t(w)+\frac{1}{2}w^c\p_c\,\textrm{tr}\,\dot a$, we obtain
\begin{align}
\Delta_{a_t}h=&\ (\p_t+w^bD_b)(D^aw_a)+\frac{1}{2}\,w^c\p_c\,\textrm{tr}\,\dot a+\dot a^{ab}S_{ab}+S_{ab}S^{ab} \nonumber\\
&+Ric_{a_t}(w,w)+2\rho_t(w)+\Delta_{a_t}V+\textrm{div}_{a_t}\dot b+\frac{1}{4}\,F_{ab}F^{ab}. \label{dh2}
\end{align}
On the other hand,
\begin{align*}
(\p_t+w^bD_b)\,\theta&=\frac{1}{2}\,\p_t\,\textrm{tr}\,\dot a+\frac{1}{2}\,w^c\p_c\,\textrm{tr}\,\dot a+(\p_t+w^bD_b)(D^aw_a), \\
\textrm{tr}\,B^2&=B_{ab}B^{ab}=\frac{1}{4}\,\dot a_{ab}\dot a^{ab}+\dot a^{ab}S_{ab}+S_{ab}S^{ab},
\end{align*}
and, by Eqs.\ (\ref{rv}) and (\ref{ri3}), $\mathcal{R}(t,q,w)=Ric_{a_t}(w,w)+2\rho_t(w)+\varrho$. Since $\frac{1}{2}\p_t\,\textrm{tr}\,\dot a+\frac{1}{4}\dot a_{ab}\dot a^{ab}-\frac{1}{2}a^{ab}\ddot a_{ab}+\frac{1}{4}\dot a_{ab}\dot a^{ab}=0$, the sum $(\p_t+w^bD_b)\theta+\textrm{tr}\,B^2+\mathcal{R}(t,q,w)$ coincides with the right-hand side of (\ref{dh2}), which proves the identity (\ref{ray2}).

Let us complete the proof of Proposition \ref{ray}. If $u$ is a $C^3$ solution of the Hamilton-Jacobi equation, then $h=0$ and the identity (\ref{ray2}) gives the equality in (\ref{ray1}) along the integral curves of $\p_t+v_t$, while the inequality follows from $\textrm{tr}\,B^2\ge (\textrm{tr}\,B)^2/d$, which holds because $B$ is self-adjoint with respect to $a_t$. If $u$ is just $C^2$, along a characteristic $t\mapsto (t,q(t))$ the Hessian of $u$ is $C^1$ and satisfies the matrix Riccati equation obtained from the linearization of the Hamiltonian flow, since $\p_q u(t,q(t))$ is a solution of Hamilton's equations and the graph of $\p_qu(t,\cdot)$ is transported by the Hamiltonian flow. For $C^3$ solutions this Riccati equation coincides with the equation obtained by differentiating twice the Hamilton-Jacobi equation, and Eq.\ (\ref{ray1}), which involves only the derivatives of $u$ up to the second order, is an algebraic consequence of it. Thus Eq.\ (\ref{ray1}) holds for $C^2$ solutions as well.

\end{document}